\documentclass[12pt,a4paper,english,intoc,bibliography=totoc,index=totoc,BCOR10mm,captions=tableheading,titlepage,fleqn]{scrbook}
\usepackage{lmodern}

\usepackage[T1]{fontenc}
\usepackage[latin9]{inputenc}
\usepackage{fancyhdr}
\usepackage{babel}
\usepackage{array}
\usepackage{url}
\usepackage{float}
\usepackage{units}
\usepackage{amsthm}
\usepackage{amsmath}
\usepackage{amssymb}
\usepackage{graphicx}
\usepackage{esint}
\usepackage[unicode=true,pdfusetitle,
 bookmarks=true,bookmarksnumbered=true,bookmarksopen=true,bookmarksopenlevel=1,
 breaklinks=false,pdfborder={0 0 0},backref=false,colorlinks=false]
 {hyperref}
\hypersetup{
 pdfpagelayout=OneColumn, pdfnewwindow=true, pdfstartview=XYZ, plainpages=false}

\makeatletter

\pdfpageheight\paperheight
\pdfpagewidth\paperwidth

\newcommand{\lyxmathsym}[1]{\ifmmode\begingroup\def\b@ld{bold}
  \text{\ifx\math@version\b@ld\bfseries\fi#1}\endgroup\else#1\fi}

\providecommand{\tabularnewline}{\\}
\floatstyle{ruled}
\newfloat{algorithm}{tbp}{loa}[chapter]
\providecommand{\algorithmname}{Algorithm}
\floatname{algorithm}{\protect\algorithmname}

  \theoremstyle{definition}
  \newtheorem{defn}{\protect\definitionname}
 \theoremstyle{definition}
  \newtheorem{example}{\protect\examplename}
  \theoremstyle{plain}
  \newtheorem{lem}{\protect\lemmaname}
  \theoremstyle{plain}
  \newtheorem{cor}{\protect\corollaryname}
\theoremstyle{plain}
\newtheorem{thm}{\protect\theoremname}
  \theoremstyle{plain}
  \newtheorem{prop}{\protect\propositionname}

\@ifundefined{date}{}{\date{}}
\newcommand{\N}{Y}

\newfloat{protocol}{tbp}{loa}
\providecommand{\protocolname}{Protocol}
\floatname{protocol}{\protect\protocolname}

\usepackage[figure]{hypcap}

\let\myTOC\tableofcontents
\renewcommand\tableofcontents{%
  \frontmatter
  \pdfbookmark[1]{\contentsname}{}
  \myTOC
  \mainmatter }

\setkomafont{captionlabel}{\bfseries}
\setcapindent{1em}

\usepackage{calc}

\let\mySection\section\renewcommand{\section}{\suppressfloats[t]\mySection}

\makeatother

  \providecommand{\definitionname}{Definition}
  \providecommand{\examplename}{Example}
  \providecommand{\lemmaname}{Lemma}
  \providecommand{\propositionname}{Proposition}
\providecommand{\corollaryname}{Corollary}
\providecommand{\theoremname}{Theorem}

\begin{document}

%
%
%
%
%
%
%
%
\begin{titlepage}
\begin{center}
\vspace*{2cm}
\noindent {\large \textbf{Jordi Soria Comas}} \\
\vspace*{1cm}
\noindent {\LARGE \textbf{Improving data utility}} \\
\noindent {\LARGE \textbf{in differential privacy}}\\
\noindent {\LARGE \textbf{and k-anonymity}} \\
\vspace*{1cm}
\noindent \LARGE \textbf{DOCTORAL THESIS} \\
\vspace*{1cm}
\noindent \Large \textbf{\textsc{Supervised by Dr. Josep Domingo-Ferrer}}\\
\vspace*{1cm}
\noindent \LARGE \textbf{Department of\\}
\noindent \LARGE \textbf{Computer Engineering and Mathematics\\}
\vspace*{3cm}
\includegraphics[width=5cm]{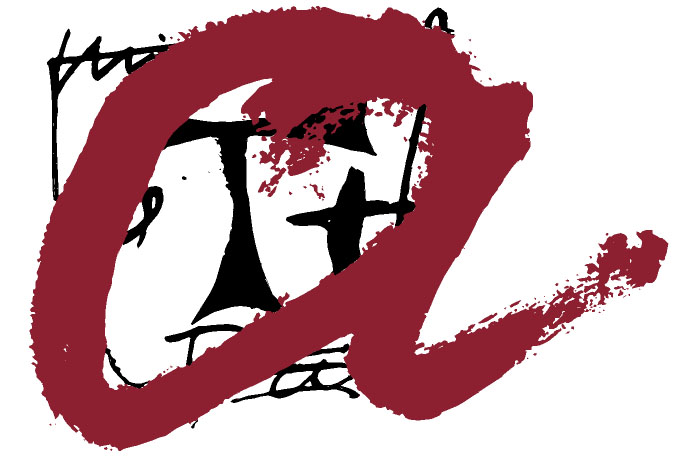}\\
\vspace*{1cm}
\noindent \Large {Tarragona} \\
\noindent \Large {2013} \\
\vspace*{0.8cm}
\end{center}

\newpage
\thispagestyle{empty}
~

\newpage
\thispagestyle{empty}
\begin{flushleft}
\includegraphics[width=8cm]{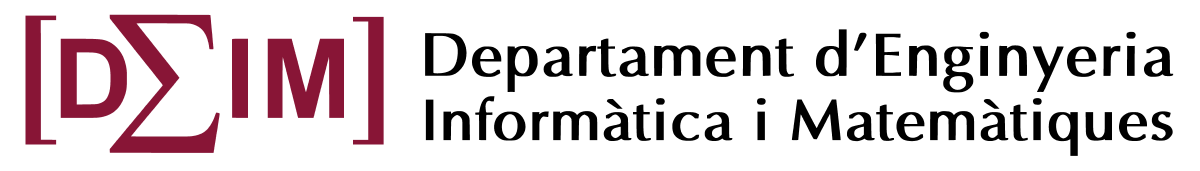}\\
\textsf{Av. Països Catalans, 26}\\
\textsf{Campus Sescelades}\\
\textsf{43007 Tarragona}\\
\textsf{Tel. (+34) 977 559 703}\\
\textsf{Fax (+34) 977 559 710}
\end{flushleft}
\vspace*{3cm}
I STATE that the present study, entitled ``Improving data utility in 
differential privacy and k-anonymity'', presented by Jordi Soria Comas
for the award of the degree of Doctor, has been carried out under my 
supervision at the Department of Computer Engineering and Mathematics
of this university, and that it fulfils all the requirements to be 
eligible for the European Doctorate Award.\\
\\
\begin{flushleft}
\vspace*{3cm}
Tarragona, 24 Apr 2013\\
\vspace*{1cm}
Doctoral Thesis Supervisor\\
\vspace*{3cm}
Dr. Josep Domingo-Ferrer
\end{flushleft}

\end{titlepage}
\sloppy

%
%
%
%
%
%
%
%
%
%
%

\cleardoublepage{}

\lhead{\rightmark}

\rhead[\leftmark]{}

\lfoot[\thepage]{}

\cfoot{}

\rfoot[]{\thepage}

\tableofcontents{}

\cleardoublepage{}

\pagestyle{plain}

\chapter*{Abstract}


\addcontentsline{toc}{chapter}{Abstract} 

Data about individuals are collected on a regular basis by governments
and companies for a variety of purposes. These data stores are valuable
resources, and there is a growing demand to access them. However,
the dissemination of data about individuals is a 
controversial task. On 
the one side, there is a demand to access accurate data; on the other side,
there is a risk of disclosing confidential information about specific
individuals. Protecting individuals' privacy usually entails some
degree of data modification, which decreases the utility of
the output. Finding a good balance between privacy and utility is
of the utmost importance in data dissemination.

The suitability of an anonymization method depends on several aspects
of the data release: the type of
release (\emph{e.g.} microdata
file, statistical table, on-line database), the specificities of the data
(\emph{e.g.} numerical, nominal, ordinal), and the desired level of disclosure
limitation. Regarding the level of disclosure limitation, 
the privacy guarantees offered by anonymization
methods have evolved with time. 
The initial approach by the statistical community focused
on masking confidential data (using either a perturbative or a non-perturbative
masking), but no formal privacy guarantees were offered. Later, the
computer science community developed several privacy models that
offer more abstract privacy guarantees; for instance, by hiding each
individual within groups of indistinguishable individuals, or by limiting
the information that may be gained from accessing the released data.

We take the approach of the computer science community. The focus
lies on two mainstream privacy models: 
$k$-anonymity and $\varepsilon$-differential
privacy. Once a privacy model has been selected, the goal is 
to enforce it while preserving as much data utility as possible.
The main objective of this thesis is to improve the
data utility in $k$-anonymous and $\varepsilon$-differentially private
data releases.

$k$-Anonymity is a widely accepted privacy model for the anonymization
of microdata sets; however, it has several drawbacks. On the disclosure
limitation side, there is a lack of protection against attribute disclosure
and against informed intruders. On the data utility side, dealing
with a large number of 
quasi-identifier attributes is problematic.
The first contribution of this
thesis is a relaxation of $k$-anonymity that improves 
protection against informed intruders, as well as data utility
in case of multiple quasi-identifier attributes.

Differential privacy limits disclosure risk through noise addition.
The Laplace distribution is commonly used for the random noise. We
show that the Laplace distribution is not optimal: the same disclosure
limitation guarantee can be attained by adding less noise. 
In this thesis, optimal
univariate and multivariate noises are characterized and constructed. 

Differential privacy seeks to limit the contribution of any single
individual on the response to a query. However, the expected response
usually depends on the user's prior knowledge. Common mechanisms
to attain differential privacy do not take into account 
the users' prior
knowledge; they implicitly assume zero initial knowledge about the
query response. As a consequence, the response provided 
may not
be very accurate for users with substantial initial knowledge.
We propose a mechanism that focuses on limiting
the knowledge gain over the prior knowledge.

$k$-Anonymity and $\varepsilon$-differential privacy are often seen
as opposed privacy notions. Supporters of $\varepsilon$-differential
privacy present $k$-anonymity as an old-fashioned privacy model 
that offers only poor disclosure limitation guarantees, while supporters
of $k$-anonymity claim that the damage done to the original data
when enforcing $\varepsilon$-differential privacy is too large.
The last contribution of this thesis shows that
microaggregation-based $k$-anonymity and $\varepsilon$-differential
privacy can be combined to produce microdata releases with the strong
privacy guarantees of $\varepsilon$-differential privacy and improved
data accuracy.

\chapter*{Resum}

\addcontentsline{toc}{chapter}{Resum} 

Els governs i les corporacions recullen de manera habitual dades
sobre individus per a una varietat de propòsits. Aquestes dades són
un recurs valuós i hi ha una demanda creixent per accedir-hi. 
Amb tot,
la disseminació de dades sobre individus és una tasca controvertida.
D'una banda hi ha una demanda d'accés a dades acurades; de l'altra,
cal tenir present el risc de revelar informació confidencial sobre
algun dels individus. La protecció de la privadesa dels individus
implica una modificació de les dades abans de llur publicació, 
cosa que en redueix la utilitat. 
És fonamental trobar un equilibri adequat entre
privadesa i utilitat.

La conveniència d'un mètode d'anonimització depèn de diversos aspectes:
el tipus de publicació (microdades, taules, base de dades interactives), 
les especificitats pròpies de les dades (numèriques,
nominals, ordinals, etc.) i el nivell de protecció desitjat. Pel que
fa al nivell de protecció, les garanties que ofereixen els
mètodes d'anonimització han anat evolucionat. Inicialment, el procediment
proposat per la comunitat estadística buscava emmascarar les dades
confidencials (mitjançant tècniques pertorbatives o no pertorbatives),
però sense oferir garanties formals de privadesa. Més tard, la comunitat
informàtica va desenvolupar diversos models
que ofereixen garanties de privadesa més abstractes; per exemple,
amagar els individus dins de grups d'individus indistingibles, 
o limitar la contribució
que cada individu pot tenir en la resposta a una consulta.

Aquesta tesi adopta el punt de vista de la comunitat informàtica. 
Ens centrem en dos models de privadesa àmpliament
acceptats: el $k$-anonimat i la privadesa $\varepsilon$-diferencial. Un
cop triat el model de privadesa, l'objectiu passa
a ser complir-ne els requisits, alhora que preservar
la màxima utilitat possible en les dades resultants. L'objectiu
principal d'aquesta tesi és la millora de la utilitat en la publicació
de dades $k$-anònimes i $\varepsilon$-diferencialment privades.

El $k$-anonimat és un model de privadesa per a fitxers de microdades
àmpliament acceptat. No obstant, presenta alguns problemes. Pel que
fa al risc de revelació, no protegeix contra la revelació d'atributs
ni contra intrusos informats. Pel que fa a la utilitat de les dades,
tractar amb fitxers amb un nombre elevat d'atributs quasi-identificadors pot
ser problemàtic. Proposem un nou model basat en la relaxació dels
estrictes requeriments d'indistingibilitat que estableix el $k$-anonimat
però que, alhora, manté la mateixa probabilitat de re-identificació.
Aquest nou model permet de millorar la protecció contra intrusos informats,
alhora que millora la utilitat de les dades en presència de múltiples
atributs quasi-identificadors.

La privadesa diferencial limita el risc de revelació afegint un soroll
aleatori al resultat de les consultes. Habitualment,
es fa servir la distribució de Laplace per al soroll aleatori. 
A la tesi, mostrem que aquesta distribució
no és òptima: es poden complir els requeriments de la privadesa
$\varepsilon$-diferencial afegint sorolls més petits. 
A més, caracteritzem i construïm
les distribucions òptimes (univariant i multivariant).

La privadesa diferencial busca limitar l'efecte que cada individu
té sobre la resposta a una consulta. La resposta que un usuari espera
depèn del coneixement previ que té de la base de dades.
Malgrat això, els mecanismes habituals per obtenir privadesa diferencial
no tenen en compte el possible coneixement previ dels usuaris; 
implícitament, se'ls suposa
un coneixement nul. Per a un usuari amb un coneixement previ elevat,
la resposta obtinguda pot ser poc precisa. Proposem un
mecanisme basat a limitar el guany
de coneixement de l'usuari respecte del seu coneixement inicial.

El $k$-anonimat i la privadesa $\varepsilon$-diferencial es presenten
sovint
com a models contraposats. D'una banda, els partidaris de
la privadesa $\varepsilon$-diferencial presenten el $k$-anonimat
com un model ja superat que ofereix unes garanties de privadesa
pobres; d'altra banda, els qui recolzen el $k$-anonimat argumenten
que la privadesa diferencial provoca danys massa importants a les
dades. La darrera contribució d'aquesta tesi mostra que 
la privadesa $\varepsilon$-diferencial i el
$k$-anonimat no són conceptes completament inconnexos: si es pren
com a punt de partida per obtenir privadesa $\varepsilon$-diferencial
un conjunt de dades $k$-anònim (obtingut mitjançant un cert 
tipus de microagregació),
la quantitat de soroll necessari 
es veu reduïda significativament.

\chapter*{Resumen}

\addcontentsline{toc}{chapter}{Resumen} 

Los gobiernos y las corporaciones recogen regularmente datos sobre
individuos para gran variedad de propósitos. Estos almacenes de datos son
unos recursos valiosos, cosa que provoca una creciente demanda de
acceso a los datos. Sin embargo, la diseminación de datos sobre individuos
es una tarea controvertida. Por un lado, hay una demanda de acceso
a datos precisos; por otro lado, existe el riesgo de revelar información
confidencial sobre algún individuo específico. La protección de la
privacidad de los individuos acarrea normalmente la modificación
de los datos originales, reduciéndose así la utilidad de los datos
publicados. Es primordial
encontrar un equilibrio adecuado entre privacidad y utilidad.

La conveniencia de un método de anonimización depende de varios aspectos:
el tipo de publicación (microdatos, datos agregados, bases de datos
interactivas), las especificidades propias de los datos (numéricos, nominales,
ordinales) y el nivel de protección deseado. En relación al nivel
de protección, ha habido una evolución en las garantías que
ofrecen los métodos de anonimización. Inicialmente, el procedimiento
propuesto por la comunidad estadística se centraba en enmascarar los
datos confidenciales (mediante técnicas perturbativas o no perturbativas),
pero sin ofrecer garantías formales de privacidad. Más tarde, la comunidad
informática desarrolló varios modelos que
ofrecen unas garantías de privacidad más abstractas; por ejemplo,
esconder a los individuos en grupos formados por varios individuos
indistinguibles,
o limitar el incremento de información que proporcionan los datos
publicados.

Adoptamos aquí el proceder de la comunidad informática y 
nos ocupamos de 
dos de los principales modelos de privacidad: $k$-anonimato y 
privacidad $\varepsilon$-diferencial. Una vez seleccionado 
un modelo de privacidad,
el objetivo pasa a ser cumplir con sus requisitos, a la vez que se trata
de preservar
la máxima utilidad posible para los datos. El objetivo principal de
la presente tesis es la mejora de la utilidad de los datos en publicaciones
$k$-anónimas y $\varepsilon$-diferencialmente privadas.

El $k$-anonimato es un modelo de privacidad para ficheros de microdatos
ampliamente aceptado; sin embargo, presenta algunos problemas. En
relación a la limitación del riesgo de revelación, no protege contra
la revelación de atributos, ni contra intrusos informados. En relación
a la utilidad de los datos, tratar con ficheros que tienen un número
elevado de atributos 
cuasi-identificadores es problemático. En esta tesis proponemos
un nuevo modelo basado en la relajación del requisito de indistinguibilidad
que establece el $k$-anonimato pero que mantiene la misma probabilidad
de re-identificación. Este nuevo modelo nos permite aumentar la
protección contra intrusos informados, a la vez que mejora la utilidad
de los datos en presencia de múltiples
atributos cuasi-identificadores.

La privacidad diferencial limita el riesgo de revelación añadiendo
un ruido aleatorio al resultado de las consultas. 
Habitualmente se utiliza la distribución de Laplace
para generar dicho ruido. 
En esta tesis mostramos
que la distribución de Laplace no es óptima para obtener privacidad
diferencial: los requisitos de la privacidad diferencial se pueden
cumplir introduciendo menos ruido. Asimismo, caracterizamos y 
construimos
las distribuciones óptimas (univariante y multivariante).

La privacidad diferencial busca limitar el efecto que cada individuo
tiene en la respuesta a una consulta. La respuesta que los usuarios
esperan depende del conocimiento previo que tienen. Sin embargo, lo
mecanismos usuales para obtener privacidad diferencial no tienen en
cuenta este conocimiento previo; implícitamente,
se supone un conocimiento
nulo. Como consecuencia, la respuesta puede ser poco precisa
cuando el usuario tiene un conocimiento previo elevado sobre ella. 
Proponemos
un mecanismo para obtener privacidad diferencial orientado a 
limitar la ganancia de conocimiento del usuario con respecto
a su conocimiento previo.

El $k$-anonimato y la privacidad $\varepsilon$-diferencial son
a menudo presentados como nociones de privacidad contrapuestas.
Por un lado, quienes apoyan la privacidad $\varepsilon$-diferencial
presentan el $k$-anonimato como un modelo de privacidad obsoleto que
ofrece unas garantías pobres; por otro lado, quienes apoyan el $k$-anonimato
argumentan que la privacidad diferencial daña demasiado los datos.
En la última contribución de esta tesis, mostramos que 
la privacidad $\varepsilon$-diferencial y el $k$-anonimato
no son nociones completamente inconexas: tomando como datos de partida
para obtener $\varepsilon$-privacidad diferencial un conjunto de
datos $k$-anónimo (construido mediante un cierto tipo de microagregación)
se reduce la cantidad de ruido necesaria y se mejora la utilidad de la
información.

\cleardoublepage{}

\pagestyle{fancy}

\lhead[\chaptername~\thechapter]{\rightmark}

\lhead[\chaptername~\thechapter]{\rightmark}

\rhead[\leftmark]{}

\lfoot[\thepage]{}

\cfoot{}

\rfoot[]{\thepage}

\chapter{Introduction}

\section{Motivation}

The collection of personal information has traditionally been limited
to surveys (where information about a specific topic is collected
from a sample population) and client-provider relationships (where
transactions carried out are recorded). One remarkable characteristic
of such situations is that the individual whose information is collected
is aware of it. Nowadays, the advances in information technologies
have dramatically changed the state of things. Information gathering
has become pervasive: vast amounts of data are collected by governments
and corporations on a daily basis, most of the times without the consent
of individuals who may even be unaware of it. For instance, Internet
stores gather data from everything that happens in their sites~\cite{eco_clicking,eco_data,KrishB2010};
not only do they keep track of the items you buy, but also of the
ones you browse but do not buy. Their objective is the generation
of a detailed profile of each individual; they can exploit this information
to guide personalized commercial communication actions, but also to
guide the strategic planning of the firm. Internet firms have long
recognized that the information they collect from customer interaction
offers them a competitive advantage over traditional firms. 

As a valuable resource, there is a growing demand to access the collected
data. For instance, many firms base their marketing and strategic
plans on publicly released census data~\cite{census}. However, when
data about individuals or entities are to be disseminated for secondary
use, special care must be taken to avoid privacy violations. Some
popular attacks against publicly released data include: the uncovering
of the medical records of the governor of Massachusetts in the data
released by the Group Insurance Commission (GIC)~\cite{Sweeney2002},
the uncovering of identities in a de-identified data set containing
a list of 20 million web search queries collected by AOL~\cite{Barbaro2006},
and the de-anonymization attacks conducted against the Netflix Prize
data set~\cite{Narayanan2008}. 

The goal of \emph{Statistical Disclosure Control} (SDC) or \emph{Statistical
Disclosure Limitation} (SDL) is to allow the release of data while
preserving the privacy of individuals. SDC techniques work by masking
the original data or statistics to be released. While reducing the
risk of disclosure, the masking also reduces the utility of the published
data. This is a fundamental trade-off that cannot be avoided: finding
a balance between privacy and utility, so that individuals' privacy
is protected and data are still useful, is the primary objective of
disclosure limitation techniques. 

SDC has traditionally evaluated the level of 
disclosure limitation experimentally;
for instance, by trying to re-identify records in the released data.
In the last few years, the computer science community has proposed
several privacy 
models that try to bring formal privacy guarantees
into the field. Usually these privacy models seek to introduce uncertainty
in the outcome of the attacks against the privacy of individuals.
The suitability of such privacy models depends on several aspects
of the data dissemination under consideration: the type of data being
released, the required level of disclosure risk limitation, etc. When
a privacy model is judged to offer enough disclosure limitation,
the next goal is to generate a data set that satisfies the selected
model and maximizes data utility. In this thesis, we focus on two
mainstream privacy models: $k$-anonymity, a model used 
to limit the risk of re-identification in microdata
releases; and $\varepsilon$-differential privacy, a privacy model
for interactive databases that seeks to limit the knowledge gain that
can be extracted from query responses. We mainly focus on data utility:
we aim at providing methods to satisfy those models, while offering improved
data utility;
but we also aim at finding a link between those models.

\section{Contributions}

We revisit two mainstream privacy models, $k$-anonymity and $\varepsilon$-differential
privacy, and we propose several improvements. These are 
mainly on the data utility
side, but also on the disclosure limitation guarantees for the case
of $k$-anonymity. Our main contributions are:
\begin{enumerate}
\item \emph{Probabilistic $k$-anonymity}. The $k$-anonymity model,
although widely accepted, suffers from certain limitations that affect
both data utility and disclosure limitation. We propose a relaxation
of the $k$-anonymity model where the requirement for indistinguishability
of records in terms of quasi-identifiers is removed, but the same
probability of uncovering a confidential attribute in the released
data set is retained. The new proposal offers two advantages. First
of all, by removing the indistinguishability requirement, the range
of feasible methods widens, and we can thus search for a method
that offers improved data utility. Apart from the improvement in
data utility, the fact that we no longer have a fixed partition in
sets of indistinguishable records opens the door to improvements on
disclosure limitation against informed intruders.
\item \emph{Optimal data-independent noise for $\varepsilon$-differential
privacy}. $\varepsilon$-Differential privacy is an output perturbation
methodology; therefore, to improve the accuracy of the responses,
the magnitude of the perturbation must be reduced. We focus on data-independent
noises, which are more frequently used due to their simplicity, and
state a strict optimality criterion for the perturbation in terms
of the concentration of the probability mass around the zero. To show
the validity of our optimality criterion, we justify that a noise
that is optimal under this criterion must be optimal under any sensible
criterion (those that prefer that less distortion is introduced). We
show that the commonly used Laplace distribution is not optimal, and
optimal univariate and multivariate distributions are built.
\item \emph{Considering prior knowledge in $\varepsilon$-differential
privacy}. $\varepsilon$-Differential privacy guarantees that the
knowledge gain that can be extracted from the response to any query
is limited by a factor of $\exp(\varepsilon)$. Such guarantee must
be enforced independently of the prior knowledge that a particular
user has. The usual approach is to assume that the user has zero prior
knowledge, and to limit the knowledge gain to $\exp(\varepsilon)$
over it. While doing so, the knowledge gain is limited to $\exp(\varepsilon)$
independently of the prior knowledge that a particular user may have.
For a user with some prior knowledge, the response may be 
less than optimal
in terms of accuracy. We propose a novel approach towards $\varepsilon$-differential
privacy where, for each query, database users also send their prior
knowledge; a knowledge gain of $\exp(\varepsilon)$ is then enforced
over prior knowledge. We also show that the greater interaction
between the database and the users that results from the communication
of the prior knowledge does not open the door for any attack.
\item {\em Improving the 
utility of $\varepsilon$-differentially private data releases
by prior micro\-ag\-grega\-tion-based $k$-anonymity}. 
Although it was introduced as a disclosure
limitation methodology for interactive databases, $\varepsilon$-differential
privacy is general enough to be used in microdata releases. However,
due to the large amount of noise introduced, general-purpose mechanisms
to generate $\varepsilon$-differentially private data have
not been developed; the focus has been on the generation of data sets
that preserve the utility for specific families of functions. A general
approach towards the construction of $\varepsilon$-differential private
data sets consists in querying for the attributes' value of each individual;
however, due to the large sensitivity of such queries, this general
approach turns out to be infeasible. Our proposal employs a prior
microaggregation step to reduce the sensitivity of those queries.
Not all microaggregation algorithms offer the reduction in the sensitivity
that we seek; we provide a characterization of those which do. 
\item {\em Differential Privacy via $t$-Closeness in Data Publishing}.
Differential privacy and $k$-anonymity are often presented as 
antagonistic privacy models. The guarantees offered by such models are
quite different: whereas $k$-anonymity seeks to limit
re-identification, $\varepsilon$-differential privacy
seeks to limit the knowledge gain that users get from query responses.
However, $t$-cloness, an improvement over $k$-anonymity to limit
attribute disclosure, offers privacy guarantees that are closer
to those of differential privacy. We show that under specific 
conditions (using a specific distance function for $t$-closeness and
given a specific users' prior knowledge) $t$-closeness
implies $\varepsilon$-differential privacy. A method to attain
$t$-cloness for such conditions (and thus also 
$\varepsilon$-differential privacy) is provided. It is worth 
noting that unlike other approaches to differential privacy,
which output a random sample from a differentially private distribution, our
proposal fits the distribution in each of the $k$-anonymous
groups of records to the differentially private distribution 
by selecting the individuals that must belong to each
of the groups. Thus not only we achieve $\varepsilon$-differential privacy,
but also preserve the thruthfulness of the data inside each of the 
$k$-anonymous groups.

\end{enumerate}

\lhead[\chaptername~\thechapter]{\rightmark}

\rhead[\leftmark]{}

\lfoot[\thepage]{}

\cfoot{}

\rfoot[]{\thepage}

\chapter{Background}

\section{The right to privacy: a brief history}

Although nowadays it is considered a fundamental right~\cite{Glancy1979,Terstegge2007},
the \emph{``right to privacy''} is a quite recent concept. It
was coined by Warren and Brandeis, back in 1890, in an article~\cite{Warren1890}
published at the Harvard Law Review. Warren and Brandeis presented
laws as dynamic systems for the protection of individuals whose evolution
is triggered by social, political, and economic changes. In particular,
the conception of the right to privacy is triggered by the technical advances
and new business models of the time. To quote Warren and Brandeis: 
\begin{quote}
Instantaneous photographs and newspaper enterprise have invaded the
sacred precincts of private and domestic life; and numerous mechanical
devices threaten to make good the prediction that \textquotedbl{}what
is whispered in the closet shall be proclaimed from the house-tops.\textquotedbl{}
\end{quote}
Warren and Bradeis argue that the ``right to privacy'' was already
existent in many areas of the common law; they only gathered all these
sparse legal concepts, and put them into focus under their common
denominator. Within the legal framework of the time, the ``right
to privacy'' was part of the right to life, one of the three
fundamental individual rights recognized by the U.S. Constitution.

Privacy concerns revived again with the invention of the computers~\cite{Feistel1973}
and information exchange networks, which skyrocketed information collection,
storage and processing capabilities. The generalization of
population surveys was a consequence. 
The focus was now on data protection.

Nowadays, the concept of privacy has gained recognition and applies
to a wide range of situations such as: avoiding external meddling
at home, 
limiting the use of surveillance technologies, 
controlling processing and dissemination of personal data, etc.
Privacy is widely considered a fundamental right, and it is 
supported by international
treaties and many constitutional laws. 
For example, the Universal Declaration of Human Rights (1948) devotes its
Article 12 to privacy.

For a more comprehensive plot of key events in the history of privacy,
see~\cite{Timeline1,Timeline2}. In~\cite{Timeline1} key privacy-related
events between 1600 (when it was a civic duty to keep an eye on your
neighbors) and 2008 (after the USA PATRIOT Act and the inception of
Facebook) are listed. In~\cite{Timeline2} key moments that have shaped privacy
related laws are depicted.

As far as the protection of individuals' data is concerned, privacy legislation
is based on several principles~\cite{OECD1980,Terstegge2007}: collection
limitation, purpose specification, use limitation, data quality, security
safeguards, openness, individual participation, and accountability.

Among all the aspects that relate to data privacy, we are especially
interested in data dissemination. Data dissemination is, for instance,
a primary task of National Statistical Offices. These 
aim at offering an accurate picture of society; to that end, 
they collect and publish statistical data on a wide range
of aspects such as economy, population, etc. Legislation usually assimilates 
privacy violations on data dissemination to individual identifiability~\cite{Directive95/46/EC,FR_2007};
for instance, Title 13 Chapter 1.1 of the U.S. Code states that ``no
individual should be re-identifiable in the released data''.

\section{Types of data}

Among all privacy-related aspects, we are mainly concerned with disclosure
risk arising from data dissemination. The type of data being released
determines the potential threat to privacy as well as the most suitable
methods to limit it. Three types of data releases are considered:
\begin{description}
\item [{Microdata releases}] The term ``microdata'' refers to a record
that contains information related to a specific individual. A microdata
release aims at publishing raw data; that is, a set of microdata records.
This kind of data release offers the greatest level of flexibility
among all types of data releases: data recipients are not limited
to a specific prefixed view of data; they are able to carry any kind
of custom analysis on the released data. However, microdata releases
incur in the greatest threat to individuals' privacy. 
\end{description}
Microdata releases seek to allow data recipient on carrying custom
data analysis; however, if strong privacy guarantees are to be provided,
data utility may be greatly lowered, which may turn the released data
unsuitable for specific analysis that require accurate data. In order
to be able to meet the requirements for accurate data, NSO sometimes
generate two data sets: a publicly accessible data set where privacy
is prioritized over accuracy, and a data set that offers improved
data accuracy, but accessible only to restricted to a set of users
(committed to non-disclosure agreements).
\begin{description}
\item [{Aggregated data releases}] The data released do not refer to a
single individual but to a group of individuals. Contingency tables,
the traditional output of NSO, belong to this category. As only aggregated
data is published, threats to individuals' privacy are diminished
in comparison to microdata releases, but data analysis is limited
to the aggregated values being published.
\item [{Dynamic Databases}] Both microdata and aggregated data releases
offer a static view of the collected data. A specific data recipient
may not be interested in all the published data, but just on a subset
of them. The problem with static approaches is that, even if not used,
all the published data accounts when dealing with disclosure risk;
in other words, if a particular data recipient were only given access
to the data that are relevant to him, improved accuracy could be provided.
This is the idea that underlies dynamic databases: the user is allowed
to submit queries to the database, and data is only provided for the
submitted queries.
\end{description}
In the present thesis we deal with microdata releases (contributions
1 and 4) and dynamic databases (contributions 2 and 3).

\section{Microdata sets}

A microdata set can be modeled as a table where each row refers to
a different individual and each column contains information regarding
one of the attributes collected. We use the notation $T(A_{1},\ldots,A_{n})$
to denote a microdata set with information about attributes $A_{1},\ldots,A_{n}$.

The attributes in a microdata set are usually classified in the following
non-exclusive categories, according to the sensitivity of the information
they convey and the risk of record re-identification they imply: 
\begin{itemize}
\item \emph{Identifiers}. An attribute is an identifier if it provides unambiguous
re-identif\-ica\-tion of the individual to which the record refers. Some
examples of identifier attributes are the social security number,
the passport number, etc. If a record contains an identifier, any
sensitive information contained in other attributes may immediately
be linked to a specific individual. To avoid direct re-identification
of an individual, identifier attributes must be removed or encrypted.
We assume in the present thesis that, when dealing with microdata
releases, identifier attributes have previously been removed; that
is, we assume that $T(A_{1},\ldots,A_{n})$ does not contain any identifier
attribute. 
\item \emph{Quasi-identifiers}. Unlike an identifier, a quasi-identifier
attribute alone does not lead to record re-identification. However,
in combination with other quasi-identifier attributes, it may allow
unambiguous re-identification of some individuals. For example, \cite{Sweeney_2000}
shows that 87\% of the population in the U.S. can be unambiguously
identified by combining a 5-digit ZIP code, birth date and sex. Removing
quasi-identifier attributes, as proposed for the identifiers, is not
possible, because quasi-identifiers are most of the times required
to perform any useful analysis of the data. Deciding whether a specific
attribute should be considered a quasi-identifier is a 
thorny issue. In practice, any information an intruder has about an individual
can be used in record re-identification. For uninformed intruders,
only the attributes available in an external non-anonymized data set
should be classified as quasi-identifiers; in presence of informed
intruders any attribute may potentially be a quasi-identifier. Thus,
to make sure all quasi-identifiers have been removed, one should remove
all attributes (!). 
\item \emph{Confidential attributes}. Confidential attributes hold sensitive
information on the individuals that took part in the data collection
process (\emph{e.g.} salary, health condition, sex orientation, etc.).
The primary goal of microdata protection techniques is to prevent
intruders from learning confidential information about a specific
individual. This goal involves not only preventing the intruder from
determining the exact value that a confidential attribute takes for
some individual, but preventing inferences on the value of that attribute
(such as bounding it). 
\item \emph{Non-confidential attributes}. Non-confidential attributes are
those that do not belong to any of the previous categories. As they
do not contain sensitive information about individuals and cannot
be used for record re-identification, they do not affect our discussion
on disclosure limitation for microdata sets. Therefore, we assume
that none of the attributes in $T(A_{1},\ldots,A_{n})$ belong to
this category. 
\end{itemize}
When publishing a microdata file, the data collector must guarantee
that no sensitive information about specific individuals is disclosed.
To do so, the data collector does not publish the original microdata
set $T(A_{1},\ldots,A_{n})$, but a modified version $T'(A_{1},\ldots,A_{n})$
where the quasi-identifiers and/or the confidential attributes have
been masked. Disclosure can be classified into two categories~\cite{hundepool2012}: 
\begin{itemize}
\item \emph{Identity disclosure}. The intruder is able to determine the
true identity of the individual corresponding to a record in the microdata
file. After re-identification, the intruder associates the values
of the confidential attributes for the record to the re-identified
individual. 
\item \emph{Attribute disclosure}. Even if identity disclosure does not
happen, it may be possible for an intruder to infer some information
for a specific individual based on the published microdata set. For
example, imagine that the salary is one of the confidential attributes
and the job is a quasi-identifier attribute; if an intruder is interested
in a specific individual whose job he knows to be ``accountant''
and there are several accountants in the data set (including the target
individual), the intruder will be unable to re-identify the individual's
record based only on her job, but he will be able to lower-bound and
upper-bound the individual's salary (which lies between the minimum
and the maximum salary of accountants in the data set).
\end{itemize}

\section{Approaches to disclosure limitation}

Given a data set that contains information about individuals ---where
an individual is a person, household, company, etc.---, the goal is
to provide statistical information (or the means to extract statistical
information, in the case of microdata releases) about the population
or a subset of individuals, 
without disclosing confidential data of specific individuals.

Disclosure limitation technologies were initially developed under
the umbrella of National Statistical Institutes (NSIs), which still remain
a primary player, with the denomination of \emph{Statistical Disclosure
Control} (SDC) or \emph{Statistical Disclosure Limitation} (SDL).
Initially for tabular data releases, and later for microdata releases,
the statistical community has proposed many methods for limiting disclosure
risk. The preservation of the statistical properties of the original
data has also been on the focus of the statistical community since
the very beginning of statistical disclosure control. Good 
reference literature on statistical disclosure control are~\cite{Adam1989,Domi2001,hundepool2012}.
For an update on the current practices in
statistical disclosure limitation at
NSIs see~\cite{Zayatz2007,Zayatz2009,hundepool2012}.

Disclosure limitation also became a topic of interest in the
computer science research community. Within the computer science community,
the terms \emph{Privacy Preserving Data Publishing} (PPDP) and \emph{Privacy
Preserving Data Mining} (PPDM) are more commonly used. Privacy Preserving
Data Mining~\cite{Agrawal2000,Aggarwal2008} brings privacy protection
concerns into traditional data mining tasks: only the results of the
data mining are released; the original data are kept secret. A prevalent
characteristic among PPDP methods is that they are tightly coupled
to the underlying data mining task. On the other side, Privacy Preserving
Data Publishing~\cite{Fung2010} focuses on the publication of data
about individuals (\emph{microdata}). PPDP allows data users to carry any
kind of analysis on the released data. Although PPDM and PPDP seem to take 
completly different approaches to disclosure limitation, they may take 
advantage of the same anonymity models; for instance, $k$-anonymity can be 
used in both the generation of anonymous microdata sets and in the 
anonymization of the results of data mining tasks~\cite{Ciriani2008}.

Although both pursue the same objective, the approaches towards disclosure
limitation taken by the statistical and computer science communities
are not coincident. The common understanding~\cite{Drescgler2011}
is that the statistical community is usually more concerned with the
statistical validity of the data (valid inferences should be obtainable)
but offers only vague privacy guarantees (no formal privacy guarantees
are provided; the level of protection is evaluated {\em a posteriori}
for each specific data set). In contrast, methods developed
by the computer science community seek to attain a predefined
notion of privacy; thus, they offer {\em a priori} privacy guarantees. 
In this work we follow the path of the computer science community by
focusing on two mainstream privacy
models.

\section{Privacy models\label{sub:Privacy-criteria}}

The first attempt to come up with a formal definition for privacy
was done by Dalenius in~\cite{Dalenius1977}. Dalenius stated that
access to the released data should not allow any attacker to increase
his knowledge about confidential information related to a specific
individual. This is a very strict notion of privacy; in fact, it was
shown in~\cite{Dwork2006} that Dalenius's view of privacy is not
feasible in presence of background information (if any utility is
to be provided). Privacy criteria used in practice offer only limited
disclosure limitation guarantees.

Two main notions are used when talking about privacy in data releases:
anonymity (it should not be possible to re-identify any individual
in the published data), and confidentiality or secrecy (access to
the released data should not allow an attacker to increase its knowledge
about confidential information related to any specific individual).
Privacy models used in practice focus on one of those two 
notions (anonymity or confidentiality) and offer certain guarantees.

Preservation of individuals' privacy entails some loss on the utility
of the protected data, in comparison to the original data. For the
data to remain useful, the privacy guarantees offered are limited. Some
assumptions on the side knowledge available to potential attackers
are made, and the privacy preservation guarantees offered hold only
for such attackers.

In this thesis we focus on two mainstream privacy models: 
$k$-anonymity~\cite{Samarati1998a,Samarati1998b},
which, based on the anonymity principle, seeks to hide individuals
within groups of indistinguishable records; and $\varepsilon$-differential
privacy~\cite{Dwork2006a,Dwork2006}, which, based on confidentiality,
seeks to limit the knowledge gain provided by the output data.

Despite the fact that $k$-anonymity is solely based on anonymity,
and $\varepsilon$-differential privacy is solely based on secrecy,
other privacy models may mix both anonymity and secrecy. This is the
case, for instance, of $l$-diversity~\cite{Machanavajjhala2007}
and $t$-closseness~\cite{Li2007} that, similarly to $k$-anonymity,
seek to hide each individual among a group of individuals, but, unlike
$k$-anonymity, they also require the confidential information of
the individuals in the group to be sufficiently diverse to improve
secrecy.

A great number of privacy criteria have been proposed. They differ
in the kind and strength of the disclosure limitation guarantees 
they offer,
and in the suitability for a certain type of data release. For a thorough
review of privacy models see~\cite{Venkatasubramanian2008}.

\section{The privacy-utility tradeoff}

Disclosure limitation in a public data release involves some degree
of modification of the data to be released. Instead of publishing
the original data $D$, a masked version $D'$ is published. The masking
improves privacy but reduces the utility of the published data, in
comparison to the original data. This tension between privacy and
utility is unavoidable: privacy and utility are two different views
of the same thing, the amount of information published. By reducing
the amount of information published, privacy improves but utility
decreases; and the other way round. Two extreme cases are: publish
the original data, which offers the greatest utility but the least
privacy; and publish 
encrypted or random data, which incurs no
disclosure risk at all, but offers no utility.

Disclosure limitation technologies seek an equilibrium between privacy
and utility: the disclosure risk must be limited, but the data need to
remain useful. Sometimes the required equilibrium between privacy
and utility does not exist; for instance, when access to very accurate
and sensitive data is required by some data recipient. As the publication
of such a data set is not feasible, data providers must rely on other
mechanisms such as data access restriction and non-disclosure agreements.

\section{Measuring utility}

Disclosure limitation entails some modifications of the original data,
which decreases the utility of the protected data; therefore, it is
important to be able to assess the quality of the protected data.

Measuring the utility of the released data is a tough task. Currently,
no single utility measure is broadly accepted~\cite{Bertino2008}.
The main problem with utility measures is related to the relativity
of the term ``data utility''~\cite{Tayi1998}: ``data utility''
can be seen as ``fitness for use''. In other words, a data set may
be useful for some kind of analysis, but not for others. The measurement
of the data utility based on the intended data usage is usually preferred~\cite{Bertino2008},
as then utility evaluation focuses on the particular type of knowledge
that is to be extracted. Often, data protection cannot be performed
with a specific data use in mind~\cite{hundepool2012} (\emph{e.g}
data uses may be very diverse or even hard to identify at the time
of data release). For such cases, a generic measure of data utility
is required to help the data collector in assessing the damage inflicted
during the disclosure limitation process.

The suitability of a utility measure also depends on the type of data
release. Measures suitable for microdata releases may not be suitable
in assessing data utility in an interactive database environment. For instance,
in a microdata release we may evaluate how well the correlation between
attributes or marginal distributions are preserved; but these utility
measures are not appropriate for interactive databases. See~\cite{hundepool2012}
for a thorougher review of utility measures used for microdata
releases, and~\cite{Bertino2008} for utility measures used in privacy
preserving data mining.

\section{\label{sub:k_anonymity}$k$-Anonymity}

A de-identified data set is a data set that has had identifier
attributes removed. Removal of identifiers is essential to
hide the individuals' identity; however, it is usually not enough. For
instance, \cite{Sweeney_2000} shows that 87\% of the population in
the United States can be uniquely identified by combining 5-digit
ZIP, gender, date of birth. 

To re-identify a record in a published data set, the intruder performs
a record linkage attack. In a record linkage attack the intruder tries
to link the records in the released data set to the records in a 
non-anonymous external data set; that is,
the intruder seeks to associate identities to the records
in the released data set. This linkage is done by matching the values
of the common attributes (the quasi-identifiers).
If the linkage is correct, the
attack succeeds and the intruder learns the value of the confidential
attributes for the re-identified individual.

For an attribute to be a quasi-identifier, it must be externally available
in a non-anonymous data set; otherwise it cannot 
be used for re-identification
of records in the released data set.
\begin{defn}[Quasi-identifier]
A quasi-identifier $QI$ of $T$ is a subset of the set of attributes
$\{A_{1},\ldots,A_{n}\}$ that is available in an external, non-anonymous
data set.
\end{defn}

A common approach to prevent record linkage attacks is to hide each
individual within a group of individuals. This is the approach that
$k$-anonymity~\cite{Samarati1998b, Ciriani2008} takes: $k$-anonymity requires each record in the published
microdata set to be indistinguishable from $k-1$ other records based
on the quasi-identifiers. This way, an intruder with access to an external
non-anonymous data set that contains the quasi-identifiers in the
released data set $T'(A_{1},\ldots,A_{n})$ is unable to perform an
exact re-identification. For any individual in the external data set,
the intruder can at most determine a set of $k$ records in the published
data set that contains the target individual.
\begin{defn}[$k$-Anonymity~\cite{Samarati1998b, Ciriani2008}]
A microdata set $T'(A_{1},\ldots,A_{n})$ is said to satisfy $k$-anonymity
if, for each record $t\in T'$, there are at least $k-1$ other records
sharing the same values for all the quasi-identifier attributes. 
\end{defn}
The determination of the attributes that are available externally
in a non-anonymous data set is a key point for $k$-anonymity to provide
the desired protection against re-identification. It was already acknowledged
in the original proposal of $k$-anonymity that it is not possible
for the data holder to determine the knowledge that each of the data
recipients may have; thus, the data holder may misjudge which attributes
need to be considered as quasi-identifier 
attributes. In such cases the released
data may be less anonymous than initially intended. Proposed solutions~\cite{Sweeney1997}
rely on policies, laws, and contracts.

The original method to generate a $k$-anonymous data set~\cite{Samarati1998a}
was based on generalization and suppression, which continue to be
the dominant techniques to achieve $k$-anonymity. Generalization
reduces the granularity of the information contained in the quasi-identifier
attributes, thus increasing the chance of several records sharing the 
values of these attributes.
A generalization hierarchy is defined for each of the quasi-identifier
attributes. Generalization is usually performed at the attribute level;
that is, either all or none of the records are generalized. Suppression
removes tuples from the original data set so that they are not released.
Suppression is usually applied to remove outlier records before applying
generalization. Suppresion seeks to reduce the amount of generalization
required to generate the $k$-anonymous data set.

The use of generalization and suppression to enforce $k$-anonymity
produces a data set that is truthful, but less precise than the original
data set. The objective is 
to obtain a $k$-anonymous data set where information
loss is minimized. Usually the goal is a minimal generalization that
produces a $k$-anonymous data set for a given level of suppression
that is considered to be acceptable. It was shown in~\cite{Meyerson2004}
that finding an optimal $k$-anonymization via generalization and
suppression is a NP-hard problem. A large number of algorithms to
attain $k$-anonymity have been proposed~\cite{Samarati1998b,Bayardo2005,LeFevre2005,LeFevre2006,conf/icdt/AggarwalFKMPTZ05};
they rely on properties of $k$-anonymous data sets or heuristics
to reduce the amount of search, or search for sub-optimal solution.

A different approach towards achieving $k$-anonymity is based on
microaggregation~\cite{Domingo2005}. Microaggregation~\cite{Domi02}
is a family of anonymization algorithms for data sets that works in
two stages: 
\begin{itemize}
\item First, the set of records in a data set is clustered in such a way
that: i) each cluster contains at least $k$ records; ii) records
within a cluster are as similar as possible. 
\item Second, records within each cluster are replaced by a representative
of the cluster, typically the centroid record. 
\end{itemize}
Clearly, when microaggregation is applied to the projection of records
on their quasi-identifier attributes, the resulting data set is $k$-anonymous.
In~\cite{Domingo2005} a simple microaggregation heuristic called
MDAV is described, in which all clusters have exactly $k$ records,
except the last one, which has between $k$ and $2k-1$ records. As
the internals of MDAV will be required in Section~\ref{sec:dp_kanon},
we recall the MDAV algorithm (See Algorithm~\ref{alg:MDAV}).

\begin{algorithm}[h]
\caption{\label{alg:MDAV}Maximum distance to average record (MDAV) }

\textbf{let} $X$ be the original data set

\textbf{let} $k$ be the minimal cluster size

\vspace{0.2cm}

\textbf{while} $|X|\ge3k$ \textbf{do}

\hspace{0.5cm}$\overline{x}\leftarrow$average record of $X$

\hspace{0.5cm}$x_{1}\leftarrow$most distant record to $\overline{x}$
in $X$

\hspace{0.5cm}$x_{2}\leftarrow$most distant record to $x_{1}$ in
$X$

\hspace{0.5cm}Form a cluster with $x_{1}$ and its $k-1$ closest
records

\hspace{0.5cm}Form a cluster with $x_{2}$ and its $k-1$ closest
records

\hspace{0.5cm}Remove the clustered records from $X$

\textbf{end while}

\vspace{0.2cm}

\textbf{if} $|X|\ge2k$ \textbf{then}

\hspace{0.5cm}$\overline{x}\leftarrow$average record of $X$

\hspace{0.5cm}$x_{1}\leftarrow$most distant record to $\overline{x}$
in $X$

\hspace{0.5cm}Form a cluster with $x_{1}$ and its closest $k-1$
records

\hspace{0.5cm}Remove the clustered records from $X$

\textbf{end if}

\vspace{0.2cm}

Form a new cluster with the remaining records.

\vspace{0.2cm}
 Within each formed cluster, replace the values of each quasi-identifier
attribute with the average value of the attribute over the cluster. 
\end{algorithm}

Despite being a widely accepted privacy model, $k$-anonymity
suffers from certain limitations. 
The most common criticism against
$k$-anonymity refers to the lack of protection against attribute disclosure:
if all the individuals within a group of indistinguishable records
share same value for a confidential attribute, then the intruder learns
the confidential attribute, even without re-identification. Some refinements
to the basic $k$-anonymity model have been proposed to improve
the protection against attribute disclosure: $l$-diversity~\cite{Machanavajjhala2007}
requires the presence of $l$ different well-represented values for
the confidential attribute in every group of records sharing the same
quasi-identifier values; $t$-closeness~\cite{Li2007} requires the
distribution of the confidential attribute in any group of records
sharing the quasi-identifier values to be close to the distribution
in the overall data set.

\section{\label{back_dp}$\varepsilon$-Differential Privacy}

Most disclosure limitation mechanisms are specifically designed to
avoid releasing information that is known to be disclosive. Such mechanisms
are instructed with the kind of data releases that may lead to a privacy
breach, and are designed to avoid them. To determine the data releases
that may lead to a privacy breach, a guess on the amount of side information
available to the intruders is usually made. As long as this guess
is accurate, the disclosure limitation mechanism accomplishes its
duty, but a privacy breach may happen if there are intruders with
greater amounts information. 

The approach of differential privacy towards disclosure limitation 
is different. Instead of enforcing a pre-specified set of rules
that seek to limit disclosure risk, it limits the effect of the
presence or absence of any single individual on any information
that can be extracted from the database.

The disclosure limitation guarantee provided by $\varepsilon$-differential
privacy is similar to that of Dalenius 
(see Section~\ref{sub:Privacy-criteria}),
being the difference that, while Dalenius compared the knowledge before
and after accessing the released data, differential privacy compares
the knowledge before and after a single individual contributes her
data. In other words, instead of limiting the knowledge provided by
the data set, it limits the knowledge provided by each individual
in the data set.

Differential privacy was introduced as an interactive (or query-response)
mechanism, where the database is held by a trusted party that catches
the queries sent by the database users and outputs a sanitized response.
Let $D$ be the database, and assume that a user wants to compute
the value of a function $f$ over $D$. The trusted party computes
the real response to the query (that is, the value of $f(D)$) and
masks it before release. The end user receives $\kappa_{f}(D)$, the
masked response. The usual way to compute the perturbed value $\kappa_{f}(D)$
is to add a random noise to $f(D)$ that depends on the variability
of the query response.

Differential privacy assumes that each record in the data set refers
to a different individual; thus, comparing the output of a query before
and after an individual has contributed her data is
equivalent to comparing
the output of that query between data sets that differ in one record.
Data sets that differ in one record are known as neighbor data sets.
Strictly speaking, a database is a data set plus some software allowing
the data to be accessed and managed. However, unless there is risk
of ambiguity, in the sequel we will use database and data set as equivalent
terms.
\begin{defn}
{[}$\varepsilon$-differential privacy, \cite{Dwork2006}{]} \label{def:dp_dwork}A
randomized function $\kappa$ gives $\varepsilon$-differential privacy
if, for all data sets $D$, $D'$ that differ in one record, and all
$S\subset Range(\kappa)$ 
\begin{equation}
P(\kappa(D)\in S)\leq e^{\varepsilon}\times P(\kappa(D')\in S)\label{eq:dp_dwork}
\end{equation}

\end{defn}
The randomized function $\kappa$ in the definition represents the
output the user gets from the database as response to the submitted
query; actually, $\kappa(D)$ is the value resulting from adding random
noise to the real query response. Inequality~(\ref{eq:dp_dwork})
can be interpreted as a bound on the knowledge gain between the responses
obtained when performing the same query on data sets $D$ and $D'$.

Two approaches to the concept of ``neighbor data sets'' are found
in the literature on differential privacy: in~\cite{Dwork2006} two
data sets are said to be neighbors if one can be obtained from the
other by adding or removing a single record; in~\cite{Nissim2008}
two data sets are said to be neighbors if one can be obtained from
the other by modifying a single record.

Let us shed some light on the disclosure risk limitation provided
by differential privacy. Assume that the data sets $D$ and $D'$
can be obtained from one another by adding or removing one record;
the case of data sets $D$ and $D'$ that can be obtained from one
another by modifying a record is similar. Let $D'=D\setminus\{r\}$;
that is, $D$ contains the record $r$ contributed by individual $i_{r}$,
but $D'$ does not. Since $D'$ does not contain $i_{r}$'s data,
the level of privacy for $i_{r}$ when querying $D'$ is maximum;
even if disclosure for individual $i_{r}$ happens, it seems unreasonable
to blame the data set $D'$ (it does not contain $i_{r}$'s data).
As differential privacy guarantees that the knowledge gain between
data sets $D$ and $D'$ is limited, the disclosure risk for $i_{r}$
is limited.

To improve the accuracy of $\varepsilon$-differentially private responses,
the magnitude of the noise must be minimized. Several methods for
calibrating the noise have been proposed. We classify them
in two categories,
according to their 
dependency on the data set: data-independent
methods, such as~\cite{Dwork2006a}, and data-dependent methods,
such as~\cite{Nissim2007}. When calibrating to a data-independent
noise, the distribution of the noise is constant across data sets;
on the other side, when calibrating to data-dependent noises the distribution
of the noise is adjusted for each data set. In general, using a data-independent
noise is simpler, but data-dependent noises provide a better adjustment
of the noise to different degrees of variability of the query function
between neighbor data sets.

For data-independent noises, a Laplace distribution is typically used.
The mean parameter is set to zero (for the expected value of the noise
to be zero), and the scale parameter is adjusted to the largest variability
of the query function between neighbor data sets. Specifically, the
density function of the Laplace noise is 
\[
p(x)=\frac{\varepsilon}{2\Delta(f)}e^{-|x|\varepsilon/\Delta(f)}
\]

To refer to the largest change of a function between neighbor data
sets, the notion of $L_{1}$-sensitivity is introduced.
\begin{defn}
{[}$L_{1}$-sensitivity{]} \label{def:sensitivity-1} The $L_{1}$-sensitivity
of $f:\mathcal{D}\rightarrow\mathbb{R}^{d}$ is 
\begin{equation}
\Delta f=\max_{D,D'}\left\Vert f(D)-f(D')\right\Vert _{1}\label{eq:sensitivity-1}
\end{equation}
for all neighbor data sets $D$,$D'$. 
\end{defn}
Using Laplace-distributed noise with zero mean and $\Delta f/\varepsilon$
scale parameter provides $\varepsilon$-differential privacy~\cite{Dwork2006a}.
This result holds independently of the number of components of $f$.
An independent Laplace-distributed noise with zero mean and $\Delta f/\varepsilon$
scale must be added to each of the components. 

In~\cite{Dinur2003,Dwork2004,Blum2005}, it was proven that
if accurate responses
are returned for a sufficiently large number of count queries, 
then the original
database can be reconstructed with great accuracy.
Initially, these results raised the belief 
that the generation of protected microdata sets that preserve the
utility for a large number of queries was unfeasible. In particular,
this motivated the presentation of differential privacy as an an interactive
query-response mechanism. However, it was later shown in~\cite{Blum2008,Dwork2009,Hardt2010,Chen11}
that differential privacy could also be enforced in the non-interactive
setting and, indeed, that the generated microdata set could preserve
the utility for an arbitrary large number of queries. 

There is a lack of methods to generate general-purpose $\varepsilon$-differentially
private data sets. Current proposals preserve utility only for restricted
classes of queries (typically count queries). This contrasts with
the general-purpose utility-preserving data release offered by the
$k$-anonymity model.

\lhead[\chaptername~\thechapter]{\rightmark}
\rhead[\leftmark]{}
\lfoot[\thepage]{}
\cfoot{}
\rfoot[]{\thepage}

\chapter{Probabilistic $k$-anonymity\label{chap:Probabilistic--anonymity}}

We propose a privacy model that, similarly to $k$-anonymity, protects
against identity disclosure (the probability of determining the true identity
for a specific value of a confidential attribute is $1/k$), but 
offers improved data accuracy (in particular, it behaves
well in presence of multiple quasi-identifier attributes). Our proposal
is based on a relaxation of the indistinguishability requirement of
$k$-anonymity. Instead of requiring records to be indistinguishable
within sets of $k$ records as far as quasi-identifiers are concerned,
we focus on the probability of re-identification. By requiring the
probability of re-identification to be $1/k$ at most, we achieve
the same level of protection against re-identification provided by
$k$-anonymity, but the range of applicable methods
to implement our model is wider and hence
the information loss can be reduced.

The contents of this chapter have been published in~\cite{fuzzieee,recsi}.

\section{\label{sub:Limitations}Limitations of the $k$-anonymity model for
disclosure limitation}

Although $k$-anonymity is a popular privacy criterion, some criticism
has been raise against it~\cite{Domingo2008}. $k$-Anonymity seeks
to prevent identity disclosure (re-identification is only possible
with probability $1/k$), but confidential information can be revealed
even if re-identification is not feasible. For example, let a medical
data set contain quasi-identifier attributes Age, Gender, Zipcode
and Race, and confidential attribute AIDS (whose values can be Yes
or No). Imagine that we 3-anonymize this data set, but a group of
three records sharing a certain combination of quasi-identifier attribute
values also shares the confidential attribute value AIDS=Yes. In this
case, if the intruder can establish that her target respondent's record
lies within that group (because it is the only group with compatible
Age, Gender, Zipcode and Race), the intruder learns that the target
respondent suffers from AIDS. This kind of disclosure is known as
\emph{attribute disclosure} and arises from the lack of variability
of the confidential attribute inside a group of indistinguishable
records. Several fixes/alternatives to $k$-anonymity which are also based on
the partitioning of the data set in groups of indistinguishable records
have appeared: $l$-diversity~\cite{Machanavajjhala2007}, $t$-closeness~\cite{Li2007},
etc. However, none of those alternatives is free from shortcomings
(see~\cite{Domingo2008}). 

On the data utility side, $k$-anonymity has been shown to provide
reasonably useful anonymized results, especially for small $k$, but
utility degrades rapidly if the number of quasi-identifiers is increased.
This is a fundamental drawback that affects any method that is based
on the partitioning of the data set in groups of indistinguishable
records. Even more dramatic is the effect of increasing the number
of quasi-identifiers on the utility. This issue is know as ``the
curse of dimensionality''~\cite{Aggarwal2005}. 

There is yet another serious concern on the disclosure limitation
provided by $k$-anonymity: the attack model considered 
is weak. $k$-Anonymity assumes that the data holder is capable of
discerning between quasi-identifier attributes and non-quasi-identifier
attributes; that is, the data holder is supposed to be
able to determine which attributes
may be available externally in a non-anonymous data set. It was already
recognized, when $k$-anonymity was first introduced~\cite{Samarati1998b},
that this is a quite stringent assumption. The proposed solution was
to rely on policies, laws, and contracts, but this is not feasible
if we aim at releasing the data
openly.

For example, consider a data set $T$ that holds the attributes 
Zipcode, Gender, Age, Income, and Disease, where the Income and Disease 
attributes
hold confidential information. As Income and Disease are confidential,
they should not be available in an external non-anonymous data set;
therefore, by following the usual approach, we would take Zipcode, Gender,
and Age as the quasi-identifiers. However, even if not available in
an external non-anonymous data set, Income and Disease may be available
to an informed intruder, which could use that knowledge to improve
the accuracy of the re-identification. For instance, let us assume that Alice
knows that Bob is in the released table $T$. By using Zipcode, Gender,
and Age, Alice is able to determine a group of $k$ records that contains
Bob's data. Now, let us assume that, as Alice and Bob are friends, Alice
knows the value of the Disease attribute for Bob's record. By using
this knowledge, Alice can perform a more precise re-identification, 
thus learning more about Bob's income than initially intended. The
extreme case happens when nobody else in the group of $k$ individuals
shares Bob's disease; Alice is then able to determine the exact value
of Bob's income with total certainty.

An even more insidious intruder can be imagined. Imagine that Alice
does not know Bob's disease, but knows the disease of some of the other
individuals that share Bob's combination 
of quasi-identifier attribute values ---{\em e.g.} Alice works
in an hospital, and happens to meet those individuals. By using the
quasi-identifiers, Alice determines a set of $k$ records that must
contain Bob's data; by using her knowledge about the Disease
attribute, Alice is then able to perform a more precise re-identification
of Bob's record than initially intended. 

For $k$-anonymity to offer protection against intruders with confidential
information, we have to assume that all the attributes can be used
in the re-identification; in other words, all attributes are quasi-identifiers.
But we have already commented that increasing the number of quasi-identifiers
has a deep negative effect on the utility of the released data. We
will show that probabilistic $k$-anonymity offers improved data accuracy
in case of multiple quasi-identifier attributes; thus, probabilistic
$k$-anonymity is able to provide disclosure limitation against informed
intruders with reasonable data accuracy. The improved data quality
comes from the ability to use multiple partitions of the data set.

\section{The probabilistic $k$-anonymity model}

$k$-Anonymity guarantees that, for any combination of values of quasi-identifier
attributes in the published microdata set $T'(A_{1},\ldots,A_{n})$,
there are at least $k$ records sharing that combination of values.
Therefore, given an individual in an external non-anonymous data set,
the probability of performing the right linkage back to the corresponding
record in the published microdata set, and thus the probability of
learning its confidential attributes, is at most $1/k$. It is in
this sense that probabilistic $k$-anonymity is defined.

A similar relaxation of the notion of $k$-anonymity was presented
in~\cite{Zhang2007}, which partitioned the data set and applied
a permutation inside each of the partition components. This is the
same strategy that we will apply in Section~\ref{sec:data-swapping}
to achieve probabilistic $k$-anonymity. However, probabilistic $k$-anonymity
is a more general framework; it is not limited to permutations,
although permutations
are a convenient choice to simplify probability calculations.
Moreover, \cite{Zhang2007} did not address the issues described in
Section~\ref{sub:Limitations}, which probabilistic $k$-anonymity
does address.

\begin{defn}[Probabilistic $k$-anonymity]
Let $T'(A_{1},\ldots,A_{n})$ be a published data set generated from
an original data set $T(A_{1},\ldots,A_{n})$ using an anonymization
mechanism $M$. The data set $T'$ is said to satisfy probabilistic
$k$-anonymity if, for any non-anonymous external data set $E$, the
probability for an intruder $I$ knowing $T'$, $M$ and $E$ to correctly
link any record $x\in E$ and its corresponding record (if any) in
$T'$ is at most $1/k$. 
\end{defn}
Note than any method used to achieve $k$-anonymity also leads to
probabilistic $k$-anonymity. In this sense, it may be said that $k$-anonymity
provides a stronger guarantee. However, from the point of view of
the probability of re-identification, both provide the same level
of protection.
Note that stating that $k$-anonymity is stronger
does not contradict the fact that a distinguishing feature
of probabilistic $k$-anonymity
is to protect against informed intruders knowing some confidential
attribute values. Indeed, $k$-anonymity can also provide such
protection, but it needs to take all attributes as quasi-identifiers.

The advantage of probabilistic $k$-anonymity in comparison to $k$-anonymity
is that, by relaxing the indistinguishability requirements
within groups of $k$ records, the range of eligible methods to enforce
probabilistic $k$-anonymity is wider, 
and therefore we may expect a reduction in the information
loss.

\begin{figure*}[ht]
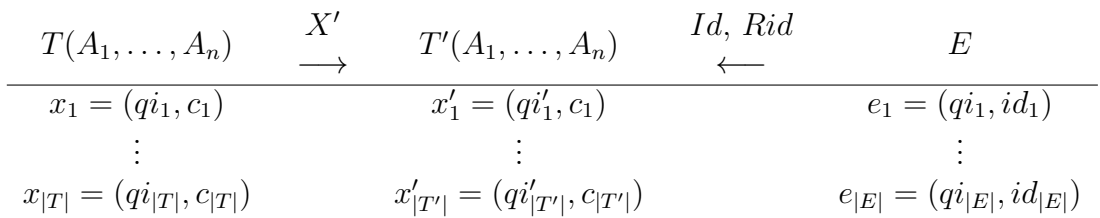

\begin{centering}
\begin{tabular}{ccccc}
$T(A_{1},\ldots,A_{n})$  & $\begin{array}{c}
X'\\
\longrightarrow
\end{array}$  & $T'(A_{1},\ldots,A_{n})$  & $\begin{array}{c}
Id,\, Rid\\
\longleftarrow
\end{array}$  & $E$\tabularnewline
\hline 
$x_{1}=(qi_{1},c_{1})$  &  & $x'_{1}=(qi'_{1},c_{1})$  &  & $e_{1}=(qi_{1},id_{1})$\tabularnewline
$\vdots$  &  & $\vdots$  &  & $\vdots$\tabularnewline
$x_{|T|}=(qi_{|T|},c_{|T|})$  &  & $x'_{|T'|}=(qi'_{|T'|},c_{|T'|})$  &  & $e_{|E|}=(qi_{|E|},id_{|E|})$\tabularnewline
\end{tabular}\caption{\label{not}Notations for probabilistic $k$-anonymity}

\par\end{centering}

\centering{} 
\end{figure*}

We start by analyzing probabilistic $k$-anonymity in presence of
non-informed intruders: confidential attributes are not available
externally, so they need not be considered as quasi-identifiers.
As probabilistic $k$-anonymity is expressed in terms of probability
of re-identification, it is natural to think of the released data
set $T'(A_{1},\ldots,A_{n})$ as a perturbation of $T(A_{1},\ldots,A_{n})$.
We use the notations in Figure~\ref{not}. The records $x_{i}$ in
$T$ have been split in two parts: the quasi-identifier attributes
$qi_{i}$, and the confidential attributes $c_{i}$. The records in
$T'$ are obtained by applying a random perturbation to the corresponding
record in $T$: $x_{i}'=X(x_{i})$. This perturbation affects only
the quasi-identifier attributes.

For the sake of simplicity, we assume that the released records in
$T'$ correspond to the first $|T'|$ records in $T$. If $|T|=|T'|$,
then all the records are released. The data set $E$ links the quasi-identifiers
$qi_{i}$ to the identifier $id_{i}$. The functions $Id$ and $Rid$
assign a record in $T'$ to the records in $E$, thus performing the
re-identification of the records in $T'$. The function $Rid$ is
the re-identification function used by the intruder, while $Id$ is
assumed to be the correct re-identification function. If there is
no record in $T'$ corresponding to the identity (\emph{i.e.} the
identified record) $e_{i}\in E$, then $Id$ returns the empty set.

The goal of probabilistic $k$-anonymity is to limit the probability
of performing the right linkage to at most $1/k$. With
the above notations this requirement can be stated as: for all $e_{i}\in E$
and for all $Rid()$ 
\[
P(Rid(e_{i})=Id(e_{i}))\le\frac{1}{k}
\]
 This formula captures the essence of the definition of probabilistic
$k$-anonymity: the probability of performing the right re-identification
must not be greater than $1/k$. However, by having the intruder use any
possible function $Rid()$ to perform the re-identification, the details
on how a rational intruder will proceed are hidden. Given a record
$e_{i}$, a rational intruder selects the record $x_{r}$ in $T'$
that has the greatest probability given the knowledge of $T'$, $E$
and $M$. The following examples will clarify how a rational intruder
acts. All examples assume that $E$ contains identities for all records
in $T$, which is the best possible knowledge that an intruder can
have.
\begin{example}
\label{ex1}Let us assume that $T$ contains two records, and that
only the first one is included in the anonymized data set. This situation
is shown in Table~\ref{figex1}. From the intruder's point of view,
$x_{1}'$ corresponds to either the individual in $e_{1}$ or $e_{2}$.
The best the intruder can do is to select the one that has the greatest
probability given the knowledge of $T'$, $E$, and the mechanism
$M$ used to generate $T'$ from $T$.

The probability that $x_{1}'$ corresponds to $e_{i}$ equals the
probability of obtaining $qi_{1}'$ from $qi_{i}^{E}$, over the total
probability of obtaining $qi'_{1}$ from any other record in $E$:
\[
P(X'(qi_{i}^{E})=qi_{1}'|T',E,M)
\]
\[
=\frac{P(X'(qi_{i}^{E})=qi_{1}'|M)}{\sum_{(qi_{j}^{E},id_{j})\in E}P(X'(qi_{j}^{E})=qi_{1}'|M)}
\]
 The intruder selects $e_{1}$ as his guess if $P(X'(qi_{1}^{E})=qi_{1}'|T',E,M)\ge P(X'(qi_{2}^{E})=qi_{1}'|T',E,M)$,
and $e_{2}$ otherwise. 
\end{example}
\begin{table}[ht]
\centering{}\caption{Data sets in Example~\ref{ex1}}
\label{figex1} %
\begin{tabular}{|c|c|c|}
\hline 
$T$  & $T'$  & $E$\tabularnewline
\hline 
$x_{1}=(qi_{1},c_{1})$  & $x_{1}'=(qi_{1}',c_{1})$  & $e_{1}=(qi_{1}^{E},id_{1})$\tabularnewline
$x_{2}=(qi_{2},c_{2})$  &  & $e_{2}=(qi_{2}^{E},id_{2})$\tabularnewline
\hline 
\end{tabular}
\end{table}

In the previous example we have seen that, given a record in $E$,
the linkage is performed to the record in $T'$ that has greatest
probability. If that probability is smaller than $1/k$, then the
probability of performing the right linkage will also be smaller than
$1/k$, as any other linkage will indeed result in a yet smaller probability.
Therefore, to achieve probabilistic $k$-anonymity, we must have for
all $qi^{E}\in E$ and all $qi'\in T'$ 
\begin{equation}
P(X'(qi^{E})=qi'|T',E,M)\le\frac{1}{k}\label{reid}
\end{equation}

\begin{example}
\label{ex2}In this example the amount of information in $T'$ has
been increased, by adding the record $x_{2}'$. The new data sets
are shown in Table~\ref{figex2}. As $E$ is assumed to exactly contain
the identities for the individuals in $T$, the intruder knows that
if one identity in $E$ corresponds to a specific record in $T'$,
the other identity in $E$ must correspond to the other record in
$T'$. This must be taken into account when computing the probabilities.
For example, the probability $P(X'(qi_{1}^{E})=qi_{1}'|T',E,M)$ that
$qi_{1}^{E}$ corresponds to $qi_{1}'$ equals $P(X'(qi_{1}^{E})=qi_{1}',X'(qi_{2}^{E})=qi_{2}'|T',E,M)$,
which can be computed as 
\[
\frac{P(X'(qi_{1}^{E})=qi_{1}',X'(qi_{2}^{E})=qi_{2}'|M)}{\sum_{\{i,j\}=\{1,2\}}P(X'(qi_{i}^{E})=qi_{1}',X'(qi_{j}^{E})=qi_{2}'|M)}
\]

\end{example}
\begin{table}[ht]
\caption{Data sets in Example~\ref{ex2}}

\centering{}\label{figex2} %
\begin{tabular}{|c|c|c|}
\hline 
$T$  & $T'$  & $E$\tabularnewline
\hline 
$x_{1}=(qi_{1},c_{1})$  & $x_{1}'=(qi_{1}',c_{1})$  & $e_{1}=(qi_{1}^{E},id_{1})$\tabularnewline
$x_{2}=(qi_{2},c_{2})$  & $x_{2}'=(qi_{2}',c_{2})$  & $e_{2}=(qi_{2}^{E},id_{2})$\tabularnewline
\hline 
\end{tabular}
\end{table}

The next example shows how the correct re-identification probability
would be computed in the most general case. 
\begin{example}
\label{ex3}Assume data sets $T$, $T'$ and $E$ as in Table~\ref{figex3}.
Contrary to Example~\ref{ex2}, fixing a correspondence between a
record in $T'$ and a record in $E$ does not completely fix the rest
of the correspondences. We still have to consider all the possible
combinations. The probability $P(X'(qi_{1}^{E})=qi_{1}'|T',E,M)$
that $qi_{1}^{E}$ corresponds to $qi_{1}'$ equals $\sum P(X'(qi_{1}^{E})=qi_{1}',X'(qi_{i_{2}}^{E})=qi_{j_{2}}',\ldots,X'(qi_{i_{M}}^{E})=qi_{j_{m}}'|T',E,M)$,
where $1<i_{2}<\ldots<i_{M}\le N$, and $\{j_{2},\cdots,j_{M}\}=\{2,\cdots,M\}$.
This probability can be computed as 
\[
\frac{\sum P(X'(qi_{1}^{E})=qi_{1}',X'(qi_{i_{2}})=qi_{j_{2}}'\ldots X'(qi_{i_{M}})=qi_{j_{M}}'|M)}{\sum P(X'(qi_{r_{1}})=qi_{s_{1}}',\ldots,X'(qi_{r_{M}})=qi_{s_{M}}'|M)}
\]
 where $1\le r_{2}<\ldots<r_{m}\le N$, and $\{s_{2},\cdots,s_{M}\}=\{2,\cdots,M\}$.
\end{example}
\begin{table}[ht]
\caption{Data sets in Example~\ref{ex3}}

\centering{}\label{figex3} %
\begin{tabular}{|c|c|c|}
\hline 
$T$  & $T'$  & $E$\tabularnewline
\hline 
$x_{1}=(qi_{1},c_{1})$  & $x_{1}'=(qi_{1}',c_{1})$  & $e_{1}=(qi_{1}^{E},id_{1})$\tabularnewline
$\vdots$  & $\vdots$  & $\vdots$\tabularnewline
$x_{N}=(qi_{N},c_{N})$  & $x_{M}'=(qi_{M}',c_{M})$  & $e_{N}=(qi_{N}^{E},id_{N})$\tabularnewline
\hline 
\end{tabular}
\end{table}

We have said that, to have probabilistic $k$-anonymity, Inequality
(\ref{reid}) must hold. However, the previous examples show that
the computation of the re-identification probability in Inequality
(\ref{reid}) for an arbitrary mechanism $M$ may be complex. In the
following section, we propose to use data swapping as $M$, which
has the advantage of making the computation of the re-identification
probability very simple.

\section{Probabilistic $k$-anonymity via microaggregation and swapping\label{sec:data-swapping}}

The proposed method consists of two main steps: (i) partition the
records in $T$ into groups of size $k$ and (ii) apply a permutation
to the quasi-identifier attributes within each of the groups. This
method can accommodate many variations, depending on how the partition
step (i) is done.

Note that, as the same permutation is applied to all quasi-identifier
attributes, the identity of the individual is not masked. However,
the quasi-identifier attributes are dissociated from the confidential
attributes, and therefore intruders can only guess the actual values
corresponding to a confidential attribute with probability at most
$1/k$. If leaking the mere presence of an individual in the data
set is itself disclosive, then some of the quasi-identifier attributes
must be considered confidential, which takes us to the informed intruder
scenario.

We introduce first the method that offers protection against uninformed
intruders. In other words, we assume that the attributes may be quasi-identifier
attributes or confidential attributes, but not both. Later we extend
our proposal to protect against informed intruders; assuming that
confidential attributes can be employed in the re-identification.

\subsection{Uninformed intruders\label{MDAV-SWAP}}

In presence of uninformed intruders there is a clear separation between
quasi-identifier and confidential attributes. Assuming that all records
in $T$ are masked and included in $T'$, we have the data sets in
Table~\ref{figex4}.

\begin{table}[ht]
\caption{Data sets in the uninformed intruder scenario}

\centering{}\label{figex4} %
\begin{tabular}{|c|c|c|}
\hline 
$T$  & $T'$  & $E$\tabularnewline
\hline 
$x_{1}=(qi_{1},c_{1})$  & $x_{1}'=(qi_{1}',c_{1})$  & $e_{1}=(qi_{1}^{E},id_{1})$\tabularnewline
$\vdots$  & $\vdots$  & $\vdots$\tabularnewline
$x_{N}=(qi_{N},c_{N})$  & $x_{N}'=(qi_{N}',c_{N})$  & $e_{N}=(qi_{N}^{E},id_{N})$\tabularnewline
\hline 
\end{tabular}
\end{table}

Selecting a random sample from $T$ to create $T'$ is a sensible
approach, as it introduces uncertainty on whether an individual whose
data was collected has been included in the published data set. However,
by assuming that all the individuals in $T$ have been included in
$T'$, we provide the intruder with the best information available.
Therefore, if we achieve probabilistic $k$-anonymity in this scenario,
then we will also achieve it in a scenario where a random sample from
$T$ is selected.

It is easy to see that the partition and swapping method described
above satisfies probabilistic $k$-anonymity because 
\[
P(X'(qi_{i}^{E})=qi|T',E,M)=\begin{cases}
\nicefrac{1}{k} & \mbox{if }qi\in G(id(qi_{i}^{E}))\\
0 & \mbox{otherwise}
\end{cases}
\]
where $G(id(qi_{i}^{E}))$ is the group of records of $T$ that contains
the record corresponding to $qi_{i}^{E}$.

The key point in the method is the partition step. A first approach
is to partition the data set $T$ into random groups. This leads indeed
not only to probabilistic $k$-anonymity, but to probabilistic $|T|$-anonymity,
as the quasi-identifiers of a record can be swapped with the quasi-identifiers
of any other record. Moreover, the risk of attribute disclosure is
small. However, the impact on data quality can be substantial, because
very different records may be swapped.

To achieve better data quality, the groups of records must be selected
to be as homogeneous as possible, although this increases the risk
of attribute disclosure. Our proposal is to generate the groups using
a microaggregation algorithm (\cite{Domi02,Domingo2005}) over the
quasi-identifier attributes. Microaggregation is a cardinality-constrained
form of clustering in which the number of clusters (groups) is not
fixed beforehand but the minimum cardinality of each group is required
to be $k$. In the section devoted to informed intruders, there are
some experimental results obtained by using the MDAV microaggregation
algorithm (\cite{Domingo2005,uARGUS}); MDAV attempts to maximize
intra-group homogeneity using the least squares criterion and it yields
groups with size $k$, except perhaps one group which has size between
$k$ and $2k-1$.

Other options in the selection of the groups of records are possible.
For example, a variant of MDAV, known as V-MDAV (\cite{Solanas2006,Solanas2010}),
may be used that performs clustering in groups of variable size and
that is known to reduce the information loss in clustered data sets.
The $\mu$-Approx microaggregation heuristic~\cite{Domingo2008b}
offers also variable-sized groups and is proven to yield a clustering
within a bound of the optimal clustering. Another possibility is to
select the groups of records in such a way that the risk of attribute
disclosure is reduced, by ensuring a certain diversity in the values
of the confidential attributes within each group.

\subsection{MDAV microaggregation for informed intruders\label{IR-SWAP}}

Consider a data set with attributes: $A_{0},A_{1},\ldots,A_{n}$,
with $A_{0}$ being a non-confidential quasi-identifier attribute,
and $A_{1},\ldots,A_{n}$ being confidential quasi-identifier attributes.
We assume the presence of several informed intruders, each of them
having knowledge of all confidential attributes except by one, whose
value wants to determine. To be more specific, intruder $I_{i}$,
for $i=1$ to $n$, is assumed to know the values of all attributes
except $A_{i}$. This is not the most stringent scenario. In the worst
case scenario, intruder $I_{i}$ would also have knowledge some of
the values of attribute $A_{i}$. However, we judge that the proposed
intruders are already strong enough. Note that the stronger the intruders,
the lower the data utility of the protected data set.

To achieve the desired level of protection against all informed intruders,
we apply the method presented for uninformed intruders once for each
informed intruder, in order to dissociate the value of the confidential
attribute unknown to this intruder from the rest of attributes. For
each informed intruder, we use the quasi-identifiers and the confidential
attribute shown in Table~\ref{tabin}.

\begin{table}[ht]
\caption{Quasi-identifiers and confidential attribute for each informed intruder}

\centering{}\label{tabin} %
\begin{tabular}{|c|c|c|}
\hline 
Intruder  & Quasi-identifier attributes  & Confidential attribute\tabularnewline
\hline 
$I_{1}$  & $A_{0},A_{2},\ldots,A_{n}$  & $A_{1}$\tabularnewline
\hline 
$I_{2}$  & $A_{0},A_{1},A_{3}\ldots,A_{n}$  & $A_{2}$\tabularnewline
\hline 
$\vdots$  & $\vdots$  & $\vdots$\tabularnewline
\hline 
$I_{n}$  & $A_{0},A_{1},\ldots,A_{n-1}$  & $A_{n}$\tabularnewline
\hline 
\end{tabular}
\end{table}

One difficulty that we face with the previous approach is that dealing
with informed intruders in sequence requires applying different permutations
over different but overlapping sets of attributes of the original
data set $T$ (the quasi-identifiers for each informed intruder).
To overcome this difficulty we take the reverse approach: instead
of performing the permutation over the quasi-identifier attributes,
we apply the reverse permutation to the single confidential attribute
unknown to the current intruder. In this way, each permutation acts
over a different attribute and there are no overlaps.

\subsection{Individual ranking microaggregation for informed intruders}

The above observation regarding the application of the inverse permutation
on the single unknown confidential attribute leads to single-attribute
microaggregation, also called individual ranking microaggregation.
Instead of multivariate microaggregation of quasi-identifier attributes,
we do individual ranking microaggregation on the unknown confidential
attribute. By doing so, the data quality of the published data set
is increased, as the confidential
attribute values
are only swapped across
records with similar values (see~\cite{Domingo2001} on the low information
loss caused by individual ranking microaggregation). It may be argued
that there is an increase in the attribute disclosure risk; however,
this increase can be mitigated by increasing $k$.

One extra benefit of this approach is that, since microaggregation
is performed on a single attribute, there is no need to normalize
attributes as required by multivariate microaggregation to avoid scale
problems.

\section{Experimental results\label{experimental}}

We have implemented the following three methods: 
\begin{itemize}
\item \emph{MDAV-ID}. MDAV microaggregation is run on the quasi-identifier
attributes to partition the data set in groups of size $k$ records.
Within each group, quasi-identifiers are replaced by the group centroid
in order to have identical quasi-identifiers for all records in the
group. This is the procedure suggested in~\cite{Domingo2005} and
it achieves the standard notion of $k$-anonymity proposed in~\cite{Samarati1998b}
in the sense that all quasi-identifiers within a group are made indistinguishable. 
\item \emph{MDAV-SWAP}. This is the method described in Section~\ref{MDAV-SWAP}
for probabilistic $k$-anonymity: MDAV microaggregation on the quasi-identifier
attributes plus swapping within groups. 
\item \emph{IR-SWAP}. This is the method described in Section~\ref{IR-SWAP}
above for probabilistic $k$-anonymity: individual ranking microaggregation
on each confidential attribute plus swapping within groups. 
\end{itemize}
The above methods have been tested with the ``Census'' and ``EIA''
reference data sets proposed in the European project CASC~\cite{refcasc}.

\subsection{``Census'' data set}

The ``Census'' data set contains 1080 records with 13 continuous
attributes. Following the approach in~\cite{Domingo2005} we consider
the first 6 attributes in ``Census'' to be non-confidential quasi-identifiers,
and the last 7 attributes to be confidential.

To assess the data quality, we evaluate the correlations from all
attributes to the confidential attributes. As the proposed methods
for probabilistic $k$-anonymity do not modify non-confidential attributes,
correlations between the latter have the same value as in the original
data set. Means and variances also remain unchanged for all attributes,
because swapping does not change the values taken by each original
attribute.

As an example, we computed the correlations for: i) the original data
set (see Table~\ref{table:corr_original}); ii) the $k$-anonymous
data set resulting from MDAV-ID with $k=12$ (see Table~\ref{table:corr_k-anonymity});
iii) the probabilistically $k$-anonymous data set resulting from
MDAV-SWAP with $k=12$ (see Table~\ref{table:correlation_qi_12});
and the probabilistically $k$-anonymous data set resulting from IR-SWAP
with $k=12$ (see Table~\ref{table:correlation_conf_12}). The values
in these tables must be taken with caution: they are results from
a single execution of the algorithms, and may change in another execution.
Despite these words of caution, we observe that MDAV-SWAP and IR-SWAP
result in correlation values closer to the original data set than
those obtained with MDAV-SWAP. The results of IR-SWAP are closest
to the original correlations.

\begin{table}[ht]
\noindent \begin{centering}
\caption{\label{table:corr_original} Correlations to the confidential attributes in the original ``Census''
data set}

\par\end{centering}

\begin{centering}

\par\end{centering}

\centering{}%
\begin{tabular}{cccccccc}
 & $A_{7}$  & $A_{8}$  & $A_{9}$  & $A_{10}$  & $A_{11}$  & $A_{12}$  & $A_{13}$\tabularnewline
\hline 
$A_{1}$  & .0038  & -.027  & -.024  & .031  & .032  & .039  & .036\tabularnewline
\hline 
$A_{2}$  & .98  & .14  & .2  & .73  & .71  & .72  & .7\tabularnewline
\hline 
$A_{3}$  & .44  & -.12  & -.058  & .56  & .55  & .56  & .55\tabularnewline
\hline 
$A_{4}$  & .98  & .2  & .28  & .73  & .69  & .71  & .69\tabularnewline
\hline 
$A_{5}$  & .78  & .27  & .27  & .9  & .85  & .88  & .86\tabularnewline
\hline 
$A_{6}$  & .79  & .13  & .22  & .59  & .57  & .57  & .56\tabularnewline
\hline 
$A_{7}$  & 1  & .17  & .23  & .72  & .7  & .71  & .69\tabularnewline
\hline 
$A_{8}$  &  & 1  & .45  & -.17  & -.19  & -.17  & -.17\tabularnewline
\hline 
$A_{9}$  &  &  & 1  & .072  & .061  & .70  & .075\tabularnewline
\hline 
$A_{10}$  &  &  &  & 1  & .96  & .98  & .96\tabularnewline
\hline 
$A_{11}$  &  &  &  &  & 1  & .91  & .89\tabularnewline
\hline 
$A_{12}$  &  &  &  &  &  & 1  & .97\tabularnewline
\hline 
$A_{13}$  &  &  &  &  &  &  & 1\tabularnewline
\hline 
\end{tabular}
\end{table}

\begin{table}[ht]
\noindent \begin{centering}
\caption{\label{table:corr_k-anonymity}Correlations to the confidential attributes in the data set obtained
using MDAV-ID with $k=12$ (``Census'' data set)}

\par\end{centering}

\begin{centering}
 
\par\end{centering}

\centering{}%
\begin{tabular}{cccccccc}
 & $A_{7}$  & $A_{8}$  & $A_{9}$  & $A_{10}$  & $A_{11}$  & $A_{12}$  & $A_{13}$\tabularnewline
\hline 
$A_{1}$  & -.0035  & -.035  & -.055  & .034  & .035  & .042  & .04\tabularnewline
\hline 
$A_{2}$  & 1  & .18  & .39  & .8  & .81  & .8  & .78\tabularnewline
\hline 
$A_{3}$  & .79  & -.17  & .084  & .89  & .9  & .89  & .89\tabularnewline
\hline 
$A_{4}$  & .99  & .23  & .45  & .82  & .8  & .81  & .8\tabularnewline
\hline 
$A_{5}$  & .86  & .18  & .4  & .94  & .92  & .94  & .93\tabularnewline
\hline 
$A_{6}$  & .95  & .2  & .43  & .77  & .76  & .76  & .75\tabularnewline
\hline 
$A_{7}$  & 1  & .2  & .41  & .8  & .8  & .79  & .78\tabularnewline
\hline 
$A_{8}$  &  & 1  & .68  & -.15  & -.18  & -.15  & -.16\tabularnewline
\hline 
$A_{9}$  &  &  & 1  & .18  & .14  & .17  & .16\tabularnewline
\hline 
$A_{10}$  &  &  &  & 1  & .98  & 1  & .99\tabularnewline
\hline 
$A_{11}$  &  &  &  &  & 1  & .97  & .97\tabularnewline
\hline 
$A_{12}$  &  &  &  &  &  & 1  & 1\tabularnewline
\hline 
$A_{13}$  &  &  &  &  &  &  & 1\tabularnewline
\hline 
\end{tabular}
\end{table}

\begin{table}[ht]
\noindent \begin{centering}
\caption{\label{table:correlation_qi_12}Correlations to the confidential attributes in the probabilistically
$k$-anonymous data set obtained using MDAV-SWAP with $k=12$ (``Census''
data set)}

\par\end{centering}

\centering{} %
\begin{tabular}{cccccccc}
 & $A_{7}$  & $A_{8}$  & $A_{9}$  & $A_{10}$  & $A_{11}$  & $A_{12}$  & $A_{13}$\tabularnewline
\hline 
$A_{1}$  & -.0011  & -.028  & -.034  & .032  & .033  & .036  & .032\tabularnewline
\hline 
$A_{2}$  & .81  & .089  & .17  & .69  & .67  & .69  & .67\tabularnewline
\hline 
$A_{3}$  & .42  & -.020  & .091  & .48  & .47  & .48  & .43\tabularnewline
\hline 
$A_{4}$  & .77  & .093  & .18  & .68  & .65  & .68  & .67\tabularnewline
\hline 
$A_{5}$  & .72  & .086  & .16  & .80  & .76  & .79  & .77\tabularnewline
\hline 
$A_{6}$  & .64  & .086  & .14  & .54  & .52  & .54  & .52\tabularnewline
\hline 
$A_{7}$  & 1  & .12  & .17  & .69  & .67  & .66  & .65\tabularnewline
\hline 
$A_{8}$  &  & 1  & .19  & -.013  & -.022  & -.042  & -.011\tabularnewline
\hline 
$A_{9}$  &  &  & 1  & .11  & .10  & .10  & .13\tabularnewline
\hline 
$A_{10}$  &  &  &  & 1  & .76  & .81  & .87\tabularnewline
\hline 
$A_{11}$  &  &  &  &  & 1  & .72  & .70\tabularnewline
\hline 
$A_{12}$  &  &  &  &  &  & 1  & .77\tabularnewline
\hline 
$A_{13}$  &  &  &  &  &  &  & 1\tabularnewline
\hline 
\end{tabular}
\end{table}

\begin{table}[ht]
\noindent \begin{centering}
\caption{\label{table:correlation_conf_12}Correlations to the confidential attributes in the probabilistically
$k$-anonymous data set obtained using IR-SWAP with $k=12$ (``Census''
data set)}

\par\end{centering}

\centering{} %
\begin{tabular}{cccccccc}
 & $A_{7}$  & $A_{8}$  & $A_{9}$  & $A_{10}$  & $A_{11}$  & $A_{12}$  & $A_{13}$\tabularnewline
\hline 
$A_{1}$  & .0041  & -.017  & -.018  & .031  & .038  & .039  & .038\tabularnewline
\hline 
$A_{2}$  & .98  & .13  & .20  & .73  & .71  & .72  & .70\tabularnewline
\hline 
$A_{3}$  & .44  & -.12  & -.041  & .56  & .55  & .56  & .55\tabularnewline
\hline 
$A_{4}$  & .98  & .19  & .27  & .73  & .68  & .71  & .69\tabularnewline
\hline 
$A_{5}$  & .78  & .26  & .26  & .90  & .85  & .88  & .86\tabularnewline
\hline 
$A_{6}$  & .79  & -.12  & .21  & .59  & .57  & .57  & .56\tabularnewline
\hline 
$A_{7}$  & 1  & .16  & .23  & .72  & .69  & .71  & .69\tabularnewline
\hline 
$A_{8}$  &  & 1  & .42  & -.17  & -.19  & -.17  & -.17\tabularnewline
\hline 
$A_{9}$  &  &  & 1  & .077  & .063  & .075  & .080\tabularnewline
\hline 
$A_{10}$  &  &  &  & 1  & .95  & .98  & .96\tabularnewline
\hline 
$A_{11}$  &  &  &  &  & 1  & .91  & .89\tabularnewline
\hline 
$A_{12}$  &  &  &  &  &  & 1  & .97\tabularnewline
\hline 
$A_{13}$  &  &  &  &  &  &  & 1\tabularnewline
\hline 
\end{tabular}
\end{table}

To obtain results with more statistical significance, we ran MDAV-ID,
MDAV-SWAP and IR-SWAP 100 times. In Table~\ref{table:abs_diff_variability}
we report the mean and the standard deviation of the absolute value
of the difference between the correlations to the confidential attributes
in the anonymized data set and the original data set. The better the
data quality of the anonymized data set, the closer the mean and standard
deviation to zero. A value close to one for the mean means that most
of the dependencies between attributes have been lost.

Table~\ref{table:abs_diff_variability} confirms what had been observed
from the previous tables based on a single run: MDAV-SWAP offers better
quality than MDAV-ID, but IR-SWAP clearly offers the best quality
among the three methods compared. For example, for the data set tried,
similar data quality is obtained using MDAV-ID with $k=11$, MDAV-SWAP
with $k=25$ and IR-SWAP with $k=300$. Hence, probabilistic $k$-anonymity
turns out to be much more information-preserving than $k$-anonymity.

\begin{table}[ht]
\begin{centering}
\caption{\label{table:abs_diff_variability} Mean and standard deviation of the absolute value of the difference
between the correlations in the original and the anonymized data sets
(``Census'' data set)}

\par\end{centering}

\begin{centering}

\par\end{centering}

\centering{}%
\begin{tabular}{|c|c|c|c|c|c|c|}
\multicolumn{1}{c}{} & \multicolumn{2}{c}{MDAV-ID} & \multicolumn{2}{c}{MDAV-SWAP} & \multicolumn{2}{c}{IR-SWAP}\tabularnewline
\hline 
k  & mean  & st.dev.  & mean  & st.dev.  & mean  & st.dev.\tabularnewline
\hline 
5  & .055  & .064  & .037  & .045  & .0021  & .0041\tabularnewline
\hline 
7  & .062  & .071  & .048  & .056  & .0022  & .0039\tabularnewline
\hline 
9  & .069  & .078  & .055  & .064  & .0028  & .0049\tabularnewline
\hline 
11  & .078  & .085  & .061  & .070  & .0038  & .0068\tabularnewline
\hline 
25  & .11  & .11  & .091  & .093  & .0061  & .012\tabularnewline
\hline 
50  & .14  & .13  & .13  & .12  & .010  & .020\tabularnewline
\hline 
100  & .17  & .15  & .19  & .17  & .020  & .030\tabularnewline
\hline 
200  & .29  & .27  & .31  & .28  & .044  & .047\tabularnewline
\hline 
300  & .38  & .39  & .37  & .34  & .087  & .071\tabularnewline
\hline 
\end{tabular}
\end{table}

\subsection{``EIA'' data set}

Empirical results for the ``EIA'' data
set are more succinctly presented, because their
interpretation is parallel to the one of the ``Census'' results. 
Table~\ref{table:abs_diff_variability_EIA}
reports an evaluation for the ``EIA'' data set analogous to the
one reported in Table~\ref{table:abs_diff_variability} for the ``Census''
data set. Like before, we observe that MDAV-SWAP performs better than
MDAV-ID, but IR-SWAP is clearly the best of the three methods.

\begin{table}[ht]
\begin{centering}
\caption{\label{table:abs_diff_variability_EIA} Mean and standard deviation of the absolute value of the difference
between the correlations in the original and the anonymized data sets
(``EIA'' data set)}

\par\end{centering}

\begin{centering}

\par\end{centering}

\centering{}%
\begin{tabular}{|c|c|c|c|c|c|c|}
\multicolumn{1}{c}{} & \multicolumn{2}{c}{MDAV-ID} & \multicolumn{2}{c}{MDAV-SWAP} & \multicolumn{2}{c}{IR-SWAP}\tabularnewline
\hline 
k  & mean  & st.dev.  & mean  & st.dev.  & mean  & st.dev.\tabularnewline
\hline 
5  & .018  & .017  & .017  & .035  & .00064  & .00075\tabularnewline
\hline 
7  & .02  & .017  & .024  & .05  & .0012  & .0018\tabularnewline
\hline 
9  & .034  & .031  & .028  & .053  & .0015  & .0018\tabularnewline
\hline 
11  & .039  & .036  & .029  & .052  & .0019  & .0023\tabularnewline
\hline 
25  & .085  & .078  & .043  & .081  & .0063  & .0072\tabularnewline
\hline 
50  & .13  & .12  & .053  & .089  & .011  & .011\tabularnewline
\hline 
100  & .15  & .14  & .058  & .092  & .029  & .037\tabularnewline
\hline 
200  & .19  & .18  & .09  & .11  & .093  & .074\tabularnewline
\hline 
300  & .2  & .18  & .12  & .13  & .14  & .091\tabularnewline
\hline 
\end{tabular}
\end{table}

\section{Conclusions\label{conc}}

$k$-Anonymity is a broadly used privacy property that focuses on
protection against identity disclosure. In a $k$-anonymous data set,
for each record there are at least $k-1$ other records sharing the
same values for all the quasi-identifier attributes. Hence, enforcing
$k$-anonymity implies variability loss and therefore, quality loss.
This is especially serious in a scenario with informed intruders
who know the values of some confidential attributes: the confidential
attributes known by the informed intruder can be viewed as additional
quasi-identifiers. The more quasi-identifier attributes, the more
data quality loss is caused by $k$-anonymity.

To mitigate the above problem, we have introduced the notion of probabilistic
$k$-anonymity. Like standard $k$-anonymity, probabilistic $k$-anonymity
guarantees that the probability of correct re-identification is  $1/k$ at
most, but without explicitly requiring that the quasi-identifier
attributes take identical values within each group of $k$ records.
We have presented two computational methods to reach probabilistic
$k$-anonymity, based on microaggregation and swapping. Experimental
work shows that, for a fixed re-identification probability $1/k$,
the new methods are much more quality-preserving than standard $k$-anonymity
enforcement.

The method based on individual ranking microaggregation is particularly
interesting. It builds on the fact that applying a permutation over
the quasi-identifiers and leaving the confidential attributes unmodified
is equivalent to applying the opposite permutation to the 
confidential attributes
and leaving the quasi-identifiers unmodified. Switching the focus
to confidential attributes has several important benefits. First,
it prevents
informed intruders from using confidential information to improve
the re-identification; the value of each confidential attribute must
be dissociated from all the other attributes. This becomes possible
after switching the focus to confidential attributes because the permutation
only affects the attribute being protected. Second, it allows using 
a different partition for  
each confidential attribute, thereby boosting accuracy and utility. 
Obviously, the reduction in
the diversity in the confidential attribute increases the chances
of attribute disclosure. Selecting a non-optimal partition (as 
done in $k$-anonymity) does not seem to be the proper approach. To
increase the variability we advocate to increase $k$, or enforce
additional criteria such as $l$-diversity or $t$-closeness. Third, 
some attributes are usually more disclosive than others. The ability
to generate a different partition for each confidential attribute
offers the possibility of selecting a different level of disclosure
limitation (the $k$ parameter) for each of the confidential attributes.

While $k$-anonymity is, in principle, only concerned with the cloaking
of individuals within groups of $k$ or more individuals (thus preventing
re-identification), the level of disclosure limitation for the confidential
attribute derives from the variability within the groups of indistinguishable
records. The level of variability is not determined by the parameter
$k$ selected; it may even happen that all the records in a group
share the same value for a confidential attribute. The criterion to
generate the partition in $k$-anonymity is based on the values of
the quasi-identifier attributes, but there is no way to determine
the optimal partition for an arbitrary user: a user
may be very interested in preserving one specific attribute that may
be meaningless for another user. When using individual ranking for
probabilistic $k$-anonymity, we advocate for the best partition for
each confidential attribute (grouping records according to the value
of the confidential attribute), even if that means that we get the
least variability (the least protection). It is obvious that the parameter
$k$ must be much larger that in regular $k$-anonymity to prevent
disclosing confidential information; however, this approach has a
great advantage: the value of $k$ is related to the level of confidentiality.

Future research will combine probabilistic $k$-anonymity with other
properties like $l$-diversity or $t$-closeness in view of reducing
the quality loss incurred to protect against attribute disclosure.
As we deal with each confidential attribute separately, the 
enforcement of additional properties
($l$-diversity and $t$-closeness) is relatively easy to achieve.

\lhead[\chaptername~\thechapter]{\rightmark}

\rhead[\leftmark]{}

\lfoot[\thepage]{}

\cfoot{}

\rfoot[]{\thepage}

\chapter{Optimal data-independent noise for $\varepsilon$-differential privacy}

To maximize the utility of the results provided by $\varepsilon$-differential
privacy, the magnitude of the random noise should be as small as possible.
Some criticisms have appeared to the data utility that results from
using Laplace noise addition as the mechanism to obtain differential
privacy~\cite{Muralidhar2010,Sarathy2010,Sarathy2011}. The question
of the optimality of Laplace noise addition arises: is it possible
to achieve $\varepsilon$-differential privacy with substantially
more data utility using other noise distributions?

Our goal is to determine the optimal distribution to achieve differential
privacy with data-independent random noise. We will limit our discussion
to absolutely continuous random noise distributions, as they provide
the greatest level of generality. Similar results can also be obtained
for discrete random noise; however, this type of noise is only applicable
in very specific circumstances.

By using an optimal noise, the distortion required to achieve a certain
level $\varepsilon$ of differential privacy is minimized. This may
lead to under-protection if the disclosure limitation offered by $\varepsilon$-differential
privacy is measured by how much noise is added to the data (as in
traditional noise addition for disclosure control, see~\cite{hundepool2012}),
rather than by the theoretical guarantee offered by differential privacy
in terms of $\varepsilon$. In what follows, we assume that a protection
level $\varepsilon$ is chosen such that the theoretical guarantee
provides sufficient protection.

We propose a general optimality criterion based on the concentration
of the probability mass of the noise distribution around zero, and
we show that any noise optimal under this criterion must be optimal
under any other sensible criterion. We also show that the Laplace
distribution, commonly used for noise in $\varepsilon$-differential
privacy, is not optimal, and we build the optimal data-independent
noise distribution. We compare the Laplace and the optimal data-independent
noise distributions. For univariate query functions, both introduce
a similar level of distortion; for multivariate query functions, optimal
data-independent noise offers responses with substantially better
data quality.

The contents of this chapter have been accepted for publication in~\cite{infsciences}.

\section{Optimal data-independent noise\label{sec:Optimal-Random-Noise}}

To improve the utility of the outputs provided by an $\varepsilon$-differentially
private access mechanism, the random noise must be adjusted to minimize
the distortion to the real query result. When using Laplace noise,
the scale parameter is set to $\Delta f/\varepsilon$ (see Section~\ref{back_dp});
this yields a noise distribution optimal within the class of Laplacian
noises, because a smaller scale parameter would no longer satisfy
$\varepsilon$-differential privacy. However, the question of the
optimality of the Laplace distribution itself within all possible
noise distributions has not been addressed in the literature: can
we improve the utility of the output by using a different noise distribution?
The answer to this question is deferred until Section~\ref{sec:Laplace_not_optimal}.
In this section we tackle a more fundamental issue: the concept of
optimality for a random noise.

Deciding which among a pair of random noises, $Y_{1}$ and $Y_{2}$,
leads to greater utility is a question that may depend on the users'
preferences. The goal of this section is to come up with an optimality
notion that is independent from the users' preferences: if $Y_{1}$
is better than $Y_{2}$ according to our criterion, any rational user
must prefer $Y_{1}$ to $Y_{2}$. Later, in Section~\ref{sec:Optimal_univariate},
we will determine the form of all optimal random noises that provide
$\varepsilon$-differential privacy to a given query function. 

Let $Y_{1}$ and $Y_{2}$ be two random noise distributions. If $Y_{1}$
can be constructed from $Y_{2}$ by moving some of the probability
mass towards zero (but without going beyond zero), then $Y_{1}$ must
always be preferred to $Y_{2}$. The reason is that the probability
mass of $Y_{1}$ is more concentrated around zero, and thus the distortion
introduced by $Y_{1}$ is smaller. A rational user always prefers
less distortion and, therefore, prefers $Y_{1}$ to $Y_{2}$.

We use the notation $\langle0,\alpha\rangle$, where $\alpha\in\mathbb{R}$,
to denote the interval $[0,\alpha]$ when $\alpha\ge0$, and the interval
$[\alpha,0]$ when $\alpha\le0$. If $Y_{1}$ can be constructed from
$Y_{2}$ by moving some of the probability mass towards zero, it must
be $P(Y_{1}\in\langle0,\alpha\rangle)\ge P(Y_{2}\in\langle0,\alpha\rangle)$
for any $\alpha\in\mathbb{R}$: otherwise, some of the probability
mass that $Y_{2}$ had in $\langle0,\alpha\rangle$ would have been
moved outside $\langle0,\alpha\rangle$, which is not possible. This
leads to the following definition.
\begin{defn}
\label{def:smaller_noise}Let $Y_{1}$ and $Y_{2}$ be two random
noise distributions on $\mathbb{R}$. We say that $Y_{1}$ is smaller
(or better) than $Y_{2}$, denoted by $Y_{1}\le Y_{2}$, if $P(Y_{1}\in\langle0,\alpha\rangle)\ge P(Y_{2}\in\langle0,\alpha\rangle)$
for any $\alpha\in\mathbb{R}$. We say that $Y_{1}$ is strictly smaller
than $Y_{2}$, denoted by $Y_{1}<Y_{2}$, if some of the previous
inequalities are strict. 
\end{defn}
For $\alpha=(\alpha_{1},\ldots,\alpha_{d})\in\mathbb{R}^{d}$, we
use $\langle0,\alpha\rangle$ to denote the set $\langle0,\alpha_{1}\rangle\times\ldots\times\langle0,\alpha_{d}\rangle$.
Consider a set $S\subset\mathbb{R}^{d}$ such that for every point
$x\in S$ we have $\langle0,x\rangle\subset S$, and a pair of random
noises $Y_{1}=(Y_{1}^{1},\ldots,Y_{d}^{1})$ and $Y_{2}=(Y_{1}^{2},\ldots,Y_{d}^{2})$
such that $Y_{1}$ can be constructed from $Y_{2}$ by moving some
probability mass towards zero. It is obvious that we must have $P(Y_{1}\in S)\ge P(Y_{2}\in S)$:
if that was not the case, it would mean that some of the probability
mass that $Y_{2}$ had in $S$ has been moved outside $S$, which
is not possible because of the form of $S$. This leads to the definition
for the multivariate case. 
\begin{defn}
\label{def:smaller_multivariate_noise} Let $Y_{1}$ and $Y_{2}$
be two random noise distributions on $\mathbb{R}^{d}$. We say that
$Y_{1}$ is smaller (or better) than $Y_{2}$, denoted by $Y_{1}\le Y_{2}$,
if $P(Y_{1}\in S)\ge P(Y_{2}\in S)$ for every set $S\subset\mathbb{R}^{d}$
such that for any $x\in S$ we have $\langle0,x\rangle\subset S$.
We say that $Y_{1}$ is strictly smaller than $Y_{2}$, denoted by
$Y_{1}<Y_{2}$, if some of the previous inequalities are strict. 
\end{defn}
Definitions~\ref{def:smaller_noise} and~\ref{def:smaller_multivariate_noise}
induce an order relationship between random noises. We use that order
relationship to define the concept of optimal random noise. 
\begin{defn}
\label{def:optimal_noise} A random noise distribution $Y_{1}$ is
optimal within a class $\mathcal{C}$ of random noise distributions
if $Y_{1}$ is minimal within $\mathcal{C}$; in other words, there
is no other random $Y_{2}\in\mathcal{C}$ such that $Y_{2}<Y_{1}$. 
\end{defn}
As stated in the previous definition, the concept of optimality is
relative to a specific class $\mathcal{C}$ of random noise distributions.
In Section~\ref{sec:Optimal_univariate} we will determine the form
of all optimal symmetric random noise distributions that provide $\varepsilon$-differential
privacy to a specific query function $f$; to do so, we will take
$\mathcal{C}$ to be the class of all symmetric random noise distributions
that provide $\varepsilon$-differential privacy for $f$.

\section{Characterization of differential privacy in terms of the density
function\label{sec:Characterization}}

To build the optimal data-independent random noise distribution satisfying
differential privacy, we will have to analyze the properties that
such a distribution must satisfy. The first step to perform this analysis
is to express the condition in the definition of differential privacy
in terms of the random noise. Assuming a data-independent random noise
$\N$, if we let $\kappa=f+\N$ then Inequality~(\ref{eq:dp_dwork})
becomes 
\[
P\left(\N\in S-f(D)\right)\leq e^{\varepsilon}P\left(\N\in S-f(D')\right)
\]
As this inequality holds for all $S$, we can think of $S$ as being
of the form $S+f(D)$. 
\begin{equation}
P\left(\N\in S\right)\leq e^{\varepsilon}P\left(\N\in S+(f(D)-f(D'))\right)\label{eq:2}
\end{equation}

For the case of absolutely continuous random noise, the characterization
in Inequality~(\ref{ref:eq:2}) can be expressed in terms of the
density function $f_{\N}$ of $\N$. To simplify the notation, we
will assume that $\N$ takes values in $\mathbb{R}$. Consider that
$f_{\N}$ is continuous except for a finite or countable set of removable
discontinuities and a finite or countable set of jump discontinuities.
If the set of jump discontinuities is countable, we will assume that
it has no accumulation points; that is, around any jump discontinuity
point in $\mathbb{R}$ we assume we can find an interval with no other
jump discontinuity points. If $f_{\N}$ has removable discontinuities
we will modify $f_{\N}$ to remove them. As we are modifying $f_{\N}$
in at most a countable set, the modification will not affect the distribution
of $\N$.

Let $x$ be a continuity point of $f_{\N}$ such that $x+d$ is also
a continuity point, where $d=f(D)-f(D')$ for some data sets $D$
and $D'$ that differ in one row. Let $I$ be an interval of size
$m$ centered at $x$ such that $f_{\N}$ is continuous in $I$ and
$I+d$. We know that such $I$ exists because there are no accumulation
points in the set of jump discontinuities. We can upper- and lower-bound
the integrals by multiplying the maximum and minimum by the size of
the interval: 
\[
\begin{array}{c}
m\times\inf_{I}(f_{\N})\leq\int_{I}f_{\N}\leq m\times\sup_{I}(f_{\N})\\
m\times\inf_{I+d}(f_{\N})\leq\int_{I+d}f_{\N}\leq m\times\sup_{I+d}(f_{\N})
\end{array}
\]
As $f_{\N}$ is continuous in $I$, the limit of $\inf_{I}(f_{\N})$
and $\sup_{I}(f_{\N})$ as the size $m$ of $I$ goes to zero is $f_{\N}(x)$.
In the same way, as $f_{\N}$ is continuous in $I+d$, the limit of
$\inf_{I+d}(f_{\N})$ and $\sup_{I+d}(f_{\N})$ as $m$ tends to 0
is $f_{\N}(x+d)$. Dividing both expressions by $m$ and taking limits
as $m$ goes to zero, we have

\[
\begin{array}{c}
f_{\N}(x)\leq lim_{m\rightarrow0}\frac{\int_{I}f_{\N}}{m}\leq f_{\N}(x)\\
f_{\N}(x+d)\leq lim_{m\rightarrow0}\frac{\int_{I+d}f_{\N}}{m}\leq f_{\N}(x+d)
\end{array}
\]
Hence, combining the above limits and Expression~(\ref{ref:eq:2})
we have 
\[
\begin{array}{ccc}
\frac{\int_{I}f_{\N}}{m} & \leq & e^{\varepsilon}\times\frac{\int_{I+d}f_{\N}}{m}\\
\downarrow &  & \downarrow\\
f_{\N}(x) &  & e^{\varepsilon}\times f_{\N}(x+d)
\end{array}
\]
Thus for all $x\in\mathbb{R}$ continuity point of $f_{\N}$, if $x+d$
is a continuity point we have 
\begin{equation}
f_{\N}(x)\leq e^{\varepsilon}\times f_{\N}(x+d),\quad d=f(D)-f(D')\label{eq:3}
\end{equation}

It is immediate to see that, if Inequality~(\ref{eq:3}) holds, by
integrating it over a set we recover Inequality~(\ref{ref:eq:2}).
Hence, Inequality~(\ref{eq:3}) is in fact an equivalent definition
of differential privacy for the case of a.c. random noise.

\section{Non-optimality of the Laplace noise\label{sec:Laplace_not_optimal}}

Since the inception of differential privacy up to now~\cite{Dwork2006a,Dwork2011},
Laplace noise addition has been proposed as a method to achieve $\varepsilon$-differential
privacy for an arbitrary function $f$ in terms of its $L_{1}$-sensitivity.
Also, as we said in the introduction, this practice has raised some
criticisms.

In this section we show, for a univariate function $f$ with values
in $\mathbb{R}$, that the Laplace distribution is not optimal in
the sense of Definition~\ref{def:optimal_noise}. To that end, we
build another distribution, based on the Laplace distribution, that
still fulfills the conditions of differential privacy and has its
probability mass more concentrated towards zero, that is, it is strictly
smaller than Laplace according to Definition~\ref{def:smaller_noise}.
Although the distribution we build is optimal, we leave the formal
proof of this assertion for Section~\ref{sec:Optimal_univariate}.

The basic idea is to concentrate the probability mass around 0 as
much as possible. This can only be done to a certain extent, because
Inequality~(\ref{eq:3}) limits our capability to do so. For example,
increasing the value of the density at a point $x$ may increase the
minimum value that $f_{\N}$ may take in the interval $[x-\Delta f,\, x+\Delta f]$.

In the construction of the distribution we will split the domain of
$f_{\N}$ into intervals of the form $[i\Delta f,\,(i+1)\Delta f]$
where $i\in\mathbb{Z}$. For each interval we will redistribute the
probability mass that $f_{X}$ assigns to that interval. The new density
function $\tilde{f_{\N}}$ will take only two values (see \figurename \ref{fig:laplace_no_optimal}):
$max_{[i\Delta f,\,(i+1)\Delta f]}\, f_{X}$ at the portion of the
interval closer to zero and $min{}_{[i\Delta f,\,(i+1)\Delta f]}\, f_{X}$
at the portion of the interval farther from zero. The result is an
absolutely continuous distribution where the probability mass has
clearly been moved towards zero. We still have to check that it fulfills
the conditions of $\varepsilon$-differential privacy.

\begin{figure}[ht]
\begin{centering}
\includegraphics[width=10cm]{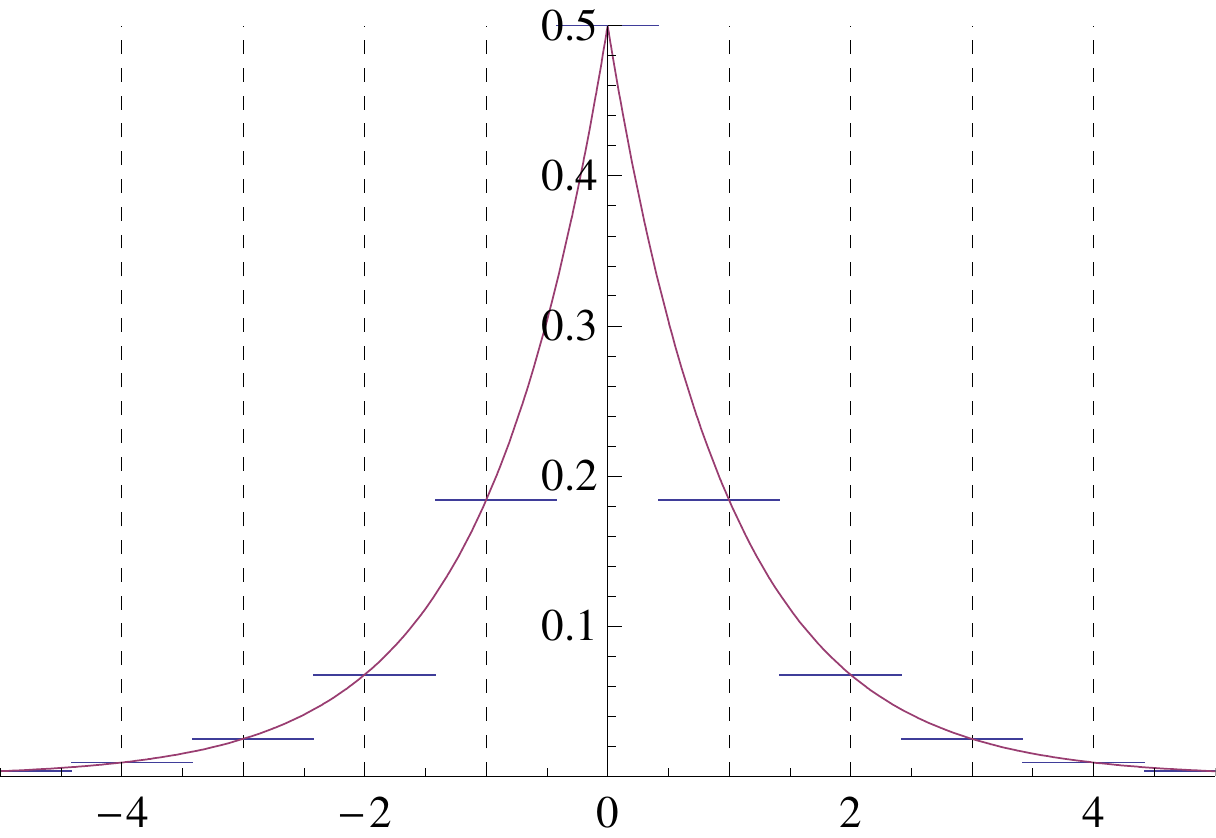} 
\par\end{centering}

\caption{Construction of the new distribution based on the Laplace(0,1) distribution\label{fig:laplace_no_optimal}}
\end{figure}

To simplify, we will detail the argument only for intervals at the
right of zero (positive reals); the argument for intervals at the
left of zero is symmetrical. The probability mass at $[i\Delta f,\,(i+1)\Delta f]$
is $e^{-i\varepsilon}\frac{1-e^{-\varepsilon}}{2}$. The maximum value
of the density of the Laplace distribution, $\frac{\varepsilon e^{-i\varepsilon}}{2\Delta f}$,
occurs at the beginning of the interval and the minimum, $\frac{\varepsilon e^{-(i+1)\varepsilon}}{2\Delta f}$,
occurs at the end. Let us determine the size $m_{i}$ of the interval
portion where the new density will be set to the maximum.

Since the probability mass of the interval must be preserved, we have
\[
\frac{\varepsilon e^{-i\varepsilon}}{2\Delta f}m_{i}+\frac{\varepsilon e^{-(i+1)\varepsilon}}{2\Delta f}(\Delta f-m_{i})=e^{-i\varepsilon}\frac{1-e^{-\varepsilon}}{2}
\]
By solving for $m_{i}$ in the above equality, we obtain: 
\[
m_{i}=\frac{\Delta f}{\varepsilon(1-e^{-\varepsilon})}(1-e^{-\varepsilon}-\varepsilon e^{-\varepsilon})
\]

The important fact about $m_{i}$ is that it does not depend on $i$.
Also, note that the maximum density of the current interval is equal
to the minimum density of the previous interval. This way, by joining
the portion of the previous interval which evaluates to the minimum
with the portion of the current interval which evaluates to the maximum,
we obtain an interval of size $(\Delta f-m_{i-1})+m_{i}=(\Delta f-m_{i})+m_{i}=\Delta f$
which evaluates to a constant density value (such joined intervals
are depicted as horizontal segments in \figurename~\ref{fig:laplace_no_optimal}.
Thus, except for the maximum of the first interval, we have split
the domain of the density function into intervals of size $\Delta f$
such that the density function evaluates to $\frac{\varepsilon e^{-i\varepsilon}}{2\Delta f}$.
This clearly satisfies the density-based characterization of differential
privacy specified by Inequality~(\ref{eq:3}).

\section{Optimal noise for univariate queries\label{sec:Optimal_univariate}}

Section~\ref{sec:Laplace_not_optimal} has shown that the Laplace
noise distribution is not optimal to achieve differential privacy.
A new distribution has been built that satisfies differential privacy
and has its probability mass more concentrated towards zero. This
section will determine the optimal data-independent absolutely continuous
random noise distribution to achieve $\varepsilon$-differential privacy
for any univariate function with finite $L_{1}$-sensitivity. Optimal
noise distributions need not be symmetric; however, we focus on the
symmetric case, because it is the most usual one.

Showing that optimal absolutely continuous noise distributions are
of a certain form requires using some properties that will be stated
as lemmata. Some of the proofs place additional regularity requirements
on the noise distribution, beyond being absolutely continuous. These
additional requirements are hardly a limitation as they are satisfied
by any practical distribution, and can be overlooked if the reader
is not interested in the proofs. In particular, we restrict the discussion
to absolutely continuous random noises, $\N$, whose density function,
$f_{\N}$, is continuous except for a finite or countable set of jump
or removable discontinuities, with the set of jump discontinuities
having no accumulation points. To avoid being unnecessarily cumbersome,
we will not mention this again in the sequel.

It was shown in Section~\ref{sec:Characterization} that for a.c.
noise distributions the definition of $\varepsilon$-differential
privacy can be stated in terms of the density function. Now we show
that if the inequality in terms of the probability function is satisfied
at the extreme, it also must be the case for the inequality in terms
of density functions.
\begin{lem}
\label{lem:2}Let $\N$ be an a.c. noise random variable that provides
$\varepsilon$-differential privacy to a function $f$ with a given
$L_{1}$-sensitivity. Consider an interval $I=\left[i_{0},i_{1}\right]\subset\mathbb{R}$.
Then $P(\N\in I+\Delta f)=e^{-\varepsilon}P(\N\in I)$ if and only
if $f_{\N}(x+\Delta f)=e^{-\varepsilon}f_{\N}(x),\;\forall x\in I$,
except at those points $x\in I$ such that $f_{\N}$ is not continuous
at $x$ or at $x+\Delta f$. Similarly, $P(\N\in I-\Delta f)=e^{-\varepsilon}P(\N\in I)$
if and only if $f_{\N}(x-\Delta f)=e^{-\varepsilon}f_{\N}(x),\;\forall x\in I$,
except at those points $x\in I$ such that $f_{\N}$ is not continuous
at $x$ or at $x-\Delta f$. \end{lem}
\begin{proof}
We will prove the first claim; the second one is completely symmetric.
The proof of $(\Longleftarrow)$ is straightforward by computing the
probability as the integral of the density function. We will focus
on the $\left(\Longrightarrow\right)$ implication. By the $\varepsilon$-differential
privacy condition we know that $f_{\N}(x+\Delta f)\geq e^{-\varepsilon}f_{\N}(x)$.
Assuming that the implication does not hold, a continuity point $a\in I$
exists such that $f_{\N}(a+\Delta f)>e^{-\varepsilon}f_{\N}(a)$.
Because of the constraints on the set of discontinuity points, an
interval $\left[a_{0},a_{1}\right]\subseteq I$ exists such that $f_{\N}(x+\Delta f)>e^{-\varepsilon}f_{\N}(x)\,\forall x\in\left[a_{0},a_{1}\right]$.
Now we can decompose the probability as follows 
\[
P(\N\in I)=\int_{i_{0}}^{a_{0}}f_{\N}(x)dx+\int_{a_{0}}^{a_{1}}f_{\N}(x)dx+\int_{a_{1}}^{i_{1}}f_{\N}(x)dx
\]
\[
P(\N\in I+\Delta f)=\int_{i_{0}}^{a_{0}}f_{\N}(x+\Delta f)dx+\int_{a_{0}}^{a_{1}}f_{\N}(x+\Delta f)dx+\int_{a_{1}}^{i_{1}}f_{\N}(x+\Delta f)dx
\]
Since $f_{\N}(a+\Delta f)\geq e^{-\varepsilon}f_{\N}(a)$ and, for
$x\in[a_{0},a_{1}]$, $f_{\N}(a+\Delta f)>e^{-\varepsilon}f_{\N}(a)$,
we have $P(\N\in I+\Delta f)>e^{-\varepsilon}P(\N\in I)$ , which
is a contradiction that comes from the assumption that a continuity
point $a\in I$ exists such that $f_{\N}(a+\Delta f)>e^{-\varepsilon}f_{\N}(a)$. 
\end{proof}
We are trying to find the optimal a.c. noise distribution that provides
$\varepsilon$-differential privacy. The goal is to concentrate as
much probability mass around the mean as possible; $\varepsilon$-differential
privacy limits our capability to do so. We will see how the probability
mass must be distributed to achieve the optimal random noise. 
\begin{lem}
\label{lem:3}Let $\N$ be a symmetric a.c. noise random variable
with zero mean that satisfies $\varepsilon$-differential privacy
for a function $f$. If $\N$ is optimal at providing $\varepsilon$-differential
privacy, then for all $i\in\mathbb{\N}$ 
\begin{eqnarray*}
P(\N\in[(i+1)\Delta f,\,(i+2)\Delta f]) & = & e^{-\varepsilon}P(\N\in[i\Delta f,\,(i+1)\Delta f])\\
P(\N\in[-(i+2)\Delta f,\,-(i+1)\Delta f]) & = & e^{-\varepsilon}P(\N\in[-(i+1)\Delta f,\,-i\Delta f])
\end{eqnarray*}

\end{lem}
The second claim is completely symmetric to the first one; a symmetric
distribution that satisfies the first claim will also satisfy the
second one. We will show that, if the claims do not hold, we can build
another distribution that fulfills $\varepsilon$-differential privacy
and has its probability mass more concentrated towards zero.
\begin{proof}
We will assume that the claim for $\N$ does not hold and we will
build another distribution $\widetilde{\N}$ that provides $\varepsilon$-differential
privacy and has $\widetilde{\N}\le\N$. If the claim held, by Lemma~\ref{lem:2},
it would be $f_{\N}(x+\Delta f)=e^{-\varepsilon}f_{\N}(x)$, $\forall x\in\mathbb{R}$
where $x$ and $x+\Delta f$ are continuity points. Let $i_{0}\geq0$
be the index of the first interval $\left[i\Delta f,\left(i+1\right)\Delta f\right]$
such that $f_{\N}(x+\Delta f)=e^{-\varepsilon}f_{\N}(x)$ does not
hold for all $x$ in the interval. Let $\tilde{f_{i_{0}}}$ be the
function defined as follows 
\[
\tilde{f_{i_{0}}}\left(x\right)=\begin{cases}
e^{-\varepsilon}f_{\N}(x+\Delta f) & x\in\left[-\left(i_{0}+1\right)\Delta f,-\Delta f\right]\\
f_{\N}(x) & x\in\left[-\Delta f,\,+\Delta f\right]\\
e^{-\varepsilon}f_{\N}(x-\Delta f) & x\in\left[\Delta f,\left(i_{0}+1\right)\Delta f\right]
\end{cases}
\]
Since $\tilde{f_{i_{0}}}$ has been defined in such a way that the
decrease of the density between points at distance $\Delta f$, as
we move away from zero, is maximum, it is clear that we will have
$f_{\N}>\tilde{f}_{i_{0}}$. As both $f_{\N}$ and $\tilde{f}_{i_{0}}$
are symmetric, we will only consider the points on the right of zero;
the same transformations must be applied to the points on the left.
For each $x\in\left[\Delta f,\left(i_{0}+1\right)\Delta f\right]$
we will consider $e_{x}=f_{\N}\left(x\right)-\tilde{f_{i_{0}}}\left(x\right)$,
the excess density of $f_{\N}$ over $\tilde{f_{i_{0}}}$. We will
build another function $f_{i_{0}}$ by distributing $e_{x}$ among
the points $\{x+i\Delta f:\,0\le i\le i_{0}\}$ in such a way that
the new function concentrates as much as possible around the mean,
and $\varepsilon$-differential privacy is satisfied. The density
added to $\tilde{f}_{i_{0}}$ at $x+i\Delta f$ will be $\alpha_{x}e^{-i\varepsilon}$
where $\alpha_{x}$ is determined by imposing $\sum_{i=0,\ldots,i_{0}}\alpha_{x}e^{-i\varepsilon}=e_{x}$.
Note that $f_{i_{0}}$ still satisfies that images of points at distance
$\Delta f$ exponentially decrease as we move away from zero, that
is $f_{i_{0}}(x+\Delta f)=e^{-\varepsilon}f_{i_{0}}(x)$.

It is important to note that the new function $f_{i_{0}}$ satisfies
$\varepsilon$-differential privacy in the range $\left[-i_{0}\Delta f,i_{0}\Delta f\right]$.
We will show that $\varepsilon$-differential privacy is satisfied
in the interval $[-\Delta f,\Delta f]$; then by using that the images
by $f_{i_{0}}$ of points at distance $\Delta f$ exponentially decrease
as we move away from zero, $\varepsilon$-differential privacy will
be satisfied in $[-i_{0}\Delta f,i_{0}\Delta f]$. In fact we will
only check that $\varepsilon$-differential privacy is satisfied in
$[0,\Delta f]$; if it is so, by the symmetry of $f_{i_{0}}$, differential
privacy will be satisfied in the whole interval $[-\Delta f,\Delta f]$.

We must check that $f_{i_{0}}(x+\delta)\leq e^{\varepsilon}\times f_{i_{0}}(x)$
for all $x\in[0,\Delta f]$ and all $\delta\in[-\Delta f,\Delta f]$.
Let us assume that there exist $x\in\left[0,\Delta f\right]$ and
$\delta\in[-\Delta f,\Delta f]$ such that the condition is not satisfied,
that is, $f_{i_{0}}\left(x+\delta\right)>e^{\varepsilon}f_{i_{0}}\left(x\right)$.
If $x+\delta\in\left[\Delta f,2\Delta f\right]$, by multiplying by
$e^{-\left(i_{0}-1\right)\varepsilon}$ we have that $x+(i_{0}-1)\Delta f$,
the corresponding point in the interval $\left[\left(i_{0}-1\right)\Delta f,i_{0}\Delta f\right]$,
does not fulfill the $\varepsilon$-differential privacy condition,
but this is not possible as we had $f_{\N}(x+i_{0}\Delta f)\le e^{\varepsilon}f_{\N}(x+(i_{0}-1)\Delta f)$
and when building $f_{0}$ we have increased the value at $x+(i_{0}-1)\Delta f$
and decreased the value at $x+i_{0}\Delta f$. If $x+\delta\in\left[0,\Delta f\right]$,
by multiplying by $e^{-i_{0}\varepsilon}$ we have that the corresponding
point in the interval $\left[i_{0}\Delta f,\left(i_{0}+1\right)\Delta f\right]$
does not satisfy the differential privacy condition. This is impossible
as we know that $\tilde{f}_{i_{0}}$ and $f_{\N}$ do satisfy it and
that $f_{i_{0}}$ lies between them; therefore $f_{i_{0}}$ must also
satisfy the differential privacy condition. In the case $x+\delta\in\left[-\Delta f,0\right]$,
the justification is different. The point $-x-\delta$ belongs to
the interval $[0,\Delta f]$ and, by the symmetry of $f_{i_{0}}$,
we have $f_{i_{0}}(-x-\delta)=f_{i_{0}}(x+\delta)$; therefore, as
we have already checked that the condition is satisfied when $x+d\in\left[0,\,\Delta f\right]$,
it must also be satisfied when $x+d\in\left[-\Delta f,0\right]$.

Now we iterate this process and define functions $f_{i},\, i\in\mathbb{N}$.
To be able to do this, it is important to note that, when defining
$f_{i}$, we are reducing the density amount in the interval $\left[i\Delta f,\left(i+1\right)\Delta f\right]$
and that $\tilde{f}_{i+1}$ is defined in $\left[\left(i+1\right)\Delta f,\left(i+2\right)\Delta f\right]$
by reducing the value in the previous interval as much as possible
while still satisfying $\varepsilon$-differential privacy. This means
that $f_{\N}>\tilde{f}_{i+1}$ at $\left[\left(i+1\right)\Delta f,\left(i+2\right)\Delta f\right]$
and thus we can compute the excess and distribute it among the corresponding
points in the previous intervals.

The resulting $\tilde{f}_{\infty}$ satisfies the $\varepsilon$-differential
privacy condition. By construction it also satisfies $f_{\N}(x+\Delta f)=e^{-\varepsilon}f_{\N}(x)\,\forall x\in\mathbb{R}$
which by integration over the desired intervals leads to the claim
of the lemma. Moreover, as all the probability mass translation has
been done towards zero, we have $\widetilde{\N}\le\N$.\end{proof}
\begin{cor}
\label{cor:1}Let $\N$ be a symmetric a.c. noise random variable
with zero mean that provides $\varepsilon$-differential privacy to
a function $f$. If $\N$ is optimal at providing $\varepsilon$-differential
privacy then 
\[
\begin{array}{c}
f_{\N}(x+\Delta f)=e^{-\varepsilon}f_{\N}(x)\quad\forall x\geq0\\
f_{\N}(x-\Delta f)=e^{-\varepsilon}f_{\N}(x)\quad\forall x\leq0
\end{array}
\]
when the points $x$ and $x+\Delta f$ in the first equality above
and $x$ and $x-\Delta f$ in the second equality are continuity points
of $f_{\N}$.\end{cor}
\begin{proof}
The proof follows from Lemmata~\ref{lem:2} and~\ref{lem:3}.
\end{proof}
Now we will show that for any symmetric a.c. noise distribution that
provides $\varepsilon$-differential privacy for a function $f$ we
can find another noise distribution, similar to the one used in the
proof that the Laplace distribution is not optimal, that performs
at least as well according to Definition~\ref{def:smaller_noise}. 
\begin{thm}
\label{thm:1}Let $\N$ be an a.c. noise random variable with zero
mean that provides $\varepsilon$-differential privacy to a query
function $f$. Then there exists a noise random variable $\widetilde{\N}$
with density function $f_{\widetilde{\N}}$ of the form 
\[
f_{\widetilde{\N}}\left(x\right)=\begin{cases}
M_{0}e^{-i\varepsilon} & x\in\left[-d-\left(i+1\right)\Delta f,-d-i\Delta f\right],\, i\in\mathbb{N}\\
M_{0} & x\in\left[-d,0\right]\\
M_{0} & x\in\left[0,d\right]\\
M_{0}e^{-i\varepsilon} & x\in\left[d+i\Delta f,d+\left(i+1\right)\Delta f\right],\, i\in\mathbb{N}
\end{cases}
\]
that provides $\varepsilon$-differential privacy to $f$ and satisfies
$\widetilde{\N}\le\N$ as per Definition~\ref{def:smaller_noise}.\end{thm}
\begin{proof}
We will assume that $\N$ is optimal and that its density function
is not of the form of $f_{\widetilde{\N}}$ for any $M_{0}$ and $d$.
The goal is to build another distribution $\widetilde{\N}$ from $\N$
such that the density $f_{\widetilde{\N}}\left(x\right)$ is as stated
above and satisfies $\widetilde{\N}\leq\N$. Note that, from the definition
of $f_{\widetilde{\N}}(x)$, the condition of $\varepsilon$-differential
privacy immediately holds for $f$.

Since $\N$ fulfills the conditions of Corollary~\ref{cor:1}, we
have 
\[
\begin{array}{c}
f_{\N}(x+\Delta f)=e^{-\varepsilon}f_{\N}(x)\quad\forall x\geq0\\
f_{\N}(x-\Delta f)=e^{-\varepsilon}f_{\N}(x)\quad\forall x\leq0
\end{array}
\]
Now we apply the same procedure we used in Section~\ref{sec:Laplace_not_optimal}
for the Laplace noise. First we split the domain of $f_{\N}$ into
intervals of the form $[i\Delta f,(i+1)\Delta f]$ where $i\in\mathbb{Z}$.
At a given interval, we redistribute the probability mass that $f_{\N}$
assigns to that interval. The new density function $f_{\widetilde{\N}}(x)$
takes only two values: $max_{[i\Delta f,(i+1)\Delta f]}\, f_{\N}$
at the portion of the interval closer to zero and $min_{[i\Delta f,(i+1)\Delta f]}\, f_{\N}$
at the portion of the interval farther from zero. The result is an
absolutely continuous distribution $\widetilde{\N}$ with $\widetilde{\N}\le\N$.

To make sure that the distribution $\widetilde{\N}$ has the specified
form, and thus satisfies $\varepsilon$-differential privacy, it remains
to check that the length of the interval where we assign maximum value
is constant across intervals.

The probability mass at $[i\Delta f,(i+1)\Delta f]$ is $e^{-i\varepsilon}\frac{1-e^{-\varepsilon}}{2}$.
It is clear from $f_{\N}(x+\Delta f)=e^{-\varepsilon}f_{\N}(x)$,
$\forall x\geq0$, that the maximum and the minimum of each interval,
$M_{i}$ and $m_{i}$ respectively, satisfy $M_{i}=e^{-i\varepsilon}M_{0}$
and $m_{i}=e^{-i\varepsilon}m_{o}$. Let $d_{i}$ be the size of the
interval where the new density evaluates to the maximum. We have 
\[
e^{-i\varepsilon}M_{0}\times d_{i}+e^{-i\varepsilon}m_{o}\times(\Delta f-d_{i})=e^{-i\varepsilon}\frac{1-e^{-\varepsilon}}{2}
\]
This formula leads to $d_{i}=\frac{1-e^{-\varepsilon}-2m_{0}\Delta f}{2\left(M_{0}-m_{0}\right)}$
which does not depend on $i$, as we wanted to see.
\end{proof}
Theorem~\ref{thm:1} states that, for any random noise that provides
$\varepsilon$-differential privacy to $f$, we can find another random
noise distribution, of the specified form, that is smaller. However,
we still have to prove that such a distribution is optimal. 
\begin{thm}
\label{thm:2} Let $\N$ be a random noise distribution with a density
function $f_{\N}$ of the form specified in Theorem~\ref{thm:1}.
Then $\N$ is optimal at providing $\varepsilon$-differential privacy. \end{thm}
\begin{proof}
To prove that $\N$ is optimal, we have to show that if we move some
probability mass of $\N$ towards zero then $\varepsilon$-differential
privacy no longer holds. We only show it for the probability mass
to the right of zero; a symmetric argument can be used for the probability
mass to the left of zero.

First of all, we must show that it is not possible to move any probability
mass from an interval $I_{i}=[i\Delta f,(i+1)\Delta f]$ to an interval
$I_{j}=[j\Delta f,(j+1)\Delta f]$ with $0\le j<i$. This is straightforward:
as the density $f_{\N}$ specified in Theorem~\ref{thm:1} has the
maximum decrease rate between consecutive intervals compatible with
the constraints of $\varepsilon$-differential privacy, moving probability
mass from $I_{i}$ to $I_{j}$ would break $\varepsilon$-differential
privacy.

To conclude the proof, we need to check that it is not possible to
redistribute the probability mass within an interval $I_{i}$ so that
it gets closer to zero. Within the interval $I_{i}$, the density
function $f_{\N}$ takes values $M_{0}\exp(-i\varepsilon)$ at $I_{i}^{l}$
(the left portion of the interval) and $M_{0}\exp(-(i+1)\varepsilon)$
at $I_{i}^{r}$ (the right portion of the interval). We cannot move
any probability mass from $I_{i}^{r}$ towards zero, because the density
would go below $M_{0}\exp(-(i+1)\varepsilon)$ and, thus, $\varepsilon$-differential
privacy would not hold. We cannot move any probability mass from $I_{i}^{l}$
towards zero, because the density would go above $M_{0}\exp(-i\varepsilon)$
and, thus, $\varepsilon$-differential privacy would not hold. 
\end{proof}
Although the theorems above are stated in terms of a fixed query function
$f$, the optimal distribution depends only on $\Delta f$; hence,
all query functions with the same $L_{1}$-sensitivity share the same
optimal noise distribution.

The values of $M_{0}$ and $d$ can be freely chosen according to
the user's preferences. In fact the two parameters $M_{0}$ and $d$
of the optimal family of distributions can be reduced to one because
\[
d=\frac{1-e^{-\varepsilon}-2M_{0}e^{-\varepsilon}\Delta f}{2(1-e^{-\varepsilon})M_{0}}
\]

For instance, let us assume that the user prefers to minimize the
noise variance. We compute the variance of candidate optimal distributions
in terms of the parameters $d$ and $M_{0}$, and find the values
that yield the minimum: 
\[
V(Z)=2M_{0}\int_{0}^{d}x^{2}dx+2M_{0}e^{-\varepsilon}\sum_{i=0\ldots\infty}e^{-i\varepsilon}\int_{d+i\Delta f}^{d+\left(i+1\right)\Delta f}x^{2}dx
\]
The variance can be computed by performing the integrals and calculating
the sum of the power series. \figurename~\ref{fig:minimal_variance}
shows the variance obtained in terms of the parameter $d$ for the
case of $\varepsilon=1$ and $\Delta f=1$. In this case, the minimum
is reached at $d=0.416737$ and the variance is 1.9181. This is below
2, the variance of the Laplace noise with scale parameter 1.

\begin{figure}[ht]
\begin{centering}
\includegraphics[width=10cm]{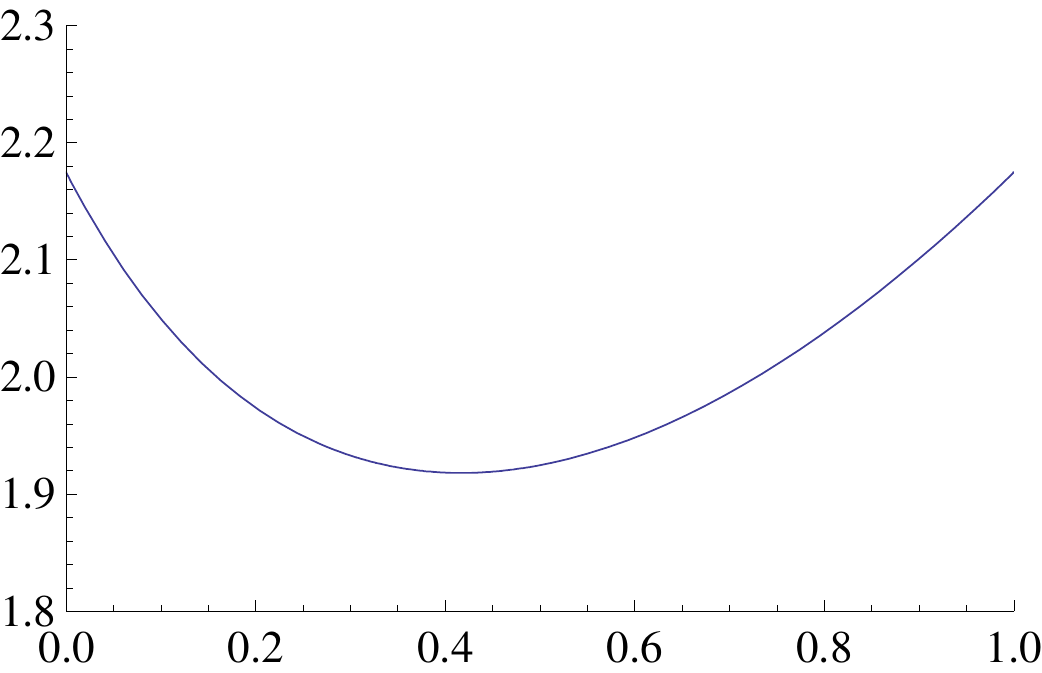} 
\par\end{centering}

\caption{\label{fig:minimal_variance}Variance for $\varepsilon=1$ and $\Delta f=1$}
\end{figure}

Table~\ref{tab:2} shows a comparison of the variance achieved by
the Laplace distribution and the optimal a.c. random noise with minimum
variance, for different values of $\varepsilon$ when $\Delta f=1$.
The table shows that the Laplace variance is only slightly greater
than the minimum variance; we may say that, for a single univariate
query, although the Laplace distribution is not optimal, it is near-optimal.
Therefore, \emph{if the utility of the differentially private answer
to a single univariate query obtained using Laplace noise is poor,
not much improvement can be expected from using a data-independent
variance-optimal random noise distribution}.

\begin{table}[ht]
\caption{\label{tab:2}Variance comparison between Laplace random noise and
a.c. optimal random noise with minimum variance, for $\Delta f=1$}

\centering{}%
\begin{tabular}{cccc}
\hline 
 & $\varepsilon=0.1$  & $\varepsilon=0,5$  & $\varepsilon=1$\tabularnewline
\hline 
\hline 
Laplace distribution  & 200.00  & 8.00  & 2.00\tabularnewline
\hline 
Optimal a.c. noise with min. var.  & 199.92  & 7.92  & 1.92\tabularnewline
\hline 
\end{tabular}
\end{table}

Assume now that the user wants the noise distribution that minimizes
the size of the symmetric confidence interval around the differentially
private query answer that contains the real query value at 95\% confidence
level. In this case, we must solve a minimization problem, as before,
but now the objective function is the size of the confidence interval
in terms of the parameters $d$ and $M_{0}$. \figurename~\ref{fig:minimize_conf_interval}
shows the size of the confidence interval, when $\Delta f=1$ and
$\varepsilon=1$, in terms of parameter $d$. The minimal length for
this case is achieved for $d=0.993$, approximately; in general, however,
the actual value of $d$ where the minimum is reached depends on $\Delta f$
and $\varepsilon$. Table~\ref{tab:3} shows a comparison between
the optimal lengths of the confidence intervals at 95\% confidence
level for several values of $\varepsilon$ when $\Delta f=1$. As
expected, the results obtained from the Laplace distribution are worse
but close to those obtained using the optimal distribution.

\begin{figure}[ht]
\begin{centering}
\includegraphics[width=10cm]{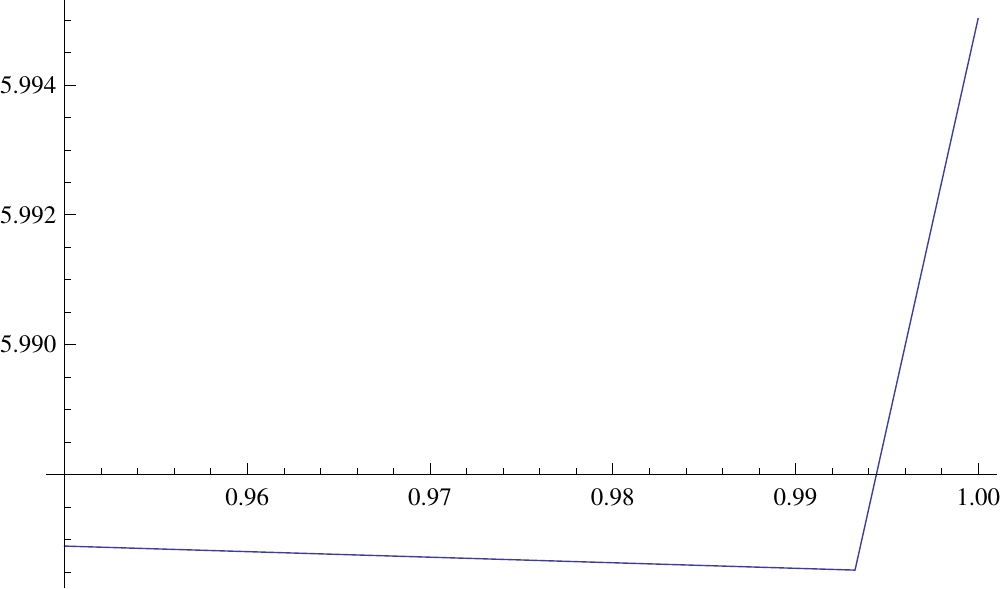} 
\par\end{centering}

\caption{\label{fig:minimize_conf_interval}Size of the 95\% symmetric confidence
interval centered at zero}
\end{figure}

\begin{table}[ht]
\caption{\label{tab:3} Comparison of the size of the symmetric 95\% confidence
interval between Laplace random noise and a.c. optimal random noise
with minimum confidence interval, for $\Delta f=1$}

\centering{}%
\begin{tabular}{cccc}
\hline 
 & $\varepsilon=0.1$  & $\varepsilon=0,5$  & $\varepsilon=1$\tabularnewline
\hline 
\hline 
Laplace distribution  & 59.91  & 11.98  & 5.99\tabularnewline
\hline 
Optimal a.c. noise with min. conf. int.  & 59.91  & 11.97  & 5.98\tabularnewline
\hline 
\end{tabular}
\end{table}

\section{Optimal noise for multivariate queries\label{sec:multi}}

In Section~\ref{sec:Optimal_univariate} we worked out the optimal
a.c. random noise for a query with values in $\mathbb{R}$. We deal
here with multiple queries or with a single query whose response is
a value in $\mathbb{R}^{d}$: both cases are equivalent, because $d$
queries with answers in $\mathbb{R}$ can be viewed as a single query
with answer in $\mathbb{R}^{d}$. Determining the form of all optimal
multivariate a.c. random noises is out of scope; we restrict to a
class of noise distributions whose density consists of several steps
(as was the case for optimal univariate distributions) and show that
they are optimal. The optimal distributions constructed will be shown
to be substantially better than Laplace. Hence, \emph{while Laplace
is near-optimal in the univariate case, in general it is far from
optimal for multivariate or multiple queries.}

We will be less formal here and, to simplify even more, examples will
be presented for the case of two queries/two dimensions, that is,
$d=2$; generalization to arbitrary $d$ is easy.

For the case of a.c. random noise for a single query, it was shown
in Section~\ref{sec:Characterization} that the $\varepsilon$-differential
privacy condition can be expressed in terms of the density function.
The result is easily generalizable to greater dimensions, and therefore
here we can also express the condition in terms of the density function.
\begin{prop}
\label{prop:4}Let $\N=(\N_{1},\ldots,\N_{d})$ be an absolutely continuous
random noise that provides $\varepsilon$-differential privacy to
a query $f:\mathcal{D}\rightarrow\mathbb{R}^{d}$. Then $\varepsilon$-differential
privacy can be characterized in terms of the density function as:
\[
f_{\N}(x)\leq e^{\varepsilon}\times f_{\N}(x+d),\quad d=f(D)-f(D')
\]
for all $x$ and $x+d$ continuity points of $f_{\N}$, where $D$
and $D'$ differ in one row. 
\end{prop}
Similarly to the case of a single univariate query, we will construct
a noise density with several steps, which reaches its maximum all
over a set that contains zero and decreases by a factor $e^{-\varepsilon}$
as we move away from it.

The main difference with other, non-optimal distributions, such as
multivariate Laplace noise, is that the various components (dimensions)
of the random noise do not need to be independent. This allows more
freedom in the definition of the distribution, which we will employ
to achieve a finer calibration to the query function. This is illustrated
below in an example, but prior to it we define a set that will be
repeatedly used in the remainder of this section.
\begin{defn}
Let $f:\mathcal{D}\rightarrow\mathbb{R}^{d}$ be a query function.
The set of differences between neighbor data sets is defined as 
\[
S_{f}=\bigcup_{D,D'}{\langle0,f(D)-f(D')\rangle}
\]
where $D$ and $D'$ data sets that differ in at most one row. 
\end{defn}
The set $S_{f}$ contains all possible variations in $f$ when one
record changes. The boundary of $S_{f}$ can be seen as a generalization
of the $L_{1}$-sensitivity used in the univariate case. Instead of
summarizing the variability of $f$ with a single figure, as $L_{1}$-sensitivity
does, $S_{f}$ keeps track of the maximum variability in each direction.
\begin{example}
\label{exex} Consider a query function $f=(f_{1},f_{2})$ such that
$S_{f}=[-1,1]\times[-1,1]$. From Definition~\ref{def:sensitivity-1},
the $L_{1}$-sensitivity of $f$ is 
\[
\Delta f=\sup_{D,D'}\Vert f(D)-f(D')\Vert_{1}=\sup_{D,D'}(|f_{1}(D)-f_{1}(D')|+|f_{2}(D)-f_{2}(D')|)=1+1=2
\]
As stated in Proposition~\ref{prop:4}, the density of the random
noise, $f_{\N}$, in each of the points of the set $[-1,1]\times[-1,1]$
must be in the range $[e^{-\varepsilon}f_{\N}(0),e^{\varepsilon}f_{\N}(0)]$.
When using independent Laplace-distributed components with zero mean
and $\Delta f/\varepsilon$ scale parameter, the top value for the
density is reached at zero, and it decreases exponentially as we move
away from it. Points with density $e^{-\varepsilon}f_{\N}(0)$ are
those that have $L_{1}$-norm equal to $\Delta f$. \figurename~\ref{fig:Laplace_approximation_1}
depicts $S_{f}$ as a gray shaded box. If all points in $S_{f}$ are
protected with independent Laplace-distributed random noise components,
all points within $[-1,1]\times[-1,1]$ must have density within the
range $[e^{-\varepsilon}f_{\N}(0),f_{\N}(0)]$.

\begin{figure}[ht]
\centering{}\includegraphics[width=14cm]{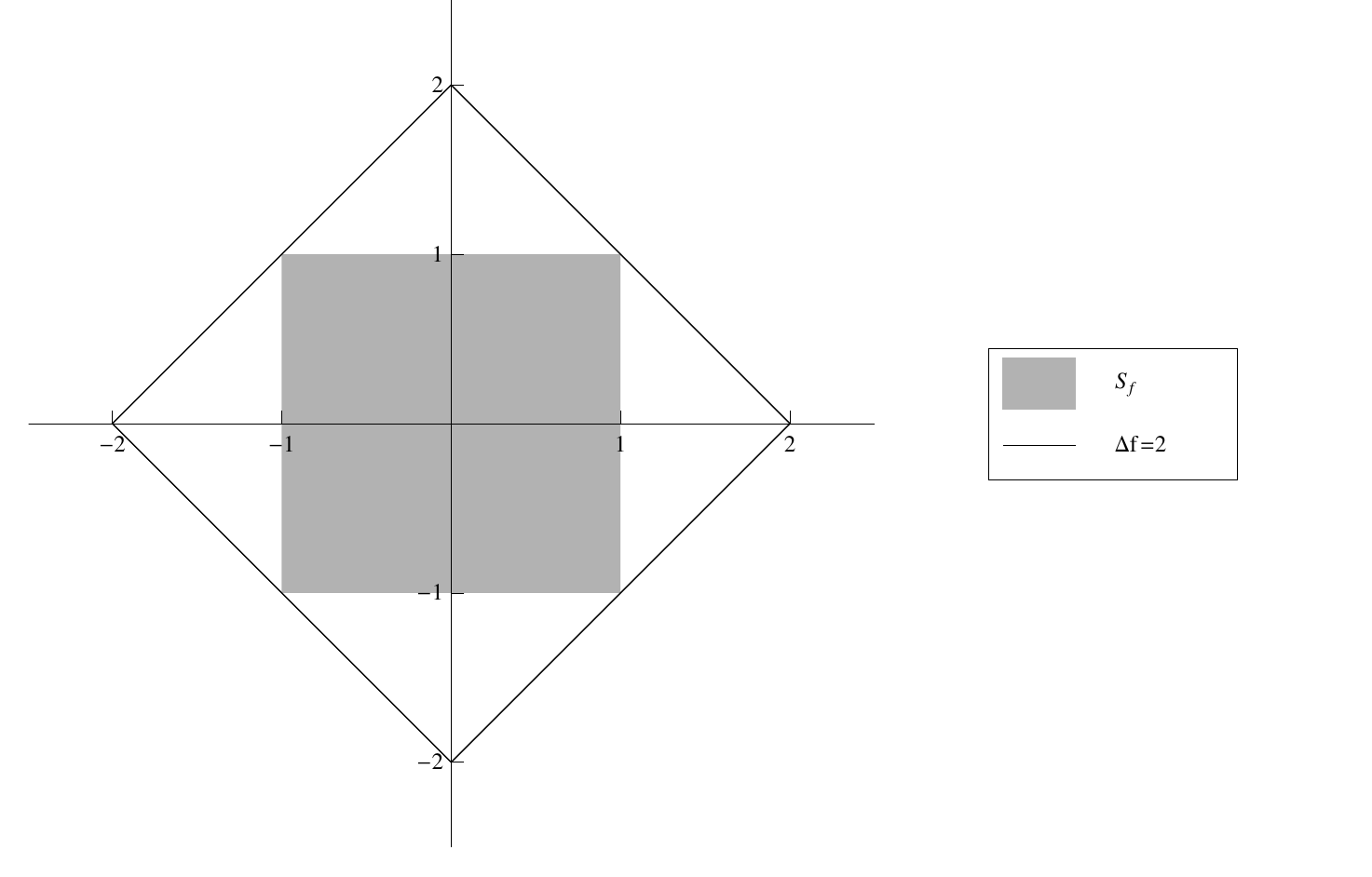}\caption{\label{fig:Laplace_approximation_1} Achieving $\varepsilon$-differential
privacy by Laplace noise addition for $S_{f}=[-1,1]\times[-1,1]$.
The shaded box represents the possible differences in the query result
between data sets that differ in one record. Differential privacy
requires the density of the noise in the shaded box to be within a
factor in $[\exp(-\varepsilon),\exp(\varepsilon)]$ of the density
at zero. The square that encloses the shaded box represents the points
that satisfy the previous condition when using Laplace noise. }
\end{figure}

As it can be appreciated in \figurename~\ref{fig:Laplace_approximation_1},
to satisfy $\varepsilon$-differential privacy at points $(1,1)$,
$(1,-1)$, $(-1,-1)$ and $(-1,1)$ with independent Laplace noise
addition for each dimension, we are overprotecting those points with
$L_{1}$-norm less than or equal to $\Delta f=2$ that do not belong
to $[-1,1]\times[-1,1]$; the density at these points is greater or
equal to $e^{-\varepsilon}f_{\N}(0)$, while this is not a requirement
of $\varepsilon$-differential privacy (which only requires a density
greater or equal to $e^{-\varepsilon}f_{\N}(0)$ for the points in
$S_{f}$).

The ratio between the size of the overprotected region and the size
of $S_{f}$ may become still larger if the variability of one of the
components is greater than the variability of the other. \figurename~\ref{fig:Laplace_approximation_10}
illustrates the case of $S_{f}$ being the set $[-1,1]\times[-10,10]$.

\begin{figure}[ht]
\begin{centering}
\includegraphics[width=14cm]{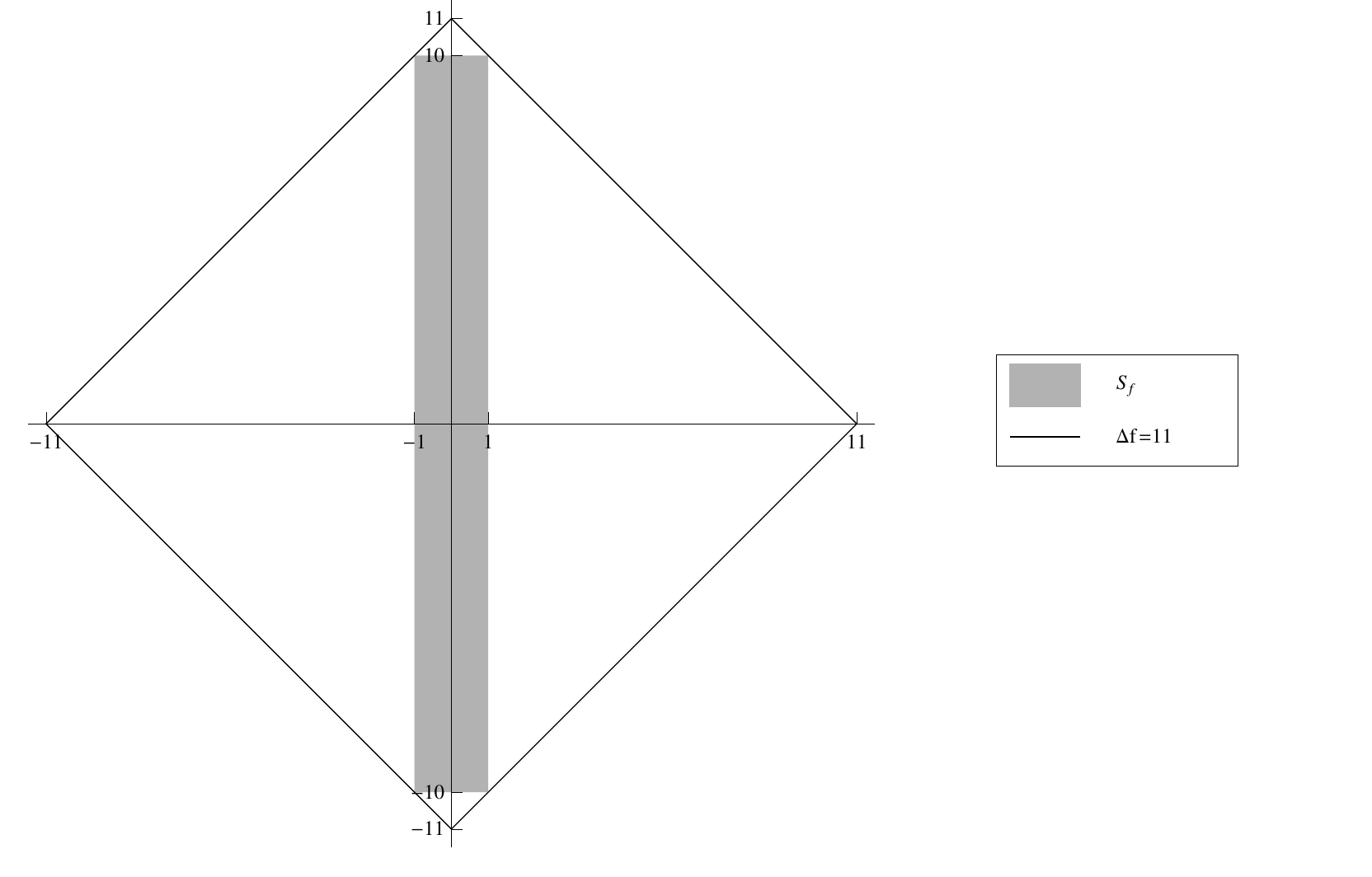} 
\par\end{centering}

\caption{\label{fig:Laplace_approximation_10}Achieving $\varepsilon$-differential
privacy by Laplace noise addition for $S_{f}=[-1,1]\times[-10,10]$
The shaded box represents the possible differences in the query result
between data sets that differ in one record. Differential privacy
requires the density of the noise in the shaded box to be within a
factor in $[\exp(-\varepsilon),\exp(\varepsilon)]$ of the density
at zero. The square that encloses the shaded box represents the points
that satisfy the previous condition when using Laplace noise.}
\end{figure}

\end{example}
In the construction of the piecewise constant noise density, we will
fix a set $S_{0}\subset S_{f}$ with $\langle0,x\rangle\subset S_{0}$
for all $x\in S_{0}$, where the maximum density will be reached.
From this $S_{0}$, we will define $S_{i}$ as the set that contains
the points that are reachable from $S_{i-1}$ in one step, that is,
by adding a value from $S_{f}$: 
\[
S_{i}=\{x\in\mathbb{R}^{d}|x=z+\delta,\, z\in S_{i-1},\,\delta\in S_{f}\}\setminus\cup_{j=0}^{i-1}{S_{j}}
\]
The density value over the points in $S_{i}$ will be $e^{-\varepsilon}$
times the density value over the points in $S_{i-1}$. Therefore,
for $x$ in $S_{i}$ it will be 
\[
f_{\N}(x)=Me^{-i\varepsilon}
\]
The value $M$ must be calibrated so that the total probability equals
1. Such calibration is possible because the density function decreases
exponentially as $i$ grows.

The following theorem shows that the constructed distribution is optimal
at providing $\varepsilon$-differential privacy to the function $f$.
\begin{thm}
\label{prop:multi_steps} Let $f=(f_{1},\ldots,f_{d})$ be a query
function with values in $\mathbb{R}^{d}$. Let $\N=(\N_{1},\ldots,\N_{d})$
be an a.c. random noise with density 
\[
f_{\N}(x)=\sum_{i\ge0}M\exp(-i\varepsilon)\mathbb{I}_{S_{i}}(x)
\]
where $\mathbb{I}_{S_{i}}(x)$ is the indicator function for set $S_{i}$
and $M$ has been calibrated to adjust the total probability mass
to one. If the following conditions hold, then $\N$ is optimal at
providing $\varepsilon$-differential privacy to f: \end{thm}
\begin{itemize}
\item $S_{0}\subset S_{f}$ 
\item $\langle0,x\rangle\subset S_{0}$ for all $x\in S_{0}$ 
\item $S_{i+1}=(S_{i}+S_{f})\setminus\cup_{j=0}^{i-1}{S_{j}}$ for all $i\ge0$\end{itemize}
\begin{proof}
First of all we check that $\N$ satisfies the $\varepsilon$-differential
privacy condition as stated in Proposition~\ref{prop:4}. Consider
$x\in\mathbb{R}^{d}$ and $\delta\in S_{f}$. The sets $S_{i}$ form
a cover of $\mathbb{R}^{d}$; therefore we have $x\in S_{i}$ for
some $i\in\mathbb{N}$. For $x+\delta$ we have one of the following
possibilities: $x+\delta\in S_{i-1}$, $x+\delta\in S_{i}$, or $x+\delta\in S_{i+1}$.
The value of the density function will, respectively, be $Me^{-(i-1)\varepsilon}$,
$Me^{-i\varepsilon}$, or $Me^{-(i+1)\varepsilon}$; in all three
cases, the $\varepsilon$-differential privacy condition is satisfied.

To show that $\N$ is optimal at providing $\varepsilon$-differential
privacy to $f$ we have to check that if we move some probability
mass towards zero, the resulting random noise does not provide $\varepsilon$-differential
privacy to $f$. We partition $\mathbb{R}^{d}$ and check, for each
set in the partition, that it is not possible to move any probability
mass towards zero and still satisfy $\varepsilon$-differential privacy.
The partition is $\{S_{f}^{i},\, i\ge1\}$ where $S_{f}^{1}=S_{f}$
and $S_{f}^{i+1}=(S_{f}^{i}+S_{f})\setminus\cup_{j=1}^{i}S_{f}^{j}$.

We start by checking that it is not possible to move any probability
mass contained in $S_{f}^{1}$ towards zero and still satisfy $\varepsilon$-differential
privacy. The density $f_{\N}$ in $S_{f}^{1}$ can be expressed as
\[
f_{\N}(x)=M\times\mathbb{I}_{S_{0}}(x)+M\exp(-\varepsilon)\times\mathbb{I}_{S_{f}^{1}\setminus S_{0}}(x)
\]
Note that $f_{\N}$ already has the maximum change in the density
that $\varepsilon$-differential privacy allows: $\exp(\varepsilon)$.
In other words, if we increase the density above $M$ or decrease
it below $M\times\exp(-\varepsilon)$, $\varepsilon$-differential
privacy will not hold. Let $U\subset S_{f}^{1}$ be the set that will
have its probability mass reduced. It must be $U\subset S_{0}$; otherwise
some points would have
their density reduced below $M\times\exp(-\varepsilon)$,
which is not possible. Now, as we have $\langle0,x\rangle\subset S_{0}$
for all $x\in S_{0}$ (\emph{i.e} for any point in $S_{0}$ the points
closer to zero are already in $S_{0}$), if we move probability mass
from $U$ towards zero, this probability mass must go to a set of
points $U'$ contained in $S_{0}$. This way the density of points
in $U'$ would be greater than $M$, which would also break $\varepsilon$-differential
privacy.

To conclude the proof we have to check that it is not possible to
move any probability mass belonging to a set $S_{f}^{i+1}$ with $i\ge1$
towards zero and still satisfy $\varepsilon$-differential privacy.
Note that the density function $f_{\N}$ decreases as fast as possible
as we move away from $S_{0}$: according to proposition Proposition~\ref{prop:4}
the density at a point $y$ reachable from a point $x$ by adding
a value from $S_{f}$ must satisfy $f_{\N}(y)\ge\exp(-\varepsilon)f_{\N}(x)$.
We have set the density $f_{\N}$ at $S_{i+1}$ to be $\exp(-\varepsilon)$
times the density at $S_{i}$; that is, the minimum value that satisfies
$\varepsilon$-differential privacy.

To move some probability mass belonging to $S_{f}^{i+1}$ towards
zero we must select a set $U\subset S_{f}^{i+1}$ and reduce its probability
mass. In other words, the density function in the points in $U$ is
to be reduced. But this is not possible, if we want to preserve $\varepsilon$-differential
privacy.\end{proof}
\begin{example}
\label{exa:laplace_vs_stepwise} Let $f$ be a function with $S_{f}=[-1,1]\times[-10,10]$,
and take $\varepsilon=1$. Hence, the sensitivity of $f$ is $\Delta f=1+10=11$
and $\varepsilon$-differential privacy with two independent Laplace-distributed
random noise components requires these components to have zero mean
and 
$11/\varepsilon$ scale parameter. Our proposal to achieve
$\varepsilon$-differential privacy is to use the piecewise constant
density construction by setting $S_{0}=[-0.1,0.1]\times[-1,1]$. \figurename~\ref{fig:Laplace_10_1}
shows the density function of both distributions. Note that with the
Laplace distribution the noise densities for both components of $f$
decrease at the same rate, even if the second component of $f$ has
ten times the sensitivity of the first one.

\begin{figure}[ht]
\begin{centering}
\includegraphics[angle=-90,width=10cm]{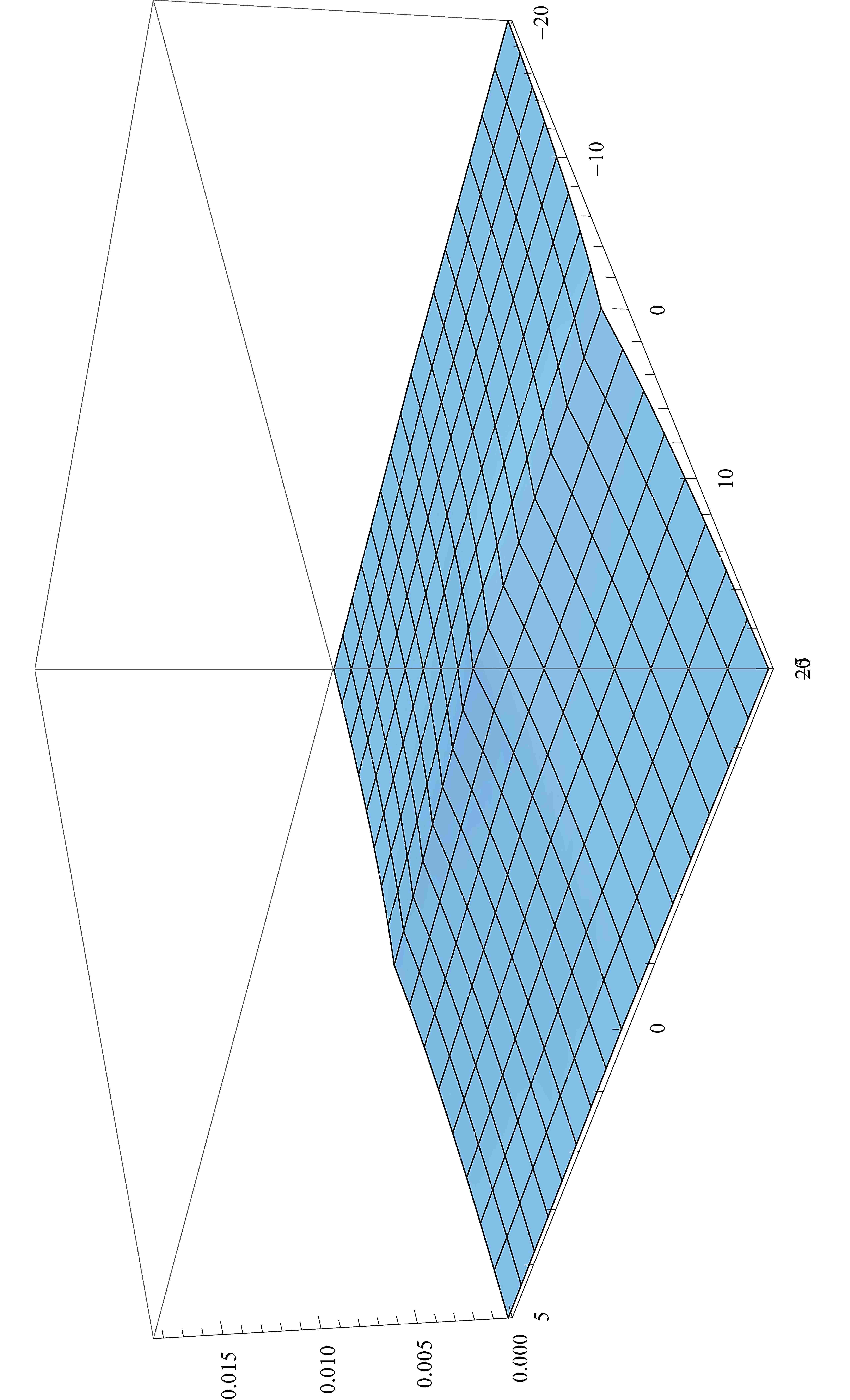} 
\par\end{centering}

\begin{centering}
\includegraphics[width=10cm]{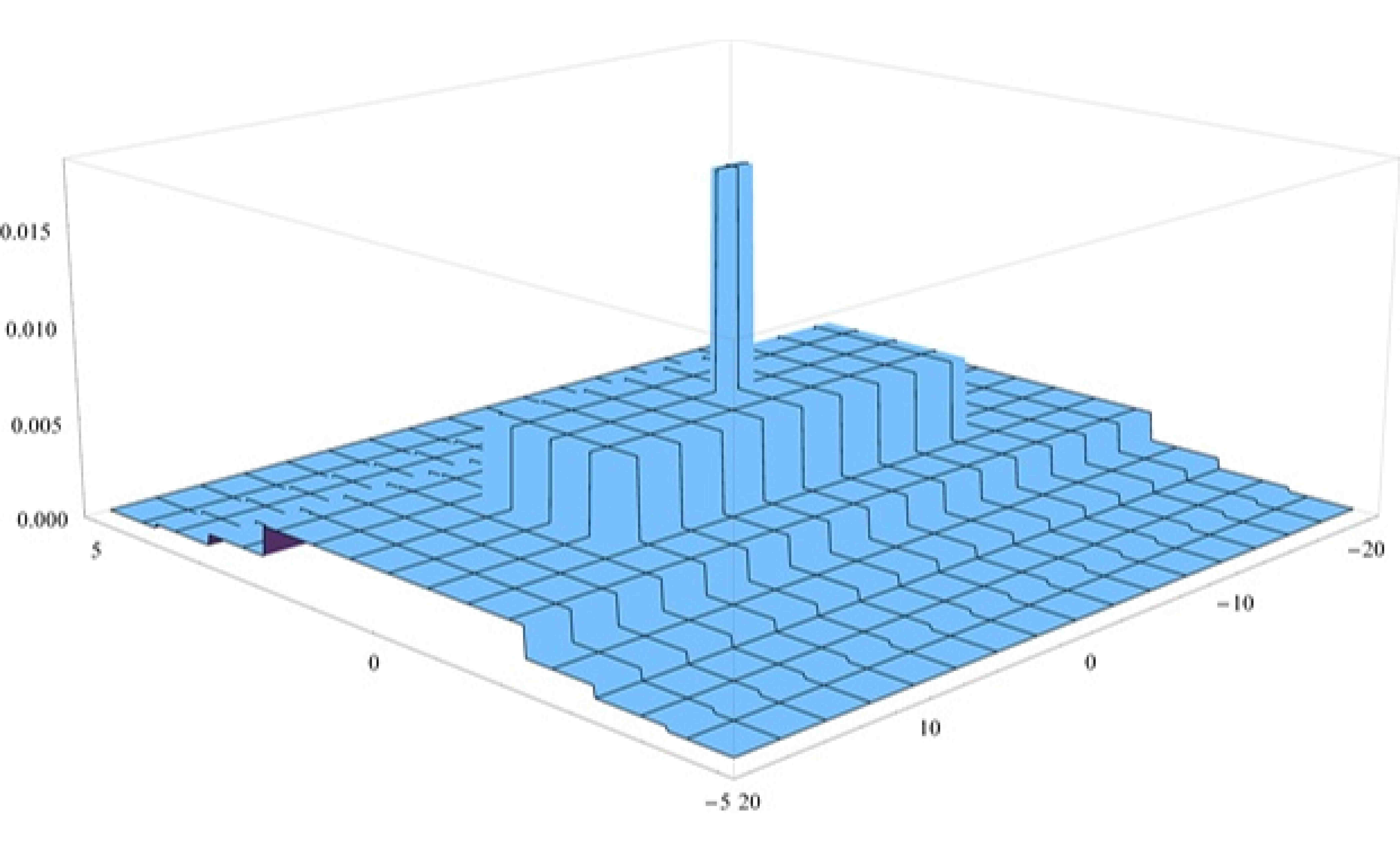} 
\par\end{centering}

\caption{\label{fig:Laplace_10_1}Density functions of the Laplace and piecewise
constant noise distributions required to achieve $1$-differential
privacy for a bivariate function $f=(f_{1},f_{2})$ with $\Delta f_{1}=1$
and $\Delta f_{2}=10$}
\end{figure}

It is easily appreciated in the figure that the piecewise constant
distribution has much more probability concentrated around zero, which
agrees with our optimality definition in Section~\ref{sec:Optimal-Random-Noise}.
To compare both distributions, we compute the variance of the components,
and the minimal size of a confidence region at some confidence levels.

For Laplace-distributed random noise $(\N_{1},\N_{2})$, the computations
are easy. Since we know that $\N_{1}$ and $\N_{2}$ follow a Laplace
distribution, their variance is twice the square of the scale factor
\begin{eqnarray*}
Var(\N_{1}) & = & 242\\
Var(\N_{2}) & = & 242
\end{eqnarray*}
With the Laplace-distributed random noise $(\N_{1},\N_{2})$ points
with equal $L_{1}$-norm are assigned the same noise density. Therefore
the confidence region with minimal size, for a given confidence level
is of the form $\{x|\,\Vert x\Vert\le\alpha\}$. Table~\ref{tab:minimal_size_laplace}
shows the size of the confidence region for several confidence levels.

\begin{table}[ht]
\caption{\label{tab:minimal_size_laplace}Minimal size of the confidence region
for two-dimensional Laplace-distributed random noise with scale parameter
11}

\centering{}%
\begin{tabular}{|c|c|c|}
\hline 
Confidence level  & $\alpha$  & Size\tabularnewline
\hline 
\hline 
0.99  & 73.02  & 10663\tabularnewline
\hline 
0.95  & 52.18  & 5445\tabularnewline
\hline 
0.90  & 42.79  & 3662\tabularnewline
\hline 
\end{tabular}
\end{table}

Computing the variance of the components of the piecewise constant
distribution will be done in terms of the sets $S_{f}$ and $S_{0}$.
If we let $S_{f}=[-s_{1},s_{1}]\times[-s_{2},s_{2}]$ and $S_{0}=[-z_{1},z_{1}]\times[-z_{2},z_{2}]$
then the density of the components $\N_{1}$ and $\N_{2}$ is 
\begin{eqnarray*}
f_{\N_{1}}(x) & = & 2Me^{-i_{1}\varepsilon}\times(z_{2}+s_{2}i_{1}+s_{2}/(e^{\varepsilon}-1))\\
f_{\N_{2}}(x) & = & 2Me^{-i_{2}\varepsilon}\times(z_{1}+s_{1}i_{2}+s_{1}/(e^{\varepsilon}-1))
\end{eqnarray*}
where $i_{1}=\lfloor(|x|-z_{1})/s_{1}+1\rfloor$ is the index of the
first set $S_{i}$ such that $(x,0)$ belongs to it, $i_{2}=\lfloor(|x|-z_{2})/s_{2}+1\rfloor$
is the index of the first set $S_{i}$ such that $(0,x)$ belongs
to it, and $M$ is a constant adjusted so that the random distribution
$(\N_{1},\N_{2})$ has probability mass one. \figurename~\ref{fig:Comparison_component_1}
compares the first and second components of the Laplace and the piecewise
constant random noise. Note that the piecewise constant distribution
seems to slightly underperform Laplace for the second component,
but it clearly outperforms Laplace for the first component.

\begin{figure}[ht]
\begin{centering}
\includegraphics[width=10cm]{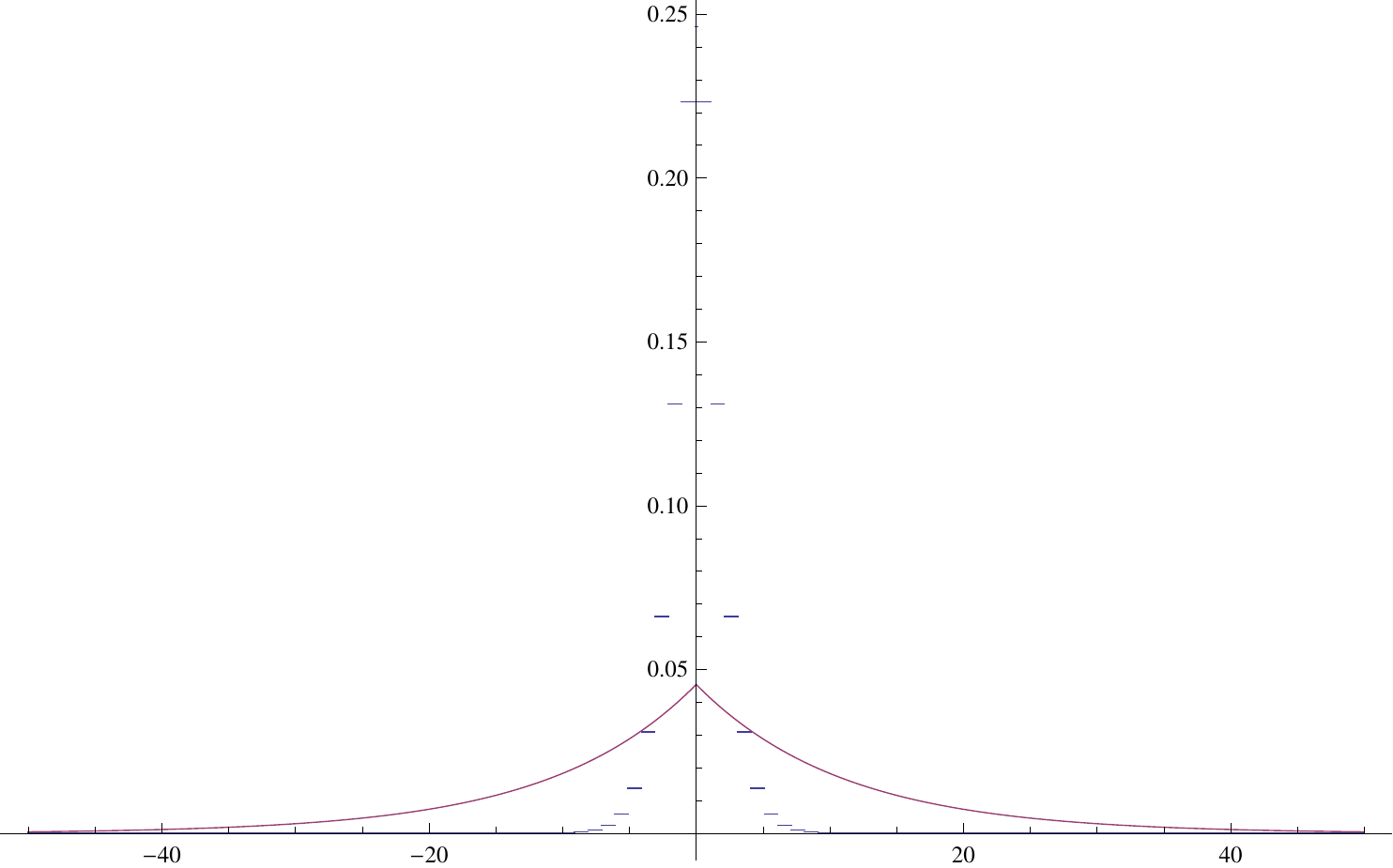} 
\par\end{centering}

\begin{centering}
\includegraphics[width=10cm]{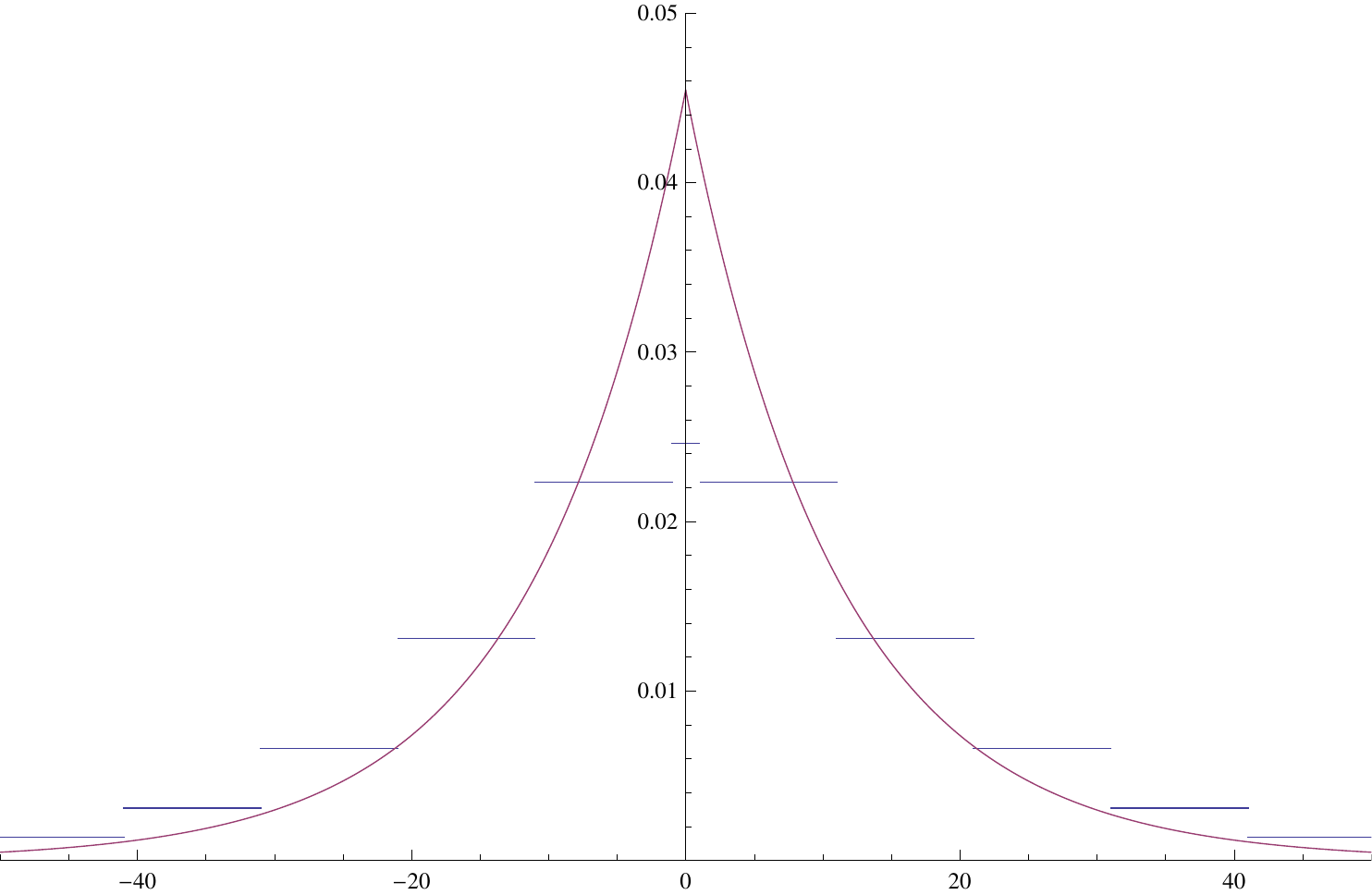} 
\par\end{centering}

\caption{\label{fig:Comparison_component_1}Comparison of
the Laplace and the piecewise constant random noise distributions
required to achieve $1$-differential privacy for a bivariate function
$f=(f_{1},f_{2})$ with $\Delta f_{1}=1$ and $\Delta f_{2}=10$.
Top, comparison for the first component; bottom, comparison for the
second component.}
\end{figure}

Since the mean of the components is zero, their variance can be computed
by integrating $\int_{\mathbb{R}}x^{2}f_{\N_{i}}(x)dx$, which results
in:

\begin{eqnarray*}
Var(\N_{1}) & = & 4.0338\\
Var(\N_{2}) & = & 403.38
\end{eqnarray*}
Compared to the variances obtained for the Laplace-distributed random
noise, we observe that the variance for $\N_{2}$ when using the piecewise
constant distribution is about twice as big as when using Laplace
distribution. On the other side, the variance of $\N_{1}$ is much
smaller when using the piecewise constant distribution. These results
are consistent with the previous observation about \figurename~\ref{fig:Comparison_component_1}.

We compute now confidence regions for the piecewise constant distribution.
To obtain a confidence region with minimal size, we make sure to include
all the points in $S_{i}$ before including any point in $S_{i+1}$.
We will consider confidence regions of the form $[-z_{1}-\beta s_{1},z_{1}+\beta s_{1}]\times[-z_{2}-\beta s_{2},z_{2}+\beta s_{2}]$.
Table~\ref{tab:confidence_region_stepwise} shows the confidence
regions obtained. By comparing with Table~\ref{tab:minimal_size_laplace},
it can be observed in the table that the minimal size for a confidence
level is much smaller when using the piecewise constant distribution.

\begin{table}[ht]
\caption{Minimal size of the confidence region for the piecewise constant noise
distribution needed for a bivariate function $f=(f_{1},f_{2})$ with
$\Delta f_{1}=1$ and $\Delta f_{2}=10$\label{tab:confidence_region_stepwise}}

\centering{}%
\begin{tabular}{|c|c|c|}
\hline 
Confidence level  & $\beta$  & Size\tabularnewline
\hline 
\hline 
0.99  & 6.99  & 1790.2\tabularnewline
\hline 
0.95  & 4.79  & 916.6\tabularnewline
\hline 
0.90  & 3.90  & 611.2\tabularnewline
\hline 
\end{tabular}
\end{table}

\end{example}
Note that in Example~\ref{exex} we considered $S_{f}$ to be the
product of two intervals. This case models the situation where the
query function components are independent, in the sense that we can
achieve any possible combination of values for the difference of the
query function. That is, $S_{f}=[-1,1]\times[-1,1]$ means that, for
any $[\delta_{1},\delta_{2}]\in[-1,1]\times[-1,1]$, we can find two
data sets $D$ and $D'$ differing in one row such that $f_{1}(D)-f_{1}(D')=\delta_{1}$
and $f_{2}(D)-f_{1}(D')=\delta_{2}$. Taking $S_{f}$ to be the product
of intervals is the natural option in the case of an interactive mechanism~\cite{Dwork2006a},
where we get to know each of the components of the query function
(\emph{i.e.} each successive query if we view the multivariate query
as a group of queries) at different times. In an interactive mechanism
it is not possible to construct the distribution that best matches
the multiquery function $f$, because at the time of the first query
we only know $f_{1}$. Clearly, it is possible to achieve a better
noise calibration for a non-interactive query than for an interactive
one, but using independent Laplace noise addition for each component
fails to exploit non-interactivity.

\section{Conclusions}

Our goal in this chapter
was to analyze the optimality of data-independent random
noise distributions to achieve $\varepsilon$-differential privacy.
The first step was to define the concept of optimal distribution as
a distribution that concentrates the probability around zero as much
as possible while ensuring differential privacy. This criterion led
to a family of optimal distributions, which can be refined by using
additional criteria. In the examples, we have computed optimal distributions
using as additional criteria the minimization of the response variance
or the minimization of the size of the confidence interval around
the response.

For a single univariate query, the optimal absolutely continuous noise
distributions to achieve $\varepsilon$-differential privacy were
built; as a result, we obtained a family of piecewise constant density
functions. The comparison with the Laplace noise distribution showed
that Laplace performs only slightly worse than the optimal absolutely
continuous distributions. Comparison figures were provided for the
variance and the size of the confidence interval.

For a multivariate query or multiple queries, a piecewise constant
construction similar to that of a single query was presented. Comparisons
in terms of variance and of size of the minimal confidence interval
showed that, \emph{for multivariate and/or multiple queries, the Laplace
distribution is far from being optimal}. Given the popularity of the
Laplace distribution, this is a very relevant result. We also observed
that the proposed mechanism provides better responses for non-interactive
queries, as it is able to exploit the global knowledge on the query
function. This is not possible for mechanisms that assume the components
of the query function to be independent, as it is the case for Laplace
noise addition.

\lhead[\chaptername~\thechapter]{\rightmark}

\rhead[\leftmark]{}

\lfoot[\thepage]{}

\cfoot{}

\rfoot[]{\thepage}

\chapter{Sensitivity-independent differential privacy via knowledge refinement}

Differential privacy states that the probability for a query response 
to belong to any subset of the query domain 
must be similar regardless of presence or absence
of any specific individual in the data set (dee Definition~\ref{def:dp_dwork}).
A usual approach to satisfy such condition is noise addition: first,
the real value of the query response is computed and, then, a random
noise is added to mask it. A Laplace distribution with zero mean and
a scale parameter that depends on the variability of the query function
is commonly used for noise addition.

Our proposal is not based on masking the true value of the response
by adding some noise, but on modifying the prior knowledge that the
database user has on the response. When a query is submitted to the
database, the user submits at the same time her knowledge/beliefs
about the response. We think of this prior knowledge as the probability
distribution that the user expects for the response. 

Our mechanism is shown to have several advantages over noise addition:
it does not require complex computations, and thus it can be easily
automated; it lets the user exploit her prior knowledge about the
response to achieve better data quality; and it is independent of
the sensitivity of the query function (although this can be a disadvantage
if the sensitivity is small). Furthermore, we give a general algorithm
for knowledge refinement and we show some compounding properties of
our mechanism for the case of multiple queries; also, we build an
interactive mechanism on top of knowledge refinement and we show that
it is safe against adaptive attacks. Finally, we give a quality assessment
for the responses to individual queries.

The contents of this chapter have been published in~\cite{passat,ijufks}.

\section{Refining prior knowledge\label{sec:2 Ref_prior_knowledge}}

Our proposal to attain $\varepsilon$-differential privacy is not
based on masking the true value of the response by adding some noise,
but on modifying the prior knowledge of the database user on the response.
When a query is submitted to the database, the user submits at the
same time her knowledge/beliefs about the response. We think of this
prior knowledge as the probability distribution that the user expects
for the response. For example, in case the user has absolutely no
idea about the possible result for a query $f$, the probability distribution
to be used is the uniform distribution over the range of $f$ (assuming
that this range is bounded). The access mechanism modifies this prior
knowledge to fit the real value of the response as much as possible
given the constraints imposed by differential privacy.
\begin{defn}
\label{def:prior_knowledge}Given a query function $f$, the \emph{prior
knowledge} about the response $f(D)$ is the probability distribution
$P_{f}$, defined over $Range(f)$, that the user expects for the
response to $f$.
\end{defn}
The more concentrated the probability mass of $P_{f}$ around the
real value of the response to $f$, the more accurate is the user's
prior knowledge. In general, as the user knows the query $f$ and
the set of possible databases $D$, one may expect her to have some
prior knowledge about the response $f(D)$. The better the knowledge
the user has on the actual database $D$, the more accurate is the
prior knowledge the user can provide to the response mechanism. If
the user's prior knowledge is wrong, the accuracy of the response
may suffer. However, whatever the prior knowledge, the refinement
procedure guarantees that the output is more accurate than the prior
knowledge.

Some users may be reluctant to provide detailed prior knowledge, because
they regard doing so as giving information about themselves to the
database. We should usually think of the prior knowledge as the information
about the response that is publicly available. Providing the database
with such a prior knowledge reveals nothing about the database user.
If the database user has information that is not publicly available,
she must decide whether to use it as prior knowledge or not; the more
accurate the prior knowledge, the more accurate the response will
be. We will see in Section~\ref{sec:Quality} that, even when little
prior knowledge is assumed, knowledge refinement may be superior,
in terms of data quality, to noise addition approaches. Therefore,
it may make sense to use knowledge refinement even if the database
user is not willing to provide all her actual prior knowledge. 

If the query function $f$ has multiple components (dimension $n>1$),
the joint probability distribution must be provided. If the components
of $f$ are independent, specifying the marginal distribution for
each component is enough to compute the joint distribution. This will
also be the case if the components are not independent but the user
has no knowledge about the relationship among them.

The access mechanism is run by the database holder as follows: 
\begin{itemize}
\item Receive the query $f$ and the prior knowledge $P_{f}$ from the database
user.
\item Compute the actual value of the query response, $f(D)$. 
\item Modify $P_{f}$ to adjust it to $f(D)$ as much as possible, given
the constraints imposed by differential privacy.
\item Randomly sample the distribution resulting from the previous step,
and return the sampled value as the response to $f$ evaluated at
$D$.
\end{itemize}
Even though knowledge refinement works by adjusting the prior knowledge,
the output is not the adjusted distribution but a sample from it.
This is the usual approach in differential privacy; only a sample
from the output distribution is returned. Returning the output distribution
itself would leak too much information; in some cases, it could be
used to determine the exact value of the query response.

Note that the user cannot pretend to have more knowledge than she
actually has: sending a guess as $P_{f}$ will most likely be wrong
and worsen the response quality. Also, we show in Section~\ref{sec:interactive}
that using several different (fake) prior knowledge distributions
to mount adaptive attacks does not succeed in breaking $\varepsilon$-differential
privacy.

The critical step is the adjustment of the prior knowledge to the
real query response. To perform this adjustment, we distinguish two
types of queries: statistical queries and individual queries. We call
statistical queries those whose outcome depends on multiple individuals,
while individual queries are those that depend on a single individual.
It will be shown below that a finer adjustment of the prior knowledge
is feasible for individual queries. We start by focusing on statistical
queries, but, before formally specifying the response mechanism, we
give an example to illustrate what we intend to do.
\begin{example}
\label{ex:1}Assume a query function $f$ that is known to return
a value within the interval $[0,1]$. Assume also that the database
user has no further knowledge about the query response, \emph{i.e.}
her prior knowledge is the uniform distribution over $[0,1]$.

To refine the prior knowledge, we modify its density by applying two
multiplicative factors: $\alpha_{u}\ge1$ to the points near $f(D)$,
and $\alpha_{d}\le1$ to the points farther from $f(D)$. In this
way, the probability of obtaining as the response a value near the
actual response $f(D)$ is increased with respect to the prior knowledge,
while the probability of obtaining a distant value is decreased. Figure~\ref{fig:distributions_ex1}
shows the probability distribution resulting from applying the procedure
described above for a pair of neighbor data sets $D$ and $D'$.

\begin{figure}[h]
\begin{centering}
\includegraphics[width=7cm]{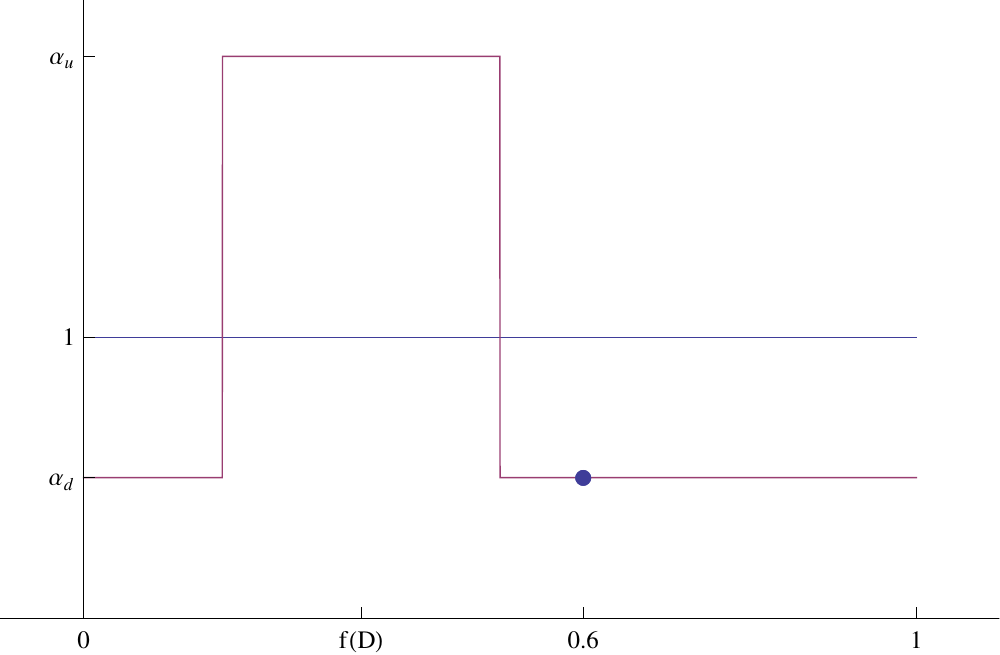}~\includegraphics[width=7cm]{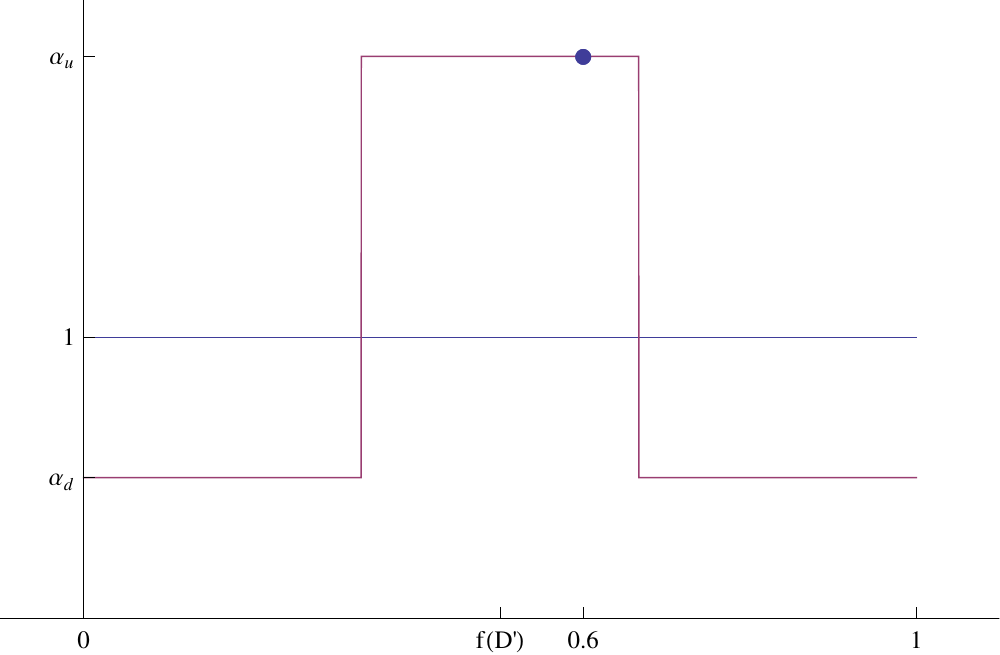}
\par\end{centering}

\caption{Distributions for the response to $f(D)$ (left) and to $f(D')$ (right)\label{fig:distributions_ex1}}
\end{figure}

To obtain $\varepsilon$-differential privacy, the density at a given
point for the response to $f(D)$ must be a factor within the interval
$[e^{-\varepsilon},e^{\varepsilon}]$ of the density at the same point
for the response to $f(D')$. Check, for example, the point 0.6 in
Figure~\ref{fig:distributions_ex1}: on the left-hand side distribution,
the point is far from the real response and thus a factor $\alpha_{d}$
is applied; on the right-hand side distribution, the point is near
the real response and the factor applied is $\alpha_{u}$. For the
$\varepsilon$-differential privacy condition to hold, it must be
$\alpha_{u}/\alpha_{d}\le e^{\varepsilon}$. We can also think in
the reverse way: given two constants $\alpha_{u}\ge1$ and $\alpha_{d}\le1$,
the level of differential privacy achieved by this response mechanism
is $\varepsilon=\ln(\alpha_{u}/\alpha_{d})$.
\end{example}
Note that, to obtain a valid density function from the above modification,
the set of points over which each of the factors $\alpha_{u}$ and
$\alpha_{d}$ are applied must be selected in such a way that the
total probability mass of the resulting distribution equals 1. If
we denote by $\mathcal{U}_{u}$ the set over which we apply the factor
$\alpha_{u}$, for the total probability mass of the adjusted distribution
to be 1, we must have $\alpha_{u}P_{f}(\mathcal{U}_{u})+\alpha_{d}(1-P_{f}(\mathcal{U}_{u}))=1$.
If the prior knowledge is an absolutely continuous distribution, as
in Example~\ref{ex:1}, for any pair of values $\alpha_{u}\ge1$
and $\alpha_{d}\le1$ it is possible to select a set $\mathcal{U}_{u}$
in such a way that $\alpha_{u}P_{f}(\mathcal{U}_{u})+\alpha_{d}(1-P_{f}(\mathcal{U}_{u}))=1$
is satisfied. The reason is that we can select the set $\mathcal{U}_{u}$
to have any probability mass between 0 and 1. If the prior knowledge
distribution is not absolutely continuous, it may not be possible
to find a set $\mathcal{U}_{u}$ with the required probability mass
for the given values $\alpha_{u}$ and $\alpha_{d}$. This section
assumes that such a set $\mathcal{U}_{u}$ exists. In Section~\ref{sec:3 gen_algorithm},
we specify a general algorithm that works for any prior knowledge
distribution. 

The following proposition formalizes the ideas discussed in the previous
example.
\begin{prop}
\label{prop:1}Let $f:\mathcal{D}\rightarrow\mathbb{R}^{n}$ be a
query function and let $P_{f}$ be the prior knowledge for $f(D)$.
Let $\alpha_{u}\ge1$ and $\alpha_{d}\le1$ be such that $\alpha_{u}=e^{\varepsilon}\alpha_{d}$.
Let $\mathcal{U}_{u}$ be an environment of $f(D)$ satisfying $\alpha_{u}P_{f}(\mathcal{U}_{u})+\alpha_{d}(1-P_{f}(\mathcal{U}_{u}))=1$.
The response mechanism that returns a value randomly sampled from
the distribution obtained by modifying $P_{f}$ through multiplication
of the probability mass of the points in $\mathcal{U}_{u}$ by $\alpha_{u}$,
and multiplication of the probability mass of the points outside $\mathcal{U}_{u}$
by $\alpha_{d}$, satisfies $\varepsilon$-differential privacy.
\end{prop}
When the query $f$ returns a value related to a single individual,
the mechanism in Proposition~\ref{prop:1} can be improved. In that
case, there are only two possibilities for the response: (i) if the
individual we are asking about is not in the database, the distribution
of the response equals the prior knowledge distribution, and (ii)
if the individual is in the database, the distribution for the response
will be the result of refining the prior knowledge. To satisfy $\varepsilon$-differential
privacy, we only need to guarantee that the distribution resulting
from (i) and (ii) does satisfy the limitation on the knowledge gain
imposed by differential privacy. In other words, the output distribution
need only be compared to the prior knowledge. The conditions that
must hold are $1\le\alpha_{u}\le e^{\varepsilon}$ and $e^{-\varepsilon}\le\alpha_{d}\le1$.

Note that, by choosing $\alpha_{u}=e^{\varepsilon}$ and $\alpha_{d}=e^{-\varepsilon}$,
the level of differential privacy that we can guarantee for a statistical
query function (depending on multiple individuals) is $2\varepsilon$,
while for an individual query (whose outcome depends on a single individual),
we double the guarantee to $\varepsilon$.
\begin{prop}
\label{prop:2}Let $f:\mathcal{D}\rightarrow\mathbb{R}^{n}$ be an
individual query in the above sense and let $P_{f}$ be the prior
knowledge distribution for $f$. Let $\alpha_{u}=e^{\varepsilon}$
and $\alpha_{d}=e^{-\varepsilon}$. Let $\mathcal{U}_{u}$ be an environment
of $f(D)$ satisfying $\alpha_{u}P_{f}(\mathcal{U}_{u})+\alpha_{d}(1-P_{f}(\mathcal{U}_{u}))=1$.
The response mechanism that returns a value randomly sampled from
the distribution obtained by modifying $P_{f}$ through multiplication
of the probability mass of the points in $\mathcal{U}_{u}$ by $\alpha_{u}$,
and multiplication of the probability mass of the points outside $\mathcal{U}_{u}$
by $\alpha_{d}$, satisfies $\varepsilon$-differential privacy.
\end{prop}

\section{A general algorithm for knowledge refinement\label{sec:3 gen_algorithm}}

Propositions~\ref{prop:1} and~\ref{prop:2} above state that, given
appropriate factors $\alpha_{u}$ and $\alpha_{d}$ and a set $\mathcal{U}_{u}$
with the required probability mass, the knowledge refinement mechanism
satisfies $\varepsilon$-differential privacy. However, some details
were left aside in the previous section: (i) how is the set $\mathcal{U}_{u}$
selected?, and (ii) can we still apply knowledge refinement if a set
$\mathcal{U}_{u}$ with the required probability mass does not exist?
This section gives a more detailed view of the knowledge refinement
mechanism and answers the two aforementioned questions.

Knowledge refinement works by increasing the probability mass of the
points near $f(D)$, and by decreasing the probability mass of the
rest of points in such a way that the total probability mass equals
one. In Example~\ref{ex:1} there was a natural way to determine
the set $\mathcal{U}_{u}$: the points closest to $f(D)$ in absolute
value. However, such a natural way does not always exist, as illustrated
in the next example.
\begin{example}
\label{extwo}To determine the form of the set $\mathcal{U}_{u}$
for a query function with two components, say $f=(f_{1},f_{2})$,
we use a distance function defined over the range of $f$, namely
$d:Range(f_{1})\times Range(f_{2})\rightarrow[0,\infty)$. If $d$
does not treat $f_{1}$ and $f_{2}$ symmetrically, then one component
is given priority over the other. In fact, there is no natural way
to define $d$ and hence $\mathcal{U}_{u}$. Such definitions are
application-dependent.

Table~\ref{tab:distance} shows some distance functions that are
appropriate for a query with a single component in terms of the type
of the result. We do not provide any distance for multivariate queries
because such distances are very application-dependent, as pointed
out in Example~\ref{extwo}.
\end{example}
\begin{table}[h]
\caption{\label{tab:distance}Example distance function for univariate query
functions depending on the type of the query result}

\centering{}%
\begin{tabular}{ccc}
\hline 
Query result & $Range(f)$ & distance\tabularnewline
\hline 
continuous & $\mathbb{R}$ & $d(x,y)=|x-y|$\tabularnewline
nominal & $\{c_{1},\ldots,c_{n}\}$ & $d(c_{i},c_{j})=\begin{cases}
0 & i=j\\
1 & i\ne j
\end{cases}$\tabularnewline
ordinal & $\{c_{1},\ldots,c_{n}\}$ & $d(c_{i},c_{j})=|i-j|$\tabularnewline
\hline 
\end{tabular}
\end{table}

Note that when we feed the knowledge refinement algorithm with a certain
distance function, we are instructing it with the sets that we want
to favor. Given a value $f(D)$, we modify the probability that the
prior knowledge assigns to the points in $Range(f)$ according to
the distance $d$. If a point at distance $r$ is being applied a
factor $\alpha_{1}$, all points at distance $r$ must be applied
the same factor, and points at a shorter distance must be applied
a factor $\alpha_{2}$ with $\alpha_{2}\ge\alpha_{1}$. Therefore,
the set $\mathcal{U}_{u}$ of points that has its probability increased
must be of the form $\mathcal{U}_{f(D),r}^{1}$ or $\mathcal{U}_{f(D),r}^{2}$,
for some $r\in[0,\infty)$, where:

\begin{equation}
\begin{array}{c}
\mathcal{U}_{f(D),r}^{1}=\{x\in Range(f):d(f(D),x)\le r\}\\
\mathcal{U}_{f(D),r}^{2}=\{x\in Range(f):d(f(D),x)<r\}
\end{array}\label{sets}
\end{equation}
The set $\mathcal{U}_{d}$ of points that has its probability decreased
is the complement of $\mathcal{U}_{u}$, that is, $\mathcal{U}_{d}=Range(f)\setminus\mathcal{U}_{u}$.

We want to choose two multiplicative factors $\alpha_{u}$ and $\alpha_{d}$
to modify the probability mass of $\mathcal{U}_{u}$ and $\mathcal{U}_{d}$,
respectively. Factors $\alpha_{u}$ and $\alpha_{d}$ must be selected
so that differential privacy holds and the total probability mass
of the resulting modified distribution equals one.

Table~\ref{tab:2 alpha} shows the form of factors $\alpha_{u}$
and $\alpha_{d}$ for the two types of queries considered in Section~\ref{sec:2 Ref_prior_knowledge}:
individual and statistical. For the case of individual queries, the
differential privacy condition need only hold between the distribution
of the response and the prior knowledge. 

Any pair of values $\alpha_{u}\in[1,e^{\varepsilon}]$ and $\alpha_{d}\in[e^{-\varepsilon},1]$
yields $\varepsilon$-differential privacy; however, $\alpha_{u}=e^{\varepsilon}$
and $\alpha_{d}=e^{-\varepsilon}$ yield the greatest knowledge gain.

For statistical queries, the condition must hold for each pair of
distributions for the response to the query over data sets that differ
in a single record. Therefore, we must have $\alpha_{u}/\alpha_{d}\leq e^{\varepsilon}$.
Same as for individual queries, the greatest knowledge gain is achieved
when $\alpha_{u}/\alpha_{d}=e^{\varepsilon}$. The actual values of
$\alpha_{u}$ and $\alpha_{d}$ must belong to the intervals $[1,e^{\varepsilon}]$
and $[e^{-\varepsilon},1]$, respectively, but they can be freely
chosen, as long as $\alpha_{u}/\alpha_{d}\lyxmathsym{ }\leq e^{\varepsilon}$
holds and the total probability mass is one:
\begin{equation}
\alpha_{u}P_{f}(\mathcal{U}_{u})+\alpha_{d}P_{f}(\mathcal{U}_{d})=1\label{totprob}
\end{equation}

\begin{table}[h]
\caption{\label{tab:2 alpha}Form of the factors $\alpha_{u}$ and $\alpha_{d}$
for individual and statistical queries}

\centering{}%
\begin{tabular}{cc}
\hline 
Type of query & Factors\tabularnewline
\hline 
individual & $\alpha_{u}=e^{\varepsilon}$, $\alpha_{d}=e^{-\varepsilon}$\tabularnewline
statistical & $\alpha_{u}\in[1,e^{\varepsilon}]$, $\alpha_{d}\in[e^{-\varepsilon},1]$
with $\alpha_{u}/\alpha_{d}=e^{\varepsilon}$\tabularnewline
\hline 
\end{tabular}
\end{table}

For statistical queries, the specific values selected for $\alpha_{u}$
and $\alpha_{d}$ determine the maximum knowledge gain for the points
in $\mathcal{U}_{u}$ and $\mathcal{U}_{d}$, where the gain is understood
as the modification w.r.t. the prior knowledge $P_{f}$. Assuming
that $\alpha_{u}/\alpha_{d}=e^{\varepsilon}$ holds, a greater value
for $\alpha_{u}$ provides increased knowledge gain for the points
in $\mathcal{U}_{u}$, but it also results in a greater value for
$\alpha_{d}$, because otherwise $\alpha_{u}/\alpha_{d}\leq e^{\varepsilon}$
would not be satisfied; this implies decreasing the knowledge gain
for the points in $\mathcal{U}_{d}$ with respect to the prior knowledge.

For fixed values of the factors $\alpha_{u}$ and $\alpha_{d}$, from
Equation (\ref{totprob}) and $P_{f}(\mathcal{U}_{d})=1-P_{f}(\mathcal{U}_{u})$,
we have: 
\[
\begin{array}{c}
P_{f}(\mathcal{U}_{u})=\frac{\alpha_{u}-1}{\alpha_{u}-\alpha_{d}}\\
P_{f}(\mathcal{U}_{d})=\frac{1-\alpha_{d}}{\alpha_{u}-\alpha_{d}}
\end{array}
\]
 For continuous prior knowledge, it is always possible to select sets
$\mathcal{U}_{u}$ and $\mathcal{U}_{d}$ with the above probability
masses. In this case, the knowledge refinement mechanism is very simple:
apply factor $\alpha_{u}$ to $\mathcal{U}_{u}$ and factor $\alpha_{d}$
to $\mathcal{U}_{d}$, as stated in Propositions~\ref{prop:1} and~\ref{prop:2}.

For other kinds of prior knowledge, the sets $\mathcal{U}_{u}$ and
$\mathcal{U}_{d}$ with the required probability masses may not exist.
In such cases, we still want to apply the factor $\alpha_{u}$ to
the greatest possible set of points closest to $f(D)$, and the factor
$\alpha_{d}$ to the greatest possible set of points farthest from
$f(D)$, thus achieving the maximum knowledge gain at such points.
We denote $\mathcal{U}_{u}'$ the set that is applied factor $\alpha_{u}$,
and $\mathcal{U}_{d}'$ the set that is applied factor $\alpha_{d}$.
For the remaining points we adjust their factor to have a total probability
mass of one. See Algorithm~\ref{alg:1} for a detailed description
of the process; this algorithm is run by the database holder.

\begin{algorithm}[h]
\caption{Knowledge refinement algorithm to respond to query $f(D)$ for a general
prior knowledge\label{alg:1}}

Input parameters: query $f$, prior knowledge $P_{f}$ of the database
user, distance function $d$, factors $\alpha_{u}$ and $\alpha_{d}$
from the database holder.
\begin{enumerate}
\item \begin{raggedright}
Compute the actual value of the query response, $f(D)$.
\par\end{raggedright}
\item Modify $P_{f}$ to adjust it to $f(D)$ as much as possible, given
the constraints imposed by differential privacy. This is done as follows:

\begin{enumerate}
\item Let $p_{u}=(\alpha_{u}-1)/(\alpha_{u}-\alpha_{d})$.
\item Let $p_{d}=(1-\alpha_{d})/(\alpha_{u}-\alpha_{d})$.
\item \textbf{if} there exists a set $\mathcal{U}_{u}$ of the form $\mathcal{U}_{f(D),r}^{1}$
or $\mathcal{U}_{f(D),r}^{2}$ (see Expression~\ref{sets}) with
$P_{f}(\mathcal{U}_{u})=p_{u}$ \textbf{then}

\begin{description}
\item [{ }] Build the distribution of the response to $f(D)$ by applying
the factor $\alpha_{u}$ to $\mathcal{U}_{u}$, and $\alpha_{d}$
to $Range(f)\setminus\mathcal{U}_{u}$.
\end{description}

\textbf{else}
\begin{enumerate}
\item Find the maximal set $\mathcal{U}_{u}'$ of the form $\mathcal{U}_{f(D),r}^{1}$
or $\mathcal{U}_{f(D),r}^{2}$ with $P_{f}(\mathcal{U}_{u}')<p_{u}$.
\item Find the maximal set $\mathcal{U}_{d}'$ of the form $Range(f)\setminus\mathcal{U}_{f(D),r}^{1}$
or $Range(f)\setminus\mathcal{U}_{f(D),r}^{2}$ with $P_{f}(\mathcal{U}_{d}')<p_{d}.$
\item Let $p_{ud}=1-P_{f}(\mathcal{U}_{u}')-P_{f}(\mathcal{U}_{d}')$ be
the probability of the points not in $\mathcal{U}_{u}'\cup\mathcal{U}_{d}'$
\item Let $\alpha_{ud}=(1-\alpha_{u}p_{u}-\alpha_{d}p_{d})/(1-p_{u}-p_{d})$
be the factor to be applied to $Range(f)\setminus(\mathcal{U}_{u}'\cup\mathcal{U}_{d}')$ 
\item Build the distribution of the response to $f(D)$ by applying:

\begin{itemize}
\item factor $\alpha_{u}$ to points in $\mathcal{U}_{u}'$
\item factor $\alpha_{d}$ to points in $\mathcal{U}_{d}'$
\item factor $\alpha_{ud}$ to points in $Range(f)\setminus(\mathcal{U}_{u}'\cup\mathcal{U}_{d}')$.
\end{itemize}
\end{enumerate}
\end{enumerate}
\item Randomly sample the distribution resulting from the previous step,
and return the sampled value as the response to $f$ evaluated at
$D$.\end{enumerate}
\end{algorithm}

It is easy to check that the total probability mass of the distribution
equals one, no matter whether the \textbf{then} or the \textbf{else}
option of the \textbf{if} statement of Algorithm~\ref{alg:1} is
taken. Regarding the differential privacy condition, we have already
seen that it holds for the \textbf{then} case. For the \textbf{else}
case, differential privacy also holds, because $\alpha_{ud}$ belongs
to the interval $[\alpha_{d},\alpha_{u}]$.

Differential privacy is usually criticized for the low utility of
the results it provides~\cite{Muralidhar2010,Sarathy2010,Sarathy2011}.
Several relaxations of $\varepsilon$-differential privacy have been
proposed; in particular, the authors of~\cite{Dwork2006b} propose
$(\varepsilon,\delta)$-differential privacy (\emph{a.k.a} $(\varepsilon,\delta)$-indistinguishability),
and $(\varepsilon,\delta)$-probabilistic differential privacy. The
former property relaxes the strict requirement of differential privacy
by adding a non-zero $\delta$. The latter property allows arbitrarily
large knowledge gains within probability $\delta$. Let us briefly
review $(\varepsilon,\delta)$-privacy and sketch how prior knowledge
refinement can achieve it. 
\begin{defn}
\label{def:delta_diff_priv}A randomized function gives $(\varepsilon,\delta)$-differential
privacy if, for all data sets $D_{1}$, $D_{2}$ such that one can
be obtained from the other by adding or removing a single record,
and all $S\subset Range(\kappa)$ 
\begin{equation}
P(\kappa(D_{1})\in S)\le\exp(\varepsilon)\times P(\kappa(D_{2})\in S)+\delta\label{ref:eq:2}
\end{equation}

\end{defn}
As $\varepsilon$-differential privacy implies $(\varepsilon,\delta)$-differential
privacy, Algorithm~\ref{alg:1} can be used to obtain $(\varepsilon,\delta)$-differential
privacy. However, a simple modification to Algorithm~\ref{alg:1}
can offer better data utility while still satisfying $(\varepsilon,\delta)$-differential
privacy (but no longer $\varepsilon$-differential privacy). We do
not provide a formal algorithm with the required modifications, but
the idea is to use the extra margin $\delta$ to increase the probability
at $f(D)$ and reduce it at the points farthest from $f(D)$. 

Just like it happened for $\varepsilon$-differential privacy, the
improvement of $(\varepsilon,\delta)$-privacy for individual queries
is greater than for statistical queries. For an individual query,
we only need to compare the distribution of the response with the
prior knowledge (see Figure~\ref{fig:distributions_indiv}). As the
prior knowledge is not modified, we can modify the response by adding
$\delta$ to the probability mass of $f(D)$, and subtract $\delta$
from the tails of the distribution. 

\begin{figure}[h]
\centering{}\includegraphics[width=7cm]{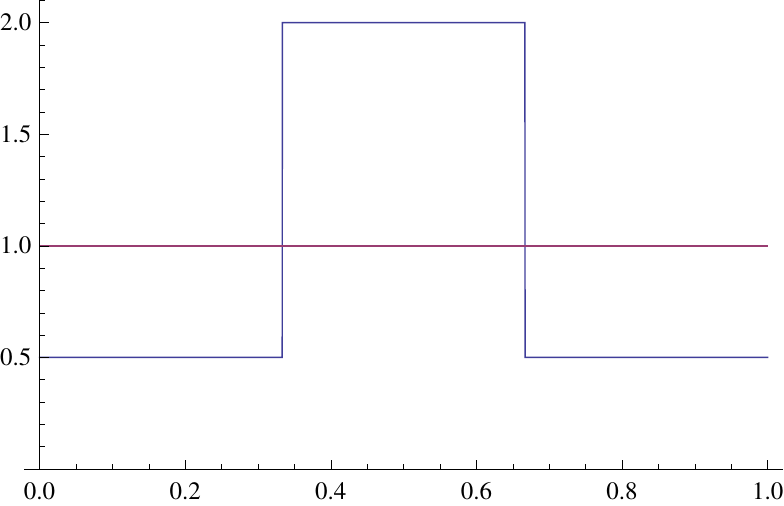}\includegraphics[width=7cm]{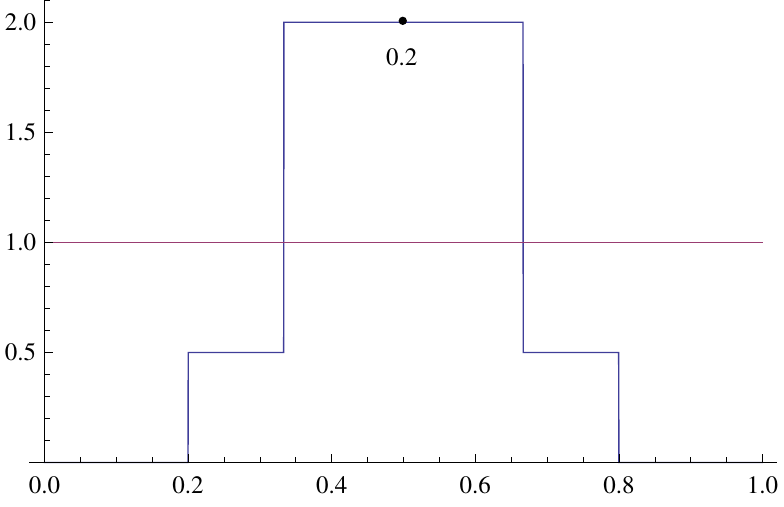}\caption{\label{fig:distributions_indiv}Distribution of the response to a
individual query when $f(D)=0.5$, for $\ln2$-differential privacy
(left), and $(\ln2,0.2)$-differential privacy (right)}
\end{figure}

For a statistical query, we also want to increase the probability
mass of the actual response $f(D)$, while reducing the probability
mass of the set $S_{f(D)}'$ of points farthest from $f(D)$. Although
other schemes are possible, a sensible choice is to have the probability
mass of $f(D)$ increased by the same amount $\delta'$, whatever
the data set $D$. As we have to keep the total probability mass equal
to one, we must decrease the probability of $S_{f(D)}'$ by $\delta'$.
Now, since we can select data sets $D_{1}$ and $D_{2}$ such that
$f(D_{1})$ belongs to $S_{f(D_{2})}'$, for Inequality~(\ref{ref:eq:2})
to hold for $S_{f(D_{2})}'$, it must be $\delta'=\delta/2$ (it can
also be $\delta'<\delta/2$, but then we are not taking advantage
of the whole $\delta$ margin).

\begin{figure}[h]
\begin{centering}
\includegraphics[width=7cm]{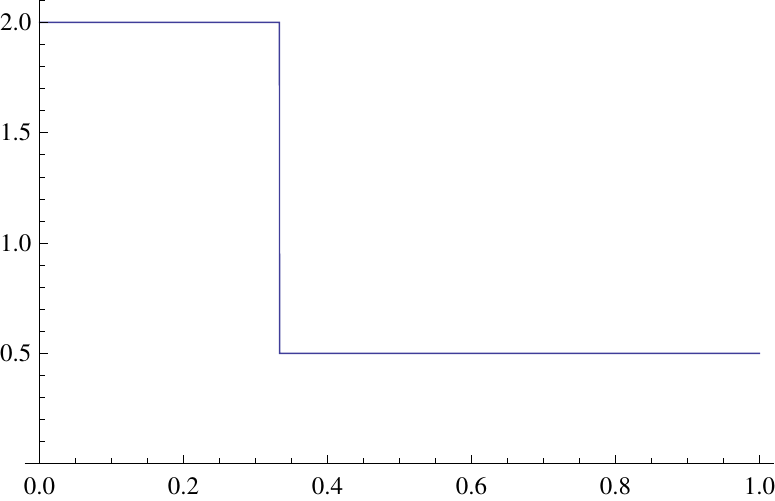}\includegraphics[width=7cm]{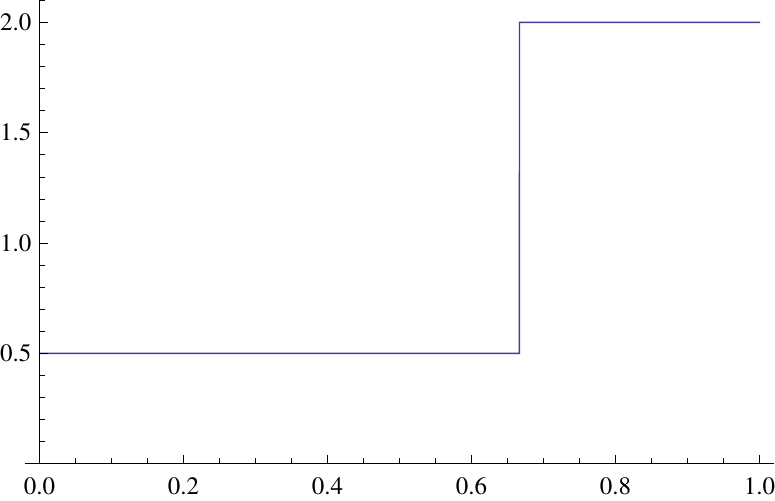}
\par\end{centering}

\begin{centering}
\includegraphics[width=7cm]{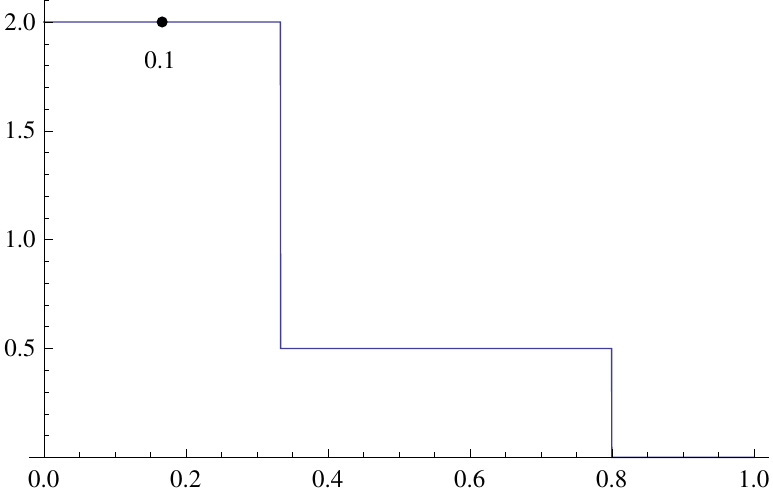}\includegraphics[width=7cm]{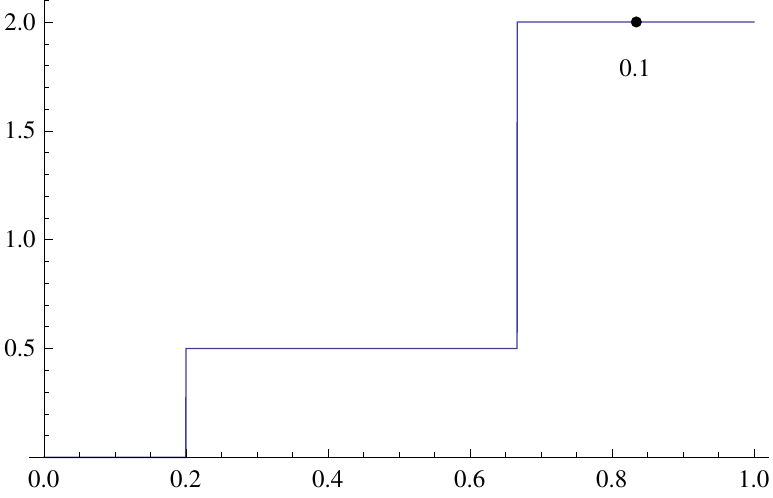}
\par\end{centering}

\caption{Distribution for the response to a statistical query when $f(D)=0.166$
(left) and $f(D)=0.833$ (right), for $\ln2$-differential privacy
(top) and $(\ln2,0.2)$-differential privacy (bottom)\label{fig:distrib_stat}}
\end{figure}

\section{Differential privacy in multicomponent queries\label{sec:compounded}}

The knowledge refinement mechanism as introduced in Section~\ref{sec:2 Ref_prior_knowledge}
is independent of the number of components of the query function.
However, for the case of multicomponent queries, we can relate the
level of differential privacy for the multicomponent query to the
level of differential privacy of the components. If we have a query
$f=(f_{1},\ldots,f_{n})$ and for each of the components, $f_{i}$,
we get an $\varepsilon_{i}$-differentially private response, then
we get a $\sum_{i=1}^{n}\varepsilon_{i}$-differentially private response
for $f$. This is in fact a property of $\varepsilon$-differential
privacy, hence a proof for our specific mechanism is not required
(see~\cite{Mcsherry07mechanismdesign}).

The above result on multicomponent queries can be improved when each
of the queries refers to a disjoint set of individuals. For the noise
addition mechanism, it easy to see that, when performing queries $f_{1},\ldots,f_{n}$
that refer each to a disjoint set of individuals, the global sensitivity
equals the maximum of the sensitivities of the individual queries~\cite{Dwork2006a}.
The reason is that, by adding or removing a single individual from
the data set, only one of the queries is affected. This is a good
property, as it guarantees $\max\{\varepsilon_{i}\}$-differential
privacy instead of $\sum\varepsilon_{i}$-differential privacy. Our
goal is to show that this property can also be achieved for our proposal.
In fact, we will show further on that this is also a general property
of differential privacy. We start with an example.
\begin{example}
Let $D$ be a database with two attributes: an identifier $ID$ and
a Boolean attribute $B$. Let $f_{1}$ and $f_{2}$ be queries that
return the value of $B$ for individuals 1 and 2, respectively. Let
the prior knowledge for both queries be the independent uniform distribution
over the set $\{0,1\}$, which assigns a prior probability 0.5 to
each of the possible outcomes for each query. To respond to $f_{1}$
in an $\varepsilon$-differentially private way with $\varepsilon=1$,
we select factors $\alpha_{u}=e^{\varepsilon}$ and $\alpha_{d}=e^{-\varepsilon}$
that modify the prior knowledge. The same factors are selected for
$f_{2}$. Now we want to check whether the combination of responses
to $f_{1}$ and $f_{2}$ is still $\varepsilon$-differentially private.

For the sake of simplicity, we assume that both individuals are in
$D$, and that $f_{1}(D)=0$ and $f_{2}(D)=0$. For the rest of cases
we would proceed in a similar way. Figure~\ref{fig:distributions_boolean}
shows the prior knowledge and the output distribution for both query
functions $f_{1}$ and $f_{2}$. Indeed, by setting $\alpha_{d}=e^{-\varepsilon}$
and adjusting the probability mass to one instead of setting $\alpha_{u}=e^{\varepsilon}$,
we have

\[
\begin{array}{c}
P(K_{f_{1}}(D)=1|f_{1}(D)=0)=P(K_{f_{1}}(D)=1|f_{2}(D)=0)=\\
=0.5\alpha_{d}=0.5e^{-1}=0.1839
\end{array}
\]

\[
\begin{array}{c}
P(K_{f_{1}}(D)=1|f_{1}(D)=1)=P(K_{f_{1}}(D)=1|f_{2}(D)=1)=\\
=1-0.5\alpha_{d}=0.8161
\end{array}
\]

\begin{figure}[h]
\begin{centering}
\includegraphics[width=10cm]{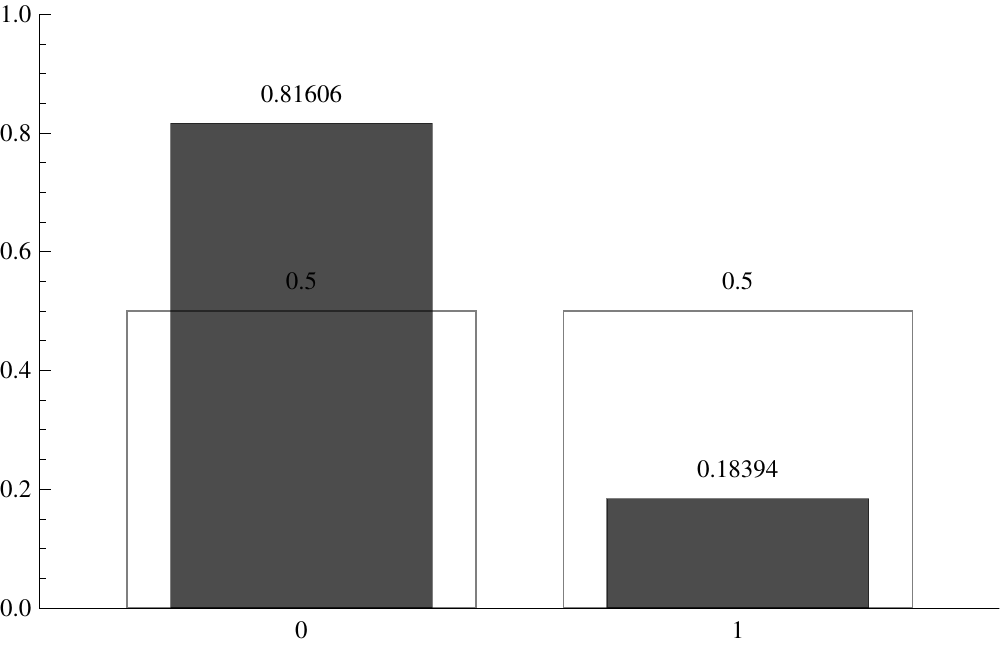}
\par\end{centering}

\caption{\label{fig:distributions_boolean}Prior knowledge about attribute
$B$ and distribution of the $\varepsilon$-differentially private
response to query functions $f_{1}$ and $f_{2}$, assuming that the
actual value for attribute $B$ is 0}
\end{figure}

Table~\ref{tab:distr_comp} shows the joint distribution for the
output of $(f_{1},f_{2})$, which is obtained by multiplying the output
distributions for $f_{1}$ and $f_{2}$. 

\begin{table}[h]
\caption{\label{tab:distr_comp}Distribution of the differentially private
response to the two-component query $(f_{1},f_{2})$ when the true
values are $f_{1}(D)=f_{2}(D)=0$}

\centering{}%
\begin{tabular}{ccc|c|c|}
 &  & \multicolumn{1}{c}{} & \multicolumn{1}{c}{0} & \multicolumn{1}{c}{1}\tabularnewline
\cline{4-5} 
 &  & $K_{f_{1}}$ & $1-0.5\alpha_{d}$ & $0.5\alpha_{d}$\tabularnewline
\cline{4-5} 
 & $K_{f_{2}}$ & \multicolumn{1}{c}{} & \multicolumn{1}{c}{} & \multicolumn{1}{c}{}\tabularnewline
\cline{2-2} \cline{4-5} 
\multicolumn{1}{c|}{0} & \multicolumn{1}{c|}{$1-0.5\alpha_{d}$} &  & $(1-0.5\alpha_{d})^{2}$ & $(1-0.5\alpha_{d})$$0.5\alpha_{d}$\tabularnewline
\cline{2-2} \cline{4-5} 
\multicolumn{1}{c|}{1} & \multicolumn{1}{c|}{$0.5\alpha_{d}$} &  & $(1-0.5\alpha_{d})$$0.5\alpha_{d}$ & $0.25\alpha_{d}^{2}$\tabularnewline
\cline{2-2} \cline{4-5} 
\end{tabular}
\end{table}

For $\varepsilon$-differential privacy to hold for the two-component
query $f=(f_{1},f_{2})$, the ratio of the response distribution at
$D$ and the response distribution at any $D'$ that results from
$D$ by adding or removing a single individual must be within the
range $[e^{-\varepsilon},e^{\varepsilon}]$. As $f_{1}$ and $f_{2}$
are related to individuals 1 and 2, any modification to $D$ that
does not affect the records for those individuals leaves the distribution
of responses unchanged. As we are assuming that individuals 1 and
2 are in $D$, the only modifications to be considered are the removal
of one of these individuals. Table~\ref{tab:distr_resp} shows the
distributions of responses when individual 1 or 2 are removed. We
use $K_{f}$ to denote the distribution of the response to query $f$.
It can be seen that the respective ratios between the distribution
in Table~\ref{tab:distr_comp} and the ones in Table~\ref{tab:distr_resp}
are within $[e^{-\varepsilon},e^{\varepsilon}]=[e^{-1},e]$; specifically,
the ratios take only two values, $\alpha_{d}=e^{-1}$ and $2-\alpha_{d}=2-e^{-1}$.

\begin{table}[h]
\caption{\label{tab:distr_resp}Distribution of the response to query $f=(f_{1},f_{2})$
when either individual 1 is missing (top) or individual 2 is missing
(bottom), and when the attribute value for the non-missing individual
is 0.}

\begin{centering}
\begin{tabular}{c>{\centering}p{2cm}c|>{\centering}p{2.5cm}|>{\centering}p{2.5cm}|}
 &  & \multicolumn{1}{c}{} & \multicolumn{1}{>{\centering}p{2.5cm}}{0} & \multicolumn{1}{>{\centering}p{2.5cm}}{1}\tabularnewline
\cline{4-5} 
 &  & $K_{f_{1}}$ & $0.5$ & $0.5$\tabularnewline
\cline{4-5} 
 & $K_{f_{2}}$ & \multicolumn{1}{c}{} & \multicolumn{1}{>{\centering}p{2.5cm}}{} & \multicolumn{1}{>{\centering}p{2.5cm}}{}\tabularnewline
\cline{2-2} \cline{4-5} 
\multicolumn{1}{c|}{0} & \multicolumn{1}{>{\centering}p{2cm}|}{$1-0.5\alpha_{d}$} &  & $0.5(1-0.5\alpha_{d})$ & $0.5(1-0.5\alpha_{d})$\tabularnewline
\cline{2-2} \cline{4-5} 
\multicolumn{1}{c|}{1} & \multicolumn{1}{>{\centering}p{2cm}|}{$0.5\alpha_{d}$} &  & $0.25\alpha_{d}$ & $0.25\alpha_{d}$\tabularnewline
\cline{2-2} \cline{4-5} 
 &  & \multicolumn{1}{c}{} & \multicolumn{1}{>{\centering}p{2.5cm}}{} & \multicolumn{1}{>{\centering}p{2.5cm}}{}\tabularnewline
\end{tabular}
\par\end{centering}

\centering{}%
\begin{tabular}{c>{\centering}p{2cm}c|>{\centering}p{2.5cm}|>{\centering}p{2.5cm}|}
 &  & \multicolumn{1}{c}{} & \multicolumn{1}{>{\centering}p{2.5cm}}{0} & \multicolumn{1}{>{\centering}p{2.5cm}}{1}\tabularnewline
\cline{4-5} 
 &  & $K_{f_{1}}$ & $1-0.5\alpha_{d}$ & $0.5\alpha_{d}$\tabularnewline
\cline{4-5} 
 & $K_{f_{2}}$ & \multicolumn{1}{c}{} & \multicolumn{1}{>{\centering}p{2.5cm}}{} & \multicolumn{1}{>{\centering}p{2.5cm}}{}\tabularnewline
\cline{2-2} \cline{4-5} 
\multicolumn{1}{c|}{0} & \multicolumn{1}{>{\centering}p{2cm}|}{$0.5$} &  & $0.5(1-0.5\alpha_{d})$ & $0.25\alpha_{d}$\tabularnewline
\cline{2-2} \cline{4-5} 
\multicolumn{1}{c|}{1} & \multicolumn{1}{>{\centering}p{2cm}|}{$0.5$} &  & $0.5(1-0.5\alpha_{d})$ & $0.25\alpha_{d}$\tabularnewline
\cline{2-2} \cline{4-5} 
\end{tabular}
\end{table}

\end{example}
We now state and prove in general the property illustrated in the
previous example.
\begin{prop}
\label{Proposition3}Let $D$ be a data set and let $(f_{1},\ldots,f_{n})$
be a set of query functions related to disjoint sets of individuals.
Let $K_{f_{i}}$ be a random variable that provides $\varepsilon_{i}$-differential
privacy for $f_{i}$, and assume that $K_{f_{i}}$ is independent
from $K_{f_{j}}$ for any $i\ne j$. Then $(K_{f_{1}},\ldots,K_{f_{n}})$
provides $\max\{\varepsilon_{i}\}$-differential privacy for $(f_{1},\ldots,f_{2})$.\end{prop}
\begin{proof}
Let $D'$ be a data set obtained from $D$ by adding or removing a
single user. We want to check that the following inequalities hold
for any subset $S$ of the range of $(K_{f_{1}},\ldots,K_{f_{n}})$:
\[
e^{-\max\{\varepsilon_{i}\}}\le\frac{P((K_{f_{1}}(D),\ldots,K_{f_{n}}(D))\in S)}{P((K_{f_{1}}(D'),\ldots,K_{f_{n}}(D'))\in S)}\le e^{\max\{\varepsilon_{i}\}}
\]
It is easy to show that the above inequality holds for the case of
$S$ being the Cartesian product of sets $S_{i}$, with $S_{i}$ a
subset of the range of $K_{f_{j}}(D)$, or when the probability distribution
of $(K_{f_{1}},\ldots,K_{f_{n}})$ is absolutely continuous. For a
general set $S$ and a non absolutely continuous distribution, the
inequalities still hold. 
However, we restrict the proof for  
$S=S_{1}\times\ldots\times S_{n}$.

The probabilities $P((K_{f_{1}}(D),\ldots,K_{f_{n}}(D))\in S)$ and
$P((K_{f_{1}}(D'),\ldots,K_{f_{n}}(D'))\in S)$ can be written as
the product of probabilities $\prod P(K_{f_{i}}(D)\in S_{i})$ and
$\prod P(K_{f_{i}}(D')\in S_{i})$, respectively. By adding or removing
a single individual, only one of the queries is affected. Say the
affected query is $f_{j}$ for some $j\in\{1,\ldots,n\}$. By removing
the factors that are both in the numerator and the denominator, the
inequalities that we need to check become 
\[
e^{-\max\{\varepsilon_{i}\}}\le\frac{P(K_{f_{j}}(D)\in S_{j})}{P(K_{f_{j}}(D')\in S_{j})}\le e^{\max\{\varepsilon_{i}\}},
\]
 which holds because $K_{f_{j}}$ satisfies $\varepsilon_{j}$-differential
privacy, and $\varepsilon_{j}\le\max\{\varepsilon_{i}\}$.
\end{proof}

\section{Interactive queries and adaptive attacks\label{sec:interactive}}

Differential privacy is usually presented as an interactive query-response
mechanism where the data set is held by a trusted party to whom users
send their queries. Despite this claimed interactivity, the formal
definition of differential privacy (Definition~\ref{def:dp_dwork})
is based on a single query, thereby removing the complexities that
interactivity would introduce. Malicious users may try to use interaction
to exploit potential vulnerabilities of the access mechanism. When
using Laplace noise addition the user can, for example, use the knowledge
acquired from previous answers to forge the new query. For knowledge
refinement the problem is even more compelling, since, besides the
query function, the user also feeds the access mechanism with a prior
knowledge distribution and optionally with a distance function.

\subsection{Interactive access mechanisms}

To implement interactivity, a protocol is built on top of the non-interactive
access mechanism. The idea is quite simple; when a query is submitted,
the access mechanism analyzes if answering the query is too disclosive,
in which case the query is simply discarded. To determine if answering
a new query is too disclosive, all the queries submitted by a user
so far, including the new query, are treated as a single multicomponent
query and $\varepsilon$-differential privacy is enforced for it.
Protocol~\ref{alg:interactive_laplace} describes the protocol for
the interactive Laplace noise access mechanism introduced in~\cite{Dwork2006a};
in the protocol, $\Delta(\cdot)$ stands for sensitivity. 

\begin{protocol} 
\caption{Interactive Laplace noise addition mechanism} \label{alg:interactive_laplace}
\begin{enumerate} 
\item The database holder initializes the access mechanism with  the following parameters: \begin{itemize} 
	\item $\varepsilon$, the maximum level of leakage allowed; 
	\item $\lambda$, the amount of noise to be added to every response ($\lambda$ is the parameter of the Laplace noise distribution); for   fixed $\varepsilon$, the greater $\lambda$, the more queries  the access mechanism will be able to answer. 
\end{itemize} 
\item Let $i:=1$. 
\item {\bf while} queries are answered by the access mechanism {\bf do} 
\begin{enumerate} 
\item The user submits a query $f_i$ (for $i >1$, $f_i$ may depend on responses to previous queries $(f_1, \cdots, f_{i-1})$). 
\item {\bf if} $\Delta(f_1,\cdots, f_i)/\lambda \leq \varepsilon$ {\bf then} the  access mechanism returns $f_i(D) + \mbox{Laplace}(\lambda)$ as  response; {\bf else} it returns nothing. 
\item $i:=i+1$ 
\end{enumerate} 
\end{enumerate} 
\end{protocol}

We now present an interactive knowledge refinement mechanism parallel
to the Laplace-based one. As knowledge refinement does not depend
on the sensitivity of the query function, our interactive mechanism
does not need to compute sensitivities and is therefore simpler than
the Laplace-based one. Also, we will allow the database user to select
the amount of leakage $\varepsilon_{i}$ independently for each query
$f_{i}$. The only requirement is that the access mechanism will refuse
answering query $f_{i}$ (and successive queries) if the leakage of
the multicomponent query $(f_{1},\cdots,f_{i})$ exceeds $\varepsilon$.

\begin{protocol} 
\caption{Interactive mechanism for knowledge refinement} 
\label{alg:interactive_refinement}
\begin{enumerate} 
	\item The database holder initializes the access mechanism with  $\varepsilon$, the maximum level of leakage allowed. 
	\item Let $i:=1$. 
	\item {\bf while} queries are answered by the access mechanism {\bf do} 
	\begin{enumerate} 
		\item The user submits a query $q_i=(f_i,P_{f_i},d_i, \varepsilon_i)$, where $f_i$ is the query function, $P_{f_i}$ is the prior knowledge distribution for the query, $d_i$ is the distance function to be used and  $\varepsilon_i$ is the desired level of leakage (for $i>1$, $q_i$ may depend on responses to previous queries $(q_1,\cdots,q_{i-1})$). 
		\item {\bf if} $\sum_{j=1}^i \varepsilon_j \leq \varepsilon$ {\bf then} the access mechanism returns a response to $f_i$ resulting from applying knowledge refinement to $P_{f_i}$ with distance $d_i$ so that $\varepsilon_i$-differential privacy is guaranteed; {\bf else} it returns nothing. 
		\item $i:=i+1$ 
	\end{enumerate} 
\end{enumerate} 
\end{protocol} 

If the $d$-th query is the last query answered by the interactive
mechanism of Protocol~\ref{alg:interactive_refinement}, by construction
the user obtains at most a knowledge gain $\varepsilon$ for $(f_{1},\cdots,f_{d})$.
This holds regardless of the prior knowledge distributions and distance
functions chosen by the user for each query.

By submitting the desired level of leakage $\varepsilon_{i}$ for
each query, in Protocol~\ref{alg:interactive_refinement} the database
user is allowed to trade more accurate answers in some queries for
less accurate answers in other queries. Protocol~\ref{alg:interactive_laplace}
could be modified to permit such flexibility as well: the user could
be asked to choose the noise parameter $\lambda_{i}$ for the $i$-th
query, and the condition checked by the access mechanism would become
\[
\sum_{j=1}^{i}\Delta(f_{j})/\lambda_{j}\leq\varepsilon
\]
 Since $\Delta(f_{1},\cdots,f_{i})\leq\Delta(f_{1})+\cdots+\Delta(f_{i})$,
when $\lambda_{1}=\cdots=\lambda_{i}$ the modified condition above
may result in less queries being answered than the condition in Protocol~\ref{alg:interactive_laplace}.

\subsection{Adaptive attacks}

The interactive mechanisms of Protocols~\ref{alg:interactive_laplace}
and~\ref{alg:interactive_refinement} guarantee, respectively for
Laplace noise and knowledge refinement, that the responses to any
sequence of adaptive queries $(q_{1},\cdots,q_{d})$ will not violate
$\varepsilon$-differential privacy. However, the following question
can be raised: is there any sequence of adaptive queries $(q_{1},\cdots,q_{d})$
and a way to combine the responses to this sequence that allows an
attacker to obtain an estimator of $f(D)$ that does not satisfy $\varepsilon$-differential
privacy? 

We show that such an attack cannot succeed. Our proof is completely
general; it does not depend on the access mechanism used to attain
differential privacy. Let $F:\mathbb{R}^{d}\rightarrow\mathbb{R}$
be the function used by the attacker to combine the responses to $q_{1},\ldots,q_{d}$;
let these responses be samples of the random vector $K_{f_{1}}(D),\cdots,K_{f_{d}}(D)$.
The attacker computes $F(K_{f_{1}}(D),\cdots,K_{f_{d}}(D))$ and takes
it as the response to $f(D)$. We are not interested in determining
$F$ or even in determining whether $F(K_{f_{1}}(D),\cdots,K_{f_{d}}(D))$
is a good estimate for $f(D)$. The following result will suffice.
\begin{prop}
\label{attack}For any function $F$, if $(K_{f_{1}}(D),\cdots,K_{f_{d}}(D))$
satisfies $\varepsilon$-differential privacy, then $F(K_{f_{1}}(D),\cdots,K_{f_{d}}(D))$
also satisfies $\varepsilon$-differential privacy.\end{prop}
\begin{proof}
We need to check that, for each pair of data sets $D$ and $D'$ that
differ in a single individual and for each set $S\in Range(F(K_{f_{1}},\cdots,K_{f_{d}}))$,
it holds that 
\[
\frac{P(F(K_{f_{1}}(D),\cdots,K_{f_{d}}(D)))\in S)}{P(F(K_{f_{1}}(D'),\cdots,K_{f_{d}}(D')))\in S)}\leq e^{\varepsilon}
\]
 Since $P(F\circ X\in S)=P(X\in F^{-1}(S))$, we can express the previous
inequality as 
\[
\frac{P((K_{f_{1}}(D),\cdots,K_{f_{d}}(D))\in F^{-1}(S))}{P((K_{f_{1}}(D'),\cdots,K_{f_{d}}(D'))\in F^{-1}(S))}\leq e^{\varepsilon}
\]
 which holds because $(K_{f_{1}}(D),\cdots,K_{f_{d}}(D))$ satisfies
$\varepsilon$-differential privacy.
\end{proof}
The following corollary follows from the previous proposition.
\begin{cor}
Whatever the attacker's strategy, her estimate for $f(D)$ always
satisfies $\varepsilon$-differential privacy.
\end{cor}

\section{Quality of the response to individual queries\label{sec:Quality}}

We have defined an individual query, $f$, to be one that depends
on a single individual. We can think of it as a query that returns
the value of some attribute for some specific individual. 

Typical differential privacy mechanisms based on noise addition provide
low data quality responses for individual queries. The reason is that,
as any individual can take any value in $Range(f)$, the sensitivity
of the query equals the length of $Range(f)$. When using knowledge
refinement, the quality of the response depends to a great extent
on the prior knowledge available. 

In this section, we provide some data quality comparisons between
Laplace noise addition and knowledge refinement for individual queries.
Comparisons will be based of specific query functions. The first one
is based on a query function that returns a Boolean value; we show
how the distribution for the differentially private response gets
closer to the real response by refining prior knowledge than by adding
Laplace noise. The second comparison is based on a continuous function
with range $[0,1]$; we show that, even if we have no prior knowledge,
knowledge refinement provides better data quality for individual queries.

\subsection{Data quality for a Boolean attribute}

Consider a simple database $D$ with two attributes: an identifier
$ID$ and a Boolean attribute $B$ that may take values 0 and 1. We
assume that $B$ is very sensitive and that, to limit the disclosure
risk, access to the database must be mediated by a query-response
mechanism satisfying differential privacy, with $\varepsilon=1$.
Let $f:\mathcal{D}\rightarrow\{0,1\}$ be a query that asks the value
of attribute $B$ for a specific individual.

To achieve differential privacy via Laplace noise addition, we must
first compute the sensitivity of function $f$. Assuming that $f$
returns $1/2$ if the individual is not in the database,
the $L_{1}$-sensitivity of $f$ is $1/2$. Therefore,
to achieve differential privacy for $\varepsilon=1$, we must add
a Laplace distribution $L(0,1/2)$ to the true value of
the query response. Figure~\ref{fig:Ex2_Laplace} shows the distribution
of the responses for both possible values of $B$, 0 and 1.

\begin{figure}[h]
\begin{centering}
\includegraphics[width=10cm]{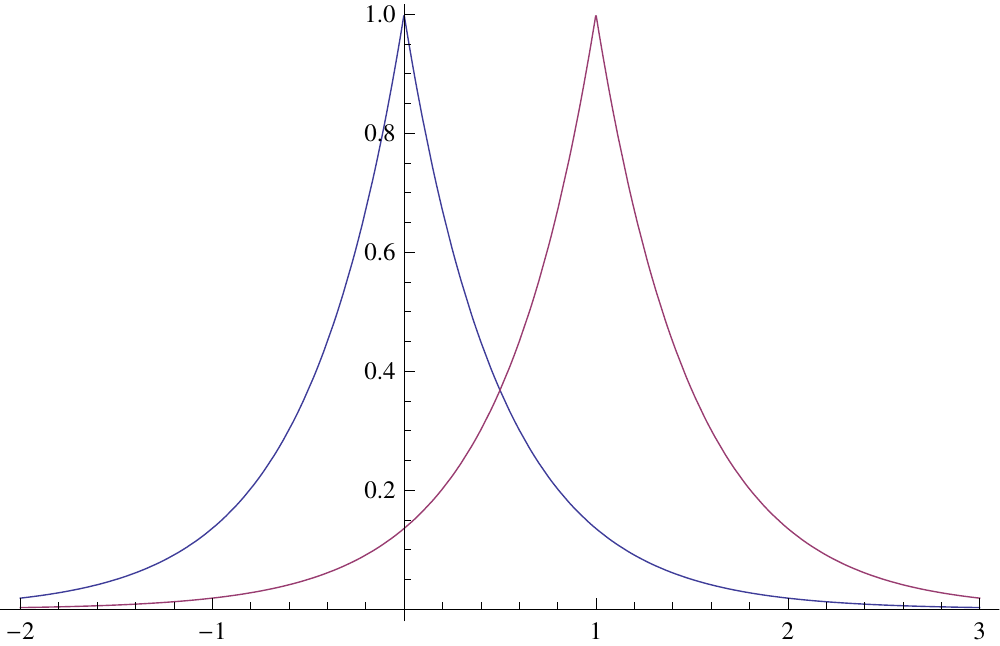}
\par\end{centering}

\caption{Response distributions with Laplace noise addition\label{fig:Ex2_Laplace}}
\end{figure}

Assuming that the user is only interested in a 0/1 response, any value
below $1/2$ is taken as 0, and any value above $1/2$
as 1. The distribution for the response thus obtained is:

\[
K_{f}(D)=\begin{cases}
0 & \mbox{if $f(D)+L(0,1/2)<0.5$}\\
1 & \mbox{otherwise.}
\end{cases}
\]

If $f(D)$ equals 0, $K_{f}(D)$ follows a Bernoulli distribution
with parameter 0.184. If $f(D)$ equals 1, the distribution of $K_{f}(D)$
is a Bernoulli with parameter 0.816. Note that this is completely
independent from the true distribution of attribute $B$, and from
any previous knowledge that the user might have on it. Hence, differential
privacy via Laplace noise addition does not let the user exploit prior
knowledge.

Let us assume that attribute $B$ is 1 only with probability 0.01.
For a user with this information, using the response obtained from
the differential privacy mechanism is actually misleading, as the
result will be 1 with probability

\[
\begin{array}{c}
P(K_{f}(D)=1)=\\
P(K_{f}(D)=1|f(D)=0)P(f(D)=0)+P(K_{f}(D)=1|f(D)=1)P(f(D)=1)\\
=0.184\cdot0.99+0.816\cdot0.01=0.19
\end{array}
\]

We could increase the parameter $\varepsilon$ to get a more accurate
response. However, by doing so we would be reducing the privacy guarantees.

Now, we turn to the refinement mechanism and, same as before, we assume
that the user knows that $B$ equals 1 with probability 0.01. Take
$\alpha_{u}=e^{\varepsilon}=e$ and $\alpha_{d}=e^{-\varepsilon}=e^{-1}$.
Hence,

\[
\begin{array}{c}
P(K_{f}(D)=1|f(D)=0)=P(f(D)=1)\cdot\alpha_{d}=0.003678\\
P(K_{f}(D)=0|f(D)=0)=1-0.003678=0.9963222\\
P(K_{f}(D)=1|f(D)=1)=P(f(D)=1))\cdot\alpha_{u}=0.027182\\
P(K_{f}(D)=0|f(D)=1)=1-0.027182=0.972817
\end{array}
\]

Note that, as this is not an absolutely continuous distribution, we
had to do some adjustment to have a total probability mass equal to
one: instead of adjusting $\alpha_{u}$ and $\alpha_{d}$, we directly
adjusted $P(K_{f}(D)=0|f(D)=0)$ and $P(K_{f}(D)=0|f(D)=1)$. Figure~\ref{fig:Response-distribution-with}
depicts the distribution of the response for both possible values
of attribute $B$ and for the prior knowledge. 

\begin{figure}[h]
\begin{centering}
\includegraphics[width=10cm]{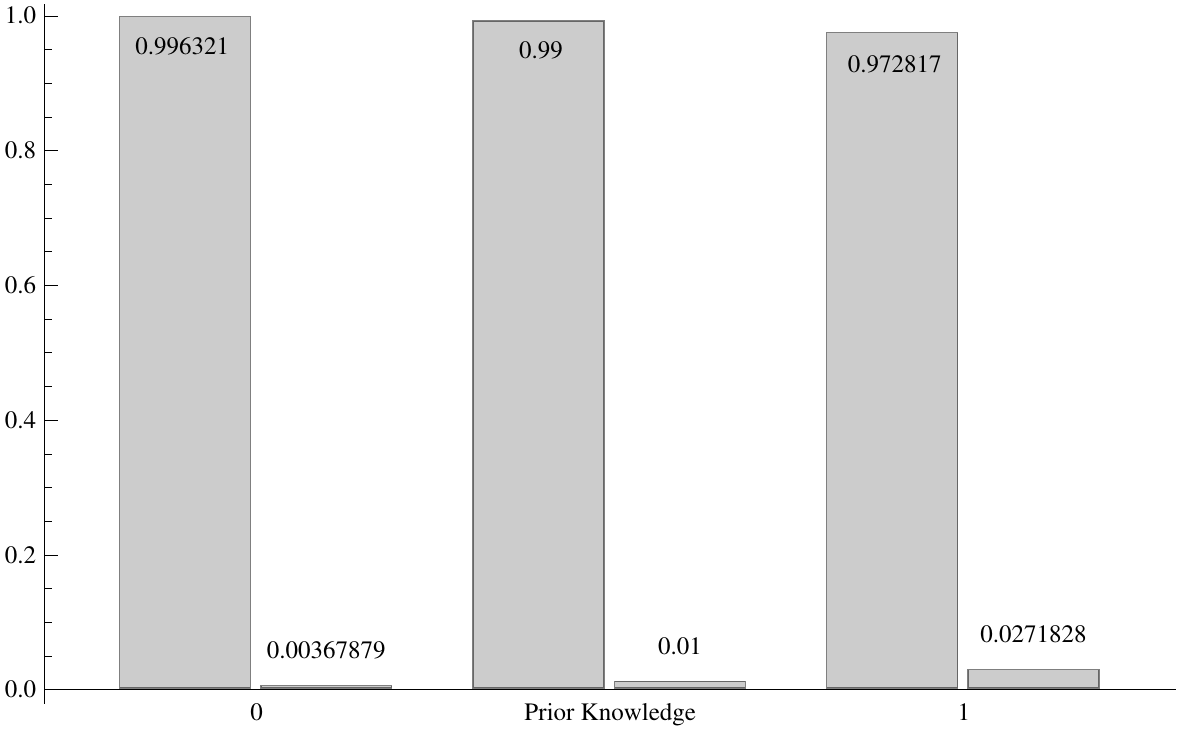}
\par\end{centering}

\caption{\label{fig:Response-distribution-with}Response distribution with
prior knowledge refinement}
\end{figure}

Now, the probability of obtaining a response 1 is

\[
\begin{array}{c}
P(K_{f}(D)=1)=\\
=P(K_{f}(D)=1|f(D)=0)P(f(D)=0)+P(K_{f}(D)=1|f(D)=1)P(f(D)=1)\\
=0.003678\cdot0.99+0.02182\cdot0.01=0.003912
\end{array}
\]

As 0.003912 is much closer to 0.01 than 0.19, we conclude that, despite
both mechanisms providing the same level of privacy, the output distribution
is much closer to the actual distribution of the attribute when using
the mechanism based on knowledge refinement. Therefore, knowledge
refinement outperforms Laplace noise addition for Boolean attributes
released under differential privacy.

\subsection{Data quality for a continuous attribute}

Let $f:\mathcal{D}\rightarrow[0,1]$ be a query function that returns
a value in the interval $[0,1]$. We have fixed the range of $f$
to be able to obtain some numerical results, but a similar comparison
can be done for other ranges. We compare the response obtained by
using Laplace noise addition and knowledge refinement with a uniform
$U[0,1]$ prior knowledge.

When using Laplace noise addition, the response to $f(D)$ is $K_{f}(D)=f(D)+Laplace(0,1/\varepsilon)$.
When using knowledge refinement, the prior knowledge is modified by
increasing the probability of the set $\mathcal{U}_{u}$ containing
the points closer to $f(D)$ by a factor $\alpha_{u}$, and decreasing
the probability of the rest by a factor $\alpha_{d}$. We saw in Section~\ref{sec:3 gen_algorithm}
that $\mathcal{U}_{u}$ must satisfy $P_{f}(\mathcal{U}_{u})=(\alpha_{u}-1)/(\alpha_{u}-\alpha_{d})$,
which in the case of a uniform prior knowledge within the interval
$[0,1]$ coincides with the size of $\mathcal{U}_{u}$. We also saw
(Table~\ref{tab:2 alpha}) that, for an individual query, the factors
are $\alpha_{u}=e^{\varepsilon}$ and $\alpha_{d}=e^{-\varepsilon}$.

Table~\ref{tab:Comparison} shows a comparison of the distribution
for the response to $f(D)$ for several values of $\varepsilon$ when
$f(D)=0.5$. For Laplace noise addition, we have computed the variance
of the response, as well as the probability for the response to be
within the range $[0,1]$. For knowledge refinement, we have computed
the variance of the response, the size of $\mathcal{U}_{u}$, and
the probability for the response to be in $\mathcal{U}_{u}$. The
results in the table show that knowledge refinement behaves much better
than Laplace noise addition, but perhaps this is better observed by
comparing the actual distributions. Figure~\ref{fig:Comparison}
shows the distributions for the response when using Laplace noise
addition and knowledge refinement with the same values of $\varepsilon$
used in the table.

\begin{table}[h]
\caption{\label{tab:Comparison}Comparison between the distribution of the
response to $f(D)$ for Laplace noise addition and knowledge refinement
for several values of $\varepsilon$ when $f(D)=0.5$}
\begin{centering}
\begin{tabular}{cccccc}
\hline 
 & \multicolumn{2}{c}{Laplace noise addition} & \multicolumn{3}{c}{Knowledge refinement}\tabularnewline
$\varepsilon$ & Variance & $P(K_{f}(D)\in[0,1])$ & Variance & $size(\mathcal{U}_{u})$ & $\ensuremath{P(K_{f}(D)\in\mathcal{U}_{u})}$\tabularnewline
\hline 
0.1 & 200 & 0.476 & 0.077 & 0.475 & 0.525\tabularnewline
$\ln(2)$ & 4.16 & 0.549 & 0.046 & 0.333 & 0.667\tabularnewline
1 & 2 & 0.607 & 0.034 & 0.269 & 0.731\tabularnewline
2 & 0.5 & 0.684 & 0.012 & 0.119 & 0.881\tabularnewline
\hline 
\end{tabular}
\par\end{centering}
\end{table}

\begin{figure}[h]
\begin{centering}
\includegraphics[width=7cm]{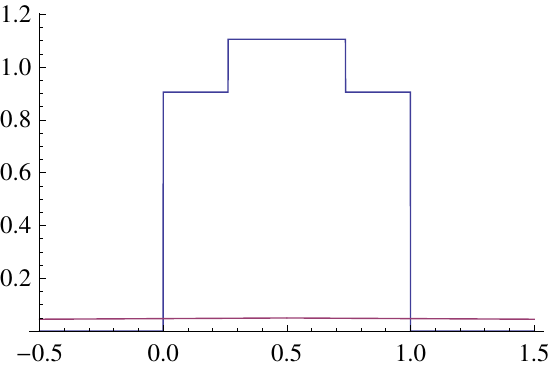}~\includegraphics[width=7cm]{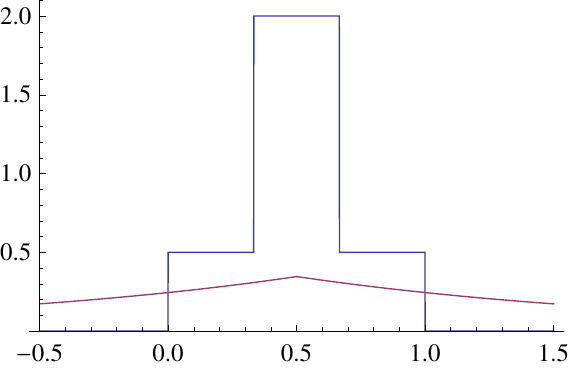}
\par\end{centering}

\begin{centering}
\includegraphics[width=7cm]{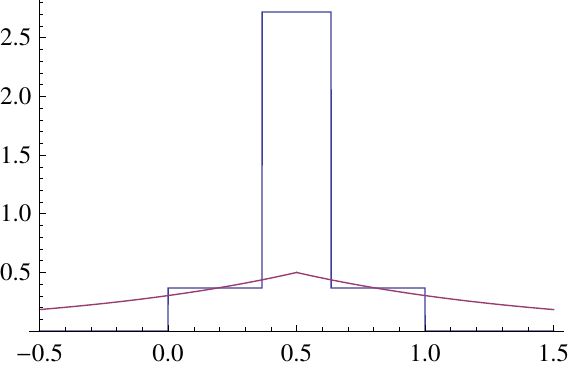}~\includegraphics[width=7cm]{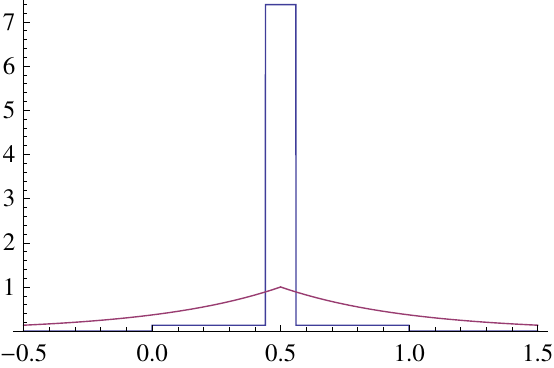}
\par\end{centering}

\caption{\label{fig:Comparison}Distribution for the response to $f(D)$, when
$f(D)=0.5$, for Laplace noise addition (distribution with unbounded
support) and knowledge refinement (distribution with support $[0,1]$)
for $\varepsilon=0.1$ (top left), $\varepsilon=\ln(2)$ (top right),
$\varepsilon=1$ (bottom left), and $\varepsilon=2$ (bottom right)}
\end{figure}

\section{Discussion\label{sec:Discussion}}

In previous sections we have highlighted that the knowledge refinement
mechanism lets the database user exploit her prior knowledge to obtain
a more accurate response. In Section~\ref{sec:Quality} we saw that,
for the case of individual queries, knowledge refinement provides
a much more accurate response even when there is no prior knowledge. 

Other advantages of prior knowledge refinement are:
\begin{itemize}
\item \emph{Simplicity}. Mechanisms such as Laplace noise addition are based
on the addition of a random noise whose magnitude depends on the variation
of the query function across neighbor data sets, also known as sensitivity.
To calibrate the random noise, the sensitivity of the function must
be computed, which may be quite complex. The mechanism based on the
refinement of the prior knowledge only depends on the prior knowledge
(it is independent from the sensitivity of the query function), and
thus it is easier to implement, especially in a non-supervised environment.
\item \emph{Generality}. As said above, Laplace noise addition requires
computing the sensitivity of the query function, and this can only
be done if the query function takes values in a metric space. This
introduces some complexities when the function returns categorical
information. The mechanism based on prior knowledge refinement does
not impose any requirement on the query function, and thus it can
be applied without extra overhead to functions returning categorical
information.
\item Consistency. Knowledge refinement lets the database user easily restrict
the response to a set of values consistent with the query function,
by having the prior knowledge assign a probability mass of zero to
the set of inconsistent values. For example, in Table~\ref{tab:Comparison}
we saw that Laplace noise sends the response outside the query function
range $[0,1]$ with great probability, while knowledge refinement
always keeps the response within range. Querying categorical attributes
is another example. It is usual to have some combinations of categories
that do not make sense. For example, if the attributes are ``employed''
(Y/N), and ``unemployment benefits'' (Y/N), a response Y for both
attributes does not make sense. When using a noise addition mechanism,
there is no way to avoid that combination of values, while, when using
knowledge refinement, to avoid that combination we only have to use
a prior knowledge distribution that assigns zero probability mass
to it.
\end{itemize}
Despite the advantages listed above, there are some situations for
which the proposed mechanism is not appropriate. If the range of values
that the function may return is large compared to the variability
between neighbor data sets, and the database user does not have precise
knowledge of the response, then a method based on noise addition produces
better data quality. This may be the case of statistical queries where
the user has no prior knowledge of the result. However, when querying
about a specific individual, the proposed method results in much greater
response quality.

\section{Conclusions}

We have introduced a novel mechanism to attain differential privacy.
This mechanism is based on refining the prior knowledge that the user
may have about the query response. This refinement is performed taking
into account the constraints imposed by differential privacy.

The refinement mechanism presents several advantages over the usual
noise addition mechanism. It is easier to implement, especially in
a non-supervised environment, as it does not require potentially complex
computations (such as determining the sensitivity of the query function).
The fact that it lets users exploit their prior knowledge may lead
to a level of data quality not reachable by mechanisms independent
of the user knowledge. For example, we showed in the examples of Section~\ref{sec:Quality}
that the distribution of the response was closer to the real distribution
when using the refinement mechanism. For query functions with great
sensitivity, the amount of noise added by noise addition mechanisms,
such as~\cite{Dwork2006a}, may render the response useless. In contrast,
the data quality that results from our proposal is independent from
the sensitivity of the query function; yet this has the drawback that,
for small sensitivities, our approach may be inferior to noise addition.

We have also analyzed the behavior of our approach for multicomponent
queries. A generic property of differential privacy guarantees that,
if a $\varepsilon_{i}$-differentially private response is provided
for a query $f_{i}$, for $i=1$ to $n$, a $\sum\varepsilon_{i}$-differentially
private response is provided for the query $(f_{1},\ldots f_{n})$.
We have seen that this can be improved if each query $f_{i}$ refers
to a disjoint set of individuals. In this case, we achieve $\max\{\varepsilon_{i}\}$-differential
privacy, instead of $\sum\varepsilon_{i}$-differential privacy. Interactive
mechanisms for Laplace noise addition and knowledge refinement have
also been described. Such interactive mechanisms take as input parameter
the maximum level of leakage $\varepsilon$ allowed by the database
holder, and queries are answered until that level of leakage is reached.
The knowledge refinement interactive mechanism is superior to the
Laplace noise interactive mechanism in that it does not need to compute
sensitivities. We have shown that any interactive mechanism providing
$\varepsilon$-differential privacy is safe against adaptive attacks;
whatever the strategy used by an attacker to combine query responses,
$\varepsilon$-differential privacy holds.

\lhead[\chaptername~\thechapter]{\rightmark}

\rhead[Enhancing utility in
differential privacy via $k$-anonymity]{}


\lfoot[\thepage]{}

\cfoot{}

\rfoot[]{\thepage}

\chapter{Enhancing data utility in differential privacy via microaggregation-based
$k$-anonymity\label{chap:Enhancing-Data-Utility}}

It is not uncommon in the data anonymization literature to oppose
the ``old'' $k$-anonymity model to the ``new'' differential privacy
model, which offers more robust privacy guarantees. Yet, it is often
disregarded that the utility of the masked results provided by differential
privacy is quite limited, due to the amount of noise that needs to
be added to the output, or because utility can only be guaranteed
for a restricted type of queries. This is in contrast with the general-purpose
anonymized data resulting from $k$-anonymity mechanisms, which also
focus on preserving data utility.
In this chapter, we show that a synergy
between differential privacy and $k$-anonymity can be found: $k$-anonymity
can help improving the utility of differentially private query responses.
We devote special attention to the utility improvement of differentially
private published data sets. Specifically, we show that the amount
of noise required to fulfill $\varepsilon$-differential privacy can
be reduced if noise is added to a $k$-anonymous version of the data
set, where $k$-anonymity is reached through a specially designed
microaggregation of all attributes. As a result of noise reduction,
the analytical utility of the anonymized output is increased. The
theoretical benefits of our proposal are illustrated in a practical
setting with an empirical evaluation on a pair of reference data sets. 

The contents of this chapter have been accepted for publication in~\cite{report}.

\section{Introduction}

Publishing microdata (\emph{e.g.}, responses to polls, census information,
healthcare records) collected by organizations such as statistical
agencies is of great interest for the data analysis community. At
the same time, microdata may contain confidential information about
individuals. To overcome this privacy threat, data should be anonymized
before making them available for secondary use~\cite{hundepool2012}.

In the last two decades, several models for data anonymization have
been proposed in the literature. One of the best-known and widely
used is $k$-anonymity~\cite{Samarati1998b}, which aims at making
each record indistinguishable from, at least, $k-1$ other records.
The usual computational procedure to reach $k$-anonymity is a combination
of attribute generalization and local suppression~\cite{Samarati01,Sweeney2002}.
An alternative procedure, especially suitable for attributes with
no obvious generalization hierarchy (like the numerical ones), is
microaggregation~\cite{Domingo2005,Domi02}. Whatever the computational
procedure, $k$-anonymity assumes that identifiers are suppressed
from the data to be released and it focuses on masking quasi-identifier
attributes; these are attributes (\emph{e.g.}, Age, Gender, Zipcode
and Race) that may enable re-identifying the respondent of a record
because they are linkable to analogous attributes available in external
identified data sources (like electoral rolls, phone books, etc.).
$k$-Anonymity does not mask confidential attributes (\emph{e.g.},
salary, health condition, political preferences, etc.) unless they
are also quasi-identifiers. While $k$-anonymity has been shown to
provide reasonably useful anonymized results, especially for small
$k$, it is also vulnerable to attacks based on the possible lack
of diversity of the non-anonymized confidential attributes or on additional
background knowledge available to the attacker~\cite{Domingo2008}.

On the other hand, $\varepsilon$-differential privacy~\cite{Dwork2006}
is a more recent and rigorous privacy model that makes no assumptions
about the attacker's background knowledge. In a nutshell, it guarantees
that the anonymization output is insensitive (up to a factor dependent
on $\varepsilon$) to modifications of individual input records. In
this way, the privacy of an individual is not compromised by her presence
in the data set, which is a much more robust guarantee than the one
offered by $k$-anonymity model. To do so, $\varepsilon$-differential
privacy requires adding an amount of noise to the anonymization output
that depends on the variability of the actual non-anonymized values.
$\varepsilon$-Differential privacy was originally proposed for the
\textit{interactive} scenario, in which, instead of releasing a masked
version of the data, the anonymizer returns noise-added answers to
interactive queries. Compared to the unrestricted and general-purpose
data publication offered by $k$-anonymity, the interactive scenario
of $\varepsilon$-differential privacy severely limits data analysis,
because it only allows answering queries whose number and type are
limited. Otherwise, an adversary could reconstruct some of the original
data~\cite{Chen11}.

It is pointed out in~\cite{Blum2008} that the previous limitation
can be circumvented by allowing an $\varepsilon$-differentially private
data publication (\emph{i.e.}, a \textit{non-interactive} setting),
which supports answering an unlimited number of potentially heterogeneous
queries. However, since $\varepsilon$-differential privacy should
ensure that the probability distribution of the published records
is not changed by \textit{any} modification of a single input record,
the amount of noise that needs to be added to the published data in
such a general setting is so large that it would severely hamper data
utility~\cite{Chen11}. This problem can be minimized in specific
scenarios, but at the expense of preserving usefulness only for restricted
classes of queries~\cite{Blum2008,Dwork2009,Hardt2010}.

In summary, we can conclude that $k$-anonymity enables general-purpose
data publication with reasonable utility at the cost of some privacy
weaknesses. On the contrary, $\varepsilon$-differential privacy offers
a very robust privacy guarantee at the cost of substantially limiting
the utility of anonymized outputs.


We show here that a synergy between both privacy models can be found
in order to achieve $\varepsilon$-differential privacy: $k$-anonymity
can help increasing the utility of differentially private query outputs.
Specifically, we show that the amount of noise required to fulfill
$\varepsilon$-differential privacy can be greatly reduced if the
query is run over a $k$-anonymous version of the data set obtained
through microaggregation of all attributes (instead of running it
on the raw input data). The rationale is that the microaggregation
performed to achieve $k$-anonymity helps reducing the sensitivity
of the input versus modifications of individual records; hence, it
helps reducing the amount of noise to be added to achieve $\varepsilon$-differential
privacy. As a result, data utility can be improved without renouncing
the strong privacy guarantee of $\varepsilon$-differential privacy.

Section~\ref{sec:dp_kanon} discusses the use of a $k$-anonymous
microaggregation step prior to the evaluation of a query function
as a means to reduce the query sensitivity, thereby reducing the noise
required to attain differential privacy. Section~\ref{sec:dp_ds_kanon}
proposes a general algorithm for generating $\varepsilon$-differentially
private data sets that employs the $k$-anonymous microaggregation
procedure described earlier. Implementation details for data sets
with numerical and categorical attributes are given. Section~\ref{sec:evaluation}
reports on an empirical evaluation of the differentially private outputs
obtained from a pair of reference data sets via $k$-anonymous microaggregation;
the output is compared against standard $k$-anonymity and $\varepsilon$-differential
privacy mechanisms regarding data utility and disclosure risk. Section~\ref{sec:conclusions}
presents the conclusions and proposes some lines of future research.

\section{Differential privacy through k-anonymous microaggregation\label{sec:dp_kanon} }

Differential privacy and microaggregation offer quite different disclosure
limitation guarantees. Differential privacy is introduced in a query-response
environment and offers probabilistic guarantees that the contribution
of any single individual to the query response is limited, while microaggregation
is used to protect microdata releases and works by clustering groups
of individuals and replacing them by the group centroid. When applied
to the quasi-identifier attributes, microaggregation achieves $k$-anonymity.
In spite of those differences, we can leverage the masking introduced
by microaggregation to decrease the amount of random noise required
to attain differential privacy.

Let $X$ be a data set with attributes $A_{1},\ldots,A_{m}$, and
$\overline{X}$ be a microaggregated $X$ with minimal cluster size
$k$. Let $M$ be a microaggregation function that takes as input
a data set, and outputs a microaggregated version of it: $M(X)=\overline{X}$.
Let $f$ be an arbitrary query function for which an $\varepsilon$-differentially
private response is requested. A typical differentially private mechanism
takes these steps: capture the query $f$, compute the real response
$f(D)$, and output a masked value $f(X)+N$, where $N$ is a random
noise whose magnitude is adjusted to the sensitivity of $f$.

To improve the utility of an $\varepsilon$-differentially private
response to $f$, we seek to minimize the distortion introduced by
the random noise $N$. Two main approaches are used in the literature.
In the first approach, a random noise is used that allows for a finer
calibration to the query $f$ under consideration. For instance, if
the variability of the query $f$ is highly dependent on the actual
data set $X$, using a data-dependent noise (such as in~\cite{Nissim2007})
would probably reduce the magnitude of the noise. In the second approach,
the query function $f$ is modified so that the new query function
is less sensitive to modifications of a record in the data set (the
abovementioned paper~\cite{Mohammed11} exemplifies this approach).

Our proposal falls into the second approach: we replace the original
query function $f$ by $f\circ M$, that is, we run the query $f$
over the microaggregated data set $\overline{X}$. For our proposal
to be meaningful, the function $f\circ M$ must be a good approximation
of $f$. Our assumption is that the microaggregated data set $\overline{X}$
preserves the statistical information contained in the original data
set $X$; therefore, any query that is only concerned with the statistical
properties of the data in $X$ can be run over the microaggregated
data set $\overline{X}$ without much deviation. The function $f\circ M$
will certainly not be a good approximation of $f$ when the output
of $f$ depends on the properties of specific individuals; however,
this is not our case, as we are only interested in the extraction
of statistical information.

Since the $k$-anonymous data set $\overline{X}$ is formed by the
centroids of the clusters (\emph{i.e.}, the average records), for
the sensitivity of the queries $f\circ M$ to be effectively reduced
the centroid must be stable against modifications of one record in
the original data set $X$. This means that modification of one record
in the original data set $X$ should only slightly affect the centroids
in the microaggregated data set. Although this will hold for most
of the clusters yielded by any microaggregation algorithm, we need
it to hold for \emph{all} clusters in order to effectively reduce
the sensitivity.

Not all microaggregation algorithms satisfy the above requirement;
for instance, if the microaggregation algorithm could generate a completely
unrelated set of clusters after modification of a single record in
$X$, the effect on the centroids could be large. As we are modifying
one record in $X$, the best we can expect is a set of clusters that
differ in one record from the original set of clusters. Microaggregation
algorithms with this property lead to the greatest reduction in the
query sensitivity; we refer to them as \emph{insensitive} microaggregation
algorithms.
\begin{defn}[Insensitive microaggregation]
 \label{def:insensitive_microaggregation-1} Let $X$ be a data set,
$M$ a microaggregation algorithm, and let $\{C_{1},\ldots,C_{n}\}$
be the set of clusters that result from running $M$ on $X$. Let
$X'$ be a data set that differs from $X$ in a single record, and
$\{C'_{1},\ldots,C'_{n}\}$ be the clusters produced by running $M$
on $X'$. We say that $M$ is insensitive to the input data if, for
every pair of data sets $X$ and $X'$ differing in a single record,
there is a bijection between the set of clusters $\{C_{1},\ldots,C_{n}\}$
and the set of clusters $\{C'_{1},\ldots,C'_{n}\}$ such that each
pair of corresponding clusters differs at most in a single record. 
\end{defn}
Since for an insensitive microaggregation algorithm corresponding
clusters differ at most in one record, bounding the variability of
the centroid is simple. For instance, for numerical data, when computing
the centroid as the mean, the maximum change for each attribute equals
the size of the range of the attribute divided by $k$. If the microaggregation
was not insensitive, a single modification in $X$ might lead to completely
different clusters, and hence to large variability in the centroids.

The output of microaggregation algorithms is usually highly dependent
on the input data. On the positive side, this leads to greater within-cluster
homogeneity and hence less information loss. On the negative side,
modifying a single record in the input data may lead to completely
different clusters; in other words, such algorithms are not insensitive
to the input data as per Definition~\ref{def:insensitive_microaggregation-1}.
We illustrate this fact for MDAV. Figure~\ref{fig:MDAV-clusters-1}
shows the clusters generated by MDAV for a toy data set $X$ consisting
of 15 records with two attributes, before and after modifying a single
record. In MDAV, we use the Euclidean distance and $k=5$. Two of
the clusters in the original data set differ by more than one record
from the respective most similar clusters in the modified data set.
Therefore, no mapping between clusters of both data sets exists that
satisfies the requirements of Definition~\ref{def:insensitive_microaggregation-1}.
The centroids of the clusters are represented by a cross. A large
change in the centroids between the original and the modified data
sets can be observed.

\begin{figure}[t]
\begin{centering}
\includegraphics[width=7cm]{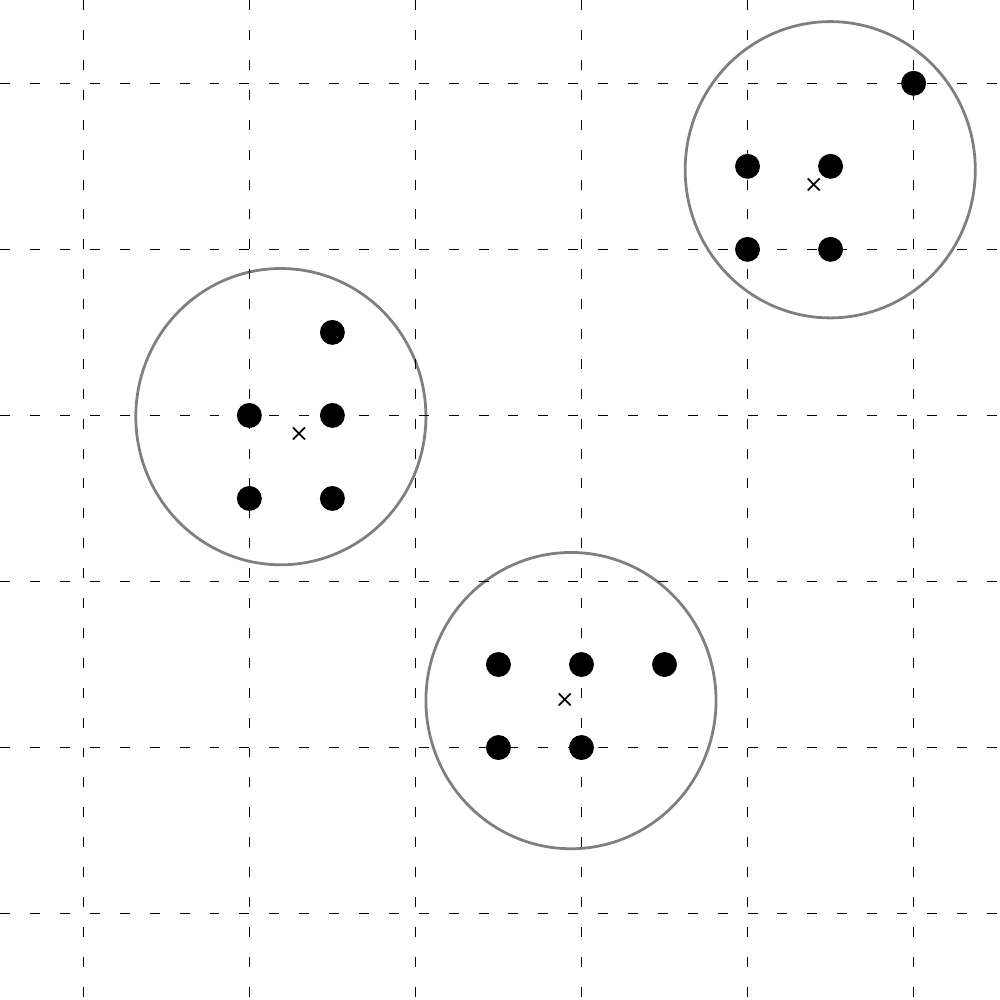}\hspace{0.5cm}\includegraphics[width=7cm]{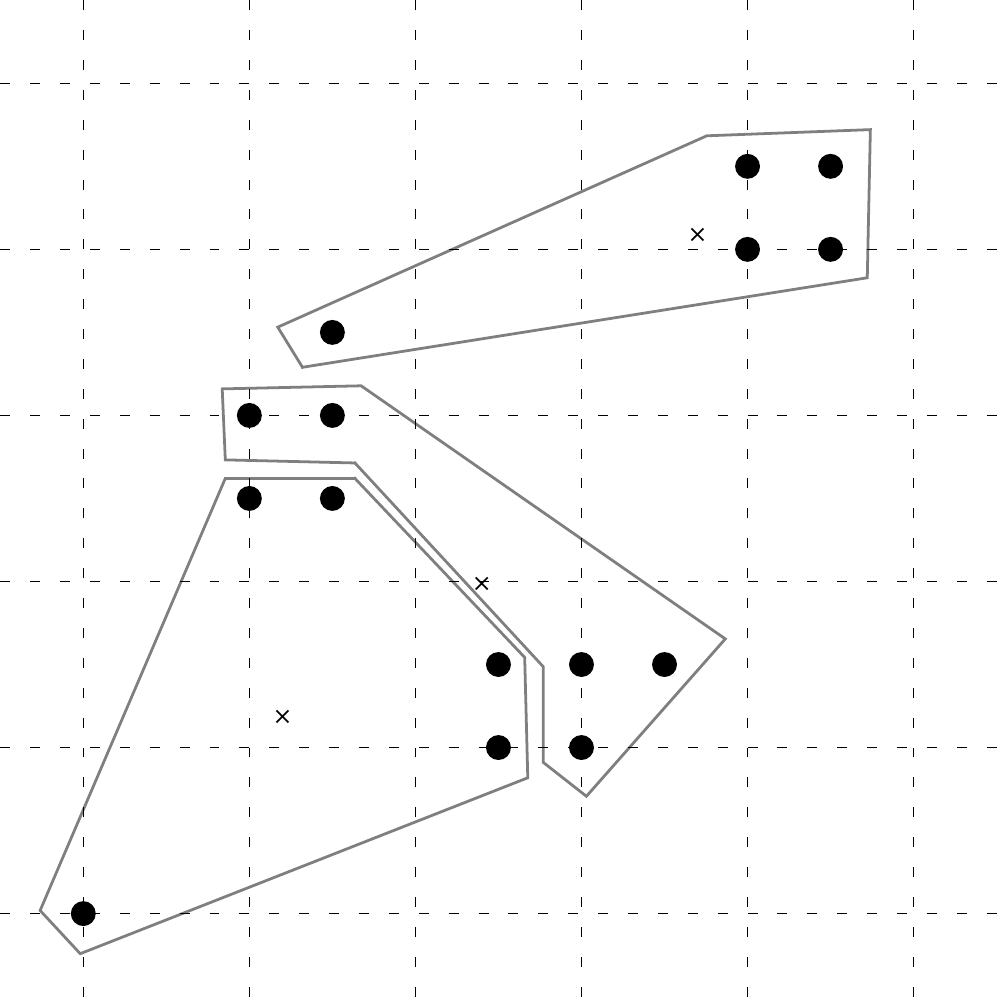} 
\par\end{centering}

\caption{\label{fig:MDAV-clusters-1}MDAV clusters and centroids with $k=5$.
Left, original data set $X$; right, data set after modifying one
record in $X$. }
\end{figure}

We want to turn MDAV into an insensitive microaggregation algorithm,
so that it can be used as the microaggregation algorithm to generate
$\overline{X}$. MDAV depends on two parameters: the minimal cluster
size $k$, and the distance function $d$ used to measure the distance
between records. Modifying $k$ does not help making MDAV insensitive:
similar examples to the ones in Figure~\ref{fig:MDAV-clusters-1}
can easily be proposed for any $k>1$; on the other hand, setting
$k=1$ does make MDAV insensitive, but it is equivalent to not performing
any microaggregation at all. Next, we see that MDAV is insensitive
if the distance function $d$ is consistent with a total order relation.
\begin{defn}
A distance function $d:X\times X\rightarrow\mathbb{R}$ is said to
be consistent with an order relation $\le_{X}$ if $d(x,y)\le d(x,z)$
whenever $x\le_{X}y\le_{X}z$. \end{defn}
\begin{prop}
\label{prop:single_total_order-1}Let $X$ be a data set equipped
with a total order relation $\le_{X}$. Let $d:X\times X\rightarrow\mathbb{R}$
be a distance function consistent with $\le_{X}$. MDAV with distance
$d$ satisfies the insensitivity condition (Definition~\ref{def:insensitive_microaggregation-1}). \end{prop}
\begin{proof}
When the distance $d$ is consistent with a total order, MDAV with
cluster size $k$ reduces to iteratively taking sets with cardinality
$k$ from the extremes, until less than $k$ records are left; the
remaining records form the last cluster. Let $x_{1},\ldots,x_{n}$
be the elements of $X$ sorted according to $\le_{X}$. MDAV generates
a set clusters of the form: 
\[
\{x_{1},\ldots,x_{k}\},\ldots,\{x_{n-k+1},\ldots,x_{n}\}
\]
We want to check that modifying a single record of $X$ leads to a
set of clusters that differ in at most one element. Suppose that we
modify record $x$ by setting it to $x'$, and let $X'$ be the modified
data set. Without loss of generality, we assume that $x\le_{X}x'$;
the proof is similar for the case $x'\le_{X}x$.

Let $C$ be the cluster of $X$ that contains $x$, and $C'$ the
cluster of $X'$ that contains $x'$. Let $m$ be the minimum of the
elements in $C$, and let $M$ be the maximum of the elements in $C'$.
As MDAV takes groups of $k$ records from the extremes, the clusters
of $X$ whose elements are all inferior to $m$, or all superior to
$M$ remain unmodified in $X'$. Therefore, we can assume that $x$
belongs to the leftmost cluster of $X$, and $x'$ belongs to the
rightmost cluster in $X'$.

Let $C_{1},\ldots,C_{m}$ and $C_{1}',\ldots,C_{m}'$ be, respectively,
the clusters of $X$ and $X'$, ordered according to $\le_{X}$. Let
$x_{1}^{i}$ and $x_{j_{i}}^{i}$ be the minimum and the maximum of
the elements of $C_{i}$: $C_{i}=\{z\in X|x_{1}^{i}\le z\le x_{j_{i}}^{i}\}$.
Cluster $C_{1}'$ contains the same elements as $C_{1}$ except for
$x$ that has been removed from $C_{1}'$ and for $x_{1}^{2}$ that
has been added to $C_{1}'$, $C_{1}'=(C_{1}\cup\{x_{1}^{2}\})\setminus\{x\}$.
Clusters $C_{2}',\ldots,C_{m-1}'$ contain the same elements as the
respective cluster $C_{2},\ldots,C_{m-1}$, except for $x_{1}^{i}$
that has been removed from $C_{i}'$ and $x_{1}^{i+1}$ that has been
added to $C_{i}'$. Cluster $C_{m}'$ contains the same elements as
$C_{m}$ except for $x_{1}^{m}$ that has been removed from $C_{m}'$and
$x'$ that has been added to $C_{m}'$. Therefore, clusters $C_{i}$
and $C_{i}'$ differ in a single record for all $i$, which completes
the proof. 
\end{proof}
We have seen that, when the distance function is consistent with a
total order relation, MDAV is insensitive. Now, we want to determine
the necessary conditions for an arbitrary microaggregation algorithm
to be insensitive. Algorithm~\ref{alg:general_microaggregation}
describes the general form of a microaggregation algorithm with fixed
cluster size $k$. Essentially it keeps selecting groups of $k$ records,
until less than $2k$ records are left; the remaining records form
the last cluster, whose size is between $k$ and $2k-1$. Generating
each cluster requires a selection criterion to prioritize some elements
over the others. We can think of this prioritization as an order relation
$\le_{i}$, and the selection criterion for constructing the cluster
$C_{i}$ to be ``select the $k$ smallest records according to $\le_{i}$''.
Note that the prioritization used to generate different clusters need
not be the same; for instance, MDAV selects the remaining element
that is farthest from the average of remaining points, and prioritizes
based on the distance to it.

\begin{algorithm}[ht]
\caption{\label{alg:general_microaggregation}General form of a microaggregation
algorithm with fixed cluster size}

\textbf{let} $X$ be the original data set

\textbf{let} $k$ be the minimal cluster size

\vspace{0.2cm}

\textbf{set} $i:=0$

\textbf{while} $|X|\ge2k$ \textbf{do}

\hspace{0.5cm}$C_{i}\leftarrow k$ smallest elements from $X$ according
to $\le_{i}$

\hspace{0.5cm}$X:=X\setminus C_{i}$

\hspace{0.5cm}$i:=i+1$

\textbf{end while}

$\overline{X}\leftarrow$Replace each record $r\in X$ by the centroid
of its cluster

\vspace{0.2cm}

\textbf{return} $\overline{X}$ 
\end{algorithm}

Let $X$ and $X'$ be a pair of data sets that differ in one record.For
Algorithm~\ref{alg:general_microaggregation} to be insensitive,
the sequence of orders $\le_{i}$ must be constant across executions
of the algorithm; to see this, note that if one of the orders $\le_{i}$
changed, we could easily construct data sets $X$ and $X'$ such that
cluster $C_{i}$ in $X$ would differ by more than one record from
its corresponding cluster in $X'$, and hence the algorithm would
not be insensitive.

Another requirement for Algorithm~\ref{alg:general_microaggregation}
to be insensitive is that the priority assigned by $\le_{i}$ to any
two different elements must be different. If there were different
elements sharing the same priority, we could end up with clusters
that differ by more than one record. For instance, assume that the
sets $X$ and $X'$ are such that $X'=(X\setminus\{x\})\cup\{x'\}$,
and assume that $x$ belongs to cluster $C_{i}$ and $x'$ belongs
to cluster $C_{i}'$. Clusters $C_{i}$ and $C_{i}'$ already differ
in one element, so for the clustering to be insensitive all the other
records in these clusters must be equal. If there was a pair of elements,
$y\ne y'$, with the same priority, and if only one of them was included
in each of the clusters $C_{i}$ and $C_{i}'$, then, as there is
no way to discriminate between $y$ and $y'$, we could, for instance,
include $y$ in $C_{i}$, and $y'$ in $C_{i}'$. In that case the
clusters $C_{i}$ and $C_{i}'$ would differ by more than one record.
Therefore, for the microaggregation to be insensitive $\le_{i}$ must
assign a different priority to each element; in other words, $\le_{i}$
must be a total order.

A similar argument to the one used in Proposition~\ref{prop:single_total_order-1}
can be used to show that when the total order relation is the same
for all the clusters ---in other words, when $\le_{i}$ and $\le_{j}$
are equal for any $i$ and $j$---, then Algorithm~\ref{alg:general_microaggregation}
is insensitive to the input data. However, we want to show that even
when the total orders $\le_{i}$ are different, insensitivity still
holds. In fact, Proposition~\ref{prop:insensitive} provides a complete
characterization of insensitive microaggregation algorithms of the
form of Algorithm~\ref{alg:general_microaggregation}.
\begin{prop}
\label{prop:insensitive}Algorithm~\ref{alg:general_microaggregation}
is insensitive to input data if and only if $\{\le_{i}\}_{i\in\mathbb{N}}$
is a fixed sequence of total order relations defined over the domain
of $X$.\end{prop}
\begin{proof}
In the discussion previous to Proposition~\ref{prop:insensitive}
we have already shown that if Algorithm~\ref{alg:general_microaggregation}
is insensitive, then $\{\le_{i}\}_{i\in\mathbb{N}}$ must be a fixed
sequence of total order relations. We show now that the reverse implication
also holds: if $\{\le_{i}\}_{i\in\mathbb{N}}$ is a fixed sequence
of total order relations, then Algorithm~\ref{alg:general_microaggregation}
is insensitive to input data.

Let $X$ and $X'$ be, respectively, the original data set and a data
set that differs from $X$ in one record. Let $C_{i}$ and $C_{i}'$
be, respectively, the clusters generated at step $i$ for the data
sets $X$ and $X'$. We want to show, for any $i$, that $C_{i}$
and $C_{i}'$ differ in at most one record.

An argument similar to the one in Proposition~\ref{prop:single_total_order-1}
shows that the clusters $C_{0}$ and $C_{0}'$ that result from the
first iteration of the algorithm differ in at most one record. To
see that Algorithm~\ref{alg:general_microaggregation} is insensitive,
it is enough to check that the sets $X\setminus C_{0}$ and $X'\setminus C_{0}'$
differ in at most one record; then, we could apply the previous argument
to $X\setminus C_{0}$ and $X'\setminus C_{0}'$ to see that $C_{1}$
and $C_{1}'$ differ in one record, and so on.

Let $x_{1},\ldots,x_{n}$ be the elements of $X$ ordered according
to $\le_{0}$, so that $C_{0}=\{x_{1},\ldots,x_{k}\}$. Assume that
$X'$ has had element $x$ replaced by $x'$: $X'=$ $\{x_{1},$ $\ldots,$
$x_{n},x'\}\setminus\{x\}$. We have the following four possibilities.
(i) If neither $x$ belongs to $C_{0}$ nor $x'$ belongs to $C_{0}'$,
then $C_{0}$ and $C_{0}'$ must be equal; therefore, $X\setminus C_{0}$
and $X'\setminus C_{0}'$ differ, at most, in one record. (ii) If
both $x$ belongs to $C_{0}$ and $x'$ belong to $C_{0}'$, then
$X\setminus C_{0}$ and $X'\setminus C_{0}'$ are equal. (iii) If
$x$ belongs to $C_{0}$ but $x'$ does not belong to $C_{0}'$, we
can write $C_{0}'$ as $\{x_{1},\ldots,x_{k+1}\}\setminus\{x\}$;
the set $X'\setminus C_{0}'$ is $\{x_{k+2},\ldots,x_{n},x'\}$, which
differs in one record from $X\setminus C_{0}=\{x_{k+1},\ldots,x_{n}\}$;
and (iv) If $x$ is not in $C_{0}$ but $x'$ is in $C_{0}'$, we
can write $C_{0}'$ as $\{x_{1},\ldots,x_{k-1},x'\}$; the set $X'\setminus C_{0}'$
is $\{x_{k},\ldots,x_{n}\}\setminus\{x\}$, which differs in one record
from $X\setminus C_{0}=\{x_{k+1},\ldots,x_{n}\}.$ Therefore, we have
seen that $X\setminus C_{0}$ and $X'\setminus C_{0}'$ differ in
at most one record, which completes the proof. 
\end{proof}
Using multiple order relations in Algorithm~\ref{alg:general_microaggregation},
as allowed by Proposition~\ref{prop:insensitive}, in contrast with
the single order relation used to turn MDAV insensitive in Proposition~\ref{prop:single_total_order-1},
allows us to increase the within-cluster homogeneity achieved in the
microaggregation (see Section~\ref{sec:evaluation} for an empirical
evaluation).

The modification of the query function $f$ to $f\circ M$ by introducing
a prior microaggregation step is intended to reduce the sensitivity
of the query function. Assume that the microaggregation function $f$
computes the centroid of each cluster as the mean of its components.
We analyze next how microaggregation affects the $L_{1}$-sensitivity
of the query function $f$.
\begin{defn}[($L_{1}$-Sensitivity)]
 The $L_{1}$-sensitivity of a function $f:D^{n}\rightarrow\mathbb{R}^{d}$
is the smallest number $\Delta(f)$ such that for all $X,X'\in D^{n}$
which differ in a single entry, 
\[
\left\Vert f(X)-f(X')\right\Vert _{1}\le\Delta(f)
\]

\end{defn}
The $L_{1}$-sensitivity of $f$, $\Delta(f)$, measures the maximum
change in $f$ that results from a modification of a single record
in $X$. Essentially, the microaggregation step $M$ in $f\circ M$
distributes the modification suffered by a single record in $X$ among
multiple records in $M(X)$. Consider, for instance, the data sets
$X$ and $X'$ depicted in Figure~\ref{fig:insensitive-2}. The record
at the top right corner in $X$ has been moved to the bottom left
corner in $X'$; all the other records remain unmodified. In the microaggregated
data sets $M(X)$ and $M(X')$ ---the crosses represent the centroids---
we observe that all the centroids have been modified but the magnitude
of the modifications is smaller: the modification suffered by the
record at the top right corner of $X$ has been distributed among
all the records in $M(X)$.

When computing the centroid as the mean, we can guarantee that the
maximum variation in any centroid is at most $1/k$ of the variation
of the record in $X$. Therefore, we can think of the $L_{1}$-sensitivity
of $f\circ M$ as the maximum change in $f$ if we allow a variation
in each record that is less than $1/k$ times the maximal variation.
In fact, this is a very rough estimate, as only a few centroids can
have a variation equaling $1/k$ of the maximal variation in $X$,
but it is useful to analyze some simple functions such as the identity.
The identity function returns the exact contents of a specific record,
and is used extensively in later sections to construct $\varepsilon$-differentially
private data sets. The sensitivity of the identity functions depends
only on the maximum variation that the selected record may suffer;
therefore, it is clear that distributing the variation among several
records lowers the sensitivity. This is formalized in the following
proposition.
\begin{prop}
\label{prop:identity_function} Let $X\in D^{n}$ be a data set with
numerical attributes only. Let $M$ be a microaggregation function
with minimal cluster size $k$ that computes the centroid by taking
the mean of the elements of each cluster. Given a record $r\in X$,
let $I_{r}()$ be the function that returns the attribute values contained
in record $r$ of $X$. Then $\Delta(I_{r}\circ M)\le\Delta(I_{r})/k$. \end{prop}
\begin{proof}
The function $I_{r}\circ M$ returns the centroid of $M(X)$ that
corresponds to the record $r$ in $X$. It was shown in the discussion
that precedes the proposition that, for a data set that contains only
numerical attributes, if the centroid is computed as the mean of the
records in the cluster, then the maximum change in any centroid is,
at most, $\Delta(I_{r})/k$; that is, $\Delta(I_{r}\circ M)\le\Delta(I_{r})/k$.
\end{proof}

\section{Differentially private data sets through $k$-anonymity\label{sec:dp_ds_kanon}}

Assume that we have an original data set $X$ and that we want to
generate a data set $X_{\varepsilon}$ ---an anonymized version of
$X$--- that satisfies $\varepsilon$-differential privacy. Even if
differential privacy was not introduced with the aim of generating
anonymized data sets, we can think of a data release as the collected
answers to successive queries for each record in the data set. Let
$I_{r}()$ be as defined in Proposition~\ref{prop:identity_function}.
We generate $X_{\varepsilon}$, by querying $X$ with $I_{r}(X)$,
for all $r\in X$. If the responses to the queries $I_{r}()$ satisfy
$\varepsilon$-differential privacy, then, as each query refers to
a different record, by the parallel composition property $X_{\varepsilon}$
also satisfies $\varepsilon$-differential privacy.

The proposed approach for generating $X_{\varepsilon}$ is general
but naive. As each query $I_{r}()$ refers to a single individual,
its sensitivity is large; therefore, the masking required to attain
$\varepsilon$-differential privacy is quite significant, and thus
the utility of such a $X_{\varepsilon}$ very limited.

To improve the utility of $X_{\varepsilon}$, we introduce a microaggregation
step as discussed in Section~\ref{sec:dp_kanon}: (i) from the original
data set $X$, we generate a $k$-anonymous data set $\overline{X}$
---by using a microaggregation algorithm with minimum cluster size
$k$, like MDAV, and assuming that all attributes are quasi-identifiers---,
and (ii) the $\varepsilon$-differentially private data set $X_{\varepsilon}$
is generated from the $k$-anonymous data set $\overline{X}$ by taking
an $\varepsilon$-differentially private response to the queries $I_{r}(\overline{X})$,
for all $r\in\overline{X}$.

By constructing the $k$-anonymous data set $\overline{X}$, we stop
thinking in terms of individuals, to start thinking in terms of groups
of $k$ individuals. Now, the sensitivity of the queries $I_{r}(\overline{X})$
used to construct $X_{\varepsilon}$ reflects the effect that modifying
a single record in $X$ has on the groups of $k$ records in $\overline{X}$.
The fact that each record in $\overline{X}$ depends on $k$ (or more)
records in $X$ is what allows the sensivity to be effectively reduced.
See Proposition~\ref{prop:identity_function} above.

Algorithm~\ref{alg:dp_data_set-1} details the procedure for generating
the differentially private data set $X_{\varepsilon}$.

\begin{algorithm}[ht]
\caption{\label{alg:dp_data_set-1}Generation of an $\varepsilon$-differentially
private data set $X_{\varepsilon}$ from $X$ via microaggregation}

\textbf{let} $X$ be the original data set

\textbf{let} $M$ be an insensitive microaggregation algorithm with
minimal cluster size $k$

\textbf{let} $S_{\varepsilon}()$ be an $\varepsilon$-differentially
private sanitizer

\textbf{let} $I_{r}()$ be the query for the attributes of record
$r$

\vspace{0.2cm}

$\overline{X}\leftarrow$ microaggregated data set $M(X)$

\textbf{for} each $r\in\overline{X}$ \textbf{do}

\hspace{0.5cm}$r_{\varepsilon}\leftarrow S_{\varepsilon}(I_{r}(\overline{X}))$

\hspace{0.5cm}\emph{insert} $r_{\varepsilon}$ \emph{into} $X_{\varepsilon}$

\textbf{end for}

\vspace{0.2cm}

\textbf{return} $X_{\varepsilon}$ 
\end{algorithm}

\subsection{Achieving differential privacy with numerical attributes\label{sec:dif-numerical}}

For a data set consisting of numerical attributes only, generating
the $\varepsilon$-differentially private data set $X_{\varepsilon}$
as previously described is quite straightforward. 

Let $X$ be a data set with $m$ numerical attributes: $A_{1}$, $\ldots$
, $A_{m}$. The first step to construct $X_{\varepsilon}$ is to generate
the $k$-anonymous data set $\overline{X}$ via an insensitive microaggregation
algorithm. As we have seen in Section~\ref{sec:dp_kanon}, the key
point of insensitive microaggregation algorithms is to define a total
order relation over $Dom(X)$, the domain of the records of the data
set $X$. The domain of $X$ contains all the possible values that
make sense, given the semantics of the attributes. In other words,
the domain is not defined by the actual records in $X$ but by the
set of values that make sense for each attribute and by the relation
between attributes.

Microaggregation algorithms use a distance function, $d:Dom(X)\times Dom(X)$
$\rightarrow\mathbb{R}$, to measure the distances between records
and generate the clusters. We assume that such a distance function
is already available and we define a total order with which the distance
is consistent. To construct a total order, we take a reference point
$R$, and define the order according to the distance to $R$. Given
a pair of elements $x,y\in Dom(X)$, we say that $x\le y$ if $d(R,x)\le d(R,y)$.
On the other hand, we still need to define the relation between elements
that are equally distant from $R$. As we assume that the data set
$X$ consists of numerical attributes only, we can take advantage
of the fact that individual attributes are equipped with a total order
---the usual numerical order--- and sort the records that are equally
distant from $R$ by means of the alphabetical order: given $x=(x_{1},\ldots,x_{m})$
and $y=(y_{1},\ldots,y_{m})$, with $d(x,R)=d(y,R)$, we say that
$x\le y$ if $(x_{1},\ldots,x_{m})\le(y_{1},\ldots,y_{m})$ according
to the alphabetical order.

Proposition~\ref{prop:identity_function} shows that, as a result
of the insensitive microaggregation, one has $\Delta(I_{r}\circ M)=\Delta(I_{r})/k$;
therefore, $\varepsilon$-differential privacy can be achieved by
adding to $\overline{X}$ an amount of Laplace noise that would only
achieve $k\varepsilon$-differential privacy if directly added to
$X$.

\subsection{Insensitive MDAV}

According to Proposition~\ref{prop:single_total_order-1}, to make
MDAV insensitive we must define a total order among the elements in
$Dom(X)$. According to the previous discussion, this total order
is constructed by selecting a reference point. To increase within-cluster
homogeneity, MDAV starts by clustering the elements at the boundaries.
For our total order to follow this guideline, the reference point
$R$ must be selected among the elements of the boundary of $Dom(X)$.
For instance, if the domain of $A_{i}$ is $[a_{b}^{i},a_{t}^{i}]$,
we can set $R$ to be the point $(a_{b}^{1},\ldots,a_{b}^{m})$.

Figure~\ref{fig:insensitive-2} illustrates the insensitive microaggregation
obtained by using MDAV with the total order defined above. The original
data set $X$ and the modified data set $X'$ are the same of Figure~\ref{fig:MDAV-clusters-1}.
We also use $k=5$ and the Euclidean distance for insensitive MDAV.
Let us take as the reference point for the above defined total order
the point $R$ at the lower left corner of the grids. Note that now
clusters $C_{1}$,$C_{2}$, and $C_{3}$ in $X$ differ in a single
record from $C_{1}'$,$C_{2}'$, and $C_{3}'$ in $X'$, respectively.
By comparing Figures~\ref{fig:MDAV-clusters-1} and~\ref{fig:insensitive-2},
we observe that the standard (non-insensitive) MDAV results in a set
of clusters with greater within-cluster homogeneity; however, in exchange
for the lost homogeneity, insensitive MDAV generates sets of clusters
that are more stable when one record of the data set changes.

\begin{figure}[t]
\begin{centering}
\includegraphics[width=7cm]{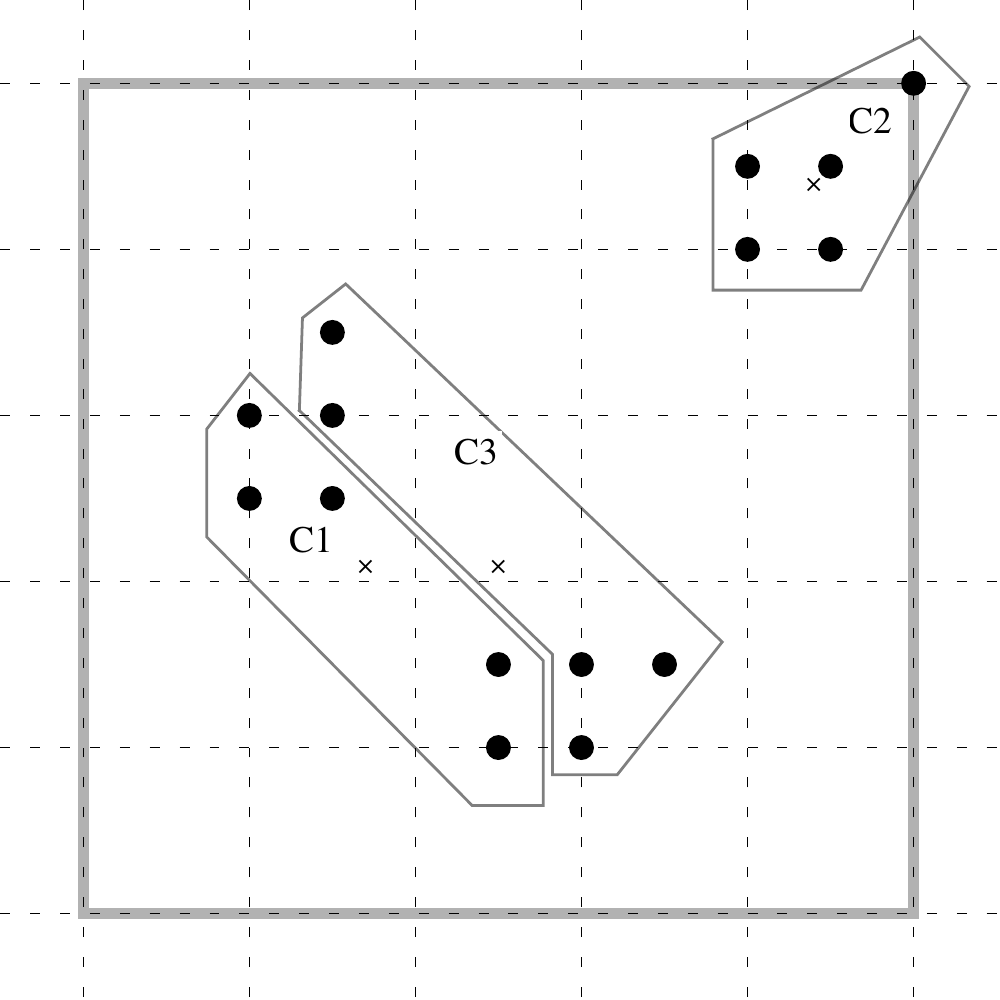}\hspace{0.5cm}\includegraphics[width=7cm]{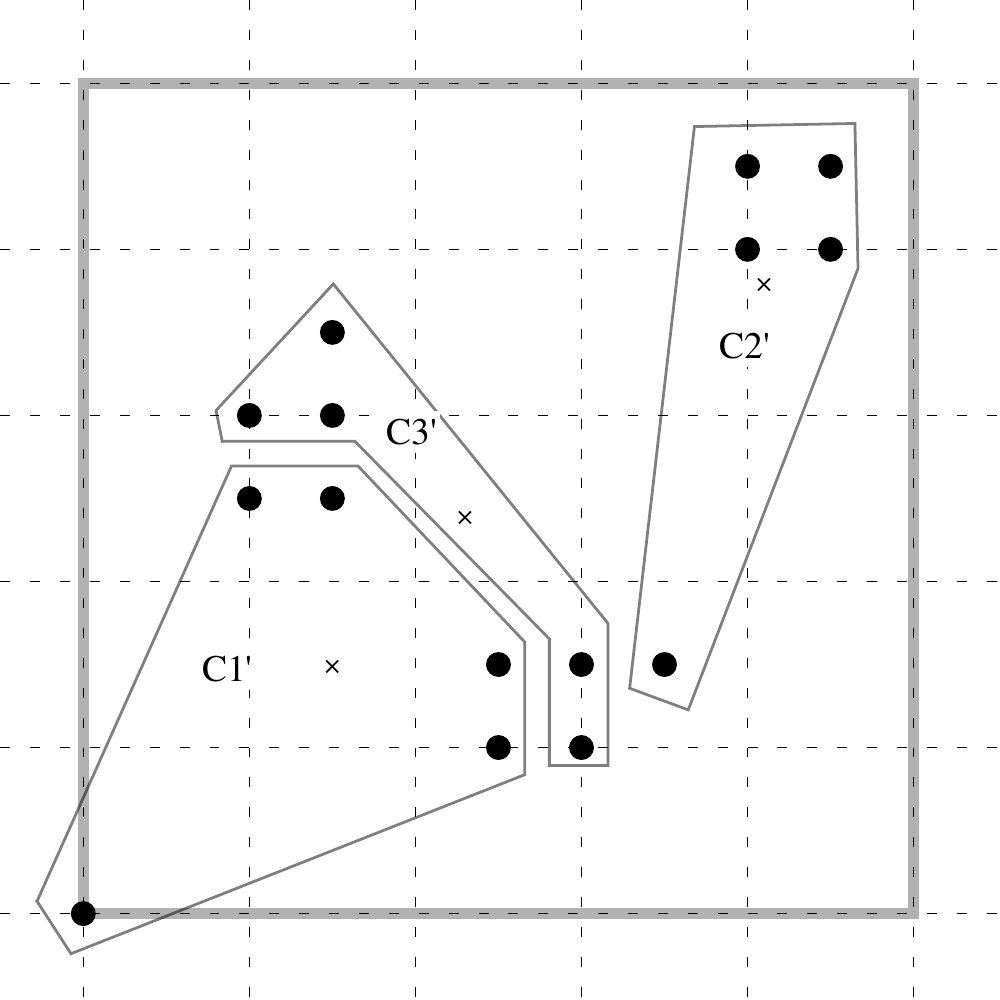} 
\par\end{centering}

\caption{\label{fig:insensitive-2}Insensitive MDAV microaggregation with $k=5$.
Left, original data set $X$; right, data set after modifying one
record in $X$.}
\end{figure}

\subsection{General insensitive microaggregation\label{sec:insensitive-micro}}

It was seen in Section~\ref{sec:dp_kanon} that each clustering step
within microaggregation can use a different total order relation,
as long as the sequence of order relations is kept constant. The advantage
of using multiple total order relations is that it allows the insensitive
microaggregation algorithm to better mimic a standard non-insensitive
microaggregation algorithm, and thus increase the within-cluster homogeneity.

The sequence of total orders is determined by a sequence of reference
points $R_{i}$. In the selection of $R_{i}$ we try to match the
criteria used by non-insensitive microaggregation algorithms to increase
within-cluster homogeneity: start clustering at the boundaries, and
generate a cluster that is far apart from the previously generated
cluster.

Let the domain of $A_{i}$ be $[a_{b}^{i},a_{t}^{i}]$. Define the
set $\mathcal{R}$ of candidate reference points at those points in
the boundaries of $Dom(X)$, that is: 
\[
\mathcal{R}=\{(a_{v_{1}}^{1},\ldots,a_{v_{m}}^{m})|v_{i}\in\{b,t\}\mbox{ for \ensuremath{1\leq i\leq m}}\}
\]
The first reference point $R_{1}$ is arbitrarily selected from $\mathcal{R}$;
for instance, $R_{1}=(a_{b}^{1},\ldots,a_{b}^{m})$. Once a point
$R_{i}$ has been selected, $R_{i+1}$ is selected among the still
unselected points in $\mathcal{R}$ so that it maximizes the Hamming
distance to $R_{i}$ ---if $R_{1}=(a_{b}^{1},\ldots,a_{b}^{m})$,
then $R_{2}=(a_{t}^{1},\ldots,a_{t}^{m})$---. If several unselected
points in $\mathcal{R}$ maximize the Hamming distance to $R_{i}$,
we select the one among them with greatest distance to $R_{i-1}$,
and so on.

Figure~\ref{fig:form_of_clusters} shows the form of the clusters
for a data set containing two numerical attributes. The graphic on
the left is for a single reference point ---this is also the form
of the clusters obtained by insensitive MDAV, which uses a single
total order relation---. The graphic on the right uses four reference
points, one for each edge of the domain, which are selected in turns
as described above.

\begin{figure}[t]
\begin{centering}
\includegraphics[width=7cm]{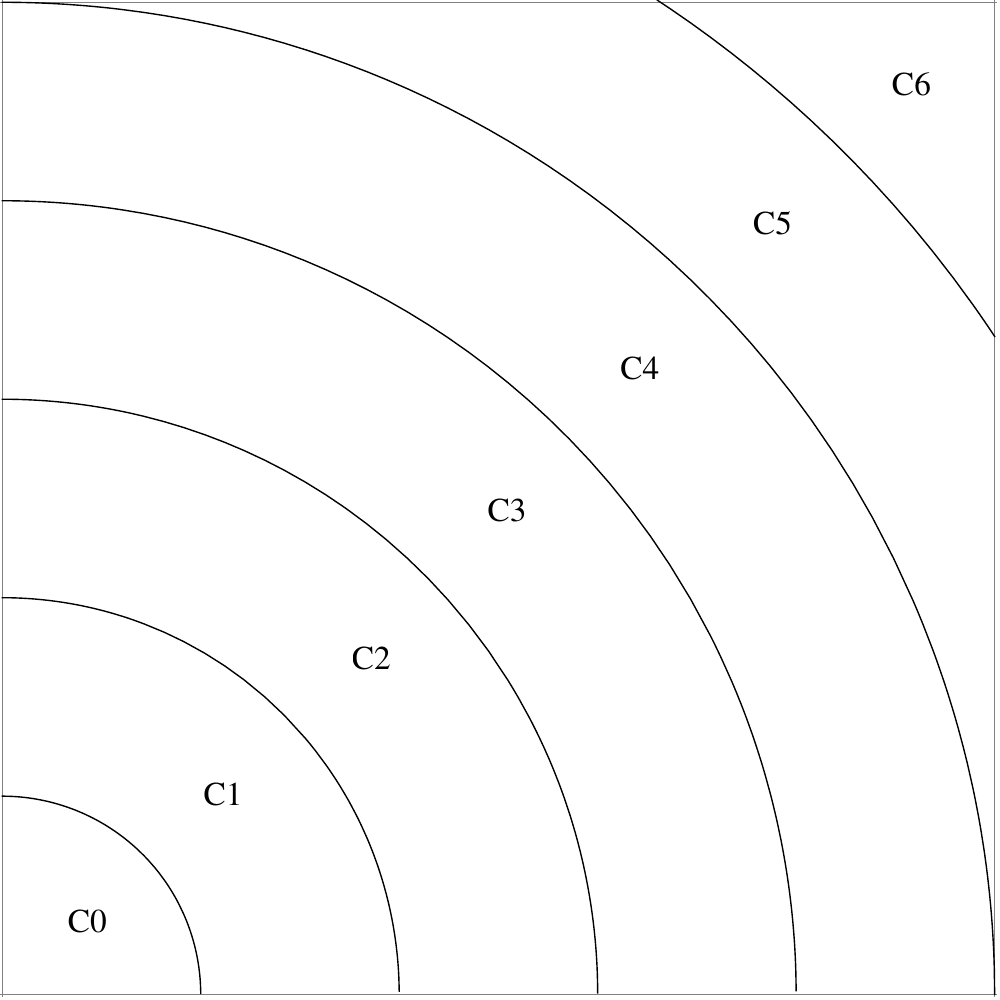}\hspace{0.5cm}\includegraphics[width=7cm]{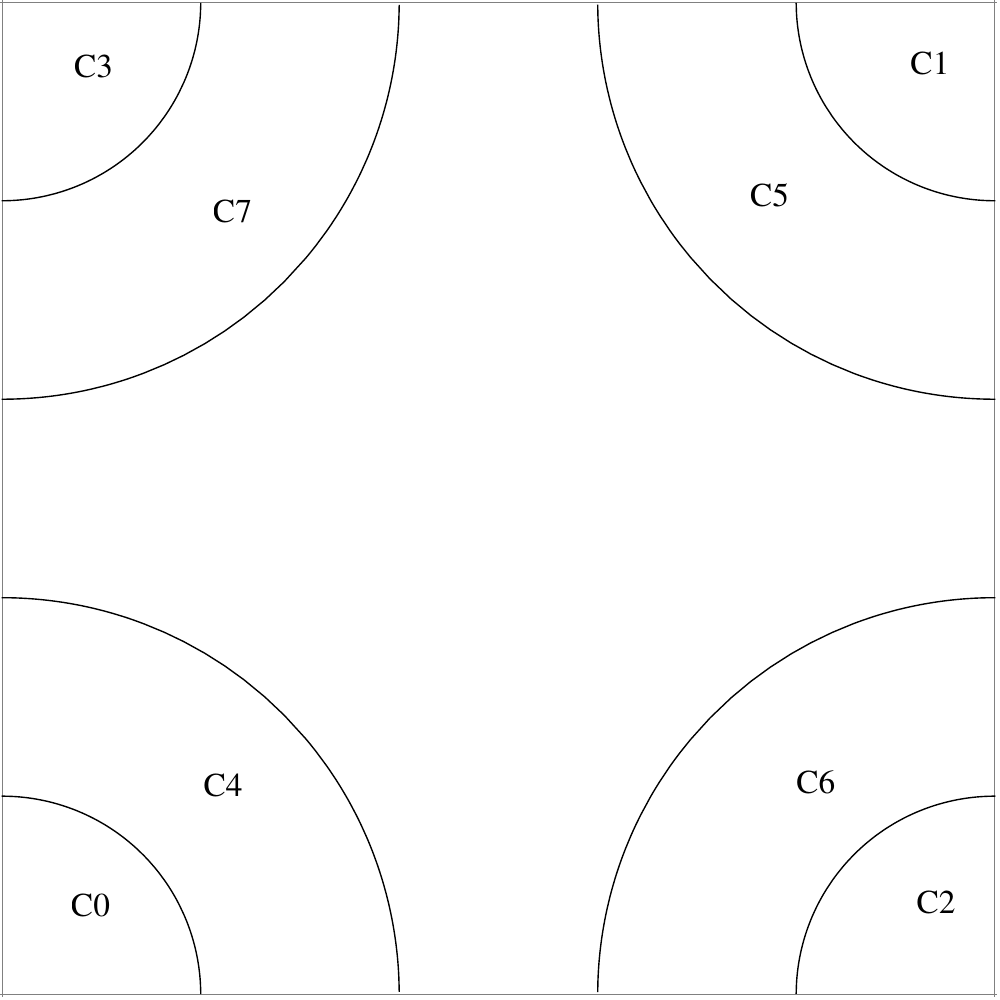} 
\par\end{centering}

\caption{\label{fig:form_of_clusters}Cluster formation. Left, using a single
reference point; right, taking each corner of the domain as a reference
point.}
\end{figure}

\subsection{Achieving differential privacy with categorical attributes\label{sec:dif-categorical}}

Many data sets contain attributes with categorical values, such as
Race, Country of birth, or Job~\cite{Blake98}. Unlike continuous-scale
numerical attributes, categorical attributes take values from a finite
set of categories for which the arithmetical operations needed to
microaggregate and add noise to the outputs do not make sense. In
the sequel, we detail alternative mechanisms that are suitable for
categorical attributes in order to achieve differential privacy as
detailed above.

Let $X$ be a data set with $m$ categorical attributes: $A_{1}$,$\ldots$,
$A_{m}$. The first challenge regards the definition of $Dom(X)$.
Unlike for numerical attributes, the universe of each categorical
attribute can only be defined by extension, listing all the possible
values. This universe can be expressed either as a flat list or it
can be structured in a hierarchic/taxonomic way. The latter scenario
is more desirable, since the taxonomy implicitly captures the semantics
inherent to conceptualizations of categorical values (\emph{e.g.},
disease categories, job categories, sports categories, etc.). In this
manner, further operations can exploit this taxonomic knowledge to
provide a semantically coherent management of attribute values~\cite{Martinez12a}.

Formally, a taxonomy $\tau$ can be defined as an upper semilattice
$\le_{\varsigma}$ on a set of concepts $\varsigma$ with a top element
$root_{\varsigma}$. We define the taxonomy $\tau(A_{i})$ associated
to an attribute $A_{i}$ as the lattice on the minimum set of concepts
that covers all values in $Dom(A_{i})$. Notice that $\tau(A_{i})$
will include all values in $Dom(A_{i})$ (\emph{e.g.}, ``skiing'',
``sailing'', ``swimming'', ``soccer'', etc., if the attribute
refers to sport names) and, usually, some additional generalizations
that are necessary to define the taxonomic structure (\emph{e.g.},``winter
sports'', ``water sports'', ``field sports'', and ``sport''
as the $root$ of the taxonomy).

If $A_{1}$, $\ldots$, $A_{m}$ are independent attributes, $Dom(X)$
can be defined as the ordered combination of values of each $Dom(A_{i})$,
as modeled in their corresponding taxonomies $\tau(A_{1})$, $\ldots$,
$\tau(A_{m})$. If $A_{1},\cdots,A_{m}$ are not independent, value
tuples in $Dom(X)$ may be restricted to a subset of valid combinations.

Next, a suitable distance function $d:Dom(X)\times Dom(X)$ $\rightarrow$
$\mathbb{R}$ to compare records should be defined. To tackle this
problem, we can exploit the taxonomy $\tau(A_{i})$ associated to
each $A_{i}$ in $X$ and the notion of \textit{semantic distance}~\cite{Sanc12}.
A semantic distance $\delta$ quantifies the amount of semantic differences
observed between two terms (\emph{i.e.}, categorical values) according
to the knowledge modeled in a taxonomy. Section~\ref{sec:distance-categorical}
discusses the adequacy of several semantic measures in the context
of differential privacy. By composing semantic distances $\delta$
for individual attributes $A_{i}$, each one computed from the corresponding
taxonomy $\tau(A_{i})$, we can define the required distance $d:Dom(X)\times Dom(X)\rightarrow\mathbb{R}$.

To construct a total order that yields insensitive and within-cluster
homogeneous microaggregation as detailed in Section~\ref{sec:insensitive-micro},
we need to define the boundaries of $Dom(X)$, from which records
will be clustered. Unlike in the numerical case, this is not straightforward
since most categorical attributes are not ordinal and, hence, a total
order cannot be trivially defined even for individual attributes.
However, since the taxonomy $\tau(A_{i})$ models the domain of $A_{i}$,
boundaries of $Dom(A_{i})$, that is, $[a_{b}^{i},a_{t}^{i}]$, can
be defined as the most distant and opposite values from the ``middle''
of $\tau(A_{i})$. From a semantic perspective, this notion of centrality
in a taxonomy can be measured by the \textit{marginality model}~\cite{mdai2012}.
This model determines the central point of the taxonomy and how far
each value is from that center, according to the semantic distance
between value pairs.

The \emph{marginality $m(\cdot,\cdot)$} of each value $a_{j}^{i}$
in $A_{i}$ with respect to its domain of values $Dom(A_{i})$ is
computed as 
\begin{equation}
m(Dom(A_{i}),a_{j}^{i})=\sum_{a_{l}^{i}\in Dom(A_{i})-\{a_{j}^{i}\}}\delta(a_{l}^{i},a_{j}^{i})\label{eq:m}
\end{equation}
where $\delta(\cdot,\cdot)$ is the semantic distance between two
values. The greater $m(Dom(A_{i}),a_{j}^{i})$, the more marginal
(\emph{i.e.}, the less central) is $a_{j}^{i}$ with regard to $Dom(A_{i})$.

Hence, for each $A_{i}$, one boundary $a_{b}^{i}$ of $Dom(A_{i})$
can be defined as the most marginal value of $Dom(A_{i})$:

\begin{equation}
a_{b}^{i}=\arg\max_{a_{j}^{i}\in Dom(A_{i})}m(Dom(A_{i}),a_{j}^{i})\label{eqbound1}
\end{equation}

The other boundary $a_{t}^{i}$ can be defined as the most distant
value from $a_{b}^{i}$ in $Dom(A_{i})$:

\begin{equation}
a_{t}^{i}=\arg\max_{a_{j}^{i}\in Dom(A_{i})}\delta(a_{j}^{i},a_{b}^{i})\label{eqbound2}
\end{equation}

By applying the above expressions to the set of attributes $A_{1},\cdots,A_{m}$
in $X$, the set $\mathcal{R}$ of candidate reference points needed
to define a total order according to the semantic distance can be
constructed as described in Section~\ref{sec:insensitive-micro}. 

If no taxonomic structure is available, other centrality measures
based on data distribution can be used (\emph{e.g.}, by selecting
the modal value as the most central value~\cite{Domingo2005}). However,
such measures omit data semantics and result in significantly less
useful anonymized results~\cite{Martinez12}.

Similarly to the numerical case, if several records are equally distant
from the reference points, the alphabetical criterion can be used
to induce an order within those equidistant records.

At this point, records in $X$ can be grouped using the insensitive
microaggregation algorithm, thereby yielding a set of clusters with
a sensitivity of only one record per cluster. The elements in each
cluster must be replaced by the cluster centroid (\emph{i.e.}, the
arithmetical mean in the numerical case) in order to obtain a $k$-anonymous
data set. Since the mean of a sample of categorical values cannot
be computed in the standard arithmetical sense, we rely again on the
notion of marginality~\cite{mdai2012}: the mean of a sample of categorical
values can be approximated by the least marginal value in the taxonomy,
which is taken as the \textit{centroid} of the set.

Formally, given a sample $S(A_{i})$ of a nominal attribute $A_{i}$
in a certain cluster, the marginality-based centroid for that cluster
is defined in~\cite{mdai2012} as: 
\begin{equation}
Centroid(S(A_{i}))=\arg\min_{a_{j}^{i}\in\tau(S(A_{i}))}m(S(A_{i}),a_{j}^{i})\label{eq:cen}
\end{equation}
where $\tau(S(A_{i}))$ is the minimum taxonomy extracted from $\tau(A_{i})$
that includes all values in $S(A_{i})$. Notice that by considering
as centroid candidates all concepts in $\tau(S(A_{i}))$, which include
all values in $S(A_{i})$ and also their taxonomic generalizations,
we improve the numerical accuracy of the centroid discretization inherent
to categorical attributes~\cite{Martinez12}.

The numerical value associated to each centroid candidate $a_{j}^{i}$
corresponds to its marginality value $m(S(A_{i}),a_{j}^{i})$, which
depends on the sample of values in the cluster. Given a cluster of
records with a set of independent attributes $A_{1},\cdots,A_{m}$,
the cluster centroid can be obtained by composing the individual centroids
of each attribute.

As in the numerical case, cluster centroids depend on input data.
To fulfill differential privacy for categorical attributes, two aspects
must be considered. On the one hand, the centroid computation should
evaluate as centroid candidates all the values in the taxonomy associated
to the \emph{domain} of each attribute ($\tau(A_{i})$), and not only
the sample of values to be aggregated ($\tau(S(A_{i}))$), since the
centroid should be insensitive to any value change of input data within
the attribute's domain. On the other hand, to achieve insensitivity,
uncertainty must be added to the centroid computation. Since adding
Laplacian noise to centroids makes no sense for categorical values,
an alternative way to obtain differentially private outputs consists
in selecting centroids in a probabilistic manner. The general idea
is to select centroids with a degree of uncertainty that is proportional
to the suitability of each centroid and the desired degree of $\varepsilon$-differential
privacy. To do so, the Exponential Mechanism proposed by McSherry
and Talwar~\cite{Mcsherry07mechanismdesign} can be applied. Given
a function with discrete outputs $t$, the mechanism chooses the output
that is close to the optimum according to the input data $D$ and
quality criterion $q(D,t)$, while preserving $\varepsilon$-differential
privacy. Each output is associated with a selection probability $\Pr(t)$,
which grows exponentially with the quality criterion, as follows:
\[
\Pr(t)\propto\exp(\frac{{\varepsilon}q(D,t)}{2\Delta(q)})
\]

In this manner, the optimal output or those that are close to it according
to the quality criterion will be more likely to be selected. Based
on the above arguments, $\varepsilon$-differentially private centroids
can be selected as indicated in Algorithm~\ref{alg:catdif}.

\begin{algorithm}[ht]
\caption{\label{alg:catdif}Computation of $\varepsilon$-differentially private
centroids for clusters with categorical attributes}

\textbf{let} $C$ be a cluster with at least $k$ records

\vspace{0.2cm}

\textbf{for} each categorical attribute $A_{i}$ \textbf{do}

\hspace{0.5cm}Take as quality criterion $q(\cdot,\cdot)$ for each
centroid candidate $a_{j}^{i}$ in $\tau(A_{i})$ the additive inverse
of its marginality towards the attribute values $S(A_{i})$ contained
in $C$, that is, $-m(S(A_{i}),a_{j}^{i})$; 

\hspace{0.5cm}Sample the centroid from a distribution that assigns
\begin{equation}
\Pr(a_{j}^{i})\propto\exp(\frac{{\varepsilon}\times(-m(S(A_{i}),a_{j}^{i}))}{2\Delta(m(A_{i}))})\label{eq:exp}
\end{equation}

\textbf{end for} 
\end{algorithm}

Notice that the inversion of the marginality function has no influence
on the relative probabilities of centroid candidates, since it is
achieved through a \textit{bijective linear transformation}.

With the algorithm we have the following result, which is parallel
to what we saw in the numerical case: if the input data are $k$-anonymous,
the higher $k$, the less the uncertainty that needs to be added to
reach $\varepsilon$-differential privacy.
\begin{prop}
Let $X$ be a data set with categorical attributes. Let $\overline{X}$
be a $k$-anonymous version of $X$ generated using an insensitive
microaggregation algorithm $M$ with minimum cluster size $k$. $\varepsilon$-Differential
privacy can be achieved by using Algorithm~\ref{alg:catdif} to obtain
cluster centroids in $\overline{X}$ with an amount of uncertainty
that decreases as $k$ grows. \end{prop}
\begin{proof}
Without loss of generality, we can write the proof for a single attribute
$A_{i}$. The argument can be composed for multi-attribute data sets.

Let $\Delta(m(A_{i}))$ be the sensitivity of the marginality function
for attribute $A_{i}$. According to the insensitive microaggregation
described earlier in Section~\ref{sec:dp_kanon}, modifying one record
in the data set will induce a change of at most one value in the set
$S(A_{i})$ of values of $A_{i}$ in a cluster. Considering that marginality
measures the sum of distances between a centroid candidate and all
the elements in $S(A_{i})$, in the worst case, in which all values
in $S(A_{i})$ correspond to the same boundary of $Dom(A_{i})$ (defined
by either Equation (\ref{eqbound1}) or Equation (\ref{eqbound2})),
and one of these is changed by the other boundary, the sensitivity
$\Delta(m(A_{i}))$ will correspond to the semantic distance between
both boundaries.

We have that: i) to compute the probabilities in Expression (\ref{eq:exp}),
the quality criterion $-m(S(A_{i}),a_{j}^{i})$ is combined with $\varepsilon$
and $\Delta(m(A_{i}))$, and the latter two magnitudes are constant
for $Dom(A_{i})$; ii) $|S(A_{i})|\geq k$; iii) $m(S(A_{i}),a_{j}^{i})$
is a sum of, at least, $k-1$ terms. Hence, as the cluster size $k$
grows, the marginalities $m(S(A_{i}),a_{j}^{i})$ of values $a_{j}^{i}$
in the cluster $S(A_{i})$ have more degrees of freedom and hence
tend to become more markedly diverse. Hence, as $k$ grows, the probabilities
computed in Expression (\ref{eq:exp}) tend to become more markedly
diverse, and the largest probability (the one of the optimum centroid
candidate) can be expected to dominate more clearly; note that probabilities
computed with Expression (\ref{eq:exp}) decrease exponentially as
marginality grows. Therefore, optimum centroids are more likely to
be selected as $k$ increases. In other words, the amount of uncertainty
added to the output to fulfill differential privacy for categorical
attributes decreases as the $k$-anonymity level of the input data
increases. 
\end{proof}

\subsection{A semantic distance suitable for differential privacy\label{sec:distance-categorical}}

As described above, the selection of differentially private outputs
for categorical attributes is based on the marginality value of centroid
candidates that, in turn, is a function of the semantic distance between
centroids and clustered values. Moreover, the total order used to
create clusters also relies on the assessment of semantic distances
between attribute values. Hence, the particular measure used to compute
semantic distances directly influences the quality of anonymized outputs.

A semantic distance $\delta:o\times o\rightarrow\mathbb{R}$ is a
function mapping a pair of concepts to a real number that quantifies
the difference between the concept meanings. A well-suited $\delta$
to achieve semantic-preserving differentially private outputs should
have the following features. First, it should capture and quantify
the semantics of the categorical values precisely, so that they can
be well differentiated, both when defining the total order and also
when selecting cluster centroids~\cite{Martinez12}. Second, from
the perspective of differential privacy, $\delta$ should have a low
numerical sensitivity to outlying values, that is, those that are
the most distant to the rest of data. In this manner, the sensitivity
of the quality criterion, which is the semantic distance of the two
most outlying values of the domain, will also be low. This will produce
less noisy and, hence, more accurate differentially private outputs.

The accuracy of a semantic measure depends on the kind of techniques
and knowledge bases used to perform the semantic assessments~\cite{Sanc12}.
Among those relying on taxonomies, feature-based measures and measures
based on intrinsic information-theoretic models usually achieve the
highest accuracy with regard to human judgments of semantic distance~\cite{Sanc12}.
The former measures~\cite{Sanc12,Petrakis06} quantify the distance
between concept pairs according to their number of common and non-common
taxonomic ancestors. The latter measures~\cite{Sanchez11,SanchezBatet11,Pirro09,SanchezBatet12}
evaluate the similarity between concept pairs according to their mutual
information, which is approximated as the number of taxonomic specializations
of their most specific common ancestor. Both approaches exploit more
taxonomic knowledge and, hence, tend to produce more accurate results,
than well-known edge-counting measures~\cite{Rada89,Wu94}, which
quantify the distance between concepts by counting the number of taxonomic
edges separating them.

On the other hand, the sensitivity to outlying values depends on the
way in which semantic evidences are quantified. Many classical methods~\cite{Rada89,Wu94}
propose distance functions that are linearly proportional to the amount
of semantic evidences observed in the taxonomy (\emph{e.g.}, number
of taxonomic links). As a result, distances associated to outlying
concepts are significantly larger than those between other more ``central''
values. This leads to a centroid quality criterion with a relatively
high sensitivity, which negatively affects the accuracy of the Exponential
Mechanism~\cite{Mcsherry07mechanismdesign}. More recent methods~\cite{Sanc12,Pirro09,mdai2012}
choose to evaluate distances in a non-linear way. Non-linear functions
provide more flexibility since they can implicitly weight the contribution
of more specific~\cite{mdai2012,Li03} or more detailed~\cite{Sanc12,Pirro09,Sanchez11,SanchezBatet12}
concepts. As a result, concept pairs become better differentiated
and semantic assessments tend to be more accurate~\cite{Sanc12}.
We can distinguish between measures that exponentially promote semantic
differences~\cite{mdai2012,Li03} and those that aggregate semantic
similarities~\cite{Pirro09,Sanchez11,SanchezBatet12} and differences~\cite{Sanc12}
in a logarithmic way. Among these, the latter one is best suited for
the differential privacy scenario, since the logarithmic assessment
of the semantic differences helps to reduce the relative numerical
distances associated to outlying concepts and, hence, to minimize
the sensitivity of the quality function used in the Exponential Mechanism.

Formally, this measure computes the distance $\delta:A_{i}\times A_{i}\rightarrow\mathbb{R}$
between two categorical values $a_{1}^{i}$ and $a_{2}^{i}$ of attribute
$A_{i}$, whose domain is modeled in the taxonomy $\tau(A_{i})$,
as a logarithmic function of their number of non-common taxonomic
ancestors divided (for normalization) by their total number of ancestors~\cite{Sanc12}:

\begin{equation}
\delta(a_{1}^{i},a_{2}^{i})=\log_{2}\Bigg(1+\frac{|\phi(a_{1}^{i})\cup\phi(a_{2}^{i})|-|\phi(a_{1}^{i})\cap\phi(a_{2}^{i})|}{|\phi(a_{1}^{i})\cup\phi(a_{2}^{i})|}\Bigg)\label{eq:d}
\end{equation}

where $\phi(a_{j}^{i})$ is the set of taxonomic ancestors of $a_{j}^{i}$
in $\tau(A_{i})$, including itself.

As demonstrated in~\cite{Sanc12} and~\cite{Batet10}, Expression
(\ref{eq:d}) satisfies \textit{non-negativity}, \textit{reflexivity},
\textit{symmetry} and \textit{subadditivity}, thereby being a distance
measure in the mathematical sense.

Moreover, thanks to the normalizing denominator, the above distance
is insensitive to the size and granularity of the background taxonomy.
and it yields positive normalized values in the $[0,1]$ range. Since
the distance $d:Dom(X)\times Dom(X)\rightarrow\mathbb{R}$ defined
in Section~\ref{sec:dif-categorical} is the composition of semantic
distances for individual attributes and their domains may be modeled
in different taxonomies, a normalized output is desirable to coherently
integrate distances computed from different sources.

\subsection{Integrating heterogeneous attribute types\label{sec:numerical-categorical}}

The above-described semantic measure provides us with a numerical
assessment of the distance between categorical attributes. As a result,
given a data set $X$ with attributes of heterogeneous data types
(\emph{i.e.}, numerical and categorical), the record distance $d:Dom(X)\times Dom(X)\rightarrow\mathbb{R}$
required for microaggregation can be defined by composing numerically
assessed distances for individual attributes, as follows:

\begin{equation}
d({\bf x}_{1},{\bf x}_{2})=\sqrt{\frac{(dist(a_{1}^{1},a_{2}^{1}))^{2}}{(dist(a_{b}^{1},a_{t}^{1}))^{2}}+\cdots+\frac{(dist(a_{1}^{m},a_{2}^{m}))^{2}}{(dist(a_{b}^{m},a_{t}^{m}))^{2}}}\label{eqdist}
\end{equation}

where $dist(a_{1}^{i},a_{2}^{i})$ is the distance (either numerical
or semantic) between the values for the $i$-th attribute $A_{i}$
in ${\bf x}_{1}$ and ${\bf x}_{2}$, and $dist(a_{b}^{i},a_{t}^{i})$
is the distance between the boundaries of $Dom(A_{i})$, which is
used to eliminate the influence of the attribute scale. 

It can be noticed that Expression (\ref{eqdist}) is similar to the
normalized Euclidean distance, but replacing attribute variances,
which depend on input data, by distances between domain boundaries,
which are insensitive to changes of input values. In this manner,
the record distance function effectively defines a total order that
fulfills differential privacy.

\section{Empirical evaluation\label{sec:evaluation}}

In this section we show some empirical results that illustrate how
$k$-anonymous microaggregation of input data reduces the amount of
noise required to fulfill differential privacy and, hence, positively
influences the utility of the anonymized outputs.

\subsection{Evaluation data}

The above-described mechanism has been applied to numerical and categorical
attributes of two reference data sets:
\begin{itemize}
\item ``Census'', which contains 1,080 records with numerical attributes~\cite{refcasc}.
This data set was used in the European project CASC and in~\cite{Domi01creta,Dand02b,Yanc02,Lasz05,Domingo2005,Domingo2008b,Domingo10}.
Like in~\cite{Domingo10}, we took attributes FICA (Social security
retirement payroll deduction), FEDTAX (Federal income tax liability),
INTVAL (Amount of interest income) and POTHVAL (Total other persons
income). To fulfill differential privacy, all four attributes were
masked, \emph{i.e.}, they were considered as quasi-identifiers in
all our tests. The resulting records were all different from each
other. Since all attributes represent non-negative amounts of money,
we took as boundaries for the attribute domains $a_{b}^{i}=0$ and
$a_{t}^{i}=1.5\times max\_attribute\_value\_in\_the\_dataset$. The
domain upper bound $a_{t}^{i}$ is a reasonable estimate if the attribute
values in the data set are representative of the attribute values
in the population, which in particular means that the population outliers
are represented in the data set. The difference between the bounds
$a_{b}^{i}$ and $a_{t}^{i}$ defines the sensitivity of each attribute
and influences the amount of Laplace noise to be added to masked outputs,
as detailed in Section~\ref{sec:dif-numerical}. Since the Laplace
distribution takes values in the range $(-\infty,+\infty)$, for consistency,
we bound noise-added outputs to the $[a_{b}^{i},a_{t}^{i}]$ range
defined above.
\item ``Adult'', a well-known data set from the UCI repository~\cite{Adult},
which has often been used in the past to evaluate privacy-preserving
methods ~\cite{Martinez12a,Domingo06,Fung05,Lin10}. Like in~\cite{Martinez12a}
we focused on two categorical attributes: OCCUPATION and NATIVE-COUNTRY.
According to the data set description $Dom(\mbox{OCCUPATION})$ includes
14 distinct categories, whereas $Dom(\mbox{NATIVE-COUNTRY})$ covers
41. The taxonomies modeling attribute domains, $\tau(\mbox{OCCUPATION})$
and $\tau(\mbox{NATIVE-COUNTRY})$, were extracted from WordNet 2.1~\cite{Fellbaum98},
a general-purpose repository that taxonomically models more than 100,000
concepts. Mappings between attribute labels and WordNet concepts are
those stated in~\cite{Martinez12a}. Considering attribute categories
and their taxonomic ancestors, the resulting taxonomies contain 122
distinct concepts for OCCUPATION and 127 for NATIVE-COUNTRY. As discussed
in Section~\ref{sec:dif-categorical}, these higher figures enable
a finer grained and more accurate discretization of cluster centroids
in comparison with approaches based on flat lists of attribute categories.
Domain boundaries for each attribute and sensitivities for centroid
quality criteria were set as described in Section~\ref{sec:dif-categorical}.
For evaluation purposes, we used the training corpus from the Adult
data set, which consists of 30,162 records after removing records
with missing values. Due to the reduced set of attribute categories,
the evaluation data contained 388 different record tuples, hence being
a much more homogeneous data set than Census. 
\end{itemize}

\subsection{Evaluation measures\label{sec:eval}}

The quality of the masked output for different combinations of $k$-anonymity
and $\varepsilon$-differential privacy levels has been evaluated
from the perspectives of \textit{information loss} and \textit{disclosure
risk}: 
\begin{itemize}
\item Information loss has been quantified by means of the well-known Sum
of Squared Errors (SSE), a measure used in a good deal of the anonymization
literature (\emph{e.g.} \cite{Domi02}). For a given anonymized data
set (\emph{i.e.}, a $k$-anonymous data set $\overline{X}$ or an
$\varepsilon$-differentially private data set $X_{\varepsilon}$),
SSE is defined as the sum of squares of attribute distances between
original records in $X$ and their versions in the anonymized data
set, that is 
\[
SSE=\sum_{x_{j}\in X}\sum_{a_{j}^{i}\in x_{j}}(dist(a_{j}^{i},(a_{j}^{i})'))^{2},
\]
where $a_{j}^{i}$ is the value of the $i$-th attribute for the $j$-th
original record and $(a_{j}^{i})'$ represents its masked version.
For numerical attributes, $dist(\cdot,\cdot)$ corresponds to the
standard Euclidean distance, whereas for categorical ones we used
the semantic distance defined in Equation (\ref{eq:d}). Hence, the
lower is SSE, the lower is information loss and the higher is the
utility of the anonymized data. 
\item The disclosure risk has been evaluated as the percentage of records
of the original data that can be correctly matched from the anonymized
data set, that is, the percentage of Record Linkages (RL) 
\[
RL=100\times\frac{\sum_{x_{j}\in X}\Pr(x'_{j})}{m},
\]
where $m$ is the number of original records and the record linkage
probability for an anonymized record ($\Pr(x'_{j})$) is calculated
as 
\[
\Pr(x'_{j})=\left\{ \begin{array}{lll}
0 & \mbox{if} & x_{j}\not\in G\\
\frac{1}{|G|} & \mbox{if} & x_{j}\in G
\end{array}\right.
\]
where $G$ is the set of original records that are at minimum distance
from $x'_{j}$. The same distance functions as for SSE have been used.
If the correct original record $x_{j}$ is in $G$, then $\Pr(x'_{j})$
is computed as the probability of guessing $x_{j}$ in $G$, that
is, $1/|G|$. Otherwise, $\Pr(x'_{j})=0$. The lower RL, the better
is the privacy of the anonymized output. 
\end{itemize}
As baseline results, we have computed SSE and RL values for a standard
$k$-anonymity scenario in which all attributes are microaggregated
by means of the original MDAV algorithm~\cite{Domingo2005}, and
also with its modified insensitive version with several reference
points (Algorithm~\ref{alg:general_microaggregation}). Furthermore,
we also considered the straightforward $\varepsilon$-differential
privacy scenario in which Laplace noise or the Exponential Mechanism
are directly applied to unaggregated inputs; this approach is equivalent
to applying our method with a $k$-anonymity level of $k=1$.

The $\varepsilon$ parameter for differential privacy has been set
to $\varepsilon$= {0.01, 0.1, 1.0, 10.0}, which covers the usual
range of differential privacy levels observed in the literature~\cite{Dwork2011,Char10,Char12,Mach08}.
The $k$-anonymity levels have been set between 1 and 100, except
for the raw sensitive and insensitive MDAV microaggregations, which
start from $k=2$, because $k=1$ would mean that input data are not
modified.

Figure~\ref{figSSERL} depicts the SSE and RL values for the different
parameterizations of $k$ and $\varepsilon$ for the Census data set,
whereas Figure~\ref{figSSERLCat} corresponds to the Adult data set.
Due to the broad ranges of the SSE and RL values, the Y-axes are represented
using a $\log_{10}$ scale. Each test involving Laplace noise shows
the averaged results of 10 runs, for the sake of stability.

\begin{figure}[ht]
\includegraphics[scale=0.3]{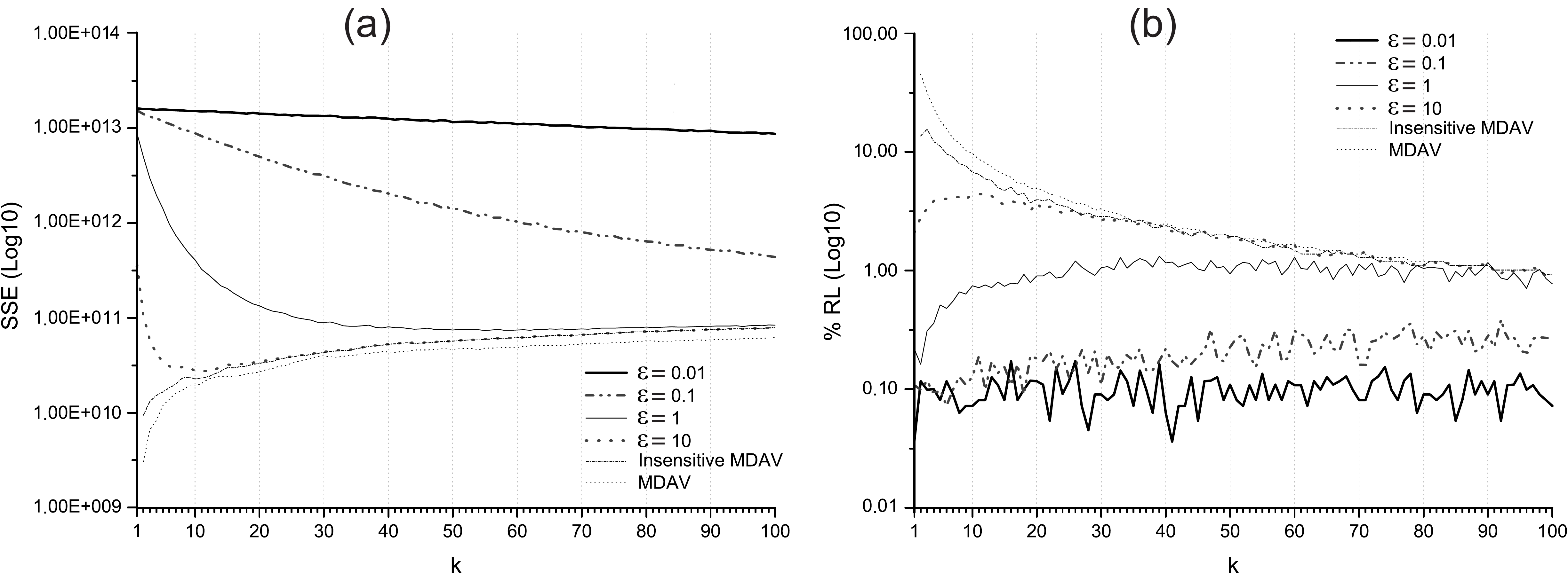}
\caption{SSE and RL values for different $k$ (varying with step 1) and $\varepsilon$
values for the ``Census'' data set. }
\label{figSSERL} 
\end{figure}

\begin{figure}[ht]
\includegraphics[scale=0.3]{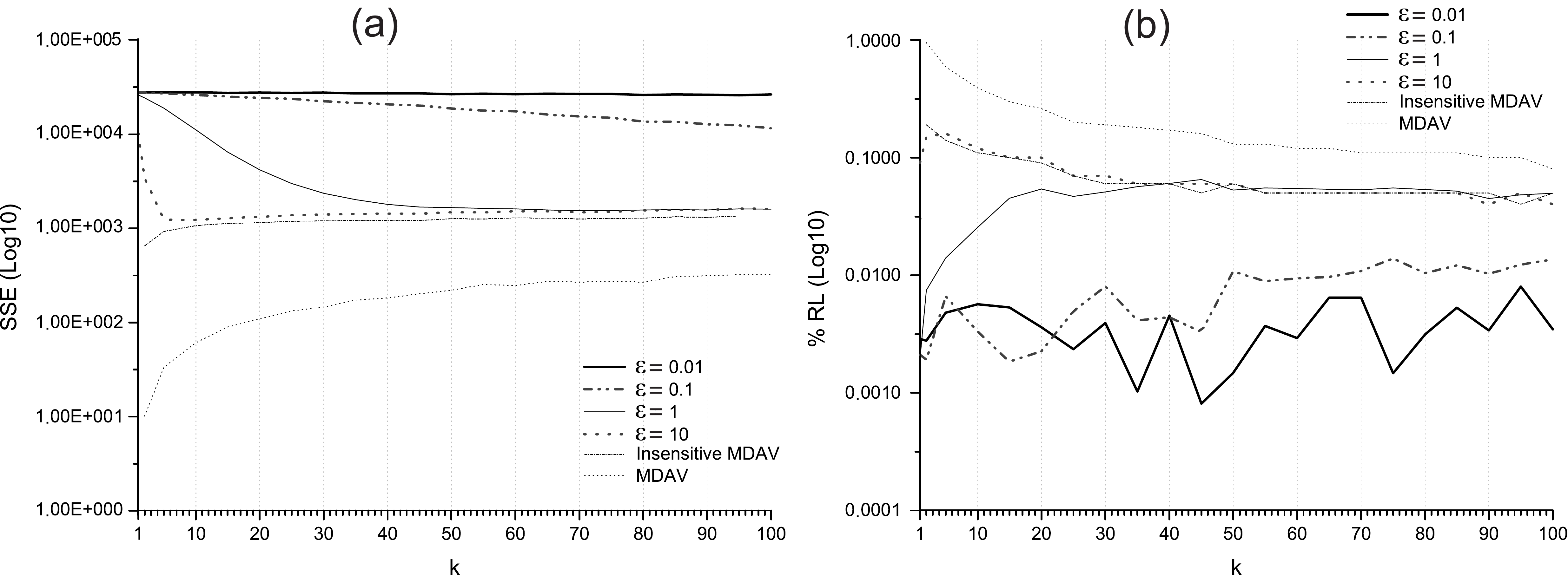}
\caption{SSE and RL values for different $k$ (varying with step 5) and $\varepsilon$
values for the ``Adult'' data set. }

\label{figSSERLCat} 
\end{figure}

To compare our method against baseline approaches regarding the \textit{balance}
between information loss and disclosure risk, we also computed the
\textit{relative} improvement of SSE and RL values for our approach
($SSE_{k\epsilon}$, $RL_{k\epsilon}$) over the baseline values ($SSE_{0}$,
$RL_{0}$) obtained with the original MDAV algorithm and with unaggregated
differential privacy. First, we computed the improvement factor of
SSE values as follows: 
\[
SSE_{f}=\frac{\sqrt{SSE_{0}}}{\sqrt{SSE_{k\epsilon}}}
\]
Then, the improvement factor of RL values was computed as: 
\[
RL_{f}=\frac{RL_{0}}{RL_{k\epsilon}}
\]
The final score that balances both dimensions was the ratio between
SSE and RL\_factors: 
\[
Score=\frac{SSE_{f}}{RL_{f}}
\]
Notice that SSE values have been square-rooted to provide a coherent
linear integration of RL and SSE, and that $Scores$ above 1.0 show
a practical improvement against baseline approaches.

Tables~\ref{mdav-num} and~\ref{dif-num} show the $SSE_{f}$ and
$RL_{f}$ factors and the resulting $Scores$ for different $\varepsilon$
values and some $k$-anonymity degrees with respect to baseline approaches
for the Census data set. Tables~\ref{mdav-cat} and~\ref{dif-cat}
correspond to the Adult data set.

\begin{table}[ht]
\caption{\label{mdav-num}Census data set. $SSE_{f}$ and $RL_{f}$ factors,
and $Scores$ for different $\varepsilon$ values against standard
MDAV microaggregation for several $k$-anonymity levels}

{\scriptsize \hspace{-3cm} }%
\begin{tabular}{|l|cc|ccc|ccc|ccc|ccc|}
\hline 
 & \multicolumn{2}{c|}{{\scriptsize $MDAV$}} & \multicolumn{3}{c|}{{\scriptsize $\epsilon=0.01$}} & \multicolumn{3}{c|}{{\scriptsize $\epsilon=0.1$}} & \multicolumn{3}{c|}{{\scriptsize $\epsilon=1.0$}} & \multicolumn{3}{c|}{{\scriptsize $\epsilon=10.0$}}\tabularnewline
 & {\scriptsize SSE$_{0}$ } & {\scriptsize RL$_{0}$ } & {\scriptsize SSE$_{f}$ } & {\scriptsize RL$_{f}$ } & \textit{\scriptsize Score}{\scriptsize{} } & {\scriptsize SSE$_{f}$ } & {\scriptsize RL$_{f}$ } & \textit{\scriptsize Score}{\scriptsize{} } & {\scriptsize SSE$_{f}$ } & {\scriptsize RL$_{f}$ } & \textit{\scriptsize Score}{\scriptsize{} } & {\scriptsize SSE$_{f}$ } & {\scriptsize RL$_{f}$ } & \textit{\scriptsize Score}\tabularnewline
\hline 
{\scriptsize $k=2$ } & {\scriptsize 3.07E+09 } & {\scriptsize 45.3 } & {\scriptsize 0.014 } & {\scriptsize 386.92 } & \textit{\scriptsize 5.37}{\scriptsize{} } & {\scriptsize 0.015 } & {\scriptsize 457.27 } & \textit{\scriptsize 6.74}{\scriptsize{} } & {\scriptsize 0.025 } & {\scriptsize 279.44 } & \textit{\scriptsize 6.92}{\scriptsize{} } & {\scriptsize 0.17 } & {\scriptsize 16.1 } & \textit{\scriptsize 2.73}{\scriptsize{} }\tabularnewline
{\scriptsize $k=5$ } & {\scriptsize 1.20E+10 } & {\scriptsize 18.4 } & {\scriptsize 0.027 } & {\scriptsize 227.41 } & \textit{\scriptsize 6.26}{\scriptsize{} } & {\scriptsize 0.032 } & {\scriptsize 186.06 } & \textit{\scriptsize 5.93}{\scriptsize{} } & {\scriptsize 0.092 } & {\scriptsize 35.87 } & \textit{\scriptsize 3.29}{\scriptsize{} } & {\scriptsize 0.588 } & {\scriptsize 4.59 } & \textit{\scriptsize 2.7}\tabularnewline
{\scriptsize $k=15$ } & {\scriptsize 2.41E+10 } & {\scriptsize 6.48 } & {\scriptsize 0.04 } & {\scriptsize 80.0 } & \textit{\scriptsize 3.24}{\scriptsize{} } & {\scriptsize 0.06 } & {\scriptsize 42.08 } & \textit{\scriptsize 2.53}{\scriptsize{} } & {\scriptsize 0.343 } & {\scriptsize 8.8 } & \textit{\scriptsize 3.02}{\scriptsize{} } & {\scriptsize 0.862 } & {\scriptsize 1.82 } & \textit{\scriptsize 1.57}\tabularnewline
{\scriptsize $k=30$ } & {\scriptsize 4.00E+10 } & {\scriptsize 3.33 } & {\scriptsize 0.054 } & {\scriptsize 37.0 } & \textit{\scriptsize 2.01}{\scriptsize{} } & {\scriptsize 0.112 } & {\scriptsize 30.83 } & \textit{\scriptsize 3.46}{\scriptsize{} } & {\scriptsize 0.667 } & {\scriptsize 3.14 } & \textit{\scriptsize 2.1}{\scriptsize{} } & {\scriptsize 0.955 } & {\scriptsize 1.27 } & \textit{\scriptsize 1.21}\tabularnewline
\hline 
\end{tabular}
\end{table}

\begin{table}[ht]
\caption{\label{dif-num}Census data set. $SSE_{f}$ and $RL_{f}$ factors,
and resulting $Scores$ for different $\varepsilon$ values and $k$-anonymity
degrees against straightforward $\varepsilon$-differential privacy
scenario (\emph{i.e.}, $k=1$)}

{\scriptsize \hspace{-1cm} }%
\begin{tabular}{|l|cc|ccc|ccc|ccc|}
\hline 
 & \multicolumn{2}{c|}{{\scriptsize $k=1$}} & \multicolumn{3}{c|}{{\scriptsize $k=5$}} & \multicolumn{3}{c|}{{\scriptsize $k=15$}} & \multicolumn{3}{c|}{{\scriptsize $k=30$}}\tabularnewline
 & {\scriptsize SSE$_{0}$ } & {\scriptsize RL$_{0}$ } & {\scriptsize SSE$_{f}$ } & {\scriptsize RL$_{f}$ } & \textit{\scriptsize Score}{\scriptsize{} } & {\scriptsize SSE$_{f}$ } & {\scriptsize RL$_{f}$ } & \textit{\scriptsize Score}{\scriptsize{} } & {\scriptsize SSE$_{f}$ } & {\scriptsize RL$_{f}$ } & \textit{\scriptsize Score}\tabularnewline
\hline 
{\scriptsize $\epsilon=0.01$ } & {\scriptsize 1.62E+13 } & {\scriptsize 0.036 } & {\scriptsize 1.01 } & {\scriptsize 0.44 } & \textit{\scriptsize 0.45}{\scriptsize{} } & {\scriptsize 1.05 } & {\scriptsize 0.44 } & \textit{\scriptsize 0.47}{\scriptsize{} } & {\scriptsize 1.10 } & {\scriptsize 0.40 } & \textit{\scriptsize 0.44}{\scriptsize{} }\tabularnewline
{\scriptsize $\epsilon=0.1$ } & {\scriptsize 1.54E+13 } & {\scriptsize 0.108 } & {\scriptsize 1.14 } & {\scriptsize 1.09 } & \textit{\scriptsize 1.24}{\scriptsize{} } & {\scriptsize 1.52 } & {\scriptsize 0.70 } & \textit{\scriptsize 1.07}{\scriptsize{} } & {\scriptsize 2.20 } & {\scriptsize 1.00 } & \textit{\scriptsize 2.20}{\scriptsize{} }\tabularnewline
{\scriptsize $\epsilon=1.0$ } & {\scriptsize 8.86E+12 } & {\scriptsize 0.218 } & {\scriptsize 2.49 } & {\scriptsize 0.42 } & \textit{\scriptsize 1.06}{\scriptsize{} } & {\scriptsize 6.57 } & {\scriptsize 0.30 } & \textit{\scriptsize 1.95}{\scriptsize{} } & {\scriptsize 9.92 } & {\scriptsize 0.21 } & \textit{\scriptsize 2.04}{\scriptsize{} }\tabularnewline
{\scriptsize $\epsilon=10.0$ } & {\scriptsize 3.69E+11 } & {\scriptsize 2.09 } & {\scriptsize 3.26 } & {\scriptsize 0.52 } & \textit{\scriptsize 1.70}{\scriptsize{} } & {\scriptsize 3.37 } & {\scriptsize 0.59 } & \textit{\scriptsize 1.98}{\scriptsize{} } & {\scriptsize 2.90 } & {\scriptsize 0.79 } & \textit{\scriptsize 2.30}{\scriptsize{} }\tabularnewline
\hline 
\end{tabular}
\end{table}

\begin{table}[ht]
\caption{\label{mdav-cat}Adult data set. $SSE_{f}$ and $RL_{f}$ factors,
and resulting $Scores$ for different $\varepsilon$ values against
standard MDAV microaggregation for several $k$-anonymity levels}

{\scriptsize \hspace{-2cm} }%
\begin{tabular}{|l|cc|ccc|ccc|ccc|ccc|}
\hline 
 & \multicolumn{2}{c|}{{\scriptsize $MDAV$}} & \multicolumn{3}{c|}{{\scriptsize $\epsilon=0.01$}} & \multicolumn{3}{c|}{{\scriptsize $\epsilon=0.1$}} & \multicolumn{3}{c|}{{\scriptsize $\epsilon=1.0$}} & \multicolumn{3}{c|}{{\scriptsize $\epsilon=10.0$}}\tabularnewline
 & {\scriptsize SSE$_{0}$ } & {\scriptsize RL$_{0}$ } & {\scriptsize SSE$_{f}$ } & {\scriptsize RL$_{f}$ } & \textit{\scriptsize Score}{\scriptsize{} } & {\scriptsize SSE$_{f}$ } & {\scriptsize RL$_{f}$ } & \textit{\scriptsize Score}{\scriptsize{} } & {\scriptsize SSE$_{f}$ } & {\scriptsize RL$_{f}$ } & \textit{\scriptsize Score}{\scriptsize{} } & {\scriptsize SSE$_{f}$ } & {\scriptsize RL$_{f}$ } & \textit{\scriptsize Score}\tabularnewline
\hline 
{\scriptsize $k=2$ } & {\scriptsize 1.01 } & {\scriptsize 0.95 } & {\scriptsize 0.019 } & {\scriptsize 343.13 } & \textit{\scriptsize 6.53}{\scriptsize{} } & {\scriptsize 0.019 } & {\scriptsize 497.53 } & \textit{\scriptsize 9.52}{\scriptsize{} } & {\scriptsize 0.02 } & {\scriptsize 127.49 } & \textit{\scriptsize 2.60}{\scriptsize{} } & {\scriptsize 0.054 } & {\scriptsize 6.33 } & \textit{\scriptsize 0.34}{\scriptsize{} }\tabularnewline
{\scriptsize $k=5$ } & {\scriptsize 3.33 } & {\scriptsize 0.59 } & {\scriptsize 0.034 } & {\scriptsize 123.31 } & \textit{\scriptsize 4.27}{\scriptsize{} } & {\scriptsize 0.035 } & {\scriptsize 89.48 } & \textit{\scriptsize 3.14}{\scriptsize{} } & {\scriptsize 0.042 } & {\scriptsize 42.21 } & \textit{\scriptsize 1.77}{\scriptsize{} } & {\scriptsize 0.164 } & {\scriptsize 3.69 } & \textit{\scriptsize 0.61}\tabularnewline
{\scriptsize $k=15$ } & {\scriptsize 8.93 } & {\scriptsize 0.3 } & {\scriptsize 0.057 } & {\scriptsize 56.49 } & \textit{\scriptsize 3.22}{\scriptsize{} } & {\scriptsize 0.059 } & {\scriptsize 163.06 } & \textit{\scriptsize 9.73}{\scriptsize{} } & {\scriptsize 0.118 } & {\scriptsize 6.66 } & \textit{\scriptsize 0.78}{\scriptsize{} } & {\scriptsize 0.263 } & {\scriptsize 3.0 } & \textit{\scriptsize 0.79}\tabularnewline
{\scriptsize $k=30$ } & {\scriptsize 14.6 } & {\scriptsize 0.19 } & {\scriptsize 0.073 } & {\scriptsize 48.55 } & \textit{\scriptsize 3.53}{\scriptsize{} } & {\scriptsize 0.08 } & {\scriptsize 23.72 } & \textit{\scriptsize 1.92}{\scriptsize{} } & {\scriptsize 0.25 } & {\scriptsize 3.73 } & \textit{\scriptsize 0.93}{\scriptsize{} } & {\scriptsize 0.32 } & {\scriptsize 2.71 } & \textit{\scriptsize 0.87}\tabularnewline
\hline 
\end{tabular}
\end{table}

\begin{table}[ht]
\caption{\label{dif-cat}Adult data set. $SSE_{f}$ and $RL_{f}$ factors,
and resulting $Scores$ for different $\varepsilon$ values and $k$-anonymity
degrees against straightforward $\varepsilon$-differential privacy
scenario (\emph{i.e.}, $k=1$)}

{\scriptsize \hspace{-1cm} }%
\begin{tabular}{|l|cc|ccc|ccc|ccc|}
\hline 
 & \multicolumn{2}{c|}{{\scriptsize $k=1$}} & \multicolumn{3}{c|}{{\scriptsize $k=5$}} & \multicolumn{3}{c|}{{\scriptsize $k=15$}} & \multicolumn{3}{c|}{{\scriptsize $k=30$}}\tabularnewline
 & {\scriptsize SSE } & {\scriptsize RL } & {\scriptsize SSE$_{f}$ } & {\scriptsize RL$_{f}$ } & \textit{\scriptsize Score}{\scriptsize{} } & {\scriptsize SSE$_{f}$ } & {\scriptsize RL$_{f}$ } & \textit{\scriptsize Score}{\scriptsize{} } & {\scriptsize SSE$_{f}$ } & {\scriptsize RL$_{f}$ } & \textit{\scriptsize Score}\tabularnewline
\hline 
{\scriptsize $\epsilon=0.01$ } & {\scriptsize 2.78E+04 } & {\scriptsize 0.0029 } & {\scriptsize 1.00 } & {\scriptsize 0.60 } & \textit{\scriptsize 0.60}{\scriptsize{} } & {\scriptsize 1.01 } & {\scriptsize 0.54 } & \textit{\scriptsize 0.54}{\scriptsize{} } & {\scriptsize 1.00 } & {\scriptsize 0.73 } & \textit{\scriptsize 0.74}{\scriptsize{} }\tabularnewline
{\scriptsize $\epsilon=0.1$ } & {\scriptsize 2.77E+04 } & {\scriptsize 0.0021 } & {\scriptsize 1.01 } & {\scriptsize 0.32 } & \textit{\scriptsize 0.32}{\scriptsize{} } & {\scriptsize 1.05 } & {\scriptsize 1.14 } & \textit{\scriptsize 1.20}{\scriptsize{} } & {\scriptsize 1.11 } & {\scriptsize 0.26 } & \textit{\scriptsize 0.29}{\scriptsize{} }\tabularnewline
{\scriptsize $\epsilon=1.0$ } & {\scriptsize 2.61E+04 } & {\scriptsize 0.0019 } & {\scriptsize 1.18 } & {\scriptsize 0.14 } & \textit{\scriptsize 0.16}{\scriptsize{} } & {\scriptsize 2.01 } & {\scriptsize 0.04 } & \textit{\scriptsize 0.09}{\scriptsize{} } & {\scriptsize 3.34 } & {\scriptsize 0.04 } & \textit{\scriptsize 0.13}{\scriptsize{} }\tabularnewline
{\scriptsize $\epsilon=10.0$ } & {\scriptsize 1.02E+04 } & {\scriptsize 0.09 } & {\scriptsize 2.89 } & {\scriptsize 0.56 } & \textit{\scriptsize 1.62}{\scriptsize{} } & {\scriptsize 2.83 } & {\scriptsize 0.90 } & \textit{\scriptsize 2.54}{\scriptsize{} } & {\scriptsize 2.70 } & {\scriptsize 1.29 } & \textit{\scriptsize 3.48}{\scriptsize{} }\tabularnewline
\hline 
\end{tabular}
\end{table}

\subsection{Discussion\label{sec:discuss}}

Regarding the evolution of SSE values in Figures~\ref{figSSERL}(a)
and ~\ref{figSSERLCat}(a), we observe for both data sets that the
$k$-anonymous microaggregation of input records effectively reduces
the required amount of noise and hence the loss of information, compared
to a straightforward implementation of $\varepsilon$-differential
privacy (with no prior microaggregation, \emph{i.e.}, $k=1$). The
relative improvement directly depends on the value of $\varepsilon$
and the best results are obtained for $\varepsilon=1.0$. As shown
in Tables~\ref{dif-num} and~\ref{dif-cat}, for $k=30$ and $\varepsilon=1.0$,
the relative improvement $SSE_{f}$ is around one order of magnitude
for the Census data set and is around 3 for the Adult data set.

Looking at Figures~\ref{figSSERL}(a) and~\ref{figSSERLCat}(a)
we observe different effects depending on the value of $\varepsilon$
for both data sets: 
\begin{itemize}
\item For small $\varepsilon$ (that is, 0.01 or 0.1), the larger $k$,
the smaller is SSE, because the noise reduction at the $\varepsilon$-differential
privacy stage more than compensates the noise increase at the microaggregation
stage due to greater aggregation. Anyway, the amount of noise involved
for these values of $\varepsilon$ is so high that even with the aforementioned
noise reduction, the output data are hardly useful. 
\item For very large $\varepsilon$ (that is, 10), there is a sharp decline
of SSE for low $k$ values (around 5); however, for larger $k$ (above
10), there is a new and slow increase in SSE, because the noise added
by $\varepsilon$-differential privacy being low, it is dominated
by the noise added by prior microaggregation in larger clusters. This
is more noticeable for the Census data set, since SSE values are more
similar to those of standard microaggregation. 
\item For medium $\varepsilon$ (that is, 1), there is a substantial decline
of SSE for low $k$ (below 30) and, for larger $k$, SSE stays nearly
constant and reasonably low. In this case, the noise added by prior
microaggregation in larger clusters is compensated by the noise reduced
at the $\varepsilon$-differential privacy stage due to decreased
sensitivity with larger $k$. 
\end{itemize}
Notice also that insensitive MDAV microaggregation incurs a higher
SSE than standard MDAV microaggregation. Indeed, the clusters formed
by insensitive microaggregation are less homogeneous, due to the total
order enforced for input records. Particularly, the Adult data set
shows a more noticeable increase of SSE figures. This is coherent
with the criterion detailed in Section~\ref{sec:insensitive-micro}
to define a total order, which alternatively picks combinations of
attribute domain boundaries as reference points to create clusters.
Since the evaluated Adult data set consists of two attributes, four
different reference points can be defined. This contrasts with the
four attributes considered for the Census data set, which provide
16 different combinations of domain boundaries, giving more degrees
of freedom and producing a more accurate clustering of input data.
Moreover, since the Adult data set consists of categorical attributes
with a limited set of possible categories (in comparison with continuous
scale numerical ranges defined by the Census attributes), the imperfections
introduced by the insensitive aggregation are amplified by the need
to discretize cluster centroids. In any case, the SSE increase caused
by insensitive microaggregation is around an order of magnitude smaller
than the noise reduction this microaggregation enables when used as
a prior step to $\varepsilon$-differential privacy.

RL values shown in Figures~\ref{figSSERL}(b) and ~\ref{figSSERLCat}(b)
behave the other way round as SSE. First, we notice that the standard
MDAV algorithm results in the highest percentage of linkages. For
the Census data set, a $k$-anonymity level $k\geq20$ is needed to
attain a percentage of linkages below 5\%. For Adult, RL is much lower
because the number of distinct records is limited by the set of categories
of each attribute (\emph{i.e.}, only 388 distinct tuples for Adult,
whereas all 1,080 records are different for Census), and because the
number of records is much higher (\emph{i.e.}, 30,162 for Adult vs.
1,080 for Census). As a result, the probability of correct record
linkage is much lower (\emph{i.e.}, below 1\% from $k=2$). Insensitive
MDAV yields slightly more privacy than MDAV for the Census data set
and significantly more privacy (less percentage of record linkages)
for the Adult data set. The superior RL reduction in Adult w.r.t.
Census is coherent with the differences in information loss observed
in SSE values, which were caused by the less homogeneous clusterization
in Adult. In both data sets, the RL values of insensitive microaggregation
are very similar to the ones obtained with $\varepsilon$-differential
privacy with $\varepsilon=10$. For $\varepsilon$ values of 0.01
and 0.1, the RL values hardly vary when the $k$-anonymity level increases,
because they are very low already with $k=1$ (no prior microaggregation).
Note that, for such low $\varepsilon$-values, the RL values stay
around 0.1\% for Census data, which, considering the data set size
of 1,080 records, corresponds to the probability of successful random
record linkage (\emph{i.e.}, 1/1,080). The fact that records are almost
randomly matched is reflected by the large spikes of the plot. For
the Adult data set, RL behaves similarly but it shows a much lower
matching probability (\emph{i.e.}, around 0.0033\%, that is, 1/30,162),
because of the larger cardinality of the data set. It can also be
seen that the top level of privacy offered by standard $\varepsilon$-differential
privacy ($k=1$) for low $\varepsilon$ is basically maintained when
using prior microaggregation ($k>1$); hence, the reduction in information
loss achieved by using microaggregation prior to noise addition does
not entail appreciable privacy penalties.

For $\varepsilon=1$, the RL results are more interesting. For the
Census data set, they show an increase of the percentage of record
linkages from 0.2\% for $k=1$ (no prior microaggregation) to around
1\% for $k=25$. For the Adult data set, RL rises from 0.02\% for
$k=1$ to around 0.05\% for $k\geq20$. This is the other side of
the very noticeable improvement of SSE values.

In all cases, as shown by $RL_{f}$ in Tables~\ref{mdav-num} and~\ref{mdav-cat},
$\varepsilon$-differential privacy reduces RL versus standard $k$-anonymity
from around 2 orders of magnitude (for $\varepsilon$ = 0.01 or 0.1)
to 1 order (for $\varepsilon$ = 1.0 or 10.0), for the considered
$k$-anonymity levels. This illustrates the practical privacy improvement
that $\varepsilon$-differential privacy brings as a result of the
more strict theoretical privacy guarantees.

By analyzing the balance ($Score$) between the SSE and RL figures
summarized in Tables~\ref{mdav-num},~\ref{dif-num},~\ref{mdav-cat}
and~\ref{dif-cat}, we can conclude that: 
\begin{itemize}
\item $Scores$ with respect to the standard MDAV algorithm (Tables~\ref{mdav-num}
and ~\ref{mdav-cat}) are above 1.0 in all cases for the Census data
set and for Adult when $\varepsilon=0.01$ or 0.1. This shows that
the improved disclosure risk brought by $\varepsilon$-differential
privacy more than compensates the relative increase of information
loss caused by noise. $Scores$ for the Adult data set are lower than
for Census because baseline $RL_{0}$ figures for Adult were so low
that the improvements brought by differential privacy are less noticeable
when evaluating disclosure risk. For the same reason, $Scores$ also
tend to decrease as the microaggregation level $k$ increases. 
\item $Scores$ with respect to standard $\varepsilon$-differential privacy
(\emph{i.e.}, $k=1$, Tables~\ref{dif-num} and ~\ref{dif-cat})
tend to increase as both $\varepsilon$ and $k$ grow. In fact, values
of $k$ and $\varepsilon$ over a threshold are needed for the balance
to show improvement (\emph{i.e.}, $Score>1.0$). We observe that for
$k\geq15$ and $\varepsilon\geq1.0$, the very substantial information
loss reduction obtained by using $k$-anonymous microaggregation prior
to $\varepsilon$-differential privacy more than compensates the small
increase in the percentage of record linkages with respect to standard
$\varepsilon$-differential privacy. 
\end{itemize}
The above observations suggest that, given a desired level $\varepsilon$
of differential privacy and a specific data set, a $k$-anonymity
level can be determined that optimizes the improvement of data utility
and/or privacy.

\subsection{Statistical analysis of anonymized results}

\label{sec:statistics-numerical}

To complement the above evaluation, in this section we provide an
attribute-level analysis of several statistics for the numerical data
set (Census). As in~\cite{Domingo10}, $\Theta$ and $\Theta'$ denote
the same statistic (\emph{e.g.}, attribute mean, attribute variance,
etc.) for each attribute over the original data set and its masked
version (by means of $k$-anonymity and/or $\varepsilon$-differential
privacy), respectively, we computed the variation of the statistic
introduced by the anonymization process as: 
\[
\Delta(\Theta)=\frac{\mid\Theta'-\Theta\mid}{\mid\Theta\mid}
\]
Variations were computed for the \textit{mean} of each attribute (named
$\Delta(m_{X})$ for FEDTAX, $\Delta(m_{P})$ for POTHVAL, $\Delta(m_{I})$
for INTVAL and $\Delta(m_{F})$ for FICA) and also for their \textit{variances}
($\Delta(\sigma_{X})$ for FEDTAX, $\Delta(\sigma_{P})$ for POTHVAL,
$\Delta(\sigma_{I})$ for INTVAL and $\Delta(\sigma_{F})$ for FICA).
In both cases, the smaller the variations, the less is the information
loss and the better is the data utility. Results are reported in Table~\ref{stats-num}.

\begin{table}[ht]
\caption{\label{stats-num}Census data set. Variation for several statistics
between the original data set and data sets anonymized with methods
using different values of $k$ and $\varepsilon$. Methods include
$\varepsilon$-differential privacy with prior $k$-anonymous microaggregation
($k=1$ amounts to plain $\varepsilon$-differential privacy), insensitive
MDAV microaggregation and plain MDAV microaggregation.}

\centering{}{\scriptsize }%
\begin{tabular}{|l|c|c|c|c|c|c|c|}
\hline 
{\scriptsize Statistic } & {\scriptsize $k$ } & {\scriptsize $\varepsilon=0.01$ } & {\scriptsize $\varepsilon=0.1$ } & {\scriptsize $\varepsilon=1.0$ } & {\scriptsize $\varepsilon=10.0$ } & {\scriptsize Insensit. MDAV } & {\scriptsize MDAV}\tabularnewline
\hline 
{\scriptsize $\Delta(m_{X})$ } & {\scriptsize 1 } & {\scriptsize 1.0947 } & {\scriptsize 1.0356 } & {\scriptsize 0.6925 } & {\scriptsize 0.0500 } & {\scriptsize 0.0 } & {\scriptsize 0.0 }\tabularnewline
 & {\scriptsize 2 } & {\scriptsize 1.1065 } & {\scriptsize 1.0088 } & {\scriptsize 0.4421 } & {\scriptsize 0.0134 } & {\scriptsize 0.0 } & {\scriptsize 0.0 }\tabularnewline
 & {\scriptsize 5 } & {\scriptsize 1.1030 } & {\scriptsize 0.8743 } & {\scriptsize 0.1461 } & {\scriptsize 0.0024 } & {\scriptsize 0.0 } & {\scriptsize 0.0 }\tabularnewline
 & {\scriptsize 15 } & {\scriptsize 1.0063 } & {\scriptsize 0.5171 } & {\scriptsize 0.0202 } & {\scriptsize 0.0003 } & {\scriptsize 0.0 } & {\scriptsize 0.0 }\tabularnewline
 & {\scriptsize 30 } & {\scriptsize 0.9677 } & {\scriptsize 0.2841 } & {\scriptsize 0.0030 } & {\scriptsize 0.0001 } & {\scriptsize 0.0 } & {\scriptsize 0.0 }\tabularnewline
\hline 
{\scriptsize $\Delta(m_{P})$ } & {\scriptsize 1 } & {\scriptsize 14.5160 } & {\scriptsize 13.7959 } & {\scriptsize 9.0362 } & {\scriptsize 1.2125 } & {\scriptsize 0.0 } & {\scriptsize 0.0 }\tabularnewline
 & {\scriptsize 2 } & {\scriptsize 14.4279 } & {\scriptsize 12.9546 } & {\scriptsize 6.2642 } & {\scriptsize 0.5245 } & {\scriptsize 0.0 } & {\scriptsize 0.0 }\tabularnewline
 & {\scriptsize 5 } & {\scriptsize 14.0466 } & {\scriptsize 11.4310 } & {\scriptsize 2.6323 } & {\scriptsize 0.1670 } & {\scriptsize 0.0 } & {\scriptsize 0.0 }\tabularnewline
 & {\scriptsize 15 } & {\scriptsize 13.4838 } & {\scriptsize 7.2579 } & {\scriptsize 0.7462 } & {\scriptsize 0.0302 } & {\scriptsize 0.0 } & {\scriptsize 0.0 }\tabularnewline
 & {\scriptsize 30 } & {\scriptsize 12.5205 } & {\scriptsize 4.4492 } & {\scriptsize 0.2970 } & {\scriptsize 0.0034 } & {\scriptsize 0.0 } & {\scriptsize 0.0 }\tabularnewline
\hline 
{\scriptsize $\Delta(m_{I})$ } & {\scriptsize 1 } & {\scriptsize 25.4754 } & {\scriptsize 24.6380 } & {\scriptsize 15.9486 } & {\scriptsize 2.2799 } & {\scriptsize 0.0 } & {\scriptsize 0.0 }\tabularnewline
 & {\scriptsize 2 } & {\scriptsize 24.8984 } & {\scriptsize 23.0202 } & {\scriptsize 10.9231 } & {\scriptsize 1.0192 } & {\scriptsize 0.0 } & {\scriptsize 0.0 }\tabularnewline
 & {\scriptsize 5 } & {\scriptsize 24.5284 } & {\scriptsize 19.3688 } & {\scriptsize 4.7766 } & {\scriptsize 0.3356 } & {\scriptsize 0.0 } & {\scriptsize 0.0 }\tabularnewline
 & {\scriptsize 15 } & {\scriptsize 23.6351 } & {\scriptsize 12.6533 } & {\scriptsize 1.4402 } & {\scriptsize 0.0656 } & {\scriptsize 0.0 } & {\scriptsize 0.0 }\tabularnewline
 & {\scriptsize 30 } & {\scriptsize 21.9726 } & {\scriptsize 7.7496 } & {\scriptsize 0.6244 } & {\scriptsize 0.0152 } & {\scriptsize 0.0 } & {\scriptsize 0.0 }\tabularnewline
\hline 
{\scriptsize $\Delta(m_{F})$ } & {\scriptsize 1 } & {\scriptsize 1.0126 } & {\scriptsize 0.9648 } & {\scriptsize 0.6151 } & {\scriptsize 0.0270 } & {\scriptsize 0.0 } & {\scriptsize 0.0 }\tabularnewline
 & {\scriptsize 2 } & {\scriptsize 1.0275 } & {\scriptsize 0.9225 } & {\scriptsize 0.4194 } & {\scriptsize 0.0090 } & {\scriptsize 0.0 } & {\scriptsize 0.0 }\tabularnewline
 & {\scriptsize 5 } & {\scriptsize 0.9927 } & {\scriptsize 0.8061 } & {\scriptsize 0.1304 } & {\scriptsize 0.0010 } & {\scriptsize 0.0 } & {\scriptsize 0.0 }\tabularnewline
 & {\scriptsize 15 } & {\scriptsize 0.9140 } & {\scriptsize 0.4945 } & {\scriptsize 0.0117 } & {\scriptsize 0.0004 } & {\scriptsize 0.0 } & {\scriptsize 0.0 }\tabularnewline
 & {\scriptsize 30 } & {\scriptsize 0.8812 } & {\scriptsize 0.2678 } & {\scriptsize 0.0026 } & {\scriptsize 0.0002 } & {\scriptsize 0.0 } & {\scriptsize 0.0 }\tabularnewline
\hline 
\hline 
{\scriptsize $\Delta(\sigma_{X})$ } & {\scriptsize 1 } & {\scriptsize 9.5299 } & {\scriptsize 9.2111 } & {\scriptsize 6.4855 } & {\scriptsize 0.5122 } & {\scriptsize 0.0 } & {\scriptsize 0.0 }\tabularnewline
 & {\scriptsize 2 } & {\scriptsize 9.4974 } & {\scriptsize 8.8891 } & {\scriptsize 4.3752 } & {\scriptsize 0.1067 } & {\scriptsize 0.0447 } & {\scriptsize 0.0053 }\tabularnewline
 & {\scriptsize 5 } & {\scriptsize 9.3824 } & {\scriptsize 7.8526 } & {\scriptsize 1.5560 } & {\scriptsize 0.0584 } & {\scriptsize 0.0804 } & {\scriptsize 0.0156 }\tabularnewline
 & {\scriptsize 15 } & {\scriptsize 9.0318 } & {\scriptsize 5.2659 } & {\scriptsize 0.1527 } & {\scriptsize 0.0972 } & {\scriptsize 0.1015 } & {\scriptsize 0.0398 }\tabularnewline
 & {\scriptsize 30 } & {\scriptsize 8.5521 } & {\scriptsize 2.9454 } & {\scriptsize 0.0461 } & {\scriptsize 0.1241 } & {\scriptsize 0.1254 } & {\scriptsize 0.0639 }\tabularnewline
\hline 
{\scriptsize $\Delta(\sigma_{P})$ } & {\scriptsize 1 } & {\scriptsize 69.3473 } & {\scriptsize 67.3757 } & {\scriptsize 45.9873 } & {\scriptsize 2.1168 } & {\scriptsize 0.0 } & {\scriptsize 0.0 }\tabularnewline
 & {\scriptsize 2 } & {\scriptsize 69.1904 } & {\scriptsize 64.9683 } & {\scriptsize 29.3618 } & {\scriptsize 0.4549 } & {\scriptsize 0.0697 } & {\scriptsize 0.0247 }\tabularnewline
 & {\scriptsize 5 } & {\scriptsize 68.5536 } & {\scriptsize 57.7628 } & {\scriptsize 8.0984 } & {\scriptsize 0.0597 } & {\scriptsize 0.1268 } & {\scriptsize 0.0991 }\tabularnewline
 & {\scriptsize 15 } & {\scriptsize 66.2145 } & {\scriptsize 35.8830 } & {\scriptsize 0.7419 } & {\scriptsize 0.2365 } & {\scriptsize 0.2429 } & {\scriptsize 0.1967 }\tabularnewline
 & {\scriptsize 30 } & {\scriptsize 62.0932 } & {\scriptsize 18.9228 } & {\scriptsize 0.0958 } & {\scriptsize 0.3370 } & {\scriptsize 0.3416 } & {\scriptsize 0.3214 }\tabularnewline
\hline 
{\scriptsize $\Delta(\sigma_{I})$ } & {\scriptsize 1 } & {\scriptsize 96.3225 } & {\scriptsize 93.3506 } & {\scriptsize 64.2854 } & {\scriptsize 3.0044 } & {\scriptsize 0.0 } & {\scriptsize 0.0 }\tabularnewline
 & {\scriptsize 2 } & {\scriptsize 96.0339 } & {\scriptsize 90.2298 } & {\scriptsize 40.5087 } & {\scriptsize 0.5915 } & {\scriptsize 0.0950 } & {\scriptsize 0.0349 }\tabularnewline
 & {\scriptsize 5 } & {\scriptsize 95.2261 } & {\scriptsize 79.1174 } & {\scriptsize 11.0832 } & {\scriptsize 0.0411 } & {\scriptsize 0.1358 } & {\scriptsize 0.1327 }\tabularnewline
 & {\scriptsize 15 } & {\scriptsize 91.7395 } & {\scriptsize 49.0829 } & {\scriptsize 1.1650 } & {\scriptsize 0.2330 } & {\scriptsize 0.2362 } & {\scriptsize 0.2614 }\tabularnewline
 & {\scriptsize 30 } & {\scriptsize 86.7387 } & {\scriptsize 24.3520 } & {\scriptsize 0.1243 } & {\scriptsize 0.4433 } & {\scriptsize 0.4476 } & {\scriptsize 0.4729 }\tabularnewline
\hline 
{\scriptsize $\Delta(\sigma_{F})$ } & {\scriptsize 1 } & {\scriptsize 16.3302 } & {\scriptsize 15.7698 } & {\scriptsize 11.3811 } & {\scriptsize 0.9712 } & {\scriptsize 0.0 } & {\scriptsize 0.0 }\tabularnewline
 & {\scriptsize 2 } & {\scriptsize 16.2702 } & {\scriptsize 15.2505 } & {\scriptsize 7.8828 } & {\scriptsize 0.2338 } & {\scriptsize 0.0593 } & {\scriptsize 0.0067 }\tabularnewline
 & {\scriptsize 5 } & {\scriptsize 16.1054 } & {\scriptsize 13.6669 } & {\scriptsize 3.0073 } & {\scriptsize 0.0625 } & {\scriptsize 0.1117 } & {\scriptsize 0.0224 }\tabularnewline
 & {\scriptsize 15} & {\scriptsize 15.4965 } & {\scriptsize 9.3544 } & {\scriptsize 0.3439 } & {\scriptsize 0.1580 } & {\scriptsize 0.1634 } & {\scriptsize 0.0670 }\tabularnewline
 & {\scriptsize 30 } & {\scriptsize 14.7710 } & {\scriptsize 5.4324 } & {\scriptsize 0.0423 } & {\scriptsize 0.1780 } & {\scriptsize 0.1797 } & {\scriptsize 0.1060 }\tabularnewline
\hline 
\end{tabular}
\end{table}

The variations of the \emph{attribute means} directly depend on the
amount of noise added to the anonymized output. Hence, for the two
$k$-anonymous MDAV implementations, attribute means are perfectly
preserved in the masked output since centroids are the exact means
of clustered values. Regarding differentially privacy implementations,
we observe a monotonic decrease for the variations of the mean for
all attributes as the $k$-anonymity factor applied to input data
increases from $k=1$ to $k=30$. This shows the benefits that data
microaggregation brings at reducing the amount of noise needed to
fulfill differential privacy. For fixed $\varepsilon$, the sharpness
of this monotonic decrease is similar for all attributes. However,
as $\varepsilon$ increases from 0.01 to 10.0, the decrease becomes
sharper and sharper for all attributes. Indeed, for $\varepsilon=0.01$
the decrease factor for the variation of the mean is around 1.1 for
all attributes (quotient of the variations of the mean for $k=1$
and $k=30$), whereas for $\varepsilon=10.0$ the decrease factor
reaches around 200. Hence, we see that $\varepsilon>0.1$ is needed
to significantly reduce baseline variations of the mean for all attributes
(we take as baseline the variations for $k=1$, that is for plain
$\varepsilon$-differential privacy without prior microaggregation).

The variations of the \textit{attribute variances} increase for the
two MDAV implementations as the $k$-anonymity level grows, since
output record values tend to be more 
homogeneous and thereby suppress
more variance as a result of the data aggregation process. The growth
factor is larger for the standard MDAV algorithm in comparison with
its insensitive version, since the latter tends to produce less homogeneous
clusters. Differential privacy implementations behave the other way
round. For $\varepsilon\leq1.0$, the variations of attribute variances
decrease as the $k$-anonymity level grows, for all attributes. This
suggests that prior microaggregation helped to decrease the large
variance introduced by the noise added to fulfill differential privacy.
Similarly to what happened for variations of means, decrease factors
for variations of variances are larger for higher $\varepsilon$ values.
Results with $\varepsilon=10.0$ are worth noting. In this case, variances
tend to increase for $k$ values above 5. As discussed in the previous
section, the noise added for such a high $\varepsilon$ value is so
low that the effect of the prior microaggregation dominates in larger
clusters. In other words, prior microaggregation followed by 10-differential
privacy behaves similarly to microaggregation alone.

The results of the above analysis of attribute-level statistics are
coherent with the results based on SSE presented in previous sections.
It becomes clear that prior microaggregation helps differentially
private data to retain the utility of original data much like standard
$k$-anonymity does.

\section{Conclusions\label{sec:conclusions}}

We have presented an approach that combines $k$-anonymity and $\varepsilon$-differential
privacy in order to reap the best of both models: namely, the reasonably
low information loss incurred by $k$-anonymity and the high privacy
level guaranteed by $\varepsilon$-differential privacy. In our approach,
we use a newly defined insensitive microaggregation to obtain a $k$-anonymous
data set by considering all attributes as quasi-identifiers; then
we take the $k$-anonymous microaggregated data set as an input to
which uncertainty is added in order to reach $\varepsilon$-differential
privacy. We have also described how our approach can be applied to
numerical and categorical attributes and also to records combining
heterogeneous attribute types.

In addition to a theoretical proposal, we have presented empirical
results for heterogeneous data sets which show that our approach reduces
the information loss of standard differential privacy by several orders
of magnitude, while preserving its theoretical privacy guarantee and
improving the practical privacy (percentage of record linkages) versus
standard $k$-anonymity.

Future work will involve at least the following research lines: 
\begin{itemize}
\item Even though special care has been exerted to avoid damaging within-cluster
homogeneity when making microaggregation insensitive, there is still
room for improvement, especially for categorical data. New criteria
to define total orders are conceivable, such as fixing sampling and
sorting strategies of data spaces, so that the within-cluster homogeneity
reaches levels more similar to the ones achieved by standard microaggregation. 
\item It would also be interesting to define a methodology that, given a
data set, a target privacy level $\varepsilon$ and fixed utility
and privacy measures, determines the most suitable $k$ for the prior
$k$-anonymous microaggregation, in view of optimizing the data utility
and/or disclosure risk. 
\end{itemize}

\lhead[\chaptername~\thechapter]{\rightmark}

\rhead[\leftmark]{}

\lfoot[\thepage]{}

\cfoot{}

\rfoot[]{\thepage}

\chapter{Differential privacy via $t$-closeness in data publishing\label{chap:t-closeness}}

$k$-Anonymity and $\varepsilon$-differential privacy are two mainstream
privacy models originated within the computer science community. 
Their approaches
towards disclosure limitation are quite different: $k$-anonymity
is a model for releases of microdata ({\em i.e.}
individual records) that seeks to prevent record
re-identification by hiding each original record 
within a group of $k$ indistinguishable
anonymized records, while $\varepsilon$-differential privacy originated
as a model
for interactive databases and seeks to limit the knowledge that users
obtain from query responses. Both models are often presented as antagonistic: 
$\varepsilon$-differential privacy supporters view 
$k$-anonymity as an old-fashioned privacy notion that offers only
poor disclosure limitation guarantees, while $\varepsilon$-differential
privacy detractors criticize the limited utility of $\varepsilon$-differentially
private outputs and the cumbersomeness of not having access to
the data set. 

We show that for data set anonymization, the $t$-closeness
extension of $k$-anonymity is closely related to $\varepsilon$-differential
privacy. 
This relation is demonstrated both versus uninformed
intruders (having access only to the released data set) and informed
intruders (having also background knowledge).
For uninformed intruders we prove that 
$\exp(\varepsilon)$-closeness is equivalent to
$\varepsilon$-differential privacy. For informed intruders,
the strict equivalence we obtain for uninformed intruders 
does not hold; however, we show that $\exp(\varepsilon)$-closeness
can be seen as a good approximation to $\varepsilon$-differential
privacy. Our approach is a constructive one: we specify a computational
procedure based on bucketization that, given
an original data set, builds a $t$-close version of it. In the 
case of uninformed intruders, this version turns out to 
be differentially private as well; in the case of informed
intruders, it is approximately differentially private.

Section~\ref{sec:Optimal-generic--anonymous} reviews partitioning
strategies used to achieve $k$-anonymity and its extensions,
including $t$-closeness. Specifically, we first
examine the shortcomings of achieving $k$-anonymity and $t$-closeness
in the classical sense, that is, by modifying the quasi-identifier
attributes to create 
groups of at least $k$ indistinguishable records. We then review
two approaches which leave the quasi-identifier attributes unaltered
and which will be used as building blocks of our computational procedure
to reach $t$-closeness and $\varepsilon$-differential privacy.
Section~\ref{sec:t_closeness} develops in detail our 
proposed construction to reach
$t$-closeness. 
Section~\ref{from} shows that $\exp(\varepsilon)$-closeness 
reached with the previous construction implies: i) $\varepsilon$-differential
privacy in the case of uninformed intruders; ii) approximate
$\varepsilon$-differential privacy in the case of informed intruders.
Conclusions are summarized in
Section~\ref{conclusion}.

The contents of this chapter have been accepted for publication in~\cite{pst2013}.

\section{Partitioning strategies for $k$-anonymity\label{sec:Optimal-generic--anonymous}}

In $k$-anonymity and its extensions (including $l$-diversity
and $t$-closeness),
the partitioning strategy to create groups of indistinguishable
records is a key point for data utility. Assume a data user 
who wants to analyze a group of individuals that has been selected based on the
value of the quasi-identifier attributes. The utility that this user
derives from the $k$-anonymous data depends on how well the
target group of individuals
can be approximated by the groups of indistinguishable
records. The best utility is achieved when the target group of individuals
can be approximated as the union of groups of indistinguishable records. 

As an example, consider a data set with 16 records,
two quasi-identifier attributes
$Q_{1}$ and $Q_{2}$, and one confidential attribute $C$. Assume
that $Q_{1}$ and $Q_{2}$ take values in the sets $\{A,B,C,D\}$
and $\{P,Q,R,S\}$, respectively. Figure~\ref{fig:projection} represents
the projection of this data set on the quasi-identifier attributes;
each point is the projection of one record on the quasi-identifiers.
In this case, we assume that each possible combination of quasi-identifiers
occurs in exactly one record. 
The dominant approach towards $k$-anonymity uses generalization
and suppression to partition the data set into groups of $k$ indistinguishable
records. Assume that the generalization hierarchies are those in Figure~\ref{fig:gen_hierarchy},
and that we want to obtain a $4$-anonymous data set. 
Figure~\ref{fig:generalized_data_sets}
depicts the $4$-anonymous data sets produced by minimal generalizations.
A data user interested in the group of individuals with $Q_{1}=A$
would prefer the $4$-anonymous data set on the left, 
which still allows distinguishing that group.
On the contrary, the
data set on the right of Figure~\ref{fig:generalized_data_sets}
is the worst option
for a user interested in the individuals 
with $Q_1=A$, because all values of $Q_1$ are lumped together.
However,
a user interested in the group of individuals with $Q_{2}=P$ would
prefer the $4$-anonymous data set on the right
of Figure~\ref{fig:generalized_data_sets}.

\begin{figure}
\begin{centering}
\includegraphics[width=5cm]{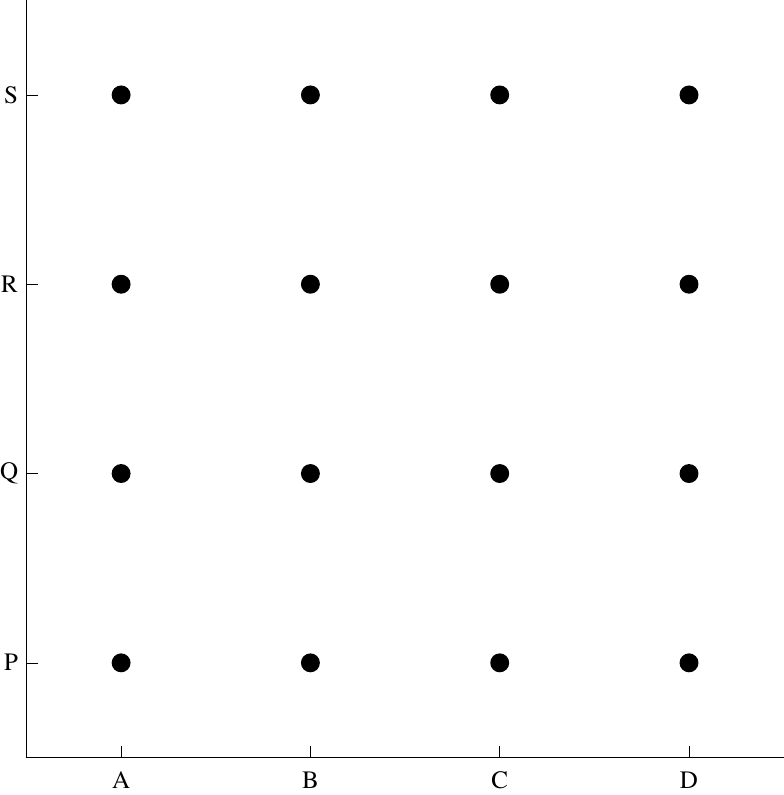}
\par\end{centering}

\caption{\label{fig:projection}Projection of the records in the data set on 
the quasi-identifier attributes $Q_{1}$ and $Q_{2}$}

\end{figure}

\begin{figure*}
\begin{centering}
\includegraphics[width=6cm]{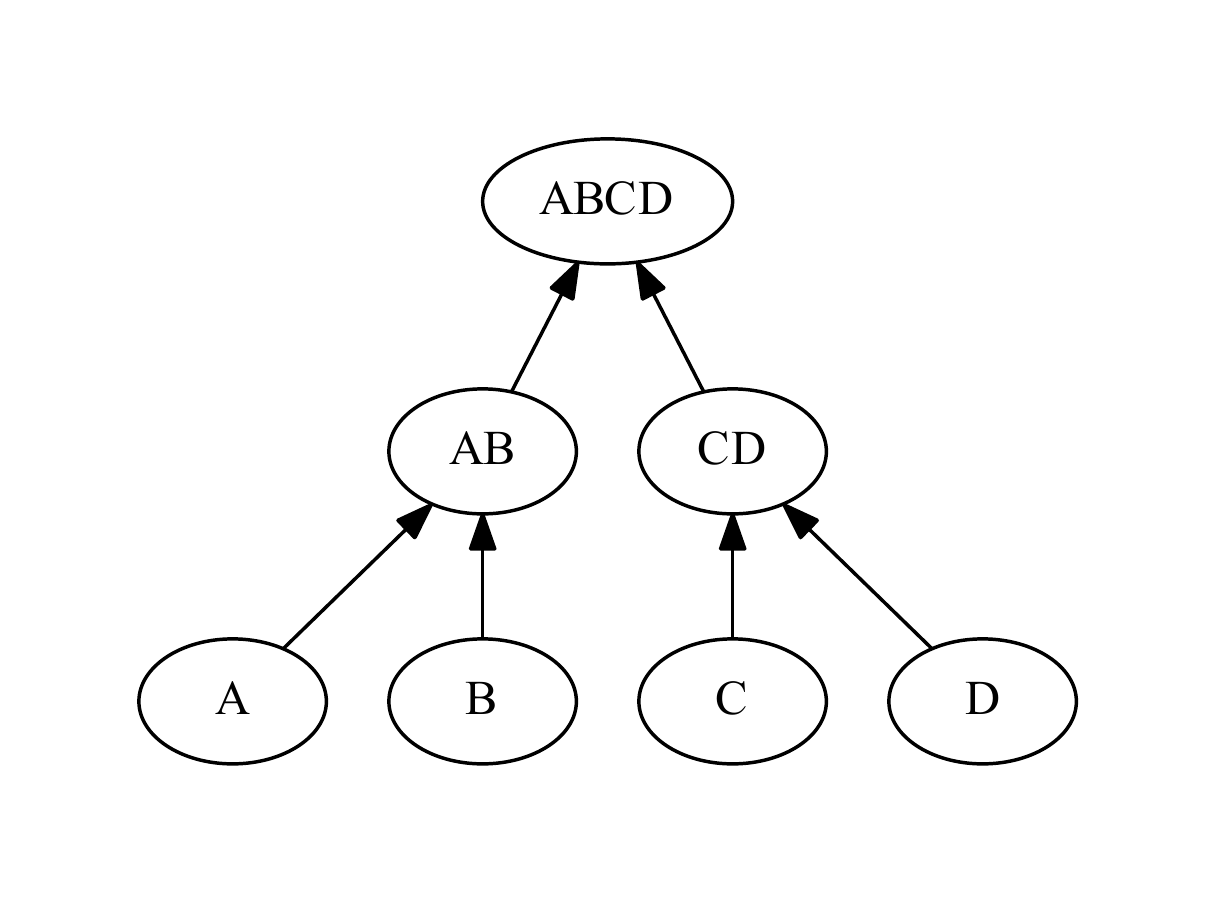}~\includegraphics[width=6cm]{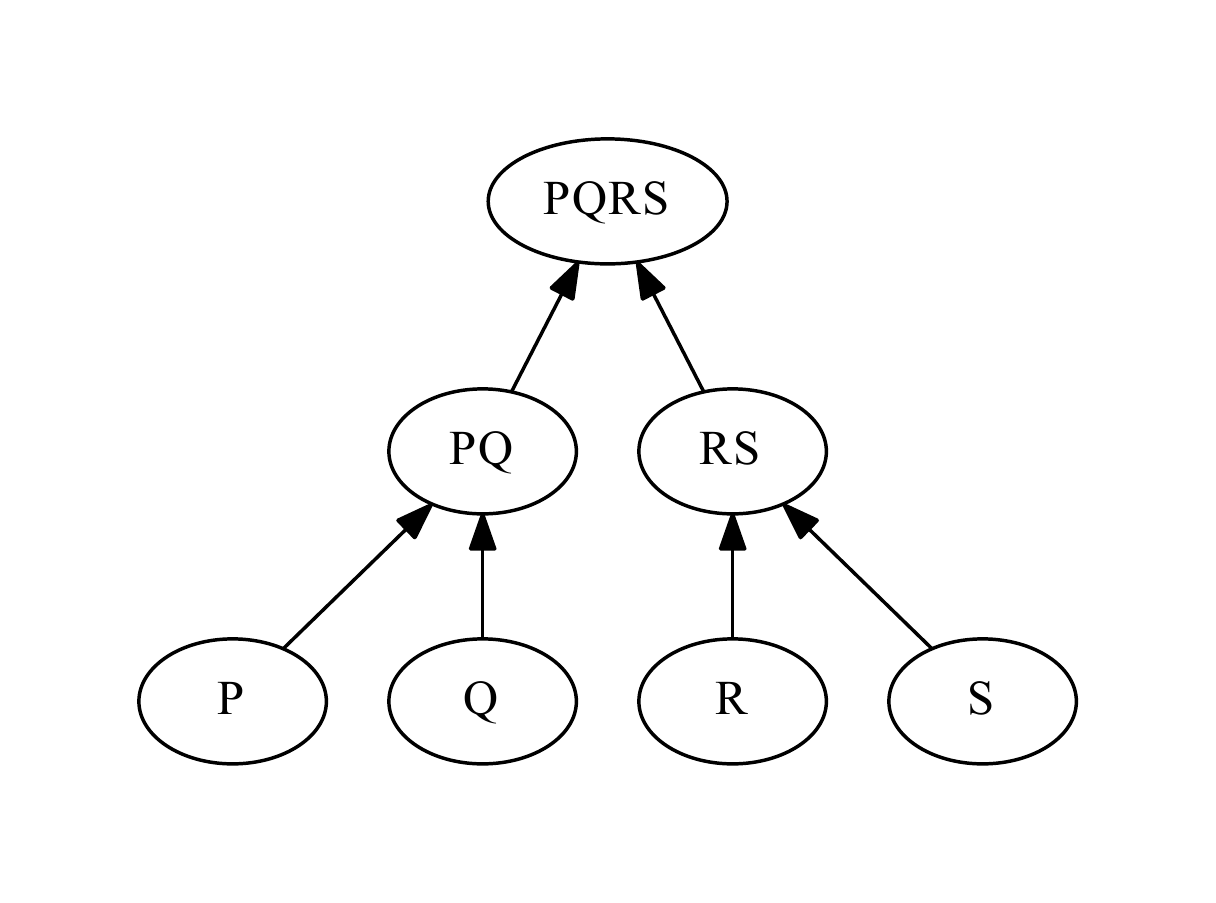}
\par\end{centering}

\caption{Generalization hierarchies for attributes $Q_{1}$ (left) and
$Q_{2}$ (right)}
\label{fig:gen_hierarchy}
\end{figure*}

\begin{figure*}

\begin{centering}
\includegraphics[width=5cm]{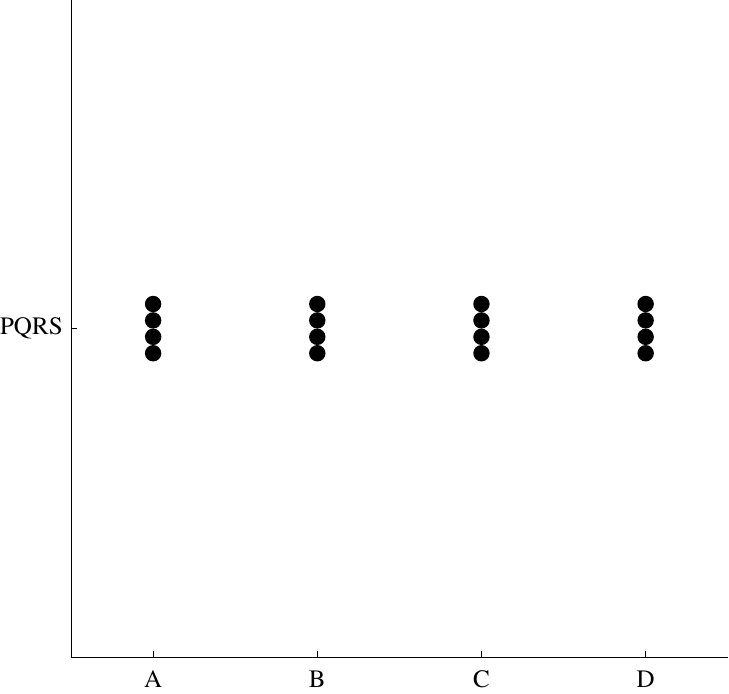}~\includegraphics[width=5cm]{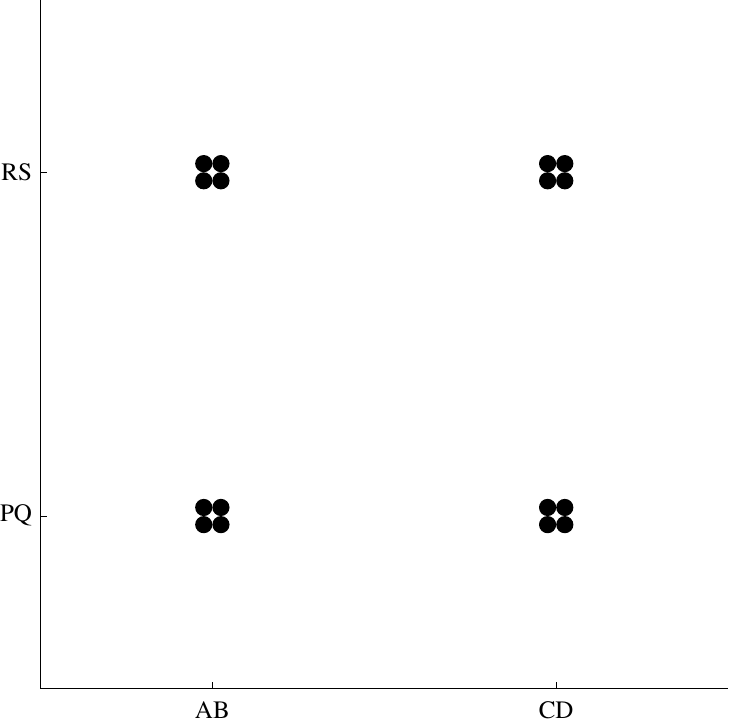}~\includegraphics[width=5cm]{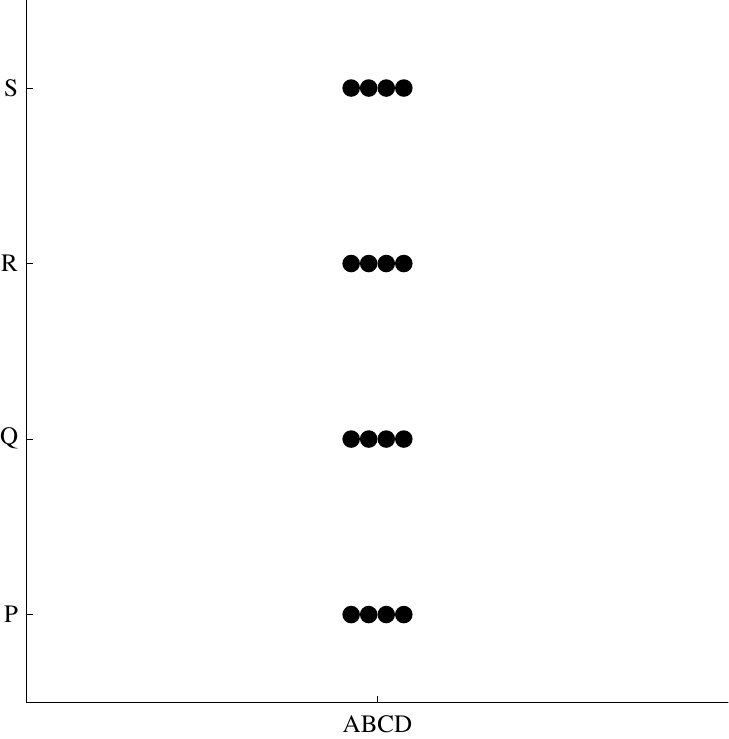}
\par\end{centering}

\caption{Minimal 4-anonymous generalizations}
\label{fig:generalized_data_sets}

\end{figure*}

Therefore, the selected partitioning of the records is essential for
the protected data set to deliver high utility. If the data collector
is aware of the kind of analyses that data users are interested in,
then the collector can tailor the partitioning to those analyses. 
However, most
of the time the data collector is unaware of the intended use of
the data; thus, a customized partitioning is not feasible. Even if
the data collector knew the relevant analyses, different analyses 
may require different partitions,
but releasing several versions of the same data set using a 
different partition each is not advisable, as it would endanger
whatever anonymity is gained by partitioning.

Another problem of the generalization approach is related to the number
of quasi-identifiers. When there is a large number of quasi-identifiers,
all of them need to be generalized
to satisfy $k$-anonymity, which
results in a large information loss. This
is known as ``the curse of dimensionality''~\cite{Aggarwal2005}.

This section aims at a method to generate $k$-anonymous data sets 
that mitigates
the issues described above:
\begin{itemize}
\item We generate a partition of the
records that preserves as much information as possible. To construct
such a partition, instead of partitioning based on the quasi-identifiers,
we will do it based on the confidential attribute. 
\item To avoid losing information
on the quasi-identifier attributes, we replace generalization
of the quasi-identifiers by an approach that preserves both 
the quasi-identifiers
and the confidential attributes. In particular,
we propose to use either
the Anatomy~\cite{Sun2009} or the probabilistic 
$k$-anonymity~\cite{fuzzieee}
methods.
\end{itemize}

If our goal is to construct a $k$-anonymous data
set, the Anatomy method is better, as it preserves more information,
namely the distribution of the confidential attribute within 
each set in the partition. However, 
if our goal is to achieve $t$-closeness,
we will show that the probabilistic $k$-anonymity approach is
preferable.

\subsection{Partitioning based on the confidential attribute\label{sub:conf_partition}}

We have argued above that customizing the $k$-anonymous partition
to specific data analysis requirements is not an option. 
Hence, the utility of the data depends
on the amount of variability of the confidential attribute. For instance,
if we target a specific individual,
the quasi-identifiers allow us to determine
a group of $k$ records that must contain that individual; thus,
we know that each of the values of the confidential attribute within
the group has probability $1/k$ of corresponding to the target individual.
The amount of knowledge we get (and thus the utility) depends on the
variability of the confidential attribute within the group: 
the more similar the confidential attribute values, 
the more knowledge for the user, but also the higher the 
risk of attribute disclosure.

To limit the variability of the confidential attribute within groups
of indistinguishable records, we propose to partition the records
in the data set based on the value of the confidential attribute.
We focus on a numerical confidential attribute. Let $D$ be a data
set with quasi-identifiers collectively denoted as $QI$,
and a confidential attribute $C$, as represented in 
Table~\ref{dataset}.

For the sake of clarity, we take a single confidential attribute. If
there are several confidential attributes, we can treat them as a single
compound confidential attribute and partition the data set according 
to a proximity criterion that takes into account all the components
({\em e.g.} microaggregation over confidential attributes~\cite{Domingo2005}).

\begin{table}
\begin{center}
\caption{Data set with quasi-identifiers $QI$ and a confidential
attribute $C$}
\label{dataset}
\begin{tabular}{ccc}
 & $QI$ & $C$\tabularnewline
\cline{2-3} 
individual 1 & $q_{1}$ & $c_{1}$\tabularnewline
\cline{2-3} 
$\vdots$ & $\vdots$ & $\vdots$\tabularnewline
\cline{2-3} 
individual N & $q_{N}$ & $c_{N}$\tabularnewline
\cline{2-3} 
\end{tabular}
\end{center}
\end{table}

To minimize the variability of $C$, we sort
the records by $C$, and generate the
partition by taking the $k$ minimal and maximal records, iteratively
(see Algorithm~\ref{alg:Optimal-partitioning}).

\begin{algorithm}
{\bf let} $D=\{(q_{i},c_{i})|i=1,\ldots,N\}$ be the original data
set

{\bf let} $P=\emptyset$ the partition of $D$ to be returned

{\bf let} $O=((oq_{1},oc_{1}),\ldots,(oq_{N},oc_{N}))$ be the list
of records of $D$ ordered by ascending values $c_i$ 

{\bf while} $|O|\geq 3k$ {\bf do}

\hspace{0.5cm}{\bf let} $P_{min}$ be the set containing the first
$k$ records of $O$

\hspace{0.5cm}insert $P_{min}$ into $P$

\hspace{0.5cm}remove the first $k$ records from $O$

\hspace{0.5cm}{\bf let} $P_{max}$ be the set containing the last
$k$ records of $O$

\hspace{0.5cm}insert $P_{max}$ into $P$

\hspace{0.5cm}remove the last $k$ records from $O$

{\bf end while}

{\bf if} $|O| \geq 2k$ {\bf then}

\hspace{0.5cm}{\bf let} $P_{min}$ be the set containing the first
$k$ records of $O$

\hspace{0.5cm}insert $P_{min}$ into $P$

\hspace{0.5cm}remove the first $k$ records from $O$

{\bf end if}

{\bf let} $P_{rest}$ be the set with the records remaining in $O$

insert $P_{rest}$ into $P$

{\bf return} $P$

\caption{\label{alg:Optimal-partitioning}Optimal partitioning based on the
confidential attribute}

\end{algorithm}

\subsection{Anatomy: reducing information loss in quasi-identifiers\label{sub:anatomy}}

We have mentioned above the ``curse of dimensionality''
information loss problem inherent to generalizations
affecting many quasi-identifier attributes.
The problem may get even worse if we construct the partition based
on the confidential attribute, as proposed in the previous section.
The values of the quasi-identifier attributes in each group of the partition
may span the whole domains of those attributes (or substantial fractions
of them).
Therefore, replacing all values of each quasi-identifier attribute
within a group by a single generalized value would lead to a great utility loss.
Moreover, note that
the generalized values for the quasi-identifiers might coincide
for different groups.

To overcome this difficulty, we propose to use the Anatomy approach
to $k$-anonymity, which preserves the original values of the quasi-identifiers.
To dissociate (break the relation between) quasi-identifiers and 
confidential attributes,
two tables are generated: the first one assigns a group identifier
to the quasi-identifiers, and the second one relates each group identifier
to the confidential attributes. We illustrate this in 
Tables~\ref{fig:deidentified_data}, \ref{fig:three_anonymous}
and~\ref{fig:group_id}. Table~\ref{fig:deidentified_data} shows
the original de-identified data. Table~\ref{fig:three_anonymous}
presents a $3$-anonymous version of the data obtained by generalization
of the attributes {\em Date of Birth} and {\em Sex}. Note that
we have used the greatest level of generalization for those attributes,
and thus the information loss is large. In contrast, we observe in
Table~\ref{fig:group_id} that, by using a group identifier 
to relate quasi-identifier attributes
and confidential attributes, we achieve exactly what we wanted: we
$k$-anonymize {\em the relation} 
between quasi-identifiers and confidential attributes,
while preserving the values of quasi-identifier attributes and 
the confidential attribute. In particular, the distribution of the 
confidential attribute within each group (records sharing the same
group identifier) is preserved.

\begin{table}
\begin{centering}
\caption{\label{fig:deidentified_data}Original de-identified medical data }
\begin{tabular}{llll}
\textbf{Ethnicity} & \textbf{Date of Birth} & \textbf{Sex} & \textbf{Problem}\tabularnewline
\hline 
asian & 09/27/64 & female & hypertension\tabularnewline
\hline 
asian & 05/08/61 & female & obesity\tabularnewline
\hline 
asian & 04/18/64 & male & chest pain\tabularnewline
\hline 
black & 03/13/63 & male & hypertension\tabularnewline
\hline 
black & 03/18/63 & male & shortness of breath\tabularnewline
\hline 
black & 09/07/64 & female & obesity\tabularnewline
\hline 
white & 05/14/61 & male & chest pain\tabularnewline
\hline 
white & 05/08/63 & male & obesity\tabularnewline
\hline 
white & 09/15/61 & female & shortness of breath\tabularnewline
\hline 
\end{tabular}
\par\end{centering}
\end{table}

\begin{table}
\begin{centering}
\caption{\label{fig:three_anonymous}$3$-Anonymous data set}
\begin{tabular}{llcl}
\textbf{Ethnicity} & \textbf{Date of Birth} & \textbf{Sex} & \textbf{Problem}\tabularnewline
\hline 
asian & {[}61,64{]} & - & hypertension\tabularnewline
\hline 
asian & {[}61,64{]} & - & obesity\tabularnewline
\hline 
asian & {[}61,64{]} & - & chest pain\tabularnewline
\hline 
black & {[}61,64{]} & - & hypertension\tabularnewline
\hline 
black & {[}61,64{]} & - & shortness of breath\tabularnewline
\hline 
black & {[}61,64{]} & - & obesity\tabularnewline
\hline 
white & {[}61,64{]} & - & chest pain\tabularnewline
\hline 
white & {[}61,64{]} & - & obesity\tabularnewline
\hline 
white & {[}61,64{]} & - & shortness of breath\tabularnewline
\hline 
\end{tabular}
\par\end{centering}
\end{table}

\begin{table}
\begin{centering}
\caption{\label{fig:group_id}Left, relation between quasi-identifiers and group
identifier. Right, relation between group identifier
and confidential attribute.}
\begin{tabular}{llll}
\textbf{Ethnicity} & \textbf{Date of Birth} & \textbf{Sex} & \textbf{ID}\tabularnewline
\hline 
asian & 09/27/64 & female & 1\tabularnewline
\hline 
asian & 05/08/61 & female & 1\tabularnewline
\hline 
asian & 04/18/64 & male & 1\tabularnewline
\hline 
black & 03/13/63 & male & 2\tabularnewline
\hline 
black & 03/18/63 & male & 2\tabularnewline
\hline 
black & 09/07/64 & female & 2\tabularnewline
\hline 
white & 05/14/61 & male & 3\tabularnewline
\hline 
white & 05/08/63 & male & 3\tabularnewline
\hline 
white & 09/15/61 & female & 3\tabularnewline
\hline 
\end{tabular}~%
\begin{tabular}{cl}
\textbf{ID} & \textbf{Problem}\tabularnewline
\hline 
1 & hypertension\tabularnewline
\hline 
1 & obesity\tabularnewline
\hline 
1 & chest pain\tabularnewline
\hline 
2 & hypertension\tabularnewline
\hline 
2 & shortness of breath\tabularnewline
\hline 
2 & obesity\tabularnewline
\hline 
3 & chest pain\tabularnewline
\hline 
3 & obesity\tabularnewline
\hline 
3 & shortness of breath\tabularnewline
\hline 
\end{tabular}
\par\end{centering}
\end{table}

\subsection{Comparison of partitioning strategies}

When generating the partition based on the quasi-identifiers, small
values of the parameter $k$ are typically used. Usually, the variability
of the confidential attribute thus obtained is large enough not to
lead to attribute disclosure. However, when basing the partition on the
confidential attribute, small values of $k$ will almost certainly
lead to attribute disclosure, because in this case 
the within-group variability of the confidential attribute is small.

However, constructing the partition based on the confidential attribute
has one important advantage: it allows fixing the desired level of
variability for the confidential attribute. 
Indeed, parameter $k$ can be increased
to a value that provides effective disclosure limitation guarantees.
For instance, by setting $k$ to $0.1\times N$, we guarantee that
the confidential attribute for any individual is hidden inside a group
of individuals that amount to a $10\%$ of the actual sample. 

Note that if partitioning
is based on the quasi-identifiers, we cannot control the level of
variability of the confidential attribute inside each of the $k$-anonymous
groups: some of them may exhibit a large variability (which offers
protection against
attribute disclosure, but poor data utility) and others may not 
(which offers good data
utility, but high risk of attribute disclosure). The underlying problem is
the impossibility of
enforcing a {\em predetermined} amount of variability: variability
increases with $k$, but the relationship between $k$ and the amount
of variability of the confidential attribute is not clear. Usually,
$k$ must be small (if any utility is to be provided), which results
in poor disclosure limitation guarantees. 

\section{A bucketization construction 
to achieve $t$-closeness\label{sec:t_closeness}}

It has been argued above that when partitioning is based on the confidential
attribute, the value of $k$ must be increased to provide effective
disclosure limitation. In this section we seek to enforce a stronger
disclosure limitation criterion: $t$-closeness. $t$-Closeness limits
the knowledge gain that an intruder can derive 
from the $k$-anonymous groups.
The distribution of the confidential attribute within each of the
$k$-anonymous groups is required to be similar to the distribution
of the confidential attribute on the whole data set. 

For $t$-closeness to
be satisfied, the distance between the data set-level 
and the group-level distribution
of the confidential attribute must be less than $t$ for any group.
When $t$-closeness was introduced, the Earth Mover's distance (EMD)
was proposed~\cite{Li2007}. 
The EMD measures the minimal amount of work required
to transform one distribution to another by moving probability mass
between each other. 

The kind of guarantee that $t$-closeness offers depends on the distance
function used. We aim at achieving an $\varepsilon$-differentially
privacy-like guarantee, and this requires us to use a different distance.
$\varepsilon$-Differential privacy guarantees that, for any two data
sets that differ in one individual, the probability for a query response
computed on either data set 
to belong to an arbitrary set $S$ differs 
at most by a factor $\exp(\varepsilon)$.
The distance function we propose mimics the $\varepsilon$-differential
privacy criterion.
\begin{defn}
\label{def:distance}Given two random distributions $\mathcal{D}_{1}$
and $\mathcal{D}_{2}$, we define the distance between $\mathcal{D}_{1}$
and $\mathcal{D}_{2}$ as:
\[
d(\mathcal{D}_{1},\mathcal{D}_{2})=
\max_{S}\{\frac{\Pr_{\mathcal{D}_{1}}(S)}{\Pr_{\mathcal{D}_{2}}(S)},
\frac{\Pr_{\mathcal{D}_{2}}(S)}{\Pr_{\mathcal{D}_{1}}(S)}\}
\]
where $S$ is an arbitrary (measurable) set, and we take the quotients
of probabilities to be zero, if both $\Pr_{\mathcal{D}_{1}}(S)$ and
$\Pr_{\mathcal{D}_{2}}(S)$ are zero, and to be infinity if only the
denominator is zero.
\end{defn}

If the distributions $\mathcal{D}_{1}$ and $\mathcal{D}_{2}$ are
discrete (as it is the case for the sampling distribution of the confidential
attribute in a microdata set), computing the distance between them
is simpler: taking the maximum over the possible individual values
suffices.
\begin{prop}
\label{prop1}
If distributions $\mathcal{D}_{1}$ and $\mathcal{D}_{2}$ take values
in a discrete set $\{x_{1},$ $\ldots,$ $x_{N}\}$, then the distance $d(\mathcal{D}_{1},\mathcal{D}_{2})$
can be computed as 
\begin{equation}
d(\mathcal{D}_{1},\mathcal{D}_{2})=\max_{i=1,\ldots,N}\{\frac{\Pr_{\mathcal{D}_{1}}(x_{i})}{\Pr_{\mathcal{D}_{2}}(x_{i})},\frac{\Pr_{\mathcal{D}_{2}}(x_{i})}{\Pr_{\mathcal{D}_{1}}(x_{i})}\}\label{eq:distance}
\end{equation}
\end{prop}

To satisfy $t$-closeness, the groups in the partition must be selected such
that the distance of the distribution of the confidential attribute
on the whole data set and the distribution on each of 
the groups is less
than $t$. When using the previously defined distance, if we work
with the sampling distribution of the confidential attribute (assuming
that at least one of the values of the confidential attribute has multiplicity
less than the cardinality of the partition), the distance is always
infinity.
The reason is that the distance due to values of the confidential
attribute that do not appear within the group is infinity (according
to Definition~\ref{def:distance}). To
avoid this issue, instead of working with the sampling distribution
of the confidential attribute, we work with a bucketized version of
it, where several points are clustered into a set of buckets $B_{1},\ldots,B_{n}$.
In Figure~\ref{fig:discretization} the values of the confidential
attribute in the original data have been clustered in buckets 
$B_1$, $B_2$ and $B_3$
that contain four points each. From this step we get a distribution
for the confidential attribute with diminished granularity.

\begin{figure}
\begin{centering}
\includegraphics[width=9cm]{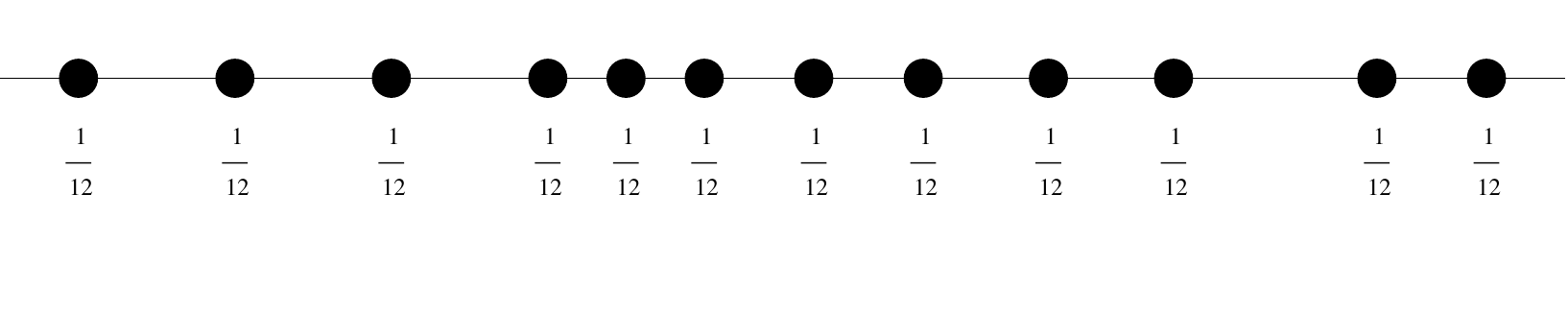}
\par\end{centering}

\begin{centering}
\includegraphics[width=9cm]{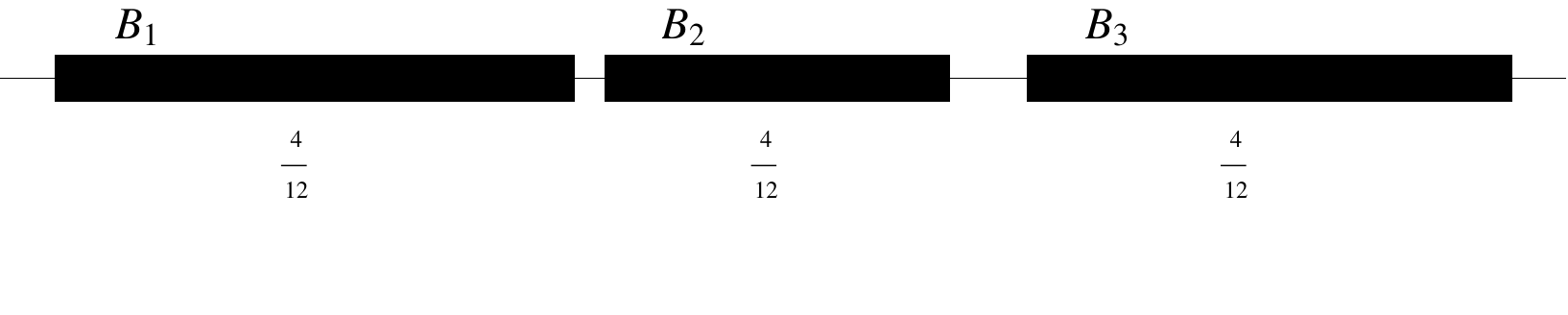}
\par\end{centering}

\caption{\label{fig:discretization}Top, original confidential attribute 
values. Bottom, bucketized confidential attribute values.}
\end{figure}

By using the proposed bucketization it is feasible to attain $t$-closeness
for a finite $t$. For instance, Figure~\ref{fig:t_close_partition}
shows a $4$-anonymous partition of the data set that satisfies $1.5$-closeness,
according to the previously defined distance. The sampling distribution
of the original data assigns probability $1/3$ to each of the buckets
$B_1$, $B_2$ and $B_3$;
hence, the bucket-level distribution $\mathcal{D}$ of the confidential
attribute in the original data set is $\Pr(B_1)=\Pr(B_2)=\Pr(B_3)=1/3$.
Each of the groups in the partition ($P_{1}$, $P_{2}$, $P_{3}$) takes 
either one or two points from each bucket.
Hence, the bucket-level
distribution $\mathcal{D}(P_1)$ of the confidential attribute for group $P_1$
is $\Pr(B_1)=1/2$ and $\Pr(B_2)=\Pr(B_3)=1/4$;
for group $P_2$ the distribution, denoted by $\mathcal{D}(P_2)$, is 
$\Pr(B_1)=\Pr(B_3)=1/4$ and $\Pr(B_2)=1/2$;
for group $P_3$ the distribution, denoted by $\mathcal{D}(P_3)$, is
$\Pr(B_1)=\Pr(B_2)=1/4$ and $\Pr(B_3)=1/2$.
By using Equation~(\ref{eq:distance})
to measure the distance between $\mathcal{D}$ and $\mathcal{D}(P_i)$,
for all $i$,
we conclude that the generated partition satisfies $1.5$-closeness.
In Table~\ref{fig:t_close_partition} we have depicted both the
original set of values of the confidential attribute, and the generated
buckets.

\begin{figure}
\begin{centering}
\begin{tabular}{cc}
\includegraphics[width=7cm]{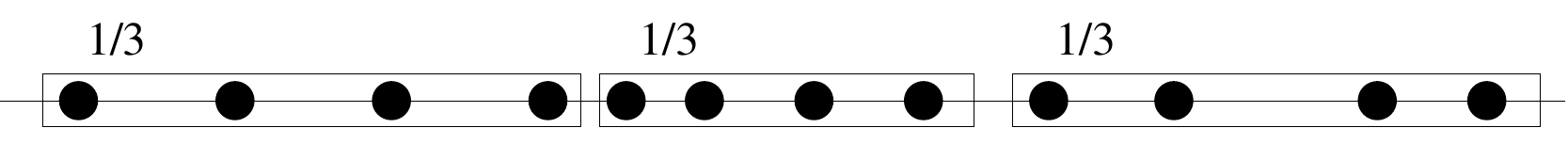} & original data\tabularnewline
\includegraphics[width=7cm]{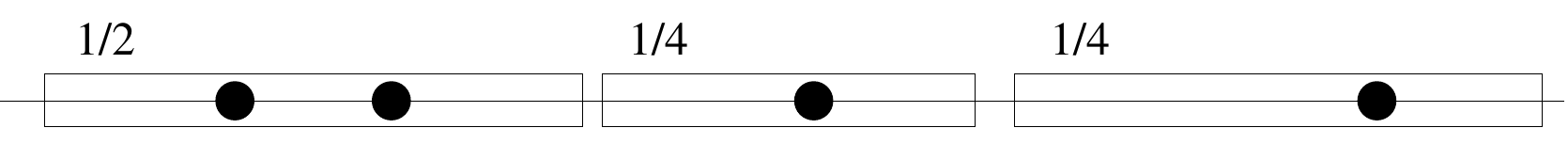} & group $P_{1}$\tabularnewline
\includegraphics[width=7cm]{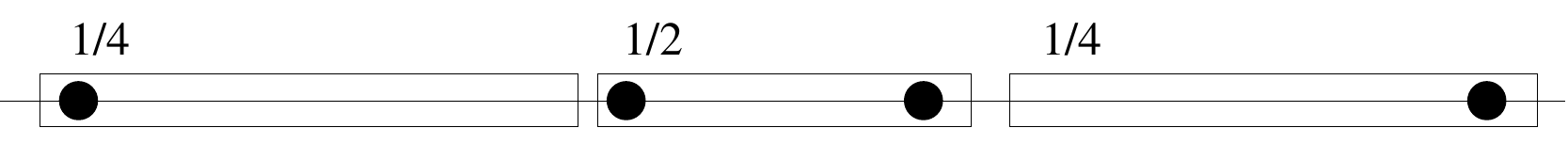} & group $P_{2}$\tabularnewline
\includegraphics[width=7cm]{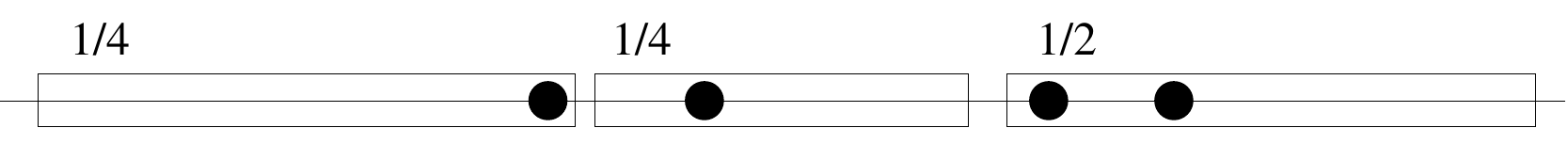} & group $P_{3}$\tabularnewline
\end{tabular}
\par\end{centering}
\caption{\label{fig:t_close_partition}Sample partition that satisfies $1.5$-closeness}
\end{figure}

Let the points in the original data set depicted 
in Figure~\ref{fig:t_close_partition}
be of the form $(qi_{i},c_{i})$, where $c_i$
the value of the confidential attribute
and $c_{i}\ge c_{j}$ for $i\ge j$.
According to the Anatomy approach to $k$-anonymity, the $4$-anonymous
$1.5$-close resultant data, associated to partition $\{P_{1},P_{2},P_{3}\}$
in Figure~\ref{fig:t_close_partition} and bucketization $\{B_{1},B_{2},B_{3}\}$
in Figure~\ref{fig:discretization}, consists of 
the two linked tables displayed 
in Table~\ref{fig:t_close_output}.

\begin{table*}
\caption{\label{fig:t_close_output}$4$-Anonymous $1.5$-close data set associated
to the partition $\{P_{1},P_{2},P_{3}\}$ in Figure~\ref{fig:t_close_partition}
and bucketization in Figure~\ref{fig:discretization}}
\begin{centering}
\begin{tabular}{ccccccccccccc}
\hline 
QI & $qi_{1}$ & $qi_{2}$ & $qi_{3}$ & $qi_{4}$ & $qi_{5}$ & $qi_{6}$ & $qi_{7}$ & $qi_{8}$ & $qi_{9}$ & $qi_{10}$ & $qi_{11}$ & $qi_{12}$\tabularnewline
\hline 
Group Id & $P_{2}$ & $P_{1}$ & $P_{1}$ & $P_{3}$ & $P_{2}$ & $P_{3}$ & $P_{1}$ & $P_{2}$ & $P_{3}$ & $P_{3}$ & $P_{1}$ & $P_{2}$\tabularnewline
\hline 
 &  &  &  &  &  &  &  &  &  &  &  & \tabularnewline
\hline 
Group Id & $P_{1}$ & $P_{1}$ & $P_{1}$ & $P_{1}$ & $P_{2}$ & $P_{2}$ & $P_{2}$ & $P_{2}$ & $P_{3}$ & $P_{3}$ & $P_{3}$ & $P_{3}$\tabularnewline
\hline 
Bucket  & $B_{1}$ & $B_{1}$ & $B_{2}$ & $B_{3}$ & $B_{1}$ & $B_{2}$ & $B_{2}$ & $B_{3}$ & $B_{1}$ & $B_{2}$ & $B_{3}$ & $B_{3}$\tabularnewline
\hline 
\end{tabular}
\par\end{centering}
\end{table*}

\subsection{Bucketization of the original data\label{sub:Bucketizing}}

The selected bucketization of the confidential attribute has a large
impact on data utility: if the bucketization is too coarse, 
the information
loss in the confidential attribute is large; if the bucketization
is too fine, it may not be possible to attain $t$-closeness. In this
section we seek to determine the optimal size (in terms of probability
mass) of the buckets.

Figure~\ref{fig:2_close_distribution} illustrates two probability
distributions:
the uniform distribution represents the global distribution of the 
confidential attribute (over the whole data set),
and the other distribution corresponds to the
confidential attribute restricted to a group 
$P_i$.
These two distributions satisfy $2$-closeness with the distance
of Definition~\ref{def:distance}: 
the density of the restriction to $P_{i}$ equals $1/2$ 
for all the
range of values of the confidential attribute, except for a range
of values that has density 2.

\begin{figure}
\begin{centering}
\includegraphics[width=6cm]{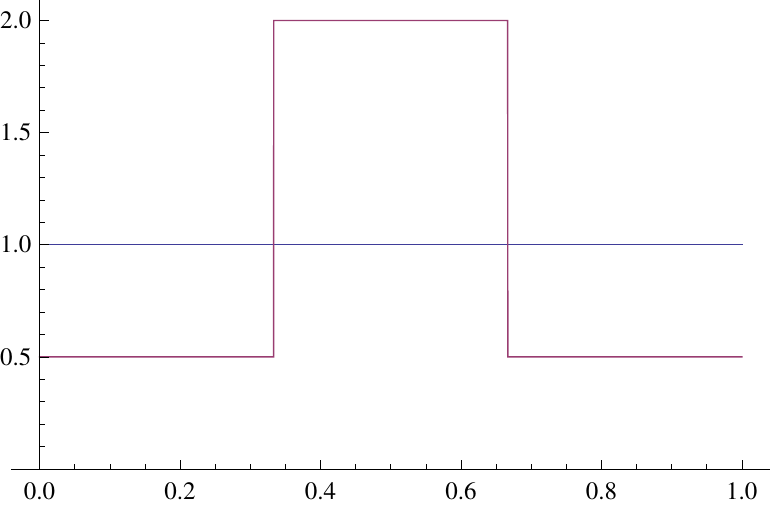}
\par\end{centering}

\caption{\label{fig:2_close_distribution}Probability distributions satisfying
$2$-closeness with the distance of Definition~\ref{def:distance}}
\end{figure}

When bucketizing the distributions in Figure~\ref{fig:2_close_distribution},
the range of values with density 2 should exactly
correspond to a bucket or a union of buckets, in order
to maximize the utility of the data. This is illustrated
in Figure~\ref{fig:compare_buckets}, whose top row shows
bucketized versions of the distributions 
of Figure~\ref{fig:2_close_distribution} using {\em three} buckets:
top left graph, bucketized version of the global distribution;
top right graph, bucketized version of the restriction to $P_i$.
Note that, for each of the buckets, the global probability
and the probability restricted to $P_i$ differ by a multiplicative
factor of two; that is, we attain $2$-closeness with equality for each
of the buckets. The bottom row of 
Figure~\ref{fig:compare_buckets} 
shows the bucketized versions of the
distributions in Figure~\ref{fig:2_close_distribution}
using {\em two} buckets. It can be seen that, with the two proposed buckets,
both bucketized distributions are identical; 
that is, we get $1$-closeness, which is
stronger than the intended $2$-closeness, but comes at the cost
of data utility loss.
Therefore, the number and hence the 
probability mass of the optimal buckets is dependent
on the level of $t$-closeness that we want. 

\begin{figure*}
\begin{center}
\begin{tabular}{ccc}
\includegraphics[width=6cm]{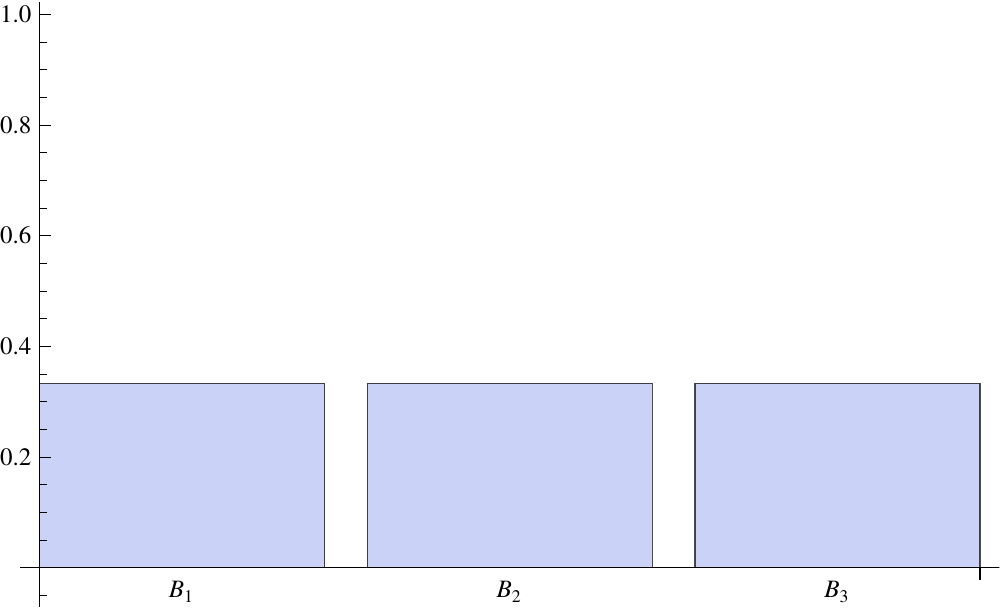}
& \rule{1cm}{0cm} &
\includegraphics[width=6cm]{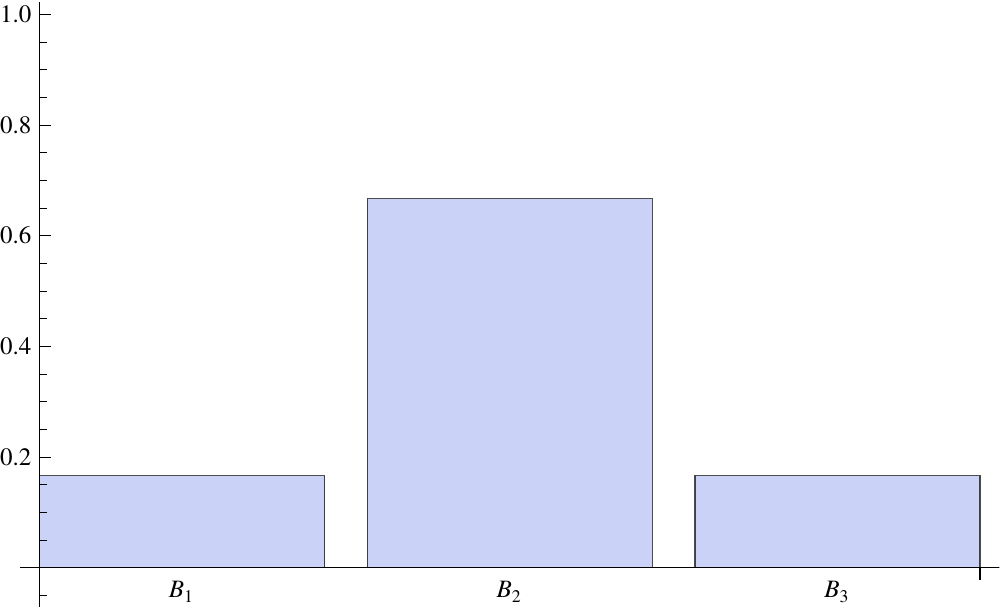} \\
\includegraphics[width=6cm]{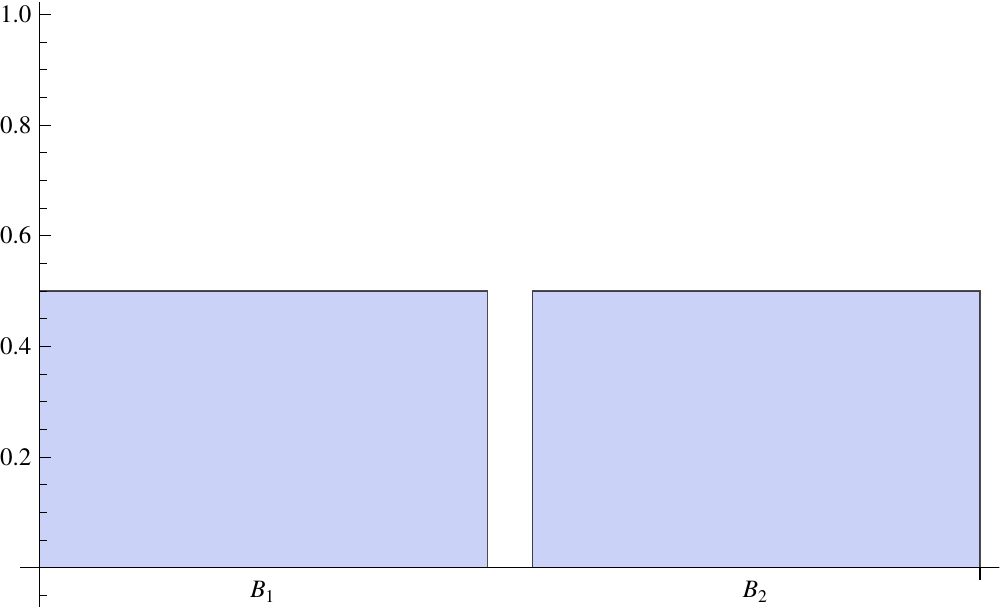} 
& \rule{1cm}{0cm} & \includegraphics[width=6cm]{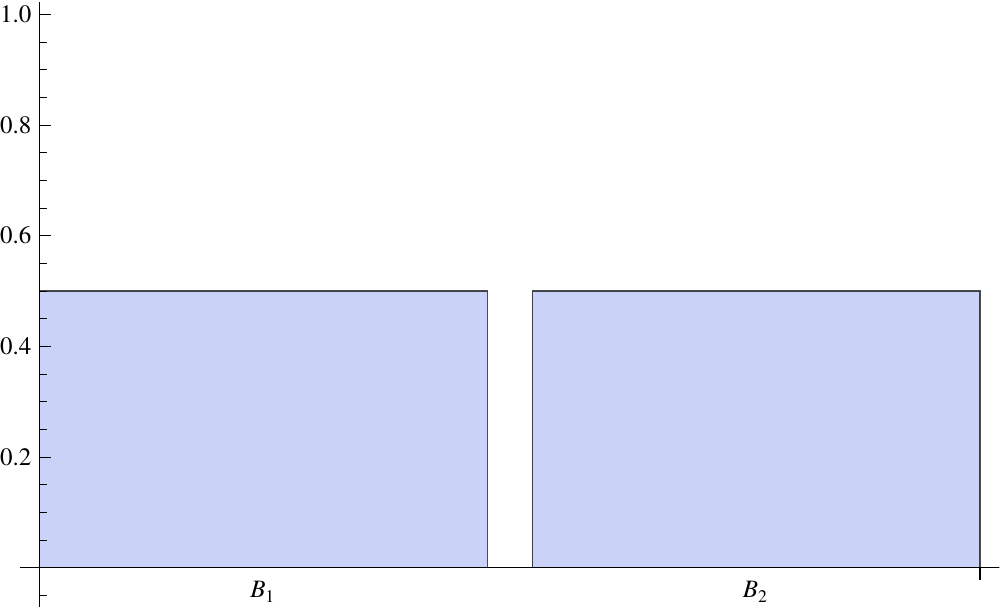}
\end{tabular}
\end{center}
\caption{\label{fig:compare_buckets}Bucketized distributions of the confidential
attribute for the whole data (left) and for a group $P_i$ (right).
Three buckets are considered in the top distributions, 
and two in the bottom ones.}
\end{figure*}

Let us now restate the bucketization process in an algorithmic way:
\begin{enumerate}
\item Let the number of records in the original data set be $N$.
\item Cluster the $N$ values of the confidential attribute in 
the original data set
into a number $b$ of buckets in such a way that:
\begin{enumerate}
\item all buckets
accumulate the same probability mass $1/b$, that is,
each bucket contains $[N/b]$ values;
\item values within a bucket are as similar as possible
({\em e.g.} for a numerical confidential value, each bucket
would contain $[N/b]$ consecutive values).
\end{enumerate}
In this way, we can 
view the bucketized distribution of the confidential
attribute in the original data set as being
uniform.
\item Partition the records in the original data set 
into a number of groups,
in such a way that every group satisfies that:
\begin{enumerate}
\item it contains $k$ (or more) records,
in view of achieving $k$-anonymity;
\item  no bucket contains a proportion
of the confidential attribute values of the group
higher than $t/b$ or lower than $1/(tb)$ (that is,
so that the bucketized distribution of the confidential attribute
in the group is at distance less than $t$ from
the bucketized distribution of the confidential 
attribute in the overall data set, according to 
Definition~\ref{def:distance}).
\end{enumerate}
\end{enumerate}

In general, the smaller the number $b$ of buckets, the easier
it is to achieve $t$-closeness, for any given $t$. 
In the extreme case $b=1$, all bucketized distributions
are 1-close ({\em e.g.} there is a single bucketized distribution).
In the other extreme case $b=N$ (no bucketization) 
it has been argued above (right after Proposition~\ref{prop1}) that
the distance between the distributions of the confidential attribute
on the global data set and on a particular group is infinity; hence,
one can only achieve $\infty$-closeness.
Hence, at most $b$ can be $k$, the number of values in each group,
and buckets should be large enough so that, when restricted to any
group, any bucket contains at least one value.

On the other hand, if the privacy requirement is $t$-closeness,
for a certain $t$, it seems reasonable to use up the allowed
distance $t$ between the global distribution of the confidential attribute
and the restriction of that distribution within each group.
Using up the allowed distance between the 
confidential attribute distributions 
enables forming groups that are more homogeneous
in terms of the quasi-identifiers, and hence decreases 
information loss.
We want each of the $k$-anonymous groups 
to emphasize a specific bucket; that is, the probability distribution
of the restriction to the partition must differ from the global distribution
by a factor of $t$ for a specific bucket, and by a factor of $1/t$
for the rest of buckets. 
Now, in the distribution of the confidential attribute for the 
original data set each bucket accumulates
probability mass $1/b$, and the total probability mass 
of the distribution restricted to a group must add to 1. Hence, we have
\[ t \times 1/b + (1/t) \times (1-1/b) = 1 \]
which yields a number of buckets $b=t+1$.

\subsection{$t$-Closeness construction}

Consider the original data set $D=\{(qi_{i},c_{i})|i=1,\ldots,N\}$,
where $qi_{i}$ refers to the quasi-identifier attributes, and $c_{i}$
to the confidential attribute. We want to generate a $k$-anonymous
$t$-close data set $D'$. 

According to Section~\ref{sub:Bucketizing}, we need to reduce the granularity
of the confidential attribute. In particular, it was proposed to group
the values of the confidential attribute in buckets of 
$[N/b]=[N/(t+1)]$
records. Assuming that the records can be ordered in terms of the confidential
attribute $c_{i}$ 
(this is possible if $c_i$ is numerical 
or ordinal) we can list the contents of the buckets as follows: 
\[
\begin{array}{c}
B_{1}=\{c_{1},\ldots,c_{\left[\frac{N}{t+1}+0.5\right]}\}\\
B_{2}=\{c_{\left[\frac{N}{t+1}+0.5\right]+1},\ldots,c_{\left[2\times\frac{N}{t+1}+0.5\right]}\}\\
\vdots\\
B_{t+1}=\{c_{\left[t\times\frac{N}{t+1}+0.5\right]+1},\ldots,c_{N}\}
\end{array}
\]
The $k$-anonymous $t$-close data set is generated as follows: 
\begin{enumerate}
\item Replace the values of the confidential attribute
in the original data set $D$ by the corresponding buckets,
and call $\bar{D}$ the resulting data set;
\item Partition $\bar{D}$ in groups of $k$ (or more) records.
\end{enumerate}

In the second step above, not all values of $k$ are
equally suitable. For instance, it must be $k \geq t+1$,
because we showed in Section~\ref{sub:Bucketizing}
that $b \leq k$ and $b=t+1$.
In fact, we can write:
\[
k=\frac{N}{(t+1)l}
\]
where $l\ge1$ is a natural number that counts the number of groups 
that emphasize each of the buckets.
In fact, if we take into account
the previous inequality $k\geq t+1$, we conclude that $l$ belongs
to the set $\{1,\ldots,\left\lfloor \frac{N}{(t+1)^{2}}\right\rfloor \}$.
Similarly to the discretization of the confidential attribute, the
value of $k$ produced by the previous formula may not be exact. In
that case we need to adjust the size $k_{i}$ of each group $P_{i}$ to
\[ k_{i}=\left[i\frac{N}{(t+1)l}\right]-\left[(i-1)\frac{N}{(t+1)l}\right] \]

\begin{table}
\caption{\label{tab:probabilities}Theoretical probability mass of the distribution
of the confidential attribute in each of the buckets corresponding
to the discretization of the confidential attribute. }

\begin{centering}
\begin{tabular}{ccccc}
 & $B_{1}$ & $B_{2}$ & $\ldots$ & $B_{t+1}$\tabularnewline
\hline 
Original data & $\nicefrac{1}{t+1}$ & $\nicefrac{1}{t+1}$ & $\ldots$ & $\nicefrac{1}{t+1}$\tabularnewline
\hline 
$P_{1}$ & $\nicefrac{t}{t+1}$ & $\nicefrac{1}{t(t+1)}$ & $\ldots$ & $\nicefrac{1}{t(t+1)}$\tabularnewline
$P_{2}$ & $\nicefrac{1}{t(t+1)}$ & $\nicefrac{t}{t+1}$ & $\ldots$ & $\nicefrac{1}{t(t+1)}$\tabularnewline
$\vdots$ & $\vdots$ & $\vdots$ &  & $\vdots$\tabularnewline
$P_{t+1}$ & $\nicefrac{1}{t(t+1)}$ & $\nicefrac{1}{t(t+1)}$ & $\ldots$ & $\nicefrac{t}{t+1}$\tabularnewline
\hline 
\end{tabular}
\par\end{centering}

\end{table}

Table~\ref{tab:probabilities} gives the theoretical probability
mass of 
each bucket of the confidential attribute for each of the groups.
We assume that $l=1$ and that group $P_{1}$
emphasizes bucket $B_{1}$, $P_{2}$ emphasizes bucket $B_{2}$, and
so on.
The exact theoretical probability masses
may not be achievable due to the discrete nature of the data. First
of all, it may not be possible to obtain a discretization of the confidential
attribute in buckets with probability mass $1/(t+1)$.
Also, when generating the $k$-anonymous partition $P_{1},\ldots,P_{t+1}$,
it may not be possible for each of the groups to contain
exactly $k$ records. Let $k_{i}$ be the number of records in $P_{i}$
and let $p_j$ be the probability that a record in the original data set
belongs to bucket $B_j$.
For $t$-closeness to be achieved, the following must 
hold for every group $P_i$: (i)
at most $\left\lfloor k_{i}p_{i}t\right\rfloor $ records must have
$B_{i}$ as the value for the confidential attribute; and (ii) at
least $\left\lceil k_{i}p_{j}/t\right\rceil $ records must have $B_{j}$
as confidential attribute. For these conditions to hold, we can start
selecting $\left\lceil k_{i}p_{j}/t\right\rceil $ records with confidential
attribute $B_{j}$, for each $j\ne i$, and complete the partition
set with $k_{i}-t\left\lceil k_{i}p_{j}/t\right\rceil $ records with
confidential attribute $B_{i}$.

\section{From $t$-closeness to $\varepsilon$-differential privacy}
\label{from}

$t$-Closeness and $\varepsilon$-differential privacy take approaches
towards disclosure limitation that are essentially different. However,
for microdata releases a link between them can be found if we make
some assumptions on the prior knowledge of intruders: 
\begin{enumerate}
\item The marginal
distribution of the confidential attribute is known to the intruder;
actually, this assumption is a requirement, because a $t$-close data 
release preserves this marginal distribution.
\item The intruder knows whether an individual's record is in the data
set; this is also a requirement, as either of the approaches proposed to
generate the $t$-close data set, Anatomy and probabilistic $k$-anonymity,
preserves the quasi-identifiers.
\item When regular $k$-anonymity is used,
another assumption on the intruder's knowledge is required: the intruder's
knowledge about the confidential attribute is limited to its marginal
distribution.
\end{enumerate}

We aim at showing that, in the case of 
a microdata release, $\exp(\varepsilon)$-closeness implies 
$\varepsilon$-differential privacy.
In other words, we
want to show that the information that
an intruder obtains from accessing the released
$\exp(\varepsilon)$-close microdata set (generated as per Section~\ref{sec:t_closeness})
satisfies the $\varepsilon$-differential privacy condition. 

Let $I$ be a specific individual in the data set. Before accessing the data
set, the intruder views the value of the confidential attribute of individual
$I$ as being distributed according to the distribution of the confidential
attribute over the whole data set. Given the assumption that limits
the prior knowledge to the marginal distribution of the confidential
attribute, that is the most precise information that the intruder
has about $I$. $\varepsilon$-Differential privacy guarantees that
the knowledge gain obtained from the response to a query that asks
for $I$'s confidential attribute is at most $\exp(\varepsilon)$;
that is, the distribution of the response must differ at most by a
factor of $\exp(\varepsilon)$ from the assumed prior knowledge. Note
that if we did not take into account the intruder's prior knowledge 
(usual $\varepsilon$-differentially private mechanisms
do not assume any prior knowledge), the $\exp(\varepsilon)$-differentially
private distribution for the confidential attribute of individual $I$'s
would be different. However, the possibility of using the available
prior knowledge exists, and thus any distribution that differs from
it by a factor of $\exp(\varepsilon)$ satisfies $\exp(\varepsilon)$-differential
privacy.

$t$-Closeness is an improvement of $k$-anonymity. As such, it
seeks to thwart record re-identification by making each record indistinguishable
from $k-1$ other records as far as the quasi-identifiers are concerned.
Apart from that, $t$-closeness requires that the sampling distribution
of the confidential attribute within each of the $k$-anonymous groups
be similar to the sampling distribution over the whole data
set. Hence, $t$-closeness effectively limits the knowledge gain that 
the intruder obtains, that is, it achieves differential privacy.

\subsection{Uninformed intruders}

Consider an uninformed intruder. By inspecting the released $\exp(\varepsilon)$-close
$k$-anonymous data set, the intruder associates a $k$-anonymous
group of records to individual $I$. In this way, 
the intruder learns the distribution
of the confidential attribute within the $k$-anonymous group 
$P$ that contains $I$. As the intruder's prior knowledge is limited
to the marginal distribution of the confidential attribute, after
accessing the data, the best the intruder can do is to associate the
distribution of the confidential attribute in $P$ to individual $I$.
As the released data set satisfies $\exp(\varepsilon)$-closeness
(generated as per Section~\ref{sec:t_closeness}), the distribution
of $P$ differs at most in a factor $\exp(\varepsilon)$ from the
distribution of the whole data set; that is, it satisfies $\varepsilon$-differential
privacy.

\subsection{Informed intruders}

For an informed intruder (whose knowledge goes beyond the distribution
of the confidential attribute over the whole data set),
in general $t$-closeness does not imply differential privacy.
To see this, consider an intruder who knows the value of the
confidential attribute for $k-1$ of the $k$ individuals in one of
the $k$-anonymous groups. Such an intruder can determine
(with certainty) the confidential attribute value for the remaining individual
in the group by simple inspection of the released data; in
differential privacy terms, access to the data set has produced
a infinite knowledge gain on the confidential attribute of that specific
individual. 

The above situation is unavoidable if, as $k$-anonymity does, we intend to
preserve the thruthfulness of the confidential attribute inside the
$k$-anonymous groups. However, as we showed in Section~\ref{sec:Optimal-generic--anonymous},
by increasing $k$, the problem is mitigated. The greatest mitigation
is attained when $k$ equals the number of records in the data set.
In such case, for the intruder to determine the confidential attribute
value of any individual with certainty, he should know the confidential
attribute for all the other individuals in the data set (strictly speaking,
it would be enough to know that none of the other individuals take
one of the values in the released data set). The problem with such
a large $k$ is that it is likely to severely damage utility. 

The destruction of data utility can be mitigated if we hide the $k$-anonymous
groups. In this way, a smaller $k$ can be used, thereby preserving the
utility of the data, and a protection equivalent to taking $k$ equal
to $N$ is attained. Hiding the $k$-anonymous groups is
feasible if instead of regular $k$-anonymity, we enforce probabilistic
$k$-anonymity~\cite{fuzzieee}.
Probabilistic $k$-anonymity can be seen as an instance
of the Anatomy method for $k$-anonymity: instead of associating a
sampling distribution to each of the groups, the values of
the confidential attribute are permuted and assigned to individual
records. This process can be viewed as taking a sample of the sampling 
distribution
for each record. When using probabilistic $k$-anonymity, 
the intruder cannot determine which records form each of the groups; 
thus, she cannot use the information about a specific
individual to increase her knowledge on the other individuals of the
group.

\section{Conclusions}
\label{conclusion}

We have shown that the $k$-anonymity family of models 
is powerful enough to achieve $\varepsilon$-differential 
privacy in the context of data publishing. Specifically,
using a suitable construction,
we have shown that $\exp(\varepsilon)$-closeness
implies $\varepsilon$-differential privacy for uninformed
intruders and approximate $\varepsilon$-differential privacy
for informed intruders. Our $t$-closeness construction
based on bucketization is also a contribution in its own
right.


\lhead[\chaptername~\thechapter]{\rightmark}

\rhead[\leftmark]{}

\lfoot[\thepage]{}

\cfoot{}

\rfoot[]{\thepage}

\chapter{Conclusions}

\section{Contributions}

This thesis has dealt with disclosure limitation in data releases.
Among the available privacy criteria, we have focused on $k$-anonymity
and $\varepsilon$-differential privacy. The focus has primarily been
placed on improving data utility, but we have also dealt with the
inherent limitations of $k$-anonymity, and with the combination of
$k$-anonymity
(or $t$-closeness)
and $\varepsilon$-differential privacy. More specifically,
our contributions are:
\begin{itemize}
\item We have reviewed $k$-anonymity and some of its limitations. In particular,
we have shown that $k$-anonymity has a suboptimal behavior in presence
of informed intruders, due to the ``curse of dimensionality''. To
improve data utility we have proposed a new 
privacy model, which relaxes the
requirements of $k$-anonymity by imposing only a probability of
re-identification equal to $1/k$. We have shown 
that our proposal offers equivalent
disclosure limitation guarantees to those of $k$-anonymity, and allows
for improved data utility. The improvement on data utility 
is chiefly due to the ability to use multiple partitions of the data
set. It allows us to offer improved privacy guarantees against informed
intruders and still keep the data useful.
\item The Laplace distribution is the most commonly used data-independent
noise distribution to attain $\varepsilon$-differential privacy. We have
shown that the Laplace distribution is not optimal: another distribution
exists which satisfies the $\varepsilon$-differential privacy condition
and has its probability mass more concentrated around zero. For the
univariate case, we have determined the form of and constructed 
all optimal data-independent
distributions. For the multivariate case, we have shown that a specific
family of distributions is optimal. Regarding data utility, we 
have shown
that for the univariate case the improvement
of the optimal distribution is small (and thus the Laplace 
distribution is near-optimal),
but for the multivariate case the improvement can be significant.
\item $\varepsilon$-Differential privacy guarantees that the knowledge
gain that can be extracted from a query response is limited. As current
methods to attain $\varepsilon$-differential privacy do not let users
specify their prior knowledge, zero knowledge is implicitly 
taken as the base
to compute the knowledge gain. We propose a mechanism to let users
specify their prior knowledge on the response: each time a user sends
a query, the user's prior knowledge is also sent. We show that this
mechanism improves data utility and, despite the increased interaction
between the database and the user, we show that it preserves privacy.
\item A synergy
between $k$-anonymity and $\varepsilon$-differential privacy has
been described for privacy-preserving data publication, 
even if both models have quite different origins.
In particular
we have shown that a specific kind of microaggregation (that results in
a $k$-anonymous data set) can be employed to reduce the sensitivity
of identity queries by a factor of $1/k$. As a result, the $\varepsilon$-differentially
private data set generated from the $k$-anonymous version offers
improved data utility.

\item We have shown that the $k$-anonymity family of models is 
powerful enough to achieve $\varepsilon$-differential privacy 
in the context of data publishing. Specifically, using a suitable
construction, we have shown that $\exp(\varepsilon)$-closeness implies $\varepsilon$-differential privacy for uninformed intruders and approximate
$\varepsilon$-differential privacy for informed intruders.
Our $t$-closeness construction based on bucketization is also a contribution in its own right. In particular, as $\varepsilon$-differential privacy
is attained through a method that provides $t$-closeness, the truthfulness
of the data inside $k$-anonymous groups is preserved. This is a remarkable
advantatge over typical methods used to attain $\varepsilon$-differential
privacy.
\end{itemize}

\section{Publications}

The publications supporting this thesis are:
\begin{itemize}
\item Jordi Soria-Comas and Josep Domingo-Ferrer. Probabilistic $k$-anonymity
through microaggregation and data swapping. In: \emph{IEEE International
Conference on Fuzzy Systems} \emph{- FUZZ-IEEE 2012}, pp. 1-8, 2012.
\item Jordi Soria-Comas and Josep Domingo-Ferrer. Optimal data-independent
noise for differential privacy.
{\em Information Sciences} (To appear).
\item Jordi Soria-Comas and Josep Domingo-Ferrer. Differential privacy through
knowledge refinement. In: \emph{4th IEEE International Conference
on Privacy, Security, Risk and Trust - PASSAT 2012}, pp.
702-707, 2012.
\item Jordi Soria-Comas and Josep Domingo-Ferrer. Sensitivity-independent 
differential
privacy via prior knowledge refinement. \emph{International Journal
of Uncertainty, Fuzziness and Knowledge-based Systems} 20(6): 855-876,
2012.
\item Jordi Soria-Comas, Josep Domingo-Ferrer and David Rebollo-Monedero.
$k$-Anonimato probabilístico. In: \emph{XII Reunión Española sobre Criptología
y Seguridad de la Información - RECSI 2012}.
\item Jordi Soria-Comas and Josep Domingo-Ferrer. On differential privacy
and data utility in SDC. In 
\emph{7th Joint UN/ECE-Eurostat Work
Session on Statistical Data Confidentiality}, 2011.
\url{http://www.unece.org/fileadmin/DAM/stats/documents/ece/ces/ge.46/2011/24_Soria-Domingo.pdf}
\item Jordi Soria-Comas, Josep Domingo-Ferrer, David S\'anchez
and Sergio Mart\'{\i}nez.
Improving the utility of differentially private data releases
via $k$-anonymity. In {\em 12th IEEE International Conference
on Trust, Security and Privacy in Computing and Communications
-IEEE TrustCom 2013}, Melbourne, Australia, July 16-18, 2013 (to appear). 
\item Jordi Soria-Comas and Josep Domingo-Ferrer. Differential privacy via $t$-closeness in data publishing. {\em 11th International Conference on Privacy, Security and Trust-PST 2013}, Tarragona, July 10-12, 2013 (to appear, IEEE Digital Library). 
\end{itemize}

\section{Future work}

The work presented in this thesis opens several avenues
for new research:
\begin{itemize}
\item In Chapter~\ref{chap:Probabilistic--anonymity} we reviewed some
of the limitations of $k$-anonymity in presence of informed intruders,
and proposed a new privacy model, probabilistic $k$-anonymity,
that offers privacy guarantees equivalent to those of $k$-anonymity.
We showed that probabilistic $k$-anonymity may offer disclosure limitation
against informed intruders and still provide useful results. The proposed
method to attain probabilistic $k$-anonymity works by generating
a different partition (the optimal one) for each confidential attributs.
As a result, the risk of attribute disclosure is increased. To deal
with this issue we proposed to increase $k$, but the enforcement
of additional criteria (\emph{e.g.} $l$-diversity, $t$-closeness)
may be a better solution.
\item The amount of noise added to attain $\varepsilon$-differential privacy
is usually large, which damages the utility of the output. One strategy
to reduce query sensitivity is based on applying some transformation
to the query. Following this strategy we have shown that for queries
returning information about specific individuals, a prior microaggregation
step can reduce sensitivity by a factor of $1/k$. It could be
interesting to determine whether microaggregation can help reducing the
sensitivity of a generic queries.
\item A common approach to the generation of $\varepsilon$-differentially
private data sets is to divide the range of possible values in fixed
buckets and then count the number of individuals within each bucket.
This approach is not suitable for dealing with sparse data: the number
of buckets with small counts is large, and therefore the added noise
may substantially change the properties of the data set. A possible
approach to make sure that the generated buckets have a similar number
of records is to use a microaggregation algorithm. Regular microaggregation
algorithms do not fit in the $\varepsilon$-differential privacy environment,
as a change in a single point may change the cluster completely; however,
the insensitive microaggregation proposed in Chapter~\ref{chap:Enhancing-Data-Utility}
guarantees a maximum change of one record per cluster. By using this approach,
the accuracy of the released data may be increased as, in practice,
we avoid considering the sparse regions.
\item In Chapter~\ref{chap:t-closeness} we presented 
a link between $t$-closeness
and $\varepsilon$-differential privacy for numeric or ordinal attributes.
Future research will include extending
the proposed approach for nominal confidential
attributes, which cannot be ordered.
We will also provide a generalization to multiple confidential
attributes.
Experimental work will be conducted to compare
the utility of the $\varepsilon$-differentially private
data sets obtained via bucketized $\exp(\varepsilon)$-closeness.
The very nature of our construction, based on $k$-anonymity,
gives reasonable hopes that more utility may be preserved than
the one offered by the Laplace noise addition typically used
to achieve $\varepsilon$-differential privacy: for example, by design,
our approach does not yield any off-range values, which may however appear
in noise addition procedures;
in fact, the anonymized values we provide are truthful, even if 
coarsened by bucketization.
We will also explore the exact privacy guarantees offered
by the approximate $\varepsilon$-differential privacy
obtained with our construction for the case of informed intruders.
\end{itemize}

\cleardoublepage{}

\lhead[]{\rightmark}

\rhead[\leftmark]{}

\bibliographystyle{plain}
\bibliography{bibliography}

\end{document}